\documentclass{article}

\usepackage{listings}
\usepackage{cancel}
\usepackage{bbm}
\usepackage{graphicx}
\usepackage{subcaption}
\usepackage{mathtools}
\usepackage{amsmath}
\usepackage{amsfonts}
\usepackage{xr}
\usepackage{varioref}
\usepackage{xr-hyper}
\usepackage{hyperref}
\usepackage{multirow}
\usepackage{booktabs}
\usepackage{threeparttable}
\usepackage{footnote}
\usepackage{mathrsfs} 
\usepackage{bm}
\usepackage[shortlabels]{enumitem}
\usepackage{tikz}
\usetikzlibrary{arrows}
\usetikzlibrary{decorations.pathreplacing}

\usepackage{amsthm} 
\usepackage[toc,page]{appendix}
\usepackage[font=small,labelfont=bf, skip=0pt]{caption}

\usepackage{array}
\newcolumntype{H}{>{\setbox0=\hbox\bgroup}c<{\egroup}@{}}

\usepackage[a4paper,margin=1in,footskip=0.25in]{geometry}

\usepackage[utf8]{inputenc}
\usepackage[style=authoryear-comp,firstinits=true,citestyle=authoryear,natbib=true,backend=bibtex,doi=false,isbn=false,url=false,maxcitenames=2,eprint=true,maxbibnames=99]{biblatex}
\AtBeginBibliography{}

\DeclareNameAlias{sortname}{last-first}

\renewbibmacro{in:}{}
\renewbibmacro*{volume+number+eid}{%
	\printfield{volume}%
	\setunit*{\addnbspace}
	\printfield{number}%
	\setunit{\addcomma\space}%
	\printfield{eid}}
\DeclareFieldFormat[article]{number}{\mkbibparens{#1}}

\bibliography{biblio}

\DeclareFieldFormat[inbook]{title}{#1}

\makeatletter
\def\thmhead@plain#1#2#3{%
	\thm@notefont{}
	\thmname{#1}\thmnumber{\@ifnotempty{#1}{ }\@upn{#2}}%
	\thmnote{ {\the\thm@notefont#3}}}
\let\thmhead\thmhead@plain
\itshape 
\makeatother

\newtheorem*{definition*}{Definition}
\newtheorem{definition}{Definition}
\newtheorem*{assumption*}{Assumption}
\newtheorem*{condition*}{Condition}
\newtheorem*{lemma*}{Lemma}
\newtheorem{lemma}{Lemma}
\newtheorem*{proposition*}{Proposition}
\newtheorem{proposition}{Proposition}
\newtheorem*{conjecture*}{Conjecture}
\newtheorem*{theorem*}{Theorem}
\newtheorem{theorem}{Theorem}
\newtheorem*{corollary*}{Corollary}
\newtheorem{corollary}{Corollary}
\newtheorem*{result*}{Result}

\newcommand{\indep}{\perp\!\!\!\!\perp} 

\usepackage{titlesec}
\titleformat*{\subsubsection}{\large\bfseries}

\graphicspath{ {images/} }

\usepackage[many]{tcolorbox}
\newsavebox{\fmbox}

\usepackage[onehalfspacing]{setspace}

\usepackage{nameref, zref-xr,zref-hyperref,zref-user}
\zxrsetup{toltxlabel}
\zexternaldocument*{outcomenonrestrictive_onlineappendices}

\title{When does IV identification not restrict outcomes?}
\author{Leonard Goff\thanks{\protect\linespread{1}\protect\selectfont Department of Economics, University of Calgary. I thank Simon Lee as well as Pat Kline, Eric Mbakop and Adam Rosen for helpful conversations about these ideas. Email: \texttt{leonard.goff@ucalgary.ca}.}}
\date{}

\begin{document}
	
	\maketitle
	
	\begin{abstract}
		Many identification results in instrumental variables (IV) models hold without requiring any restrictions on the distribution of potential outcomes, or how those outcomes are correlated with selection behavior. This enables IV models to allow for arbitrary heterogeneity in treatment effects and the possibility of selection on gains in the outcome. I provide a necessary and sufficient condition for treatment effects to be point identified in a manner that does not restrict outcomes, when the instruments take a finite number of values. The condition generalizes the well-known LATE monotonicity assumption, and unifies a wide variety of other known IV identification results. The result also yields a brute-force approach to reveal all selection models that allow for point identification of treatment effects without restricting outcomes, and then enumerate all of the identified parameters within each such selection model. The search uncovers new selection models that yield identification, provides impossibility results for others, and offers opportunities to relax assumptions on selection used in existing literature. An application considers the identification of complementarities between two cross-randomized treatments, obtaining a necessary and sufficient condition on selection for local average complementarities among compliers to be identified in a manner that does not restrict outcomes. I use this result to revisit two empirical settings, one in which the data are incompatible with this restriction on selection, and another in which the data are compatible with the restriction.
	\end{abstract}
	
	\large 
	\section{Introduction}
	
	To leverage instrumental variables (IV) with heterogeneous treatment effects, researchers often make assumptions about selection behavior, such as the ``monotonicity'' assumption of the seminal local average treatment effects (LATE) model \citep{Imbens2018}. Given this monotonicity assumption, the average treatment effect among compliers is point identified, with no restrictions imposed on the distribution of potential outcomes beyond them being independent of the instrument.
	
	This paper shows that a similar result holds broadly across IV models. If restrictions on selection behavior are sufficient to establish a particular generalization of the monotonicity assumption, then corresponding local average treatment effects are identified. Strikingly, when the treatments and instruments have finite support, this identification result also has a converse: for any treatment effect that conditions on selection behavior to be point identified without further restrictions on the distribution of potential outcomes, the selection model must permit the generalization of monotonicity to hold.
	
	Together this yields a necessary and sufficient condition for when identification of a given treatment effect or counterfactual mean is possible without ad-hoc assumptions regarding outcomes such as treatment effect homogeneity. I say that a parameter is identified in an \textit{outcome-nonrestrictive} way when it is point identified without any restrictions on the distribution of potential outcomes, beyond statistical independence between those outcomes and the instruments. This paper shows that  quite generally, the price of outcome-nonrestrictive identification is the need to make assumptions about selection into treatment. This trade-off may be quite appealing, particularly in ``design-based'' studies in which the researcher has contextual knowledge about a factor that affects treatment uptake in a given setting \citep{carddesignbased}, but may be reluctant to make assumptions about the very causal effect of interest (e.g. that it is homogeneous across individuals).
	
	Outcome-nonrestrictive identification results also have the practical benefit of paving the way for analysis to be repeated across multiple outcome variables, with maintained assumptions about selection in a given (natural) experiment aiding identification for each outcome. An outcome-nonrestrictive identification result is not fully indifferent to \textit{which} variable one uses as an outcome: one must make the standard independence and exclusion restrictions for each one. But in settings where treatment is as-good-as-randomly assigned and there are limited opportunities for the assignment to affect anything except via treatment (e.g. in certain experimental settings), such assumptions may be quite natural without much further justification specific to each outcome.
	
	When specialized to the case of a binary treatment, the main result of this paper can be stated succinctly as follows. Let $D_i(z)$ denote the counterfactual treatment of individual $i$ when the available instruments take value $z$, for each $z$ in a set $\mathcal{Z}$.
	\begin{result*}[(in the case of binary treatments)]
		Suppose $|\mathcal{Z}|$ is finite, and consider any subgroup of the population defined by their counterfactual selection behavior $D_i(\cdot)$. Then the average treatment effect among this subgroup is identified in an outcome-nonrestrictive way if and only if there exists a function $\alpha: \mathcal{Z} \rightarrow \mathbbm{R}$ such that
		\begin{equation} \label{eq:mainthmsummary}
			\left\{\sum_{z} \alpha(z) \cdot D_i(z)\right\} \in \{0,1\} \quad \textrm{ for all } i,
		\end{equation}
		where the subgroup is composed of the individuals $i$ for whom the term in brackets is equal to one. The ``if'' direction of the above further assumes that $\sum_{z} \alpha(z)=0$, and this restriction is needed for the ``only if'' direction if there are always-takers or never-takers.
	\end{result*}
	\noindent The above result is formalized in Theorem \ref{thm:suff} (which establishes sufficiency) and Theorem \ref{thm:necc} (which establishes necessity) of this paper, where each result is provided more generally for any finite set of treatments that are not necessarily binary, and is established for means of each potential outcome alone in addition to holding for treatment effects.\footnote{To see how Eq. \eqref{eq:mainthmsummary} nests the LATE identification result of \citet{Imbens2018} in the case of a single binary instrument, take the function $\alpha(z) = (-1)^{z+1}$, so that $\alpha(0)=-1$ and $\alpha(1)=1$. The quantity $\sum_{z} \alpha(z) \cdot D_i(z) = D_i(1)-D_i(0)$ is equal to either $0$ or $1$ for all individuals $i$, and is equal to unity only for the compliers.}
	
	When generalized beyond the binary-treatment case, the above result delivers a simple geometric characterization of outcome-nonrestrictive identification in IV models. Units in the population can be partitioned by which ``response type'' $G_i \in \mathcal{G}$ they belong to, and the conditioning event that defines a target parameter by functions $c:  \mathcal{G}\rightarrow \{0,1\}$ that indicates which response types are considered by that parameter. In models with a finite selection model $\mathcal{G}$, a local counterfactual mean $\mathbbm{E}[Y_i(t)|c(G_i)=1]$ is point-identified in an outcome-nonrestrictive manner if and only if a vector representation of $c$ belongs to a particular linear subspace of $\mathbbm{R}^{|\mathcal{G}|}$, where the subspace depends on what restrictions are assumed about selection behavior through the choice of $\mathcal{G}$. Intuitively, ``linearity'' arises throughout from the law of iterated expectations over the latent types $g$. While a counterfactual mean $\mathbbm{E}[Y_i(t)|c(G_i)=1]$ can be expressed as a convex linear combination of $\mathbbm{E}[Y_i(t)|G_i=g]$ over the $g$ for which $c(g)=1$, the observable distribution of $Y$ conditioned treatment realization $t$ similarly amounts to a mixture over conditional distributions of $Y_i(t)$ that condition on the various response types.\footnote{A similar structure is exploited in previous work that has used linear programming approaches to identification under LATE monotonicity (e.g. \citealt{MTS}). These authors generally focus on partial identification under specific monotonicity assumptions (see also \citet{Mogstad2020b,kamat2023identification} for related results). I characterize point identification for arbitrary selection models.}
	
	The perspective of outcome-nonrestrictive identification turns out to unify a wide variety of existing IV identification results in the literature. Theorem \ref{thm:suff} provides a simple and common proof of point identification for settings including: i) the original LATE model \citep{Imbens2018}; ii) the marginal treatment effect (\citealt{Heckman2001, Heckman2005}) and its generalization to multivalued treatments \citep{Lee2018a}; iii) unordered monotonicity \citep{Heckman2018}; iv) vector monotonicity with multiple instruments \citep{goff2024vector}; v) restrictions on choice and/or knowledge of second-best options \citep{kirkeboenleuvenmogstad}; vi) interaction effects between two treatments \citep{blackwell2017}; and vii) recent notions of monotonicity that are only required to hold between particular \textit{pairs} of instrument values \citep{sun2024pairwise,cclate, sigstad2024marginal}. I contrast the above results with other identification results from the IV literature that weaken assumptions about selection while leveraging additional assumptions about outcomes and are therefore not outcome-nonrestrictive \citep{imbensangristrubin,Kolesar2013,toleratingdefiance,comey2023supercompliers}.
	
	In the other direction, the necessary part of my result (Theorem \ref{thm:necc}) motivates a comprehensive search over all outcome-nonrestrictive identification results, which is feasible in settings in which the instruments and treatments have small support.\footnote{Further, I provide code that enables quick enumeration of identified parameters given a user-provided selection model.} For a given finite set of support points of the instrument and treatment variables, the set of possible selection models is finite and can in principle be enumerated by brute force. Within each such selection model, I show that the set of possible functions $\alpha(\cdot)$ can also be enumerated. We can therefore search systematically over the opportunities for IV identification that place no modeling restrictions on outcomes and have \textit{not} yet been revealed in the literature. Section \ref{sec:bruteforce} proposes two algorithms that implement this insight, which I apply to uncover all outcome-nonrestrictive identification results for binary or ternary instruments and treatments. In many cases the selection models underlying these results can be empirically motivated, and in some cases they offer opportunities to relax assumptions made in existing work. Section \ref{sec:empirical} presents an extended application of this search to the identification of interaction effects in cross-randomized experiments. This setting illustrates how the computational approach can distill a unified and economically meaningful model of selection that is necessary (and not just sufficient) for identification.
	
	In establishing sufficient \textit{and} necessary conditions for identification in IV models, the perspective of this paper is related to recent results by \citet*{navjeevan2023identification} (NPS). NPS begin with an IV setup that is similar to that of this paper, but consider the identification of unconditional moments of latent heterogeneity in general rather than focusing on moments of potential outcomes that condition on selection. The conditions for identification of ``conditional'' and ``unconditional'' versions of such moments are not equivalent in general, but my Theorem \ref{thm:necc} establishes that they are when the former is outcome-nonrestrictive. In Appendix \ref{sec:nps} I show that while my Theorem \ref{thm:suff} can be obtained as a corollary to results found in NPS, the same is not true of Theorem \ref{thm:necc}. Without Theorem \ref{thm:necc}, we would lack a guarantee that the search for outcome-nonrestrictive identified local average treatment effects executed in this paper is exhaustive. 
	
	The structure of the paper is as follows. Section \ref{sec:oaid} begins by formalizing a definition of the notion of ``outcome-nonrestrictive'' identification in IV models. Section \ref{sec:posi} introduces the idea of \textit{binary combinations}, a name I give to instances in which an analog of \eqref{eq:mainthmsummary} holds for general unordered discrete treatments. I show there that whether the instruments are discrete or continuous, binary combinations are sufficient for outcome-nonrestrictive identification of local counterfactual means. Further, collections of binary combinations across different treatment values (but isolating the same response types) are sufficient to identify local average treatment effects. I call such collections of binary combinations \textit{binary collections}. Appendix \ref{sec:examples} details how the notion of binary collections nests a broad range of identification results for treatment effects from the literature.
	
	Section \ref{sec:discrete} then specializes to the case of discrete and finite instruments, and shows that in such settings binary combinations and binary collections are the \textit{only} cases in which local counterfactual means or local average treatment effects can be identified in an outcome-nonrestrictive way. Building on an algebraic characterization of binary combinations, Section \ref{sec:bruteforce} presents new results on the identification of various local average treatment effect parameters after enumerating all binary collections by brute-force, with some detailed examples examined in Appendix \ref{app:examples}. A full catalog of such identification results is presented in Appendix \ref{sec:catalog}, for settings with small instrument/treatment support. 
	
	Section \ref{sec:empirical} turns to a specific application of the results to the identification of interaction effects in order to assess the ``complementarity'' between two treatments in cross-randomized designs. I find that the average interaction effect between the treatments can be identified among a complier subgroup without restricting outcomes if and only if the selection model allows just five response types, composed of the individuals who decide ``separately'' whether to select into each of the two treatments. I use this result to revisit two empirical applications: one in which the data are incompatible with observable implications of this selection model, and another in which the data are. In this application, I proceed to estimating the extent of complementarity, yielding new evidence on the interaction between pharmacotherapy and livelihoods assistance in combatting depression.
	
	\section{Defining outcome-nonrestrictive IV identification} \label{sec:oaid}
\subsection{Setup and notation}
Let treatment $t$ take values in a finite set $\mathcal{T}$. Denote potential outcomes as $Y_i(t,z)$ and potential treatments as $T_i(z)$, where $Z_i$ are instruments with support $\mathcal{Z}$. I'll refer to $Z_i$ as ``the instruments'', since in general it can be a vector of instrumental variables. The index $i$ corresponds to observational units, i.e. ``individuals''.

\subsection{IV model assumptions} \label{sec:ivmodel}
Let $D^{[t]}_i(z) = \mathbbm{1}(T_i(z) = t)$ be an indicator for $i$ taking treatment $t$ when the instruments are equal to $z$. I throughout impose the exclusion restriction that $Z_i$ only affects $Y_i$ through $T_i = T_i(Z_i)$, so that $Y_i(t,z)=Y_i(t)$ for all $i$ and $z \in \mathcal{T}$. The observed outcome $Y_i$ is then:
\begin{equation} \label{eq:outcomeeq}
	Y_i= Y_i(T_i(Z_i))=\sum_{t \in \mathcal{T}}D^{[t]}_i(Z_i) \cdot Y_i(t)
\end{equation} 
Let $G_i: \mathcal{Z} \rightarrow \mathcal{T}$ be $i$'s ``response type'', i.e. the function that yields individual $i$'s counterfactual treatment value for each possible instrument value $z$. Let $\mathcal{G}$ be the set of all admissible response types. Let us denote the full set of conceivable functions from $\mathcal{Z}$ to $\mathcal{T}$ as $\mathcal{T}^{\mathcal{Z}}$. Any $\mathcal{G} \subset \mathcal{T}^{\mathcal{Z}}$ reflects a restriction on the response types, which I refer to as a \textit{selection model} or \textit{choice model}. 

I will also assume throughout that the instruments are exogenous in the sense that
\begin{equation} \label{eq:independence}
	Z_i \indep (\tilde{Y}_i,G_i)
\end{equation}
where $\tilde{Y} = \{Y_i(t)\}_{t \in \mathcal{T}}$ is a vector of potential outcomes across all treatments $t$. Eq. \eqref{eq:independence} says that potential outcomes and potential treatments are jointly independent of the instruments. In applications, researchers often defend a \textit{conditional} version of \eqref{eq:independence}, i.e. $\{Z_i \indep (\tilde{Y}_i,G_i)\} | X_i$, where $X_i$ are observed covariates unaffected by treatment. Since my focus in this paper is on identification and not estimation of treatment effects, I suppress throughout conditioning on any such covariates for ease of exposition, and consider them in Appendix \ref{app:covs}.

Intuitively, the notion of \textit{outcome-nonrestrictive} identification amounts to identification of a causal parameter that holds without restrictions on the distribution of $(\tilde{Y}_i,G_i)$, apart from \eqref{eq:independence} and that $supp\{G_i\} \subseteq \mathcal{G}$, where $supp\{G_i\}$ is the support of the response types $G_i$. However, defining the notion of outcome-nonrestrictive identification in a formal way requires accounting for possible restrictions on observables implied by a given selection model $\mathcal{G}$. The remainder of this Section (subsections \ref{sec:observable} and \ref{sec:outcomenonrestrictive} below) develops some notation to give this formal definition, before turning to the first main result in Section \ref{sec:posi}. Note that neither the definition of outcome-nonrestrictive identification---not my results concerning it---restrict the support of $Y_i$ (e.g. that it be discrete or bounded). 

\subsection{Observable restrictions implied by the model} \label{sec:observable}
Let $\mathcal{P}$ denote the joint distribution of the model fundamentals $(G_i,\tilde{Y}_i,Z_i)$. Given \eqref{eq:independence}, we can decompose $\mathcal{P}$ as
$$\mathcal{P} = \mathcal{P}_{latent} \times \mathcal{P}_Z,$$
where $\mathcal{P}_Z$ denotes the distribution of the instruments $Z_i$ and $\mathcal{P}_{latent}$ denotes the distribution of the latent (counterfactual) variables of the model $\tilde{Y}$ and $G$.\footnote{By $\mathcal{P} = \mathcal{P}_{latent} \times \mathcal{P}_Z$, I mean that for any Borel set $\mathcal{B}_{L}$ of values for $(G_i,\tilde{Y}_i)$ and $\mathcal{B}_{Z}$ of values for $\mathcal{P}_Z$ we have $\mathcal{P}(\mathcal{B}_{L} \times \mathcal{B}_{Z}) = \mathcal{P}_{latent}(\mathcal{B}_{L}) \cdot \mathcal{P}_{Z}(\mathcal{B}_{Z})$, where $\mathcal{B}_{L} \times \mathcal{B}_{Z}$ is the Cartesian product of $\mathcal{B}_{L}$ and $\mathcal{B}_{Z}$.}

A generic causal parameter of interest $\theta$ is a functional $\theta(\mathcal{P})$ of the distribution $\mathcal{P}$ of model  variables. Let $\mathcal{P}_{obs}$ denote the distribution of observable variables $(Y_i,T_i,Z_i)$. Note that $\mathcal{P}_Z$ is a marginalization of $\mathcal{P}_{obs}$ over $Y_i$ and $T_i$. I make use of the following notational convention: for a sub-vector $W_0$ of a random vector $W$, let $\mathcal{P}_{W_0}(\mathcal{P}_{W})$ be the distribution of $W_0$ that arises after marginalizing distribution $\mathcal{P}_W$ over the components of $W$ not included in $W_0$. In this notation, for example, $\mathcal{P}_Z = \mathcal{P}_Z(\mathcal{P}_{obs})$. 

Define $\mathscr{P}_{latent}(\mathcal{G})$ to be the set of $\mathcal{P}_{latent}$ compatible with a given selection model $\mathcal{G}$ and admitting of finite moments:
\begin{equation} \label{eq:platentdef}
	\mathscr{P}_{latent}(\mathcal{G}):= \{ \mathcal{P}_{latent} \in \mathscr{P}_{\tilde{Y}G}:  supp(\mathcal{P}_{\mathcal{G}}(\mathcal{P}_{latent})) \subseteq \mathcal{G}\}
\end{equation}
where we let $\mathscr{P}_{\tilde{Y}G}$ denote the set of all distributions over $(\tilde{Y}_i,G_i)$, such that $\mathbbm{E}[Y_i(t)|G_i=g]$ exists and is finite for each $t \in \mathcal{T}$ and $g \in \mathcal{G}$. Employing a similar notation, we let $\mathscr{P}_Z$ be the set of distributions over instrument values that embed any maintained support restrictions (e.g. that $Z_i$ is binary with $P(Z_i=1) \in (0,1)$).

Note that for any $\mathcal{P} = \mathcal{P}_{latent} \times \mathcal{P}_Z$, Eq. \eqref{eq:outcomeeq} and $T_i=T_i(Z_i)$ imply a distribution of observables. Let $\phi$ denote this map so that $ \mathcal{P}_{obs} = \phi(\mathcal{P})$. The set of possible distributions of observables given a selection model $\mathcal{G}$ is $$\mathscr{P}_{obs}(\mathcal{G}) := \{\phi(\mathcal{P}_{latent} \times \mathcal{P}_Z): \mathcal{P}_{latent} \in \mathscr{P}_{latent}(\mathcal{G}), \mathcal{P}_Z \in \mathscr{P}_Z\}$$
All together, we can think of the basic IV model as the set of distributions $M=\{\mathcal{P}_{latent} \times \mathcal{P}_{Z}: \mathcal{P}_{latent} \in \mathscr{P}_{latent}(\mathcal{G}), \mathcal{P}_{Z} \in \mathscr{P}_{Z}\}$. In this notation note that $\mathscr{P}_{obs}(\mathcal{G})=\phi(M)$.


In general, $\mathscr{P}_{obs}(\mathcal{G})$ is a strict subset of the set of all joint distributions of $(Y_i,T_i,Z_i)$, i.e. restrictions on $\mathcal{G}$ coupled with Eq. (\ref{eq:independence}) imply testable implications on $\mathcal{P}_{obs}$. These testable implications have been studied in the case of the classic LATE model (see e.g. \citealt{kitagawatesting,mourifiewan,kedagniandmourifie}, see also \citealt{jiang2023testing}). Such restrictions are discussed further in Section \ref{sec:testable}. 

\subsection{Outcome nonrestrictive IV identification} \label{sec:outcomenonrestrictive}
In defining outcome-nonrestrictive identification, I focus on parameters $\theta$ that take the form of a conditional counterfactual mean $\mu_c^t := \mathbbm{E}[Y_i(t)|c(G_i)=1]$, a conditional treatment effect $\Delta^{t,t'}_c := \mathbbm{E}[Y_i(t')-Y_i(t)|c(G_i)=1]=\mu_c^{t'}-\mu_c^t$, or a probability $P(c(G_i)=1)$. In all three cases, the target parameter is defined from of a function $c: \mathcal{G} \rightarrow \{0,1\}$ that represents inclusion in some collection of response types. For example, in the LATE model, the LATE is a conditional treatment effect $\Delta^{t,t'}_c$ with $t'=1,t=0$ and $c(g) = \mathbbm{1}(g = \textrm{complier})$. 

Given such a function $c(\cdot)$, it will be useful to denote the subset of $\mathscr{P}_{latent}(\mathcal{G})$ for which $P(c(G_i)=1)>0$ given the distribution $\mathcal{P}_G$ of $G_i$ as:
\begin{equation} \label{eq:platentcdef}
	\mathscr{P}_{latent,c}(\mathcal{G}):= \{ \mathcal{P}_{latent} \in \mathscr{P}_{latent}(\mathcal{G}) \textrm{ and } P(c(G_i)=1)>0 \textrm{ according to } \mathcal{P}_{G}(\mathcal{P}_{latent})\}
\end{equation}
Similarly, let $\mathscr{P}_{obs,c}(\mathcal{G}) := \{\phi(\mathcal{P}_{latent} \times
 \mathcal{P}_Z): \mathcal{P}_{latent} \in \mathscr{P}_{latent,c}(\mathcal{G}), \mathcal{P}_Z \in \mathscr{P}_Z\}$. $\mathscr{P}_{obs,c}(\mathcal{G})$ consist of the distributions of observables that respect selection model $\mathcal{G}$ and put positive probability on the groups $g \in \mathcal{G}$ such that $c(g)=1$. These sets are used in defining outcome-nonrestrictive identification as a simple guarantee that the target parameters $\mu^t_c$ and $\Delta^{t,t'}_c$ are well-defined. But causal parameters that condition on a probability-zero event---such as the marginal treatment effect---can also be accommodated in this framework, as limiting cases of a sequence of parameters for $c_j$ satisfying $P(c_j(G_i)=1)>0$ (see Appendix \ref{sec:examples}).
 
We are now ready to give a definition of outcome-nonrestrictive identification, where the target parameter $\theta$ is expressed as a function $\theta=\theta(\mathcal{P})$ of the data generating process $\mathcal{P}$:
\begin{definition} \label{def:oai}
	Given a choice model $\mathcal{G}$, we say that parameter $\theta$ with conditioning function $c$ is \textbf{outcome-nonrestrictive} identified under $\mathcal{G}$ if the set $$\{\theta(\mathcal{P}): \phi(\mathcal{P}) = \mathcal{P}_{obs} \textrm{ and } \mathcal{P} = (\mathcal{P}_{latent} \times \mathcal{P}_{Z}) \textrm{ for some } \mathcal{P}_{latent} \in \mathscr{P}_{latent,c}(\mathcal{G}) \textit{ and } \mathcal{P}_{Z} \in \mathscr{P}_{Z}\}$$
	is a singleton for all $\mathcal{P}_{obs} \in \mathscr{P}_{obs,c}(\mathcal{G})$.
\end{definition}
\noindent Point identification in general says there is a unique value $\theta(\mathcal{P})$ compatible with Eq. \eqref{eq:independence} and $\phi(\mathcal{P})$, for all $\mathcal{P}$ in some set defined by the model. The key requirement that identification be \textit{outcome-nonrestrictive} is that this model is broad enough to include \textit{all} of $\mathscr{P}_{latent,c}(\mathcal{G})$.\footnote{\label{fn:idzoo}Definition \ref{def:oai} represents a case of point identification as defined in \citet{idzoo} (see also \citealt{matzkin2007}), where the known information ($\phi$ in Lewbel's notation) is the distribution $\mathcal{P}_{obs}$, the model value ($m \in M$ in Lewbel's notation) is $\mathcal{P} = \mathcal{P}_{latent} \times \mathcal{P}_Z$, and the model $M$ is the Cartesian product of $\mathscr{P}_{latent,c}(\mathcal{G})$ and $\mathscr{P}_Z$.} The set $\mathscr{P}_{latent,c}(\mathcal{G})$ allows what \citet{Heckman2004} call \textit{essential heterogeneity}. The only restrictions on outcomes amount to IV independence (imposed by taking the product measure $\mathcal{P}=\mathcal{P}_{latent} \times \mathcal{P}_Z$), exclusion (implicit in the notation $Y_i(t)$), and finite group-specific means of $\tilde{Y}_i$ (imposed through $\mathcal{P}_{\tilde{Y}G}$ in \eqref{eq:platentdef}).\footnote{$\mathscr{P}_{latent,c}(\mathcal{G})$ does restrict the marginal distributions of $G_i$ and $Z_i$: through $\mathcal{G}$, $P(c(G_i)=1)>0$, and $\mathscr{P}_Z$.} Thus $\mathscr{P}_{latent,c}(\mathcal{G})$ is compatible with any marginal distribution $\mathcal{P}_{\tilde{Y}}$ of $\tilde{Y} = \{Y_i(t)\}_{t \in \mathcal{T}}$ or selection-type conditioned distributions $\mathcal{P}_{\tilde{Y}|G=g}$ across various $g \in \mathcal{G}$ whatsoever (provided that they have finite means), so there is no assumption that e.g. treatment effects are homogeneous across units, or are unrelated to counterfactual selection behavior $G_i$.  

Note that identification of $\mu_c^t$ and $\mu_c^{t'}$ immediately implies identification of $\Delta^{t,t'}_c=\mu_c^{t'}-\mu_c^{t}$. With outcome-nonrestrictive identification, this implication in fact goes the other way as well: outcome nonrestrictive identification of $\Delta^{t,t'}_c$ requires outcome-nonrestrictive identification of each constituent part $\mu_c^t$, $\mu_c^{t'}$ . The intuition is that absent assumptions about the joint distribution of potential outcomes, data from individuals with $T_i=t'$ provide no information about $Y_i(t')$ for a different treatment $t'\ne t$, and vice-versa. Thus, we have:
\begin{proposition} \label{prop:ifTEthenmeans}
	$\Delta^{t,t'}_c$ for $t'\ne t$ is outcome-nonrestrictive identified if and only if $\mu_c^t$ and $\mu_c^{t'}$ are.
\end{proposition}
\noindent See Appendix \ref{proofsec} for proofs. Although researchers are typically more interested in treatment effect parameters like $\Delta^{t,t'}_c$ than they are in counterfactual means, we can---informed by Proposition \ref{prop:ifTEthenmeans}---begin our analysis of outcome-nonrestrictive identification with the simpler counterfactual means, before later considering treatment effects.\footnote{The proof in Section \ref{proofsec} extends Proposition \ref{prop:ifTEthenmeans} to cover treatment effect parameters that involve more than two separate treatment states, which is useful for the study of interaction effects in Section \ref{sec:empirical}.} 

	\section{A generic outcome-nonrestrictive identification result} \label{sec:posi}
\subsection{Identifying counterfactual means through ``binary combinations''}
We begin with a very simple sufficient condition for outcome-nonrestrictive identification of counterfactual means of potential outcomes taking the form $\mu_c^t = \mathbbm{E}[Y_i(t)|c(G_i)=1]$. Consider any finite collection of distinct instrument values $z_k \in \mathcal{Z}$ for $k=1 \dots K$ and corresponding coefficients $\alpha_k$. The following quantity is then identified:
\begin{align*}
	\sum_{k=1}^K\alpha_k \cdot \mathbbm{E}\left[Y_i\cdot D^{[t]}_i|Z_i=z_k\right] &= \sum_{k=1}^K \alpha_k \cdot \mathbbm{E}\left[Y_i(t)\cdot D^{[t]}_i(z_k)|Z_i=z_k\right]\\
	&= \mathbbm{E}\left[Y_i(t)\cdot \left(\sum_{k=1}^K\alpha_{k}\cdot D^{[t]}_i(z_k)\right)\right]
\end{align*}
where $D_i^{[t]} = D_i^{[t]}(Z_i) = \mathbbm{1}(T_i=t)$ and the second equality follows from independence (\ref{eq:independence}).

Suppose that the $z_k$ and $\alpha_k$ could be chosen in such a way as to guarantee that the linear combination $\sum_{k}\alpha_{k}\cdot D^{[t]}_i(z_k)$ (in parentheses above) could \textit{only} take values of 0 or 1 for any given $i$. In such a case, the above simplifies to
$$ \sum_{k=1}^K \alpha_k \cdot \mathbbm{E}\left[Y_i\cdot D^{[t]}_i|Z_i=z_k\right] = P\left(\sum_{k}\alpha_{k}\cdot D^{[t]}_i(z_k)=1\right)\cdot \mathbbm{E}\left[Y_i(t)\left|\sum_{k}\alpha_{k}\cdot D^{[t]}_i(z_k)=1\right.\right]$$
Meanwhile
\begin{equation} \label{idresultp} \sum_{k=1}^K\alpha_{k}\cdot \mathbbm{E}\left[D^{[t]}_i|Z_i=z_k\right] = P\left(\sum_{k}\alpha_{k}\cdot D^{[t]}_i(z_k)=1\right)
\end{equation}
Therefore, provided that $P\left(\sum_{k}\alpha_{k}\cdot D^{[t]}_i(z_k)=1\right)>0:$
\begin{equation} \label{idresult}
	\mathbbm{E}\left[Y_i(t)\left|\sum_{k}\alpha_{k}\cdot D^{[t]}_i(z_k)=1\right.\right] = \frac{\sum_{k=1}^K\alpha_{k}\cdot \mathbbm{E}\left[Y_i\cdot D^{[t]}_i|Z_i=z_k\right]}{\sum_{k=1}^K\alpha_{k}\cdot \mathbbm{E}\left[D^{[t]}_i|Z_i=z_k\right]}
\end{equation}
Eq. (\ref{idresult}) represents a generalization of the ``Wald ratio'' form common among IV estimands, and turns out to nest a surprising variety of point identification results from the IV literature, both for counterfactual means as well as treatment effects.

To make our way from the former to the latter, let us establish some terminology to refer to situations in which result \eqref{idresult} can be applied.\\

\noindent \textbf{Definition.} Given selection model $\mathcal{G}$, a \textit{binary combination} is a treatment value $t \in \mathcal{T}$ and a function $\alpha: \mathcal{Z} \rightarrow \mathbb{R}$ of finite support $\mathcal{Z}_K=\{z_k\}_{k=1}^K$ such that $\sum_{k=1}^K\alpha(z_k)\cdot D^{[t]}_i(z_k) \in \{0,1\}$ for all $i$, according to $\mathcal{G}$.\\

\noindent In section \ref{sec:discrete}, I will restrict to discrete instruments and represent the coefficients $\alpha_k$ as a vector in $\mathbbm{R}^{|\mathcal{Z}|}$. But for generality, this section continues to allow the instruments to have arbitrary support and we can think of the coefficients $\alpha_k$ as a function $\alpha(z_k)=\alpha_k$ from $\mathcal{Z}_K\rightarrow\mathbbm{R}$ where $\mathcal{Z}_K=\{z_k\}_{k=1}^K$, or equivalently a function $\alpha$ on all of $\mathcal{Z}$ whose support (the set of points $z$ where $\alpha(z)$ differs from zero) is contained within a finite set $\mathcal{Z}_K \subseteq \mathcal{Z}$.

Binary combinations can thus be indexed by the pair $(t, \alpha)$. Note that given a binary combination, the value of $\sum_{k}\alpha_{k}\cdot D^{[t]}_i(z_k)$ depends only on the response type $G_i$ of individual $i$, so given a binary combination we can write the event that $\sum_{k}\alpha_{k}\cdot D^{[t]}_i(z_k)=1$ as $c^{[t,\alpha]}(G_i)=1$, where $c^{[t,\alpha]}: \mathcal{G} \rightarrow \{0,1\}$ is a function whose definition depends on the treatment $t$ and coefficients $\alpha$ of that binary combination.

With this terminology we can formalize the result of Equation (\ref{idresult}) as follows:

\begin{theorem} \label{thm:suff} Given independence Eq. (\ref{eq:independence}) and a binary combination $(t, \alpha)$, $P(c^{[t,\alpha]}(G_i)=1)$ is identified by the LHS of \eqref{idresultp}. If additionally $P(c^{[t,\alpha]}(G_i)=1)>0$, the conditional counterfactual mean $\mathbbm{E}[Y_i(t)|c^{[t,\alpha]}(G_i)=1]$ is identified by the RHS of (\ref{idresult}). Since the only restrictions placed on $\mathcal{P}$ in deriving \eqref{idresult} are that Eq. \eqref{eq:independence} holds, that $\mathcal{Z}_K \subseteq \mathcal{Z}$, and that $P(c^{[t,\alpha]}(G_i)=1)>0$, identification of $\mathbbm{E}[Y_i(t)|c^{[t,\alpha]}(G_i)=1]$ is outcome nonrestrictive. 
\end{theorem}

\noindent The key to the observation that $\mathbbm{E}[Y_i(t)|c^{[t,\alpha]}(G_i)=1]$ is outcome-nonrestrictive identified is that whether the requirement $P(c^{[t,\alpha]}(G_i)=1)>0$ holds depends only on the marginal distribution of $G_i$, and whether $\mathcal{Z}_K \subseteq \mathcal{Z}$ depends only on the marginal distribution of $Z_i$. This implies nothing about the distribution of potential outcomes or their relation to $G_i$.



\subsection{Identifying treatment effects through ``binary collections''} \label{sec:binarycollections}
While Theorem \ref{thm:suff} yields identification of conditional counterfactual means, we can furthermore identify \textit{treatment effects} when two treatment values $t$ and $t'$ admit of binary combinations that yield the same conditioning events.

Consider a collection of binary combinations that apply to at least two distinct values $t \in \mathcal{T}$. Let us denote set of coefficients $\alpha$ in each binary combination by $\alpha^{[t]}$, indexed by the treatment value $t$ it will be applied to. In this notation, $\alpha_k^{[t]}$ is the coefficient on $z_k$ in the binary combination corresponding to treatment $t$.\\

\noindent \textbf{Definition.} A \textit{binary collection} is a set of binary combinations $\{(t,\alpha^{[t]})\}_{t \in \psi}$ for treatment values in set $\psi \subseteq \mathcal{T}$ where $|\psi| \ge 2$, with the property that given the selection model $\mathcal{G}$, the functions $c^{[t,\alpha^{[t]}]}$ and $c^{[t',\alpha^{[t']}]}$ are identical, for any $t,t' \in \psi$.\\ 

\noindent For a given binary collection, let us for brevity denote the common function $c^{[t,\alpha^{[t]}]}$ for all $t \in \psi$ as $c$. It follows immediately from Theorem \ref{thm:suff} that treatment effects $\mathbbm{E}[Y_i(t')-Y_i(t)|c(G_i)=1] = \mathbbm{E}[Y_i(t')|c(G_i)=1] - \mathbbm{E}[Y_i(t)|c(G_i)=1]$ are identified for any pair $t, t' \in \psi$.\\

\noindent \textit{Note: }By replacing $Y_i(t)$ by $\mathbbm{1}(Y_i(t) \le y)$ we can also identify the conditional distributions $F_{Y(t)|c(G)=1}$ for all $t \in \psi$ and compute e.g. quantile treatment effects or establish bounds on the distribution of treatment effects among the $c(G_i)=1$ group \citet{fanpark}.\\

\noindent When treatment is itself binary, we can generate binary collections from any binary combination where the coefficients sum to zero:
\begin{proposition} \label{prop:binarytreatment}
	Let $\mathcal{T} = \{0,1\}$, and suppose $(t,\alpha)$ is a binary combination such that $\sum_k \alpha_k = 0$. Then there exists a binary collection with $\psi = \mathcal{T}$. In particular, the coefficients for $t=0$ are simply $-1$ times the corresponding coefficients for $t=1$.
\end{proposition}
\begin{proof} 
	See alternative statement of this result in Section \ref{sec:discrete}.
\end{proof}
\noindent The restriction that $\sum_k \alpha_k = 0$ is a natural one, in the following sense:
\begin{proposition} \label{prop:sumtozero}
	Let $\Delta_c^{t,t'}=\mathbbm{E}[Y_i(t')-Y_i(t)|c(G_i)=1]$ be outcome-nonrestrictive identified from a binary collection with $t' \ne t$. Then if $\mathcal{G}$ contains a group $g_0$ that always takes  treatment $t$, it must be the the case that $\sum_k \alpha^{[t]}_k = 0$.
\end{proposition}
\begin{proof}
	Since $P(T_i=t'|G_i=g_0) = 0$, the data provide no information on $Y(t')|G_i=g_0$, so we must have $c(g_0)=0$ (see proof of Proposition \ref{prop:ifTEthenmeans}). Thus $c(g_0)=\sum_{k} \alpha^{[t]}_k \cdot \mathbbm{1}(T_{g_0}(z_k)=t) = \sum_{k} \alpha^{[t]}_k = 0$.
\end{proof}


\noindent For example, in the LATE model of \citet{Imbens2018}, allowing for ``always-takers'' (who always take treatment $t=1$, regardless of $Z_i$) implies that $\sum_z \alpha^{[1]}_z = 0$, while allowing for ``never-takers'' (who always take treatment $t=0$) implies that $\sum_z \alpha^{[0]}_z = 0$. Consistent with this, identification of the compliers LATE follows from the binary collection in which  $\alpha^{[1]}_1 = 1$, $\alpha^{[1]}_0 = -1$, $\alpha^{[0]}_1 = -1$, and $\alpha^{[0]}_0 = 1$.

\subsection{Using binary combinations and collections for testing the model} \label{sec:testable}
The existence of binary combinations with $K>1$ generally yields overidentification restrictions that can used to test the IV model (including exclusion, independence, and the choice of selection model $\mathcal{G}$). In particular, suppose that $|\mathcal{G}| < \infty$ and note that for any Borel set $\mathcal{B}$ of $\mathbbm{R}$ and binary combination $(t, \alpha)$, we have that:
\begin{equation} \label{eq:testable}
	\sum_{k=1}^K \alpha_k \cdot P(Y_i \in \mathcal{B}, T_i=t|Z_i=z_k)  = P(Y_i(t) \in \mathcal{B}, c(G_i)=1)
\end{equation}
using Eq. \eqref{eq:independence} and that $P(Y_i(t) \in \mathcal{B}, T_i(z_k)=t) = \sum_{g \in \mathcal{G}} P(G_i=g)\cdot P(Y_i(t) \in \mathcal{B}|G_i=g) \cdot A^{[t]}_{z_k,g}$. Since the RHS of Eq. \eqref{eq:testable} represents a probability, the LHS must be weakly positive. Provided that not all of the $\alpha_k$ are positive, the implication that $\sum_{k=1}^K \alpha_k \cdot P(Y_i \in \mathcal{B}, T_i=t|Z_i=z_k) \ge 0$ is not guaranteed and therefore can be used to test the model assumptions.

Furthermore, finding binary \textit{collections} may yield further overidentification restrictions that make use of the ``first stage'' data alone. Depending on the selection model, the equality $\sum_{k=1}^K\alpha^{[t]}_{k}\cdot \mathbbm{E}\left[D^{[t]}_i|Z_i=z_k\right] =\sum_{k=1}^K\alpha^{[t']}_{k}\cdot \mathbbm{E}\left[D^{[t']}_i|Z_i=z_k\right]$ may not be trivially satisfied, even in the case of a binary treatment. See Section \ref{sec:empirical} for an example of such equality restrictions in the context of an empirical application, and Appendix \ref{sec:lp} for further linear inequality constraints that are based upon first stage empirical moments. Still further testable restrictions hold if one has a binary collection and Eq. \eqref{eq:independence} holds conditional on observed covariates $X_i$. See Appendix \ref{app:covs} for details.
	\section{Outcome-nonrestrictive identification with discrete instruments} \label{sec:discrete}
The remainder of this paper now specializes to settings in which the instruments $Z_i$ are discrete and take only a finite number of values $\mathcal{Z}$. This simplification allows us to establish that binary combinations are also \textit{necessary} for outcome-nonrestrictive identification to occur. Combining with Theorem \ref{thm:suff} and Proposition \ref{prop:ifTEthenmeans}, this provides a full characterization of outcome-nonrestrictive identification which yields a simple geometric interpretation.\\

\noindent \textit{Notation: selection models with finite instrument values.} When $\mathcal{Z}$ is discrete and finite, any binary combination $\alpha$ can be associated with a vector in $\mathbbm{R}^{|\mathcal{Z}|}$(having non-zero components only for values $z \in \mathcal{Z}_K$). Across a finite set of treatments, note that $\mathcal{G}$ can then also only take finitely many values, i.e. $|\mathcal{G}| \le |\mathcal{T}|^{|\mathcal{Z}|}$. A function $c: \mathcal{G} \rightarrow \{0,1\}$ defining a causal parameter like $\mu_c^t=\mathbbm{E}[Y_i(t)|c(G_i)=1]$ can now be associated with a $|\mathcal{G}|$-component vector $c$ with components $c_g = c(g)$ for each $g \in \mathcal{G}$.

In this setting, we can also express the content of the selection model $\mathcal{G}$ through a $|\mathcal{Z}| \times |\mathcal{G}|$ matrix $A$, where component $A_{zg}$ gives the  common treatment $T_i(z)$ that all units $i$ with $G_i=g$ take, when the instruments are equal to $z$. The restrictions imposed by selection model $\mathcal{G}$ correspond to deleting columns from a $|\mathcal{Z}| \times |\mathcal{T}|^{|\mathcal{Z}|}$ matrix that would include all $|\mathcal{T}|^{|\mathcal{Z}|}$ imaginable response types given $\mathcal{T}$ and $\mathcal{Z}$.

Define for any treatment $t$ the binary matrix $A^{[t]}$ having components $[A^{[t]}]_{zg} = \mathbbm{1}(A_{zg}=t)$, which records the common value of $D^{[t]}_i(z)$ for any individual $i$ having $G_i=g$. The matrix $A^{[t]}$ simply tells us whether units take treatment $t$ versus any other treatment.\footnote{A matrix analogous to $A^{[t]}$ is used heavily in \citet{Heckman2018}.}

\subsection{A necessary condition for outcome-nonrestrictive identification}
Consider a parameter of the form $E[Y_i(t)|c(G_i)=1]$. Recall from above the representation of $c(\cdot)$ as a binary vector $c \in \mathbbm{R}^{|\mathcal{G}|}$ with $c_g \in \{0,1\}$ for each $g \in \mathcal{G}$. For any matrix $B$ let $rowspace(B)$ or $rs(B)$ denote its rowspace, and $B'$ its transpose.

\begin{theorem} \label{thm:necc}
	Suppose that $\mathcal{Z}$ and $\mathcal{T}$ are finite. Then if $\mu_c^t:=E[Y_i(t)|c(G_i)=1]$ is outcome-nonrestrictive identified, then $c' = \alpha'A^{[t]}$, for some $\alpha \in \mathbbm{R}^{|\mathcal{Z}|}$, i.e. $c \in rowspace(A^{[t]})$.
\end{theorem}
\noindent Note that given finite $|\mathcal{Z}|$ and hence a finite space of response types, $c \in rs(A^{[t]})$ occurs exactly when there exists a binary combination $(t,\alpha)$ with conditioning function $c=c^{[\alpha,t]}$. Theorem \ref{thm:necc} thus establishes that if the instruments are finite, then Theorem \ref{thm:suff} covers \textit{all} instances in which a counterfactual mean that conditions on response types can be identified in an outcome-nonrestrictive way. In other words, binary combinations are both necessary and sufficient for outcome-nonrestrictive identification of counterfactual means. Together with Proposition \ref{prop:ifTEthenmeans}, it follows that binary collections are similarly both necessary and sufficient for outcome-nonrestrictive identification of local average treatment effects.\\

\noindent \textit{Remark: } Theorems \ref{thm:suff} and \ref{thm:necc} both extend to the more general family of target parameters that can be defined by functions $c$ that depend on $Z_i$ in addition to response types $G_i$. This is useful to nest parameters like the average treatment effect on the treated, or certain parameters that can arise in settings with multiple instruments. See Appendix \ref{app:deponz} for details.

\subsection{Examples to which Theorem \ref{thm:necc} does \textit{not} apply} \label{sec:notapply}
Although Theorem \ref{thm:necc} synthesizes a wide variety of existing IV identification results (detailed in Appendix \ref{sec:examples}), the presumption that \textit{outcome-nonrestrictive} identification holds---rather than point identification in general---is important.

The structure of Theorem \ref{thm:necc} can be summarized as follows. With $M:=\{\mathcal{P}_{latent} \times \mathcal{P}_{Z}: \mathcal{P}_{latent} \in \mathscr{P}_{latent,c}(\mathcal{G}), \mathcal{P}_{Z} \in \mathscr{P}_{Z}\}$, outcome-nonrestrictive identification says that $\{\theta(\mathcal{P}): \mathcal{P} \in M \textrm{ and } \phi(\mathcal{P}) = \mathcal{P}_{obs}\}$ is a singleton for all $\mathcal{P}_{obs} \in \mathscr{P}_{obs,c}(\mathcal{G})$. This requires that there be no $\mathcal{P},\mathcal{P}' \in M$ such that $\phi(\mathcal{P})=\phi(\mathcal{P}')$ but $\theta(\mathcal{P}) \ne \theta(\mathcal{P}')$. The proof of Theorem \ref{thm:necc} uses that there always exist $\mathcal{P} \in M$ that satisfy a certain regularity condition, which enables us to construct from $\mathcal{P}$ such a $\mathcal{P}'$, provided that $c \notin rs(A^{[t]})$. But if one restricts the model space $M$ of permissible DGPs by maintaining further assumptions about the distribution of potential outcomes (or how they are correlated with response types $G_i$), it can be that the constructed distribution $\mathcal{P}'$ violates those assumptions and therefore do not belong to $M$.

For example, it is known for example that $\mathbbm{E}[Y_i(t')-Y_i(t)]$---the unconditional average treatment effect (ATE) between $t$ and $t'$---is identified under an assumption of ``no selection on gains'' (NSOG): that is that $Y_i(t)-Y_i(t')$ is mean independent of $T_i$ and $Z_i$ for all $t,t' \in \mathcal{T}$ \citep{Kolesar2013, aroragoffhjort}. Homogeneous treatment effects are a special case of NSOG.\footnote{Another stronger restriction is when potential outcomes $Y_i(t)$ alone---and not just treatment effects---are mean independent of $T_i$ and $Z_i$ for all $t$. This essentially rules out endogeneity: $\mathbbm{E}[Y_i(t)]=\mathbbm{E}[Y_i|T_i=t]$, so identification is unsurprising under this stronger restriction.} In Appendix \ref{app:nsog}, I show that the result that the ATE between $t$ and $t'$ is identified under NSOG can be extended to see that unconditional counterfactual means $\mathbbm{E}[Y_i(t)]$ are also identified under NSOG, requiring no assumptions on selection (beyond an order condition that can be verified in the data).

Without a selection model $\mathcal{G}$ that imposes substantive restrictions, the vector $c = (1,1 \dots 1)'$ will not be in the row space of $A^{[t]}$. In particular, as long as there is a ``never-takers'' group $g_0(t)$ for treatment $t$ such that $T_i(z) \ne t$ for all $z \in \mathcal{Z}$ when $G_i = g_0(t)$, then $(1,1 \dots 1)' \in rs(A^{[t]})$ cannot hold. However there is no contradiction with Theorem \ref{thm:necc}, since identification based on NSOG does not hold for \textit{all} joint distributions between $\tilde{Y}=\{Y_i(t)\}_{t \in \mathcal{T}}$ and $G_i$. Rather, the assumption of NSOG eliminates some such distributions that are compatible with the data and the basic model of Eq. \eqref{eq:independence}, shrinking $M$.\footnote{In fact, the NSOG assumption is strong enough to let the researcher \textit{impute} the value of $\mathbbm{E}[Y_i(t)|G_i = g_0(t)]$ if such a never-taker group exists, thus eliminating the dependence of the estimand on the distribution of $Y_i$ among individuals such that $G_i = g_0(t)$ and $T_i=t$ (which would not be identified from observable data).} Indeed I show explicitly in Section \ref{app:nsog} that absent a selection model $\mathcal{G}$ such that $(1,1 \dots 1)' \in rs(A^{[t]})$ for all $t \in \mathcal{T}$, the construction $\mathcal{P}'$ in the proof of Theorem \ref{thm:necc} will not satisfy the additional restriction of NSOG.

Another example of an IV identification result that is not covered by Theorem \ref{thm:necc} is the ``compliers--defiers'' result of  \citet{toleratingdefiance} that the local average treatment effect among a subset of compliers is identified in a setting with a binary treatment and instrument, if there are more compliers than defiers and a subset of the compliers have the same average treatment effect as the defiers. Again, this additional assumption places restrictions on the joint distribution of response types $G_i$ and potential outcomes $\tilde{Y}_i$. Further, the identified parameter conditions on an event (a particular subgroup of the compliers) that is less course than the groups $G_i$ that are defined simply by counterfactual selection behavior, so does not fit the form $\Delta_{c}^{t,t'}=\mu_c^{t'}-\mu_c^{t}$ that Theorem \ref{thm:necc} and Proposition \ref{prop:ifTEthenmeans} speak to. Similar considerations apply to recent results of \citep{comey2023supercompliers} that show identification of the local average treatment effect among ``supercompliers'' in a setting in which $\mathcal{Y}=\mathcal{T}=\mathcal{Z}=\{0,1\}$, where the supercompliers are defined as the subset of compliers that have a strictly positive treatment effect. This model imposes monotonicity in the outcome equation, and the conditioning event for the supercomplier LATE conditions both on selection behavior \textit{and} a property of outcomes, namely that $Y_i(1) > Y_i(0)$.

Another type of identification result that is not covered by Theorem \ref{thm:necc} above---although it is outcome-nonrestrictive---is identification of a treatment effect parameter that does not maintain two fixed treatment values $t$ and $t'$ across all units included in the parameter. An example of this kind arises in \citet{klinewalters}, in which the identified causal parameter compares the effect of Head Start to one of two next-best alternatives (either traditional pre-school or no pre-school). This estimand combines two response types for which this next-best alternative is generally different. See Section \ref{sec:klinewalters} for details.

When $c \notin rs(A^{[t]})$, Theorem \ref{thm:necc} establishes that the parameter $\mu_c^t$ is not \textit{point} identified in an outcome-nonrestrictive manner. However, the data may still provide identifying information about the value of $\mu_c^t$ if auxiliary conditions are maintained, for example that the support of $Y_i$ is bounded with known bounds. Appendix \ref{sec:partialid} considers partial identification of $\mu_c^t$ in such settings, and also relates the results of this paper to recent results by \citet{bai2024identifyingpowermonotonicityaverage}, who focus on bounding the ATE and unconditional means in particular.

\subsection{Summary: combining the necessary and sufficient conditions for identification} \label{sec:combining}
Combining Theorems \ref{thm:suff} and \ref{thm:necc}, we have in the case of finite discrete instruments that given a selection model $\mathcal{G}$, a conditional counterfactual mean of the form $\mathbbm{E}[Y_i(t)|c(G_i)=1]$ is outcome-nonrestrictive identified \textit{if and only if} the vector representation $c \in \{0,1\}^{|\mathcal{G}|}$ of $c(\cdot)$ lies in the rowspace of $A^{[t]}$, i.e. $c \in rs(A^{[t]})$. Analogously, a treatment effect parameter $\mathbbm{E}[Y_i(t')-Y_i(t)|c(G_i)=1]$ is outcome-nonrestrictive identified if and only if $c$ lies in the rowspaces of both $A^{[t]}$ \textit{and} $A^{[t']}$, i.e. $c \in (rs(A^{[t']}) \cap rs(A^{[t]}))$.

Appendix \ref{sec:nps} discusses how Theorems \ref{thm:suff} and \ref{thm:necc} relate to recent necessary and sufficient conditions for identification in IV models by \citet{navjeevan2023identification}. While Theorem \ref{thm:suff} can be seen a special case of their results, Theorem \ref{thm:necc} cannot. Further, \citet{navjeevan2023identification} do not derive the condition $c \in (rs(A^{[t']}) \cap rs(A^{[t]}))$ for the identification of treatment effect parameters (though sufficiency of this condition would represent a corollary to their results). I turn now to the analysis of this key condition, which enables an exhaustive search for outcome-nonrestrictive identification results for treatment effects when the instruments are discrete and finite.
\section{Making use of the equivalence result for identification} \label{sec:bruteforce}

Continuing our focus on settings in which the instruments have finite points of support, this Section shows how the condition $c \in rs(A^{[t]})$ characterizing identification can be useful in understanding existing identification results, generating new ones, and ruling out further opportunities for identification in a given selection model.

\subsection{A geometric characterization of identification with discrete instruments} \label{sec:geom}
Let $\mathcal{C}(t)$ be the set of $c$ in the rowspace of $A^{[t]}$ that have entries of only zero or one: i.e. 
\begin{equation} \label{eq:intersectionbinarycombo}
	\mathcal{C}(t) = rs(A^{[t]}) \cap \{0,1\}^{|\mathcal{G}|}
\end{equation} 
It is always the case that $C(t) \ne \emptyset$ provided that $P(T_i=t)>0$.\footnote{To see this, note that  $\mathbbm{E}[Y_i(t)|T_i(z)=t] = \frac{\mathbbm{E}[Y_i\cdot D^{[t]}_i|Z_i=z]}{\mathbbm{E}[D^{[t]}_i|Z_i=z]}$, which considers all units that take treatment $t$ when $Z_i=z$. This corresponds to a binary combination with $\alpha_{z'} = \mathbbm{1}(z'=z)$ and $c_g = \mathbbm{1}(T_g(z)=t)$.} A binary collection in turn occurs when $\mathcal{C}(t) \cap \mathcal{C}(t') \ne \emptyset$.\footnote{Binary combinations can occur in choice models where no binary collections exist. The choice model described in Proposition 8 of \citet{lee2023treatment} for example has this property.} When treatment is binary, this observation yields a simple proof of Proposition \ref{prop:binarytreatment}. In this discrete setting, we can rewrite Proposition \ref{prop:binarytreatment} as \begin{proposition*} [(alternative statement of Proposition \ref{prop:binarytreatment})]
	Let $\mathbbm{1}_{n}$ denote a vector of ones in $\mathbbm{R}^{n}$. If $\mathcal{T} = \{0,1\}$ and $\alpha'\mathbbm{1}_{|\mathcal{Z}|}=0$ and ${A^{[1]}}'\alpha \in \mathcal{C}(1)$, then ${A^{[1]}}'\alpha=-{A^{[0]}}'\alpha$ so that ${A^{[0]}}'(-\alpha) \in \mathcal{C}(0)$ and hence $\mathbb{E}[Y_i(1)-Y_i(0)|c_{G_i}]$ is outcome-nonrestrictive identified.
\end{proposition*} 
\begin{proof}
	Since $A^{[0]} = \mathbbm{1}_{|\mathcal{Z}|}\mathbbm{1}_{|\mathcal{G}|}'-A^{[1]}$, so $\alpha'A^{[0]}=\cancel{\alpha'\mathbbm{1}_{|\mathcal{Z}|}}\mathbbm{1}_{|\mathcal{G}|}'-\alpha'A^{[1]}=(-\alpha')A^{[1]}$.
\end{proof}

\subsubsection*{Example selection model: the classic LATE model with a binary instrument} Consider the model of \citet{Imbens2018} with a binary instrument, where $\mathcal{Z} = \mathcal{T} = \{0,1\}$ and we rule out ``defiers'', i.e. those who would have $T_i(0)=1, T_i(1)=0$. Then:
$$A^{[1]} = \begin{bmatrix}
	0 & 1 & 0\\  0 & 1 & 1
\end{bmatrix} \quad \quad \quad  \textrm{ and } \quad \quad \quad A^{[0]} = \begin{bmatrix}
	1 & 0 & 1\\  1 & 0 & 0
\end{bmatrix}$$
where the first row of each matrix represents $z=0$ and the second $z=1$, while the columns correspond to never-takers, always-takers, and compliers, respectively.

The rowspaces of $A^{[1]}$ and $A^{[0]}$ can be found by row-reducing each matrix, yielding:
$$rs(A^{[1]}) = span\left\{\begin{pmatrix}
	0 \\ 1 \\  0
\end{pmatrix},\begin{pmatrix}
	0 \\ 0 \\  1
\end{pmatrix}\right\} \quad \quad \quad  \textrm{ and } \quad \quad \quad rs(A^{[0]}) = span\left\{\begin{pmatrix}
	1 \\ 0 \\  0
\end{pmatrix},\begin{pmatrix}
	0 \\ 0 \\  1
\end{pmatrix}\right\}$$
\begin{figure}[ht!]
	\begin{center}
		\hspace{1cm}\includegraphics[width=4in]{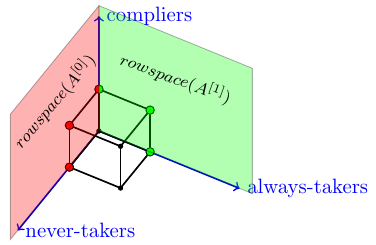}	
	\end{center}
	\caption{Identification in the LATE model with a binary instrument. The vector $c=(0,0,1)'$ belongs to both $rs(A^{[1]})$ and $rs(A^{[0]})$ and hence the LATE parameter $\mathbbm{E}[Y_i(1)-Y_i(0)|i \textrm{ is a complier}]$ is identified, and this identification is outcome-nonrestrictive. $(0,0,1)'$ is the only vertex of the unit cube that belongs to $rs(A^{[1]}) \cap rs(A^{[0]})$, and thus represents the only response type for which treatment effects can be identified. \label{fig:cube}}
\end{figure}

By Theorem \ref{thm:suff}, we can thus identify the mean of $Y_i(1)$ among always-takers or among compliers (or among both), and we can identify the mean of $Y_i(0)$ among never-takers or among compliers (or both). As depicted in Figure \ref{fig:cube}, these correspond to the non-zero vertices of the unit cube in $\mathbbm{R}^3$ that take a value of zero in the never-takers ``direction'', or a value of zero in the always-takers ``direction'', respectively.

Note that $(0,0,1)'$ the \textit{unique} non-zero vertex of the unit cube in $\mathbbm{R}^3$ that belongs to both $rs(A^{[1]})$ and to $rs(A^{[0]})$. Theorem \ref{thm:necc} demonstrates that the LATE among compliers is then in fact the \textit{only} treatment effect parameter $\Delta_c$ that is outcome-nonrestrictive identified in the LATE model. The local average treatment effect $\Delta_c$ is outcome-nonrestrictive identified for the compliers $c=(0,0,1)'$, because this this $c$ belongs to \textit{both} $rs(A^{[1]})$ and to $rs(A^{[0]})$.\\

\noindent \textit{Remark: } \citet{melowinter} study the cardinality of the intersection between the unit cube in $\mathbbm{R}^n$ and any linear subspace of $\mathbbm{R}^n$. For a matrix $A$ with $rs(A)$ of dimension $k$, their result implies that the $rs(A) \cap \{0,1\}^{n}$ has a cardinality of at most $2^k$. In the binary-binary LATE model, $k=|\mathcal{Z}|=2$ for either of $A^{[0]}$ or $A^{[1]}$, and in either case $|rs(A^{[t]}) \cap \{0,1\}^{n}|=2$, which does not meet this upper bound of $2^k=4$. However the result does imply that there can be no more than $2^{|\mathcal{Z}|}$ binary combinations, even though typically $2^{|\mathcal{Z}|} < 2^{|\mathcal{G}|}$ and there are $2^{|\mathcal{G}|}$ potential values of $c$ to consider ex-ante.

\subsubsection*{Example target parameter: the unconditional average treatment effect} With a binary treatment, a well-studied parameter of interest is the overall population average treatment effect (ATE): $\Delta^{0,1}=\mathbbm{E}[Y_i(1)-Y_i(0)]$. This can be seen as a parameter $\mu_{c}^{t',t}=\mathbbm{E}[Y_i(t')-Y_i(t)|c(G_i)=1]$ in which the function $c(g)=1$ for all $g \in \mathcal{G}$, or in vector form $c=(1,1, \dots 1)'$.

As a Corollary of Theorems \ref{thm:suff} and \ref{thm:necc} we thus have the following:
\begin{corollary} \label{corr:ate}
	Suppose that $\mathcal{Z}$ and $\mathcal{T}$ are finite. Then the unconditional counterfactual mean $\mathbbm{E}[Y_i(t)]=\mu^t_{(1,1, \dots 1)'}$ is outcome-nonrestrictive identified if and only if $(1,1, \dots 1)' \in rs(A^{[t]})$, and the unconditional average treatment effect $\mathbbm{E}[Y_i(t')-Y_i(t)]=\Delta^{t,t'}_{(1,1, \dots 1)'}$ is outcome-nonrestrictive identified if and only if $(1,1, \dots 1)' \in (rs(A^{[t']}) \cap rs(A^{[t]}))$.
\end{corollary}
\noindent Since the presence of never-takers with respect to treatment $t$ implies that $(1,1, \dots 1)' \notin rs(A^{[t]}), $\footnote{If such never-takers are alowed in $\mathcal{G}$, this introduces a column of all zeroes in the matrix $A^{[t]}$.} Corollary \ref{corr:ate} implies that ATEs and unconditional counterfactual means are never point-identified in an outcome-nonrestrictive manner absent restrictions on selection.

Corollary \ref{corr:ate} also relates my results to recent work by \citet{bai2024identifyingpowermonotonicityaverage} on the partial identification power of monotonicity for these parameters, as described in Appendix \ref{sec:bhmsv}. \citet{bai2024identifyingpowermonotonicityaverage} show that selection models can have limited additional identifying power for ATEs provided that they include a restriction that the authors call \textit{generalized monotonicity}, and the outcome is discrete and bounded. These results underscore the upside to focusing on target parameters beyond the ATE (i.e. $c \ne (1,1,\dots 1)'$) when one is willing to impose restrictions on selection, or finding restrictions on selection that do not imply generalized monotonicity but still aid in identification. Section \ref{sec:detailedexamples} shows some examples of this kind.

\subsection{Applying the characterization to search for identified treatment effect parameters} \label{sec:geomTEs}
What then can we say about the set of possible identified treatment effect parameters for a given $t' \ne t$, that is: $\mathbbm{E}[Y_i(t')-Y_i(t)|c_{G_i}=1]$ where $c \in rs(A^{[t]}) \cap rs(A^{[t']}) \cap \{0,1\}^{|\mathcal{G}|}$?

For ease of notation, let us for the moment label the treatment values of interest $t'=1$ and $t=0$, without loss of generality. Let us similarly denote $\alpha^{[t']}$ by $\alpha_1$ and $\alpha^{[t]}$ by $\alpha_0$ (each of these is a $|\mathcal{Z}|$-component vector). Then for some $c \in \{0,1\}^{|\mathcal{G}|}$ and $\alpha_0, \alpha_1 \in \mathbbm{R}^{|\mathcal{G}|}$, we have a binary collection when:
$$c' = \alpha_1'A^{[1]} = \alpha_0'A^{[0]}$$
which occurs if and only if
\begin{equation} \label{eq:Anullspace}
	(\alpha_1',-\alpha_0')\begin{bmatrix}
		A^{[1]}\\
		A^{[0]}
	\end{bmatrix}:=\alpha'A^{[1,0]} = \mathbf{0}^{|\mathcal{G}|}
\end{equation}
where we let $A^{[1,0]}$ denote a $2\cdot|\mathcal{Z}| \times |\mathcal{G}|$ matrix composed of the rows of $A^{[1]}$ followed by the rows of $A^{[0]}$, and $\alpha = (\alpha_1',-\alpha_0')'$ is a $2\cdot|\mathcal{Z}| \times 1$ vector. For any $\alpha$ in the left null-space $ns({A^{[1,0]}})$ of ${A^{[1,0]}}$, let $c(\alpha)$ denote the value $c = {A^{[1]}}'\alpha_1 = {A^{[0]}}'\alpha_0$ where $\alpha_1$ is a vector of the the first $|\mathcal{Z}|$ components of $\alpha$ and $\alpha_0$ is a vector of minus one times each of the last $|\mathcal{Z}|$ components of $\alpha$. In general then $\mathcal{C}(t) \cap \mathcal{C}(t') = \{c(\alpha): \alpha \in ns(A^{[t',t]})\} \cap \{0,1\}^{|\mathcal{G}|}$, where $A^{[t',t]}$ is composed from $A^{[t']}$ and $A^{[t]}$ as above. This characterization proves useful in the search for new IV identification results to follow.

The following result further aids in implementing a practical search for binary combinations (and hence binary collections): 
\begin{proposition} \label{prop:limitsearch}
	If $c \in rs(A^{[t]})$ for some $t$ and $c \in \{0,1\}^{|\mathcal{G}|}$, then the equation $c'=\alpha'A^{[t]}$ can be satisfied by a vector $\alpha$ having elements that are rational and belong to the set $$\mathcal{C}_{n} := \left\{\frac{a}{b}: a,b, \in \mathcal{D}_{|\mathcal{Z}|}\right\}$$
	where $\mathcal{D}_{n}:=\{det(B): B \in \{0,1\}^{n \times n}\}$ is the set of possible determinant values for an $n \times n$ matrix $B$ having entries in $\{0,1\}$.
\end{proposition}

\noindent Proposition \ref{prop:limitsearch} implies that when searching for binary combinations, we can always restrict the components of $\alpha$ to belong to the finite set $\mathcal{D}_{|\mathcal{Z}|}$. For $n \le 7$, the set $\mathcal{D}_{n}$ is known to consist of consecutive integers symmetric about zero \citep{craigen}:\footnote{For $n \ge 8$, $\mathcal{D}_n$ remains a bounded set of integers for any given $n$, but $\mathcal{D}_n$ generally skips some consecutive integers. For example, it is not possible for a $7 \times 7$ binary matrix to have a determinant of $28$ but one can achieve a determinant of $32$ \citep{craigen}.}
\begin{align*}
	\mathcal{D}_1 &= \mathcal{D}_2 = \{-1,0,1\}, \quad \mathcal{D}_3 = \{-2,-1,0,1,2\}, \quad 
	\mathcal{D}_4 = \{-3, \dots, -1,0,1,\dots ,3\}\\
	&\mathcal{D}_5, \quad = \{-5, \dots, -1,0,1,\dots ,5\}, \quad \mathcal{D}_6 = \{-8, \dots, -1,0,1,\dots ,8\}
\end{align*}
As an implication it for example follows that for $|\mathcal{Z}|\le 2$, all $\alpha_z$ must be in the set $\mathcal{C}_{1}=\mathcal{C}_{2}=\{-1,0,1\}$, in the set $\mathcal{C}_{3}=\{-2,-1,-1/2,0,1/2,1,2\}$ for $|\mathcal{Z}|= 3$, and for $|\mathcal{Z}|= 4$: $$\mathcal{C}_{4}=\{-3,-2,-3/2,-1,-2/3,-1/2,-1/3,0,1/3,1/2,2/3,1,3/2,2,3\}$$

\noindent \textit{Remark:} Since each $\alpha_z$ is rational, we can without loss of generality rewrite Eq. \eqref{idresult} with \textit{integer} coefficients $\alpha_k$ (by multiplying the numerator and denominator of \eqref{idresult} by the least common multiple of the denominators of all $\alpha_z$). However, these integer coefficients do not necessary need to add up to zero, as they do in the binary treatment LATE model and various extensions of it.

\subsection{Algorithms for generating binary collections}

In this section I implement a brute-force algorithm that uses the results thus far to perform an exhaustive search for binary collections in settings with $|\mathcal{Z}|, |\mathcal{T}| \le 3$, uncovering several novel identification results for treatment effects in IV models.\footnote{I am grateful to Simon Lee for suggesting this idea to me.} 

I compare two versions of the algorithm, which are laid out explicitly in Appendix \ref{app:algoone}. The first is a ``naive'' approach that iterates over all possible selection models $\mathcal{G}$ given $\mathcal{Z}$ and $\mathcal{T}$ and then finds binary collections within that selection model. For a given selection model $\mathcal{G}$, there are $2^{|\mathcal{G}|}$ possible values of the vector $c$, and a certificate of whether $c$ corresponds to a binary collection for a given $t',t$ can be verified by testing whether $c=c(\alpha)$ for some $\alpha$ in the left nullspace of matrix $A^{[t',t]}$ defined in Eq. \eqref{eq:Anullspace}. A second algorithm makes use of Proposition \ref{prop:limitsearch} to instead iterate over the possible $2 \cdot |\mathcal{Z}|$-component vectors $\alpha$, rather than over selection models $\mathcal{G}$. This comes at great computational benefit, as computations for a single $\alpha$ are useful for studying many selection models at once. Given the results of Section \ref{sec:geomTEs}, we can without loss of generality restrict the search over $\alpha$ to those having components in the discrete and finite set $\mathcal{C}_{|\mathcal{Z}|}$. Compared with Algorithm 1 above, which quickly becomes infeasible for $|\mathcal{Z}|\ge 3$, this second approach runs on $|\mathcal{Z}|=3$ within minutes. The reason is that the number of possible selection models $2^{|\mathcal{T}|^{|\mathcal{Z}|}}$ scales much more quickly with $|\mathcal{Z}|$ than the number $(\mathcal{C}_{|\mathcal{Z}|})^{2|\mathcal{Z}|}$ of possible $\alpha$ vectors, as shown in Appendix Table \ref{table:complexity}.

\subsection{Overview of computational results}

Table \ref{table:bruteforceresults} presents an overview of results of the two algorithms for settings with $|\mathcal{Z}| ,|\mathcal{T}| \le 3$.\footnotemark  While the next section highlights examples from each combination $(|\mathcal{Z}|,|\mathcal{T}|)$ in detail, a full catalog of the identification results is provided in Appendix \ref{sec:catalog}. While the settings reported in Table \ref{table:bruteforceresults} are ``small'', they turn out to contain a rich structure of identification results, which varies considerably by $\mathcal{T}$ and $\mathcal{Z}$.

\begin{table}[ht!]
	\begin{center}
		\begin{tabular}{cc|ccccc}
			$|\mathcal{T}|$ & $|\mathcal{Z}|$ & \# SMs & \# BCs  & \multicolumn{2}{c}{Algorithm 2 run-time}& Algorithm 1 run-time \\
			\cmidrule(r){5-6}&&&&Initial search & Organizing/Paring \\ \hline\hline
			2 & 2 & 2 & 4 & 0.08 seconds &.06 seconds & 0.11 seconds \\
			3 & 2 & 5  & 5 & 0.08 seconds &.12 seconds & 13.9 seconds \\
			2 & 3 & 11  & 30 & 55 seconds & .20 seconds& 4.3 seconds\\
			3 & 3 & 251 & 251 & 18 minutes & 65 minutes & N/A (estimate: 22 days)\\
		\end{tabular} \\ \vspace{.25cm}
		\begin{tablenotes}\footnotesize
			\item[*] ``\# BCs'' = number of distinct binary collections, ``\# SMs'' = number of distinct maximal selection models.
			\item[*] See text below for precise definitions.
		\end{tablenotes} \caption{\label{table:bruteforceresults}}
	\end{center}
\end{table}
\noindent The third column in Table \ref{table:bruteforceresults} counts the number of distinct selection models for a given support of the instruments and treatments, that are \textit{maximal} for some binary collection. The detailed description of Algorithm 2 in Appendix \ref{app:algoone} describes how given a binary collection indexed by $\alpha \in \mathbbm{R}^{2 \cdot |\mathcal{Z}|}$ (and a choice of $t,t'$, implicit), we can define a maximal selection model $\mathcal{G}(\alpha)$ with the property that $\alpha$ continues to deliver a binary collection for $t',t$ within any smaller selection model $\mathcal{G} \subseteq \mathcal{G}(\alpha)$ that is more restrictive that $\mathcal{G}(\alpha)$.\footnotetext{Run times are with R version 4.3.2 with a 3600MHz processor (AMD Ryzen Threadripper PRO 5975WX), 128GB RAM. While Algorithm 1 is parallelized across 31 cores, Algorithm 2 computation uses a single core. Algorithm 2 is not trivial to parallelize across processors given the need to check for redundancies, but does enable Algorithm 2 to be feasibly extended to $|\mathcal{Z}|=4$ on this computer setup.}\footnote{For example, let $\mathcal{G}$ be the choice model described in Section \ref{sec:3by3} from case ii of Proposition 2 of \citet{kirkeboenleuvenmogstad}. After removing two response types from $\mathcal{G}$, a second treatment effect parameter becomes identified, which is listed under a different selection model $\mathcal{G'} \subset \mathcal{G}$ counted in Table \ref{table:bruteforceresults}.} 


The fourth column in Table \ref{table:bruteforceresults} counts the number of distinct binary collection vectors $\alpha$ for a given support of the instruments and treatments. Although a given $\alpha$ generates a valid binary collection under any $\mathcal{G} \subseteq \mathcal{G}(\alpha)$, this column only counts a given $\alpha$ one time, to avoid double counting of the same identification result. Note that for some $|\mathcal{T}|,|\mathcal{Z}|$ there are the same number of distinct identification results for conditional average treatment effects as there are distinct selection models admitting such identification results, this does not mean that exactly one treatment effect parameter is identified in any given selection model. The reason is that the selection models may be nested as described above. The preamble to Appendix \ref{sec:catalog} provides a detailed example.

\subsection{Detailed examples and new identification results} \label{sec:detailedexamples}

Appendix \ref{app:examples} makes several illustrative observations from identification results that are summarized in Table \ref{table:bruteforceresults}, and reported in full in the catalog of Appendix \ref{sec:catalog}.
	\section{Application: interaction effects in cross-randomized designs} \label{sec:empirical}

This section applies Theorems \ref{thm:suff} and \ref{thm:suff} to study the identification of complementarities between two binary treatment variables. This represents a setting in which $|\mathcal{T}|=|\mathcal{Z}| = 4$. Appendix \ref{sec:spillovers} considers a second application for this case, which studies treatment effects when there can be spillovers between pairs of observational units.

\subsection{Background and empirical practice} \label{sec:backgroundempiricalpractice}
In many experimental settings, researchers cross randomize two treatments $A$ and $B$, and investigate whether there are interaction effects between the treatments, i.e. whether the effect of receiving both $A$ and $B$ differs from the sum of the effects of each of $A$ and $B$ alone. In some such settings instrumental variables methods are not needed, because compliance is perfect or the intent-to-treat effect is the policy-relevant effect of direct interest (see e.g. \citealt{dufloetal,mbitietal}). Given randomization, intent-to-treat (ITT) effects can be straightforwardly estimated by the regression:
\begin{equation} \label{eq:ittinteraction}
	Y_i = \gamma_0+\gamma_1 \cdot \mathbbm{1}(Z_i=A)+\gamma_2 \cdot \mathbbm{1}(Z_i=B) +\gamma_3 \cdot \mathbbm{1}(Z_i=C) + \nu_i
\end{equation} 
where $Z_i=C$ indicates the treatment arm for \textit{both} treatments $A$ and $B$. Such cross-randomized experiments are often referred to as ``factorial designs''.\footnote{See \citet{crosscuts} for a review of empirical practice in factorial designs, especially regarding the $\gamma_3$ term.}

However in many factorial designs the treatment arms $Z_i \in \{A,B,C\}$ represent \textit{offers} for treatments $A$ or $B$ or both, respectively, and researchers obtain data on whether the treatments were actually received. For example, \citet{depression} study the effects of pharmacotherapy (medication) and livelihood assistance (personalized training and support around income generation), among adults with depression in Karnataka, India. Across the three treatment arms of the cross-randomized experiment, roughly 65\% of participants actually undertake pharmacotherapy (defined as attending at least one psychiatric consultation), receive livelihoods assistance (attending at least one livelihoods workshop), or both. Further many adults assigned to receive both pharmacotherapy \textit{and} assistance undertake only one of the two treatments, although they are offered both. The population studied does not typically have access to pharmacotherapy or livelihoods assistance except through the field experiment, so the non-compliance is one-sided.

To move beyond analysis that is limited to intent-to-treat effects, let us denote the possible treatments as $\mathcal{T} = \{0,A,B,C\}$, with associated potential outcomes $Y_i(t)$ for $t \in \mathcal{T}$. For example, $Y_i(0)$ is the outcome $i$ would experience with neither of the two treatments $A$ and $B$. Meanwhile the set of instrument values is $$\mathcal{Z} = \{\textrm{offered neither},\textrm{offered just A},\textrm{offered just B},\textrm{offered both}\}$$ If subjects are offered both $A$ and $B$, they may choose to take treatment $A$ only, treatment $B$ only, or both treatments $C$. Assume that treatments $A$ and $B$ are otherwise not available to participants, so non-compliance is one-sided.

For a single individual, we can say that $A$ and $B$ exhibit \textit{complementarity} if $|Y_i(C)-Y_i(0)| > |Y_i(A)-Y_i(0)| + |Y_i(B)-Y_i(0)|$. Of course, testing for complementarity at the individual is infeasible due to the fundamental problem of causal inference. Let
\begin{equation} \label{eq:h0avg}
	H_0: \mathbbm{E}[Y_i(C)-Y_i(A)-Y_i(B)+Y_i(0)] = 0
\end{equation}
instead be the two-sided hypothesis of no interaction \textit{on average}, where the interaction effect is $\{Y_i(C)-Y_i(0)\}-\{Y_i(A)-Y_i(0)\} + \{Y_i(B)-Y_i(0)\}=Y_i(C)-Y_i(A)-Y_i(B)+Y_i(0)$. Under perfect compliance, $H_0$ is equivalent to the hypothesis $\gamma_3-\gamma_1-\gamma_2 > 0$ from the ITT regression \eqref{eq:ittinteraction}. This test is employed for example by \citet{depression}, using only data on assignment and ignoring information about compliance. However, the interpretation of this test may be misleading if compliance is not perfect:
\begin{proposition} \label{prop:misleading}
	If there is imperfect compliance, the parameter $\gamma_3-\gamma_1-\gamma_2$ in Eq. \eqref{eq:ittinteraction} may be zero even when $H_0$ does not hold, and may be non-zero even when $H_0$ holds.
\end{proposition}
\noindent The intuition behind Proposition \ref{prop:misleading} is that regression \eqref{eq:ittinteraction} tells us nothing about complementarity effects among individuals who do not align their actual treatments $T_i$ with their treatment assignment $Z_i$. The result suggests that the common empirical practice of using ITT regressions (rather than focusing on treatment effects per-se) is problematic given that compliance is often known to be far from perfect. However, there are limited identification results for researchers to make use of to estimate interaction effects with imperfect compliance and effect heterogeneity \citep{blackwell2017,interactingtreatments}.

One solution is to restrict outcomes, assuming sufficient treatment effect homogeneity to get around Proposition \ref{prop:misleading}. For example, if we assume that no selection on gains (NSOG) holds, the four unconditional counterfactual means $\mathbbm{E}[Y_i(C)]$, $\mathbbm{E}[Y_i(B)]$, $\mathbbm{E}[Y_i(A)]$, and $\mathbbm{E}[Y_i(0)]$ are identified under general conditions given in Appendix \ref{app:nsog}. Identification is constructive and corresponds to the estimand of a two-stage least squares (2SLS) regression of $Y_i$ on indicators for each of the four treatments (and no constant), instrumented by indicators for each of the four treatment assignment arms.\footnote{In particular, since there are four instrument values and four treatment values, we can use a result derived in Appendix \ref{app:nsog} under NSOG, that $\mathbbm{E}[Y_i(t)] = \sum_z \Sigma^{-1}_{tz} \cdot \mathbbm{E}[Y_i\cdot \mathbbm{1}(Z_i=z)]$, provided that the matrix with entries $\Sigma_{zt} = P(Z_i=z,T_i=t)$ is invertible. Some algebra shows that this coincides with the two-stage least squares estimand mentioned above.} We can then test $H_0$ by testing $\beta_3=0$ in the equation $Y_i = \beta_0+\beta_1 \cdot \mathbbm{1}(T_i \in \{A,C\})+\beta_2 \cdot \mathbbm{1}(T_i \in \{B,C\}) +\beta_3 \cdot \mathbbm{1}(T_i=C) + \epsilon_i$, estimated using the instruments $\mathbbm{1}(Z_i=A)$, $\mathbbm{1}(Z_i=B)$, $\mathbbm{1}(Z_i=\textrm{both})$ and a constant.

Nevertheless, NSOG is a very restrictive assumption. It suggests for example that individuals do not have some knowledge of their specific gains from the various treatments that informs their selection behavior. When NSOG does not hold, \citet{interactingtreatments} detail how the 2SLS estimand $\beta_3$ generally mixes interaction effects with terms that simply reflect treatment effect heterogeneity. It is thus desirable to pursue an alternative approach that leads to an interpretable causal estimand without restricting outcomes. 

\subsection{Identifying the local average interaction effect among compliers}
We now use Theorems \ref{thm:suff} and \ref{thm:necc} to examine to what extent NSOG can be meaningfully relaxed. Ex-ante, there are $2\times 2\times 4 = 16$ response types that respect one-sided non-compliance.\footnote{Response types correspond to the choices individuals would make across three decisions: whether to take treatment A if A only is offered, B if B only is offered, and which of the four treatment combinations to take if both are offered.} However, assuming that the weak-axiom of revealed preference (WARP) holds, we obtain the additional restrictions that $\{T_{i}(\textrm{offered both})=A \implies T_{i}(\textrm{offered A})=A\}$,  $\{T_{i}(\textrm{offered both})=B \implies T_{i}(\textrm{offered B})=B\}$, $\{T_{i}(\textrm{offered both})=0 \implies T_{i}(\textrm{offered A})=T_{i}(\textrm{offered B})=0\}$, $\{T_{i}(\textrm{offered A})=A \implies T_{i}(\textrm{offered both}) \ne 0\}$, and $\{T_{i}(\textrm{offered B})=B \implies T_{i}(\textrm{offered both}) \ne 0\}$. These restrictions eliminate seven response types in total. The nine that remain are enumerated in Table \ref{table:cross}. I refer to the nine remaining response types as $\mathcal{G}^{WARP}$. $\mathcal{G}^{WARP}$ represents the weakest selection model consistent with rational choice and one-sided non-compliance in a factorial design. 

\begin{table}[h!]
	\centering \small
	\begin{tabular}{c|cccccccccc}
		\textbf{offered} $\downarrow$ &n.t. & complier & A only & B only & \color{gray}{only both} & \color{gray}{A+} & \color{gray}{B+} & \color{gray}{favor A} & \color{gray}{favor B}\\
		\hline
		neither & 0 & 0 & 0 & 0& \color{gray}{0} & \color{gray}{0} & \color{gray}{0} & \color{gray}{0} & \color{gray}{0}\\
		just A & 0 &A& A & 0 & \color{gray}{0} & \color{gray}{A} & \color{gray}{0} & \color{gray}{A}& \color{gray}{A}\\
		just B & 0 & B & 0 & B& \color{gray}{0} & \color{gray}{0} & \color{gray}{B}& \color{gray}{B} & \color{gray}{B}\\
		both & 0 & C & A & B& \color{gray}{C} & \color{gray}{C} & \color{gray}{C} & \color{gray}{A} & \color{gray}{B}\\
		\hline
	\end{tabular} \vspace{.25cm}
	\caption{Response types that satisfy WARP in the cross-randomized offer design. The columns in black correpond to the response types allowed by Proposition \ref{prop:interactionid}, while the gray columns correspond to the remaining response types that are compatible with WARP. \label{table:cross}}
\end{table}

Given a selection model $\mathcal{G}$ and a function $c: \mathcal{G} \rightarrow \{0,1\}$, let us refer to $LAIE(c):=\mathbbm{E}[Y_i(C)-Y_i(A)-Y_i(B)+Y_i(0)|c(G_i)=1]$ as the local \textit{average interaction effect} among the subgroup of $g \in \mathcal{G}$ such that $c(g)=1$. LAIEs are causal quantities like the local treatment effect parameters introduced in Section \ref{sec:oaid}, except that they involve the potential outcomes for all four treatments rather than just two. The following Proposition uses Theorems \ref{thm:suff} and \ref{thm:necc} to establish when $LAIE(c)$ is identified in a manner that does not restrict outcomes:
\begin{proposition} \label{prop:interactionid}
	Given one-sided noncompliance and WARP, a local average interaction effect parameter $LAIE(c)$ is outcome-nonrestrictive identified if and only if $c(g)=\mathbbm{1}(g=\textrm{complier})$ and $\mathcal{G} \subseteq \{\textrm{n.t.}, \textrm{complier}, \textrm{A only}, \textrm{B only}\}$.
\end{proposition}
\noindent Proposition \ref{prop:interactionid} follows from a brute-force enumeration over all of the 511 selection models $\mathcal{G} \subseteq \mathcal{G}^{WARP}$, and the $c \in \{0,1\}^{|\mathcal{G}|}$ within each of them. Theorems \ref{thm:suff} and \ref{thm:necc} along with an extension of Proposition \ref{prop:ifTEthenmeans} to parameters that involve more than two treatment states (proved in Appendix \ref{proofsec}) shows that $LAIE(c)$ is outcome non-restrictive identified iff $c \in rs(A^{[0]}) \cap rs(A^{[A]}) \cap rs(A^{[B]}) \cap rs(A^{[C]})$. Thus it is sufficient to enumerate all binary collections with $\psi = \{0,A,B,C\}$.\footnote{The search also shows that average interaction effects among compliers can also be identified in a selection model in which the ``A only'' group is traded for a group that takes treatment $A$ only when offered both, and takes neither treatment under all other treatment assignments. I omit an extended discussion of this novel result here for brevity since this final group violates WARP and the model is asymmetric with respect to the treatments (details are available upon request).}

Proposition \ref{prop:interactionid} establishes that identifying complementarities in a cross-randomized design without outcome restrictions requires substantive restrictions on selection: many of the response types in $\mathcal{G}^{WARP}$ not included in $\{\textrm{n.t.}, \textrm{complier}, \textrm{A only}, \textrm{B only}\}$ are ex-ante plausible. For example, let $U_i(t)$ denote the interpret $U_i(t)$ as the net utility of treatment $t \in \mathcal{T}$ relative to no treatment for individual $i$ (thus normalizing $U_i(0)=0$). Without loss of generality, consider a random coefficients form for the utility function: $U_i(t) = \pi_{Ai} \cdot \mathbbm{1}(t=A)+\pi_{Bi} \cdot \mathbbm{1}(t=B)+\pi_{Ci} \cdot \mathbbm{1}(t=C)$. If the vector $\pi_i = (\pi_{Ai},\pi_{Bi},\pi_{Ci})'$ has support in an open neighborhood of the origin in $\mathbbm{R}^3$, all nine groups from Table \ref{table:cross} will be present in the population. 

Appendix \ref{sec:sepchoice} shows that we can rationalize the restriction made in Proposition \ref{prop:interactionid} by supposing that individuals choose \textit{separately} whether to receive treatment $A$ or $B$, rather than as a single joint decision. That is, individuals choose as if they evaluate the costs and benefits of each treatment $A$ or $B$ separately, and choose all treatments offered to them for which benefits outweigh costs. For this reason, let us denote the largest selection model in which the local average interaction effect among compliers is identified as $\mathcal{G}^{sep}:=\{\textrm{n.t.}, \textrm{complier}, \textrm{A only}, \textrm{B only}\}$. Given one-sided noncompliance, the selection model $\mathcal{G}^{sep}$ is also equivalent to what \citet{blackwell2017} calls a ``treatment exclusion'' restriction that the instrument for treatment $A$ does not affect uptake of treatment $B$ (and vice versa). \citet{blackwell2017} shows that in this case the interaction coefficient $\beta_3$ from a 2SLS regression identifies the local average interaction effect among compliers. Proposition \ref{prop:interactionid} shows that treatment exclusion is furthermore \textit{necessary} to identify this parameter without restricting outcomes. Given WARP and one-sided noncompliance, treatment exclusion cannot be relaxed without restricting outcomes.\\


\textbf{Estimating treatment effects and interaction effects among compliers:}  In the language of Section \ref{sec:posi}, the positive side of Proposition \ref{prop:interactionid} can be understood as the existence of a binary collection $\{(t,\alpha^{[t]})\}_{t \in \psi}$ in which $\psi = \mathcal{T} = \{0,A,B,C\}$. The function $c$ describing this binary collection is $c(g)=\mathbbm{1}(g=\textrm{complier})$, or in vector notation $c' = (0,1,0,0)'$. Let us denote the functions $\alpha^{[t]}(z)$ in vector form as $\alpha_t$, in which case
\begin{align}
	\alpha_0 = (1, -1, -1, 1 )', \quad \alpha_A = (0, 1, 0, -1 )', \quad \alpha_B = (0, 0, 1, -1 )', \quad \alpha_C = (0, 0, 0, 1 )' \label{eq:bcs} 
\end{align}
One can verify directly that for each $t \in \mathcal{R}$, $\alpha_t'A^{[t]} = (0,1,0,0)$ where the matrix $A$ is defined from the first four columns of Table \ref{table:cross}. For brevity, let us denote $LAIE((0,1,0,0)') = \mathbbm{E}[Y_i(C)-Y_i(A)-Y_i(B)+Y_i(0)|g = \textrm{complier}]$ as simply LAIE (with this $c$ implicit).

The binary collection \eqref{eq:bcs} implies a cumbersome expression for $LAIE$, but some simplification shows that $LAIE=\theta^{ITT}/p$, where $p=P(G_i=\textrm{complier})$ and $\theta^{ITT}:=\gamma_3-\gamma_1-\gamma_2$ is the measure of average complimentary from the intent-to-treat regression Eq. \eqref{eq:ittinteraction}. This delivers the following important consequence of Proposition \ref{prop:interactionid}:
\begin{corollary} \label{corr:itt}
	Given $\mathcal{G} \subseteq \mathcal{G}^{sep}$, the sign of the local average interaction effect among compliers $LAIE$ is the same as $\gamma_3-\gamma_1-\gamma_2$ from the  intent-to-treat regression \eqref{eq:ittinteraction}.
\end{corollary}
\noindent The algebra that leads to $LAIE=\theta^{ITT}/p$ is given in Appendix \ref{sec:ittexpression}, where it is also extended to the case in which covariates are included in Eq. \eqref{eq:ittinteraction}.

While the ITT condition $\gamma_3-\gamma_1-\gamma_2$ cannot be used to test the hypothesis of overall unconditional complementarity (without outcome restrictions), it can by Propisition \ref{corr:itt} be used to test for the sign of local average interaction effect among compliers (the only group for whom interaction effects can be point identified in any way without restricting outcomes). This latter interpretation requires no outcome restrictions, but instead the non-trivial selection model $\mathcal{G} \subseteq \mathcal{G}^{sep}$. Corollary \ref{corr:itt} formally justifies the test for complementarity used by \citet{depression} within this selection model. 

Note finally that the binary collection \eqref{eq:bcs} implied by the selection model $\mathcal{G} \subseteq \mathcal{G}^{sep}$ yields three overidentification restrictions for the share of compliers:
\begin{align}
	p&:=P(T_i=C|Z_i=\textrm{both}) \label{eq:fourps} \\
	&=P(T_i=A|Z_i=\textrm{just A})-P(T_i=A|Z_i=\textrm{both}) \nonumber \\
	&=P(T_i=B|Z_i=\textrm{just B})-P(T_i=B|Z_i=\textrm{both}) \nonumber \\
	&=1+P(T_i=0|Z_i=\textrm{both})-P(T_i=0|Z_i=\textrm{just A})-P(T_i=0|Z_i=\textrm{just B}) \nonumber
\end{align}
for some value $p \in [0,1]$ which identifies $P(G_i = \textrm{complier})$. This testable implication is new to the literature and can be used to assess the substantive assumption $\mathcal{G} \subseteq \mathcal{G}^{sep}$.\footnote{These are stronger than testable implications mentioned by \citet{blackwell2017}, which give $P(A\textrm{ or }C|\textrm{both})=P(A|\textrm{just A})$ and $P(B\textrm{ or }C|\textrm{both})=P(B|\textrm{just B})$ in the case of one-sided noncompliance. Those do not imply the last line of Eq. \eqref{eq:fourps}.} 

\subsection{Empirical application}
I use the replication data from \citet{depression} to implement the above findings empirically. First, we assess the testable implications of $\mathcal{G} \subseteq \mathcal{G}^{sep}$. Unable to reject the over-identifying restrictions, I then estimate $LAIE$ following Proposition \ref{prop:interactionid}. 

In Appendix \ref{sec:star} I consider a second empirical setting, from \citet{angristlangoreopoulos}, in which students were cross-randomized into academic support and financial incentives for good grades. In that setting, I find that the testable implications of $\mathcal{G} \subseteq \mathcal{G}^{sep}$ are rejected, and therefore local average interaction effect parameters cannot be identified absent restrictions on outcomes. This illustrates that the general over-identifying restrictions highlighted in Section \ref{sec:testable} have power in an empirically relevant way.

Since the experiment reported in \citet{depression} stratifies randomization into nine strata (defined by district and terciles of a village poverty index), the implementation below requires some extensions to the basic results of this paper that allow randomization to hold based on observed covariates $X_i$. Implementation is described in Appendix \ref{app:additionalinteraction}. For simplicity, conditional expectations that need to be estimated are assumed to be additively separable between instruments $Z_i$ and indicators for the strata $X_i$.\\

\noindent \textbf{Testing the overidentification restrictions:}
Each of the four expressions the proportion $p$ of compliers in Eq. \eqref{eq:fourps} can be estimated using regressions of the various $D_i^{[t]}$ on instrument indicators as well as indicators for strata $X_i$. The extension of \eqref{eq:fourps} to the case with strata fixed effects is given in Appendix \ref{sec:covsempirical}. Following \citet{depression}, I use cluster robust inference by village (the level of treatment assignment).

The point estimates for $p:=P(G_i=\textrm{complier})$ are $36.7\%$, $40.2\%$, $39.7\%$, and $43.7\%$, respectively. A chi-squared test for equality of all four estimates of $p$ can be implemented using standard seemingly unrelated regression routines, and yields a p-value of 65\%. This indicates that we cannot reject these overidentification restrictions at all conventional levels. This provides some initial evidence in favor of the choice model $\mathcal{G} \subseteq \mathcal{G}^{sep}$. 

However, the equality restrictions \eqref{eq:fourps} are not the only observable implications of $\mathcal{G} \subset \mathcal{G}^{sep}$. In Appendix \ref{sec:lp}, I describe how all of the observable first-stage information can be aggregated into a system of linear equations $\mathcal{A}x = \beta$, where $\mathcal{A}$ is a known matrix defined from the $A^{[t]}$, $\beta$ is a vector of observed treatment choice probabilities, and $x$ is a vector of the (non-negative) unobserved occupancies $x_g = P(G_i=g)$ of each response type. Maintaining the weaker assumption of $\mathcal{G} \subseteq \mathcal{G}^{WARP}$, we can test whether $\mathcal{G}$ is furthermore a subset of $\mathcal{G}^{set}$ by computing a lower bound on the sum of the components of $x_g$ for $g \in \mathcal{G}^{WARP}-\mathcal{G}^{sep}$, subject to the constraints that $\mathcal{A}x = \beta$, each $x_g \ge 0$ and the $x_g$ sum to unity. In principle, this exercise could be implemented by strata to test $\mathcal{G} \subseteq \mathcal{G}^{sep}$ among the individuals within each. To increase statistical power given the small sample, I pool the data across all strata for this exercise. This is valid under the assumption that the response-type distribution is common across strata.\footnote{Note that this resrtiction does not require \textit{potential outcomes} to be uncorrelated with stratum.}

Ignoring sampling uncertainty in the observed treatment choice probabilities, the data suggest that $P(G_i \in \mathcal{G}^{WARP}-\mathcal{G}^{sep})$ is at least $6.3\%$ (and is no more than $80.8\%$). Since $6.3\% > 0$, this provides some evidence against the restriction $\mathcal{G} \subseteq \mathcal{G}^{sep}$.\footnote{Point estimates further yield $p_{oboth} \in [0,3\%]$, $p_{A+} \in [0,3\%]$, $p_{B+} \in [0,34\%]$, $p_{favor A} \in [3\%,6\%]$, $p_{favor B} \in [3\%,38\%]$.} However, accounting for uncertainty in $\beta$ suggests that this lower bound for $P(G_i \in \mathcal{G}^{WARP}-\mathcal{G}^{sep})$ is not statistically significant. \citet{fsst} provide a method for testing whether there exists componentwise non-negative solutions $x$ to systems of the form $\mathcal{A}x=\beta$ like the above, when $\beta$ is estimated from the data. This method yields a 95\% confidence interval of $[0, 0.83]$ for the share of offending response types. This confidence interval includes zero (up to machine-precision), indicating that we cannot reject that $\mathcal{G} \subseteq \mathcal{G}^{sep}$ within the weaker assumption $\mathcal{G} \subseteq \mathcal{G}^{WARP}$, even using the full observable information on treatment uptake. Appendix \ref{sec:lp} provides details on this procedure.\footnote{I implement the FSST method using the $\mathtt{R}$ package $\mathtt{lpinfer}$. This method does not involve any clustering and is designed for $i.i.d.$ data, so the confidence interval reported above may undercover the parameter $P(G_i \in \mathcal{G}^{WARP}-\mathcal{G}^{sep})$ if one considers uncertainty as arising from treatment assignment as well. Since the proportion of each cluster (in this case village) that is sampled is small (on average about two individuals), results for OLS suggest that the influence of clustering in treatment assignment may be minimal, even considering both uncertainty arising from clustered treatment assignment as well as sampling \citep{abadieetalclustering}. Similar results are provided using alternative methods for inference on linear systems introduced by \citet{romanoshaikh} and \citet{chorussell}.}\\ 

\noindent \textbf{Estimates of local average interaction among compliers:} 
The data from \citet{depression} follow 1,000 respondents over five survey waves. I use their main outcome variable, which is a standardized version of the PHQ-9 score for depression, with higher values indicating more severe depression. I focus on longer-run outcomes in the fourth and fifth waves, which occured between one and two years after treatment. In these longer waves, the authors estimate $\gamma_3-\gamma_1-\gamma_2$ from the ITT regression to be marginally significant at the 10\% level with a p-value of $.10$. Meanwhile, they find that the combination  (treatment ``C'') of pharmacotherapy (treatment ``A'') and livelihoods assistance (treatment ``B'') reduces depression symptoms even after the intervention that is significant at the 95\% level, while the effects of treatments A or B alone are insignificant (cf. their Table 2, panel B). However, these estimates come from an ITT regression that ignores actual treatment uptake, and may be attenuated or otherwise distorted when interpreted as effects of the treatments themselves rather than as effects of assignment.

\begin{table}[h!]
	\centering
	{
\def\sym#1{\ifmmode^{#1}\else\(^{#1}\)\fi}
\begin{tabular}{l*{4}{c}}
\hline\hline
                                        &\multicolumn{1}{c}{(1)}&\multicolumn{1}{c}{(2)}&\multicolumn{1}{c}{(3)}&\multicolumn{1}{c}{(4)}\\
                                        &\multicolumn{1}{c}{ITT}&\multicolumn{1}{c}{2SLS}&\multicolumn{1}{c}{Binary collections}&\multicolumn{1}{c}{GMM}\\
\hline
$\mathbbm{E}[Y(C)-Y(0)|c(G)=1]$         &   -0.277\sym{**} &   -0.676\sym{**} &   -0.489         &   -0.291\sym{*}  \\
                                        & (0.0846)         &  (0.224)         &  (0.253)         &  (0.115)         \\
$\mathbbm{E}[Y(A)-Y(0)|c(G)=1]$         &  -0.0468         &  -0.0946         &  -0.0699         &  -0.0922         \\
                                        & (0.0759)         &  (0.153)         &  (0.211)         &  (0.213)         \\
$\mathbbm{E}[Y(B)-Y(0)|c(G)=1]$         &  -0.0282         &  -0.0355         &    0.134         &    0.247         \\
                                        & (0.0786)         &  (0.100)         &  (0.243)         &  (0.167)         \\
\hline
Local avg. interaction effect           &   -0.202         &   -0.545         &   -0.553         &   -0.504         \\
H0: no interaction (p-val)          &    0.101         &   0.0814         &    0.101         &    0.146         \\
c(G)                                    &compliers/all         &all individuals         &compliers         &compliers         \\
p(c(G)=1)                               &        1         &        1         &       .4         &       .4         \\
Identifying assumption                  &perfect compliance         &     NSOG         &$\mathcal{G} \subseteq \mathcal{G}^{sep}$         &$\mathcal{G} \subseteq \mathcal{G}^{sep}$         \\
\hline
Sample size                             &     1650         &     1650         &     1650         &     1650         \\
\hline\hline
\multicolumn{5}{l}{\footnotesize Standard errors in parentheses}\\
\multicolumn{5}{l}{\footnotesize \sym{*} \(p<0.05\), \sym{**} \(p<0.01\), \sym{***} \(p<0.001\)}\\
\end{tabular}
}

	\caption{Treatment effects and interaction effect estimates, where $A$ is pharmacotherapty (``PC''), treatment $B$ is livelihoods assistance (``LA''), and treatment C is receiving both. Outcome variable is the PHQ-9 depression score, expressed in units of its sample standard deviation. ITT estimates could be interpreted as effects of receiving treatment under perfect compliance, though this is rejected by the data. 2SLS estimates assume no-selection-on-gains (NSOG). Estimators for local effects among compliers via binary collections are described in Appendix \ref{sec:covsempirical}, and the GMM refinement is described in Appendix \ref{sec:gmm}. All columns include strata controls and cluster standard errors by village.} \label{table:interaction}
\end{table}

Column (1) of Table \ref{table:interaction} implements this ITT regression of the outcome on instrument indicators (and strata fixed effects). Departing slightly from \citet{depression}, I focus on a minimal specification and do not control for baseline values of the outcome as they do. However, the findings are qualitatively the same and quantitatively similar. In line with \citet{depression} only the effect of treatment $C$ (pharmacotherapy and livelihoods assistance) is statistically significant. Column (2) uses data on treatment uptake and implements a 2SLS regression as described in Section \ref{sec:backgroundempiricalpractice}. Consistent with intuition given imperfect compliance, these treatment effects estimates are larger in magnitude and have the same pattern of significance. However, the main treatment effect estimates (besides the LAIE) invoke the strong assumption of no-selection-on-gains (NSOG) to be interpreted causally and as averaging over the same response types.\footnote{Thm. 2 of \citet{blackwell2017} shows how the various 2SLS coefficients average effects over different groups of response types even given $\mathcal{G} \subseteq \mathcal{G}^{sep}$, making them not comparable to one another without outcome restrictions like NSOG.}

By contrast, none of the estimates reported in Columns (3) and (4) require restrictions on outcomes to be causally interpreted and compared. Column (3) uses simple sample estimators of the expectations from Eq. \eqref{idresult} (extended for strata fixed effects) along with the binary combinations Eq. \eqref{eq:bcs} that isolate compliers. See Appendix \ref{sec:covsempirical} for details. Column (4) re-estimates Column (3) while further imposing the overidentification restrictions \eqref{eq:fourps} for the share of compliers, using a generalized method of moments (GMM) estimator (see Appendix \ref{sec:gmm} for details). This estimator penalizes the differing numerical estimates of $p=P(i \textit{ is complier})$ in each of the four terms that make up $LAIE$. While the GMM estimator does not end up reducing the standard error of $LAIE$ in this setting, it does restore statistical significance to $\mathbbm{E}[Y_i(C)-Y_i(0)|i \textit{ is complier}]$ at the 95\% level.

The three estimates of $LAIE$ in columns (3)-(5) are valid under the same assumption that $\mathcal{G} \subseteq \mathcal{G}^{sep}$, and suggest that pharmacotherapy and livelihoods assistance are complementary: they have an interaction effect of about half of a standard deviation of PHQ-9 among compliers. 2SLS provides the most precise estimate of this parameter, which is signficant at the 10\% level. The relatively large positive magnitude of the effect of livelihoods assistance (treatment B) in columns (3) and (4) raises the question of whether this intervention may in fact exacerbate depression symptoms among compliers, when it is not accompanied by pharmacotherapy (treatment A). This economically meaningful magnitude is not evident in the ITT estimates from \citet{depression} that do not adjust for non-compliance. However, the estimate is not quite significant at the 10\% level even with the GMM estimator (t-statistic $2.47/1.67=1.48$), and thus should be interpreted with caution. Otherwise, the  estimates reported in Table \ref{table:interaction} confirm the qualitative findings of \citet{depression}, while offering quantitative treatment effects that account for the partial compliance. 
	
	\section{Conclusion}
	This paper has formalized the notion of ``outcome-nonrestrictive'' identification in IV models, and shown that it is equivalent (with discrete instruments) to the existence of linear combinations of counterfactual treatment indicators that add up to zero or one for all response types in the assumed selection model. A selection model only allows for treatment effects to be identified in an outcome-nonrestrictive way when a particular matrix that summarizes the selection behavior allowed by the model with respect by two treatment values has a non-trivial null-space that intersects the unit cube in the space of types allowed by the model. This insight yields a systematic approach to enumerating all selection models that afford identification of treatment effects in a manner that does not restrict outcomes. The search delivers a multiplicity of new identification results, despite its computational complexity scaling rapidly with the size of support of the instruments and treatment. Future work could leverage the algorithms proposed in this paper to construct a searchable database of identification results for still more complex settings.
	
	\printbibliography
	\nocite{depressiondata}
	
	\begin{appendices}

		\section{Proofs} \label{proofsec}

\subsection{Proof of Proposition \ref{prop:ifTEthenmeans}}
To ease notation, write $\Delta_c^{t,t'}$ as $\Delta$, $\mu_c^t$ as $\mu(t)$, and $\mu_c^{t'}$ as $\mu(t')$, with $c$ fixed. It is apparent that if $\mu(t')$ and $\mu(t)$ are outcome-nonrestrictive identified, then $\Delta = \mu(t')-\mu(t)$ is too.

Now let us consider the other direction. Suppose that $\mu(t)$ is not outcome-nonrestrictive identified (an analogous argument holds if $\mu(t')$ is not outcome-nonrestrictive identified). Then for some $\mathcal{P}_{obs} \in \mathscr{P}_{obs,c}(\mathcal{G})$, the set $\{\theta_{\mu(t)}(\mathcal{P}): \mathcal{P} \in M \textrm{ and } \phi(\mathcal{P}) = \mathcal{P}_{obs}\}$ has at least two elements, where $M:=\{\mathcal{P}_{latent} \times \mathcal{P}_{Z}: \mathcal{P}_{latent} \in \mathscr{P}_{latent,c}(\mathcal{G}), \mathcal{P}_{Z} \in \mathscr{P}_{Z}\}$ and we let $\theta_{\mu(t)}(\cdot)$ be the map that yields the value of $\mu(t)$ as a function of $\mathcal{P}$.\footnote{Note that the set $M$ will vary with $c$, but since we are considering a fixed $c$ this is left implicit to ease notation.} Accordingly, let $\mathcal{P}_1,\mathcal{P}_2 \in M$ where $\theta_{\mu(t)}(\mathcal{P}_1)=a$ and $\theta_{\mu(t)}(\mathcal{P}_2)=b$ where $a \ne b$ despite $\phi(\mathcal{P}_1)=\phi(\mathcal{P}_2)=\mathcal{P}_{obs}$. 

Let us decompose $\mathcal{P}_1$ as $\left( \left\{\mathcal{P}_{Y(s)|G=g}\right\}_{\substack{s \in \mathcal{T}\\ g \in \mathcal{G}}}, \mathcal{P}_G, \mathcal{P}_Z\right)$, which is possible because $\mathcal{P}_1$ satisfies independence Eq. \eqref{eq:independence} between the instruments and the latent variables. Let $\mathcal{P}(0)$ denote a degenerate distribution at zero in $\mathbb{R}$. Now consider the distribution $\tilde{\mathcal{P}}_1=\left( \{\mathcal{P}(0)\}_{g \in \mathcal{G}}, \left\{\mathcal{P}_{Y(t)|G=g}\right\}_{\substack{s \in \mathcal{T}, s \ne t'\\ g \in \mathcal{G}}}, \mathcal{P}_G, \mathcal{P}_Z\right)$. That is, $Y_i(t')=0$ with probability one according to $\tilde{\mathcal{P}}_1$, but the joint distribution of $Z_i$, $G_i$ and all of the other potential outcomes $s \ne t'$ are the same under $\tilde{\mathcal{P}}_1$ as they are under $\mathcal{P}_1$. Note that given this construction: $\theta_{\mu(t)}(\tilde{\mathcal{P}}_1)=\theta_{\mu(t)}(\mathcal{P}_1)=a$, since $\mu(t)$ only depends on the distributions $\mathcal{P}_{Y(t)|G=g}$ and $\mathcal{P}_G$, and $t \ne t'$. Note as well that from $\mathcal{P}_1 \in M$ we know that $\mathcal{P}_{latent}(\mathcal{P}_1) \in \mathscr{P}_{latent,c}(\mathcal{G})$. Since $\mathcal{P}_G$ has not been changed in defining $\tilde{\mathcal{P}}_1$ from $\mathcal{P}_1$, and a degenerate random variable at zero has a finite expectation, it follows that $\mathcal{P}_{latent}(\tilde{\mathcal{P}}_1) \in \mathscr{P}_{latent,c}(\mathcal{G})$ as well. Since $\mathcal{P}_Z$ has also not been changed, we further have that $\tilde{\mathcal{P}}_1 \in M$. 

Define $\tilde{\mathcal{P}}_2$ analogously from $\mathcal{P}_2$, and observe that similarly $\theta_{\mu(t)}(\tilde{\mathcal{P}}_2)=\theta_{\mu(t)}(\mathcal{P}_2)=b$ and again that $\tilde{\mathcal{P}}_2 \in M$.

Observe furthermore that $\theta_\Delta(\tilde{\mathcal{P}}_1)=\theta_{\mu(t')}(\tilde{\mathcal{P}}_1)-\theta_{\mu(t)}(\tilde{\mathcal{P}}_1)=0-a=-a$, and similarly $\theta_\Delta(\tilde{\mathcal{P}}_2)=\theta_{\mu(t')}(\tilde{\mathcal{P}}_2)-\theta_{\mu(t)}(\tilde{\mathcal{P}}_2)=0-b=-b$. Thus since since $b \ne a$:
\begin{equation} \label{eq:noteq}
	\theta_\Delta(\tilde{\mathcal{P}}_1) \ne \theta_\Delta(\tilde{\mathcal{P}}_2)
\end{equation}
I now show that this contradicts $\Delta$ being outcome-nonrestrictive identified.

To see this, decompose $\mathcal{P}_{obs}$ as $\left( \left\{\mathcal{P}_{Y|T=s,Z=z}\right\}_{\substack{s \in \mathcal{T}\\ z \in \mathcal{Z}}}, \left\{\mathcal{P}_{T|Z=z}\right\}_{z \in \mathcal{Z}}, \mathcal{P}_Z\right)$ and define $\tilde{\mathcal{P}}_{obs}=\left(\{\mathcal{P}(0)\}_{z \in \mathcal{Z}},\left\{\mathcal{P}_{Y|T=s,Z=z}\right\}_{\substack{s \in \mathcal{T}, s \ne t'\\ z \in \mathcal{Z}}}, \left\{\mathcal{P}_{T|Z=z}\right\}_{z \in \mathcal{Z}}, \mathcal{P}_Z\right)$ where the $\{\mathcal{P}(0)\}_{z \in \mathcal{Z}}$ indicate that $P(Y_i=0|T_i=t',Z_i=z)=1$ for all $z \in \mathcal{Z}$ according to $\tilde{\mathcal{P}}_{obs}$. That is, the marginal distribution $\mathcal{P}_{TZ}$ and the conditional distributions $\mathcal{P}_{Y|T=s,Z=z}$ for all $s \ne t'$ and $z$ are unchanged from $\mathcal{P}_{obs}$, but $Y_i=0$ with probability one conditional on $T_i=t'$.

The next step is to observe that $\phi(\tilde{\mathcal{P}}_1)=\tilde{\mathcal{P}}_{obs}$ and $\phi(\tilde{\mathcal{P}}_2)=\tilde{\mathcal{P}}_{obs}$. To see this, note that $Y_i(t')=0$ with probability one implies that $Y_i=0$ with probability one conditional on $T_i=t'$ (provided that $P(T_i=t)>0$). Now since $\mathcal{P}_1$ and $\tilde{\mathcal{P}}_1$ only differ in $\mathcal{P}_{Y(t')|G=g}$ (leaving $\mathcal{P}_{TZ}$ and and $\mathcal{P}_{Y|T_i=s,Z=z}$ for all $s \ne t'$ and $z$ unchanged), it follows from $\phi(\mathcal{P}_1)=\mathcal{P}_{obs}$ that $\phi(\tilde{\mathcal{P}}_1)=\tilde{\mathcal{P}}_{obs}$, and analogously for $\tilde{\mathcal{P}}_2$. This further implies that $\tilde{\mathcal{P}}_{obs} \in \mathscr{P}_{obs,c}(\mathcal{G})$. 

Since $\Delta$ is outcome-nonrestrictive identified and $\tilde{\mathcal{P}}_{obs} \in \mathscr{P}_{obs,c}(\mathcal{G})$, the set $\{\theta_\Delta(\mathcal{P}): \mathcal{P} \in M \textrm{ and } \phi(\mathcal{P})=\tilde{\mathcal{P}}_{obs}\}$ must be a singleton. Given that $\phi(\tilde{\mathcal{P}}_1)=\phi(\tilde{\mathcal{P}}_2)=\tilde{\mathcal{P}}_{obs}$ and $\tilde{\mathcal{P}}_1, \tilde{\mathcal{P}}_2 \in M$ we must then have $\theta_{\Delta}(\tilde{\mathcal{P}}_1)=\theta_{\Delta}(\tilde{\mathcal{P}}_2)$. This yields a contradiction with \eqref{eq:noteq}. 

We can generalize Proposition \ref{prop:ifTEthenmeans} as follows. For any vector of coefficients $\rho_t$ for each $t \in \mathcal{T}$, define $\theta^\rho_c:=\sum_t \rho_t \cdot \mu^t_c $. Clearly $\mu^{t}_c$ is a special case of $\theta_\rho^c$ in which $\rho_t$ is equal to one for a single treatment, and zero for all others. Similarly, $\Delta^{t,t'}_c$ is a case of $\theta^\rho_c$ in which $\rho_{t'} = 1$, $\rho_t = -1$, and all other components of $\rho$ are equal to zero. In Section \ref{sec:empirical}, the local average complimentarity parameter $\lambda_c = \mu_{c}^C-\mu_{c}^A-\mu_{c}^B+\mu_{c}^0$ is an example of $\theta_c^\rho$ where $\rho_C=\rho_0 = 1$ and $\rho_A=\rho_B = -1$.

In general, let $\psi(\rho) \subseteq \mathcal{T}$ be the set of treatments for which $\rho_t \ne 0$. Clearly $\theta^\rho_c$ is outcome-nonrestrictive identified if $\mu_c^t$ is for each $t \in \psi(\rho)$. The above argument articulated for treatment effects extends immediately to show that $\theta^\rho_c$ is also outcome-nonrestrictive identified \textit{only }if $\mu_c^t$ is for each $t \in \psi(\rho)$. To see this, we again begin with a value $t \in \psi(\rho)$ such that $\mu(t)$ is not outcome-nonrestrictive identified, i.e. $\theta_{\mu(t)}(\mathcal{P}_1)=a$ and $\theta_{\mu(t)}(\mathcal{P}_2)=b$ with $a \ne b$, where $\mathcal{P}_1$ and $\mathcal{P}_2$ are the corresponding latent variable distributions in $M$ such that $\phi(\mathcal{P}_1)=\phi(\mathcal{P}_2)=\mathcal{P}_{obs}$. Let $d_1 = \sum_{s\ne t} \rho_s \cdot \theta_{\mu(s)}(\mathcal{P}_1)$ and $d_2 = \sum_{s\ne t} \rho_s \cdot \theta_{\mu(s)}(\mathcal{P}_2)$ such that $\theta^\rho_c = \rho_t \cdot a + d_1$ under $\mathcal{P}_1$ and $\theta^\rho_c = \rho_t \cdot b + d_2$ under $\mathcal{P}_2$.

Suppose that $\theta^\rho_c$ is outcome-nonrestrictive identified. In this case, we must have that $d_2 = d_1+ \rho_t \cdot (a-b)$. Now consider the distributions $\tilde{P}_1$, $\tilde{P}_2$ and $\tilde{P}_{obs}$ defined above, where we take $t' \ne t$ to be any other treatment in $\psi(\rho)$ other than $t$. We have already seen above that $\tilde{\mathcal{P}}_{obs} \in \mathscr{P}_{obs,c}(\mathcal{G})$, $\tilde{\mathcal{P}}_1,\tilde{\mathcal{P}}_2 \in M$ and  $\phi(\tilde{\mathcal{P}}_1)=\phi(\tilde{\mathcal{P}}_2)= \tilde{\mathcal{P}}_{obs}$. Thus we must have that $\theta^\rho_c$ is the same under both $\tilde{\mathcal{P}}_1$ and $\tilde{\mathcal{P}}_2$. Instead, we have that under $\tilde{\mathcal{P}}_1$, $\theta^\rho_c$ is equal to $\rho_t \cdot a + d_1 - \rho_{t'} \cdot \theta_{\mu(t')}(\mathcal{P}_1)$, and under $\tilde{\mathcal{P}}_1$, $\theta^\rho_c$ is equal to
$$\rho_t \cdot b + d_2 - \rho_{t'} \cdot \theta_{\mu(t')}(\mathcal{P}_2) = \rho_t \cdot b + d_1+ \rho_t \cdot (a-b) - \rho_{t'} \cdot \theta_{\mu(t')}(\mathcal{P}_2) = \{\rho_t \cdot a + d_1\}- \rho_{t'} \cdot \theta_{\mu(t')}(\mathcal{P}_2)$$
Thus we must have that $\theta_{\mu(t')}(\mathcal{P}_2) = \theta_{\mu(t')}(\mathcal{P}_1)$. This argument can be repeated for every $t' \in \psi(\rho)$, $t' \ne t$, and we then have that $d_1=d_2$. This in turn implies that $\rho_t \cdot (a-b) = 0$, which contradicts $a \ne b$ with $\rho_t \ne 0$. We have thus arrived at a contradiction.

\subsection{Proof of Theorem \ref{thm:necc}}
\subsubsection*{Setup and notation}

Let $\mathcal{Y} \subseteq \mathbbm{R}$ be the support of $Y$. For any $y \in \mathcal{Y}, z \in \mathcal{Z}$ and $t \in \mathcal{T}$, define $F_{(YD)|Z=z}(y,t) := \mathbbm{E}[\mathbbm{1}(Y_i \le y)\mathbbm{1}(T_i=t)|Z_i=z]$. This function acts like a CDF for $Y_i$ and a probability mass function for $T_i$, conditional on $Z_i=z$. We begin with the observation that knowing the distribution $\mathcal{P}_{obs}$ of $(Y_i,T_i,Z_i)$ is equivalent to knowing the value of $F_{(YD)|Z=z}(y,t)$ for all $(y,t,z)$ along with the observable distribution of the instruments $\mathcal{P}_{Z}$.

By the law of iterated expectations over $G_i$ and using independence \eqref{eq:independence}: \begin{align}&F_{(YD)|Z=z}(y,t) = \mathbbm{E}\left\{\mathbbm{E}[\mathbbm{1}(Y_i(t) \le y)\mathbbm{1}(T_i(z)=t)|Z_i=z,G_i]\right\} \nonumber \\&= \sum_{g: A^{[t]}_{zg}=1} P(G_i=g)\cdot\mathbbm{E}[\mathbbm{1}(Y_i(t) \le y)|G_i=g] \nonumber\\
	&= \sum_{g: A^{[t]}_{zg}=1} P(G_i=g)\cdot F_{Y(t)|G=g}(y): = \sum_{g \in \mathcal{G}} A^{[t]}_{zg} \cdot P(G_i=g) \cdot F_{Y(t)|G=g}(y) \label{Aeq}  \end{align}

\noindent I use the following Lemma to assume that $A^{[t]}$ has full row rank, without loss of generality:
\begin{lemma} \label{lemma:fullrowrank}
	If $\mu_g^t$ is outcome-nonrestrictive identified given instrument support $\mathcal{Z}$, it remains outcome-nonrestrictive identified using data from $Z_i \in \mathcal{Z}_0$, where $\mathcal{Z}_0 \subseteq \mathcal{Z}$ is a subset of instrument values for which the rows of $A^{[t]}$ for $z \in \mathcal{Z}_0$ are linearly independent of one another.
\end{lemma}
\noindent A special case of Lemma \ref{lemma:fullrowrank} is an observation by \citet{Heckman2018} that one can remove any rows of $A^{[t]}$ that is an exact copy of another row (i.e. there are two instrument values for which all response types behave the same regarding whether they choose treatment $t$ or not), and there is hence a direct redundancy over instrument values.

\subsubsection*{Outcome-nonrestrictive identification}
Now define $\mathbf{F}_{(YD)|Z}(y)$ to be a $|\mathcal{T}|\cdot|\mathcal{Z}| \times 1$ vector of $F_{(YD)|Z=z}(y,t)$ over $z$ and $t$ and $\mathbf{G}^*(y)$ to be the unknown $|\mathcal{T}|\cdot |\mathcal{G}|$-component vector of $P(G_i=g)\cdot F_{Y(t)|G=g}(y)$ over $g$ and $t$, for a fixed $y$. Now let $\mathbf{G}^*$ represent the whole vector-valued function $\mathbf{G}^*: \mathcal{Y} \rightarrow \mathbbm{R}^{|\mathcal{T}|\cdot|\mathcal{G}|}$, and define $\mathbf{F}_{(YD)|Z}$ similarly as the function $\mathcal{Y} \rightarrow \mathbbm{R}^{|\mathcal{T}|\cdot|\mathcal{Z}|}$ yielding the vector $\mathbf{F}_{(YD)|Z}(y) $. Note that $\mathcal{P}_Z$ and $\mathbf{F}_{(YD)|Z}$ encode the entire distribution $\mathcal{P}_{obs}$ of observables $(Y,T,Z)$ while $\mathcal{P}_Z$ and $\mathbf{G}^*$ encode the entire distribution $\mathcal{P}$ of model primitives $(\tilde{Y},G,Z)$.

The relationship between the two can be characterized by writing Eq. (\ref{Aeq}) as: 
\begin{equation} \label{Aeq2} 
	\mathbf{F}_{(YD)|Z} = \mathcal{A} \circ \mathbf{G}^*
\end{equation}
where $\mathcal{A}$ is the linear map of functions  $\mathcal{Y} \rightarrow \mathbbm{R}^{|\mathcal{T}|\cdot |\mathcal{G}|}$ to functions $\mathcal{Y} \rightarrow\mathbbm{R}^{|\mathcal{T}|\cdot|\mathcal{Z}|}$ defined by:
$$ \left[\mathcal{A} \circ \bm{\mu}(y)\right]_{tz} = \sum_{g} A^{[t]}_{z,g} \cdot \bm{\mu}(y)_{t g} $$
holding for each $y$, for any vector-valued function $\bm{\mu}: \mathcal{Y} \rightarrow\mathbbm{R}^{|\mathcal{T}|\cdot|\mathcal{G}|}$.

Let $\theta=\mathbbm{E}[Y_i(t)|c(G_i)=1]$ be the parameter of interest. Note that similar to \eqref{Aeq2}, $\theta$ can also be written as a linear map applied to the function $\mathbf{G}^*$. In particular $\theta = \Theta \circ \mathbf{G}^*$, where for any function $\bm{\mu}: \mathcal{Y}$ to $\mathbbm{R}^{|\mathcal{T}|\cdot|\mathcal{G}|}$, $\Theta \circ \bm{\mu}$ is the scalar:
\begin{equation} \label{thetaeq} \sum_{g \in \mathcal{G}} \frac{c_g}{P(c(G_i)=1)} \cdot \int_\mathcal{Y} y\cdot  d\bm{\mu}(y)_{t,g}
\end{equation}
The set of such $\bm{\mu}$ that recover the distribution of observables can be written as: $$\mathcal{S}:= \{\bm{\mu}: \mathcal{A} \circ \bm{\mu} = \mathbf{F}_{(YD)|Z}\}$$
However, some such candidate values $\bm{\mu} \in \mathcal{S}$ for $\mathbf{G}^*$ may correspond to $F_{Y(t)|G=g}(\cdot)$ that do not represent valid CDFs. Accordingly, let us define
$$\mathcal{R}:= \{\bm{\mu}: [\bm{\mu}(y)]_{tg}/P(G_i=g) \textrm{ is a proper CDF  for each }t \in \mathcal{T} \textrm{ and }g \in \mathcal{G} \textrm{ s.t. } P(G_i=g)>0\}$$
The remainder of this section establishes that for $\theta$ to be outcome-nonrestrictive identified, the set $\mathcal{S} \cap \mathcal{R}$ must map to a singleton under $\Theta$.

Note that the sets $\mathcal{R}$ and $\mathcal{S}$ as well as the map $\Theta$ depend on the distribution $\mathcal{P}_{latent}$ (through $\mathbf{F}_{(YD)|Z}$ for $\mathcal{S}$ and through the $P(G_i=g)$ for $\mathcal{R}$ and $\Theta$).\footnote{Note that the map $\Theta$ depends on $t$ and the vector $c$ as well, also left implicit for ease of exposition.} Let us denote this dependence by $\mathcal{S}(\mathcal{P}_{latent})$, $\mathcal{R}(\mathcal{P}_{latent})$ and $\Theta(\mathcal{P}_{latent})$, though I will later leave this dependence implicit to ease notation. 

Definition \ref{def:oai} of outcome-nonrestrictive identification, translated into this notation, says that \begin{equation} \label{eq:oaidthm2}
	\left\{\Theta(\mathcal{P}_{latent}) \circ \bm{\mu}: \bm{\mu} \in \mathcal{R}(\mathcal{P}_{latent}) \textrm{ and } \bm{\mu} \in \mathcal{S}(\mathcal{P}_{latent})\right\} \textrm{ is a singleton } \forall \mathcal{P}_{latent} \in \mathscr{P}_{latent,c}(\mathcal{G})
\end{equation}
The following regularity condition will prove to be useful later in the proof:
\begin{condition*}[REG] \label{ass:reg}
	Fix a $t \in \mathcal{T}$. For some $g^* \in \mathcal{G}$, there exists a $\underline{L}>0$ and $\bar{L}< \infty$ such that for any $g' \in \mathcal{G}$ and $y'>y$:
	$$\underline{L} \le\frac{F_{Y(t)|G=g'}(y')-F_{Y(t)|G=g'}(y)}{F_{Y(t)|G=g^*}(y')-F_{Y(t)|G=g^*}(y)} \le \bar{L}$$
\end{condition*}
\noindent Note that whether or not Condition REG holds is a property of $\mathcal{P}_{latent}$. A sufficient condition is that $Y$ is discrete and finite and the support of $Y(t)|G=g$ is the same for all $g$. Another sufficient condition is that i) $Y$ is continuously distributed with the support of the density $f_{Y(t)|G=g}(y)$ the same for all $g$ and $t$; ii) the density on this set $\mathcal{Y}$ is bounded from below by $\underline{M}>0$ for all $g$, and iii) similarly $\sup_{y \in \mathcal{Y}} f_{Y(t)|G=g}(y)\le \bar{M}$ for some $\bar{M} < \infty$, for all $g$.\footnote{In the discrete case, let $\underline{L} = \min_{y \in \mathcal{Y},g \in \mathcal{G}} P(Y(t)=y|G=g) P(Y(t)=y|G=g^*)$ and $\bar{L} = 1/\min_{y \in \mathcal{Y}} P(Y(t)=y|G=g^*)$. In the continuous case let $\bar{L} = \frac{\max{g \in \mathcal{G}} \sup_{y \in \mathcal{Y}} f_{Y(t)|G=g}(y)}{\min{g \in \mathcal{G}} \inf_{y \in \mathcal{Y}} f_{Y(t)|G=g}(y)} \le \bar{M}/\underline{M}$ and $\underline{L} = \frac{\min{g \in \mathcal{G}} \inf_{y \in \mathcal{Y}} f_{Y(t)|G=g}(y)}{\max{g \in \mathcal{G}} \sup_{y \in \mathcal{Y}} f_{Y(t)|G=g}(y)} \ge \underline{M}/\bar{M}$.} A mixture of distributions satisfying the above will also satisfy REG.

Let $\bar{\mathscr{P}}_{latent,c}(\mathcal{G})$ denote the set of distributions $\mathcal{P}_{latent} \in \mathscr{P}_{latent,c}(\mathcal{G})$ that satisfy Condition REG. $\bar{\mathscr{P}}_{latent,c}(\mathcal{G})$ is never empty (given $\mathcal{G} \neq \emptyset$), since we have seen above that for any $|\mathcal{G}|>0$ there are always distributions that satisfy REG (with examples for each of discrete, continuous or mixed $Y$). Further, $\mathscr{P}_{latent,c}(\mathcal{G})$ only limits the support of $G$ and places no constraint on the distribution of $\tilde{Y}|G$. Note from \eqref{eq:oaidthm2} that if $\theta$ is outcome-nonrestrictive identified, $\left\{\Theta \cdot \bm{\mu}\right\}_{\bm{\mu} \in (\mathcal{S}(\mathcal{P}_{latent}) \cap \mathcal{R}(\mathcal{P}_{latent})}$ must be a singleton for all $\mathcal{P}_{latent}$ such that $supp\{\mathcal{P}_{G}(\mathcal{P}_{latent})\} \subseteq \mathcal{G}$, including any $\mathcal{P}_{latent} \in \bar{\mathscr{P}}_{latent,c}(\mathcal{G})$. 

The remainder of the proof of Theorem \ref{thm:necc} shows that if $c \notin rs(A^{[t]})$, it is always possible to find $\mathcal{P}_{latent} \in \bar{\mathscr{P}}_{latent,c}(\mathcal{G})$ such that $\left\{\Theta(\mathcal{P}_{latent}) \cdot \bm{\mu}\right\}_{\bm{\mu} \in (\mathcal{S}(\mathcal{P}_{latent}) \cap \mathcal{R}(\mathcal{P}_{latent})}$ is not in fact a singleton.

\subsubsection*{A candidate for $\mathbf{G}^*$ that recovers observables}
To see this, we will explicitly construct a functional $\mathbf{G}$ of $\mathcal{P}_{latent}$, that generally differs from $\mathbf{G^*}$ and lets us define an ``alternative'' to $\mathcal{P}_{latent}$ but still recovers observables.

Consider the vector-valued function $\mathbf{G}$, where the $t,g$ component of $\mathbf{G}(y)$ is:
$$\left[\mathbf{G}(y)\right]_{t,g}:= \begin{cases}
	P(G_i=g)\cdot F_{Y(t)|G}(y|g) & \textrm{ if } \max_{z \in \mathcal{Z}} \mathbbm{1}(T_g(z)=t)=0\\
	\sum_z [(A^{[t]})^+]_{g,z} ]\cdot F_{(YD)|Z}(y,t|z) & \textrm{ if }  \max_{z \in \mathcal{Z}} \mathbbm{1}(T_g(z)=t)=1
\end{cases} $$
and $(A^{[t]})^+$ indicates the Moore-Penrose pseudoinverse of the matrix $A^{[t]}$.

The reason for separating out the two cases in the definition of $\mathbf{G}$ is that if there exists a group $g$ that acts as a ``never-taker'' with respect to treatment $t$ such that $\max_{z \in \mathcal{Z}} \mathbbm{1}(T_g(z)=t)=0$, then this corresponds to a column of all zeros in $A^{[t]}$. A property of the Moore-Penrose inverse is that if column $g$ of $A^{[t]}$ is all zeros, then the corresponding row $g$ of $(A^{[t]})^+$ is also all zeros (see e.g. \citealt{mppartitioned}) which would leave $\left[\mathbf{G}(y)\right]_{t,g}=0$ for all $y$ if we did not separate out this case. This would make it impossible for $\mathbf{G}$ to represent a possible candidate for $\mathbf{G}^*$ (i.e. $\mathbf{G} \in \mathcal{R}$). The above construction avoids this problem by simply replacing such problematic combinations of $(g,t)$ by using the actual $[\mathbf{G}^*(y)]_{t,g}$ (which are unknown). Note that if the first case holds for \textit{all} $g \in \mathcal{G}$, then the matrix $A^{[t]}$ is simply the zero matrix, and outcome-nonrestrictive identification cannot hold, by Lemma \ref{lemma:fullrowrank}. Thus, we can continue under the assumption that the second case holds for at least some $g \in \mathcal{G}$. 

Let use see now that $\mathbf{G}$ ``recovers observables'', by which I mean that $\mathcal{A} \circ \bm{\mu} = \mathbf{F}_{(YD)|Z}$ and hence $\mathbf{G} \in \mathcal{S}$. Indeed:
\begin{align*}[\mathcal{A} \circ \mathbf{G}(y)]_{t,z} &= \sum_{g} A^{[t]}_{z,g}\left[\mathbf{G}(y)\right]_{t,g}\\
	&=\cancel{\sum_{g: \max_{z \in \mathcal{Z}} \mathbbm{1}(T_g(z)=t)=0} A^{[t]}_{z,g}\cdot  P(G_i=g)\cdot F_{Y(t)|G}(y|g)}\\
	&\hspace{1in} + \sum_{g: \max_{z \in \mathcal{Z}} \mathbbm{1}(T_g(z)=t)=1}\sum_{z'} A^{[t]}_{z,g}[(A^{[t]})^+]_{g,z'} F_{(YD)|Z}(y,t|z')\\
	&=\sum_{g,z'} A^{[t]}_{z,g}[(A^{[t]})^+]_{g,z'} F_{(YD)|Z}(y,t|z')\\
	&= \sum_{z'} [A^{[t]}(A^{[t]})^+]_{z,z'} F_{(YD)|Z}(y,t|z')= [F_{(YD)|Z}(y)]_{tz} \end{align*}
where the second and third equalities use that $A^{[t]}_{z,g} = 0$ for all $z$, if $g$ is such that $\max_{z \in \mathcal{Z}} \mathbbm{1}(T_g(z)=t)=0$. The final equality follows from $A^{[t]}(A^{[t]})^+ = I_{|\mathcal{Z}|}$, which in turn follows from $(A^{[t]})^+ = {A^{[t]}}'({A^{[t]}}{A^{[t]}}')^{-1}$ since we can by Lemma \ref{lemma:fullrowrank} assume that $A^{[t]}$ has full row rank.

$\mathbf{G}$ may still however not be in $\mathcal{R}$, as its definition above does not ensure that each $F_{Y(t)|G}(y|g)$ is necessarily weakly increasing in $y$ with a limit of unity as $y \uparrow \infty$. Note that $[\mathbf{G}]_{t,g}/P(G_i=g)$ does have the final two properties of a CDF: right-continuity and a left limit of zero. To see this, substitute (\ref{Aeq2}) into the definition of $\mathbf{G}$, to rewrite as:
\begin{equation} \label{eq:Fstarnew}
	\left[\mathbf{G}(y)\right]_{t,g}:= \begin{cases}
		P(G_i=g)\cdot F_{Y(t)|G}(y|g) & \textrm{ if } \max_{z \in \mathcal{Z}} \mathbbm{1}(T_g(z)=t)=0\\
		\sum_{g'} [(A^{[t]})^+ A^{[t]}]_{g,g'}\cdot P(G_i=g')\cdot F_{Y(t)|G}(y|g') & \textrm{ if } \max_{z \in \mathcal{Z}} \mathbbm{1}(T_g(z)=t)=1
	\end{cases}
\end{equation}
Right continuity of each element of $\mathbf{G}(y)$ in $y$ follows from right-continuity of the $F_{Y(t)|G}(y|g')$. Note that $\lim_{y \downarrow -\infty} \left[\mathbf{G}(y)\right]_{t,g} = 0$ follows from each of the CDFs $F_{(YD)|Z}$ approaching zero as $y \downarrow -\infty$, given that the components of $A^{[t]}$ and $P(G_i=g)$ are finite.

Let $\beta_{t,g} := \lim_{y \uparrow \infty} \left[\mathbf{G}(y)\right]_{t,g}$. For any $t,g$ such that $\max_{z \in \mathcal{Z}} \mathbbm{1}(T_g(z)=t)=0$, it follows from the definition of
$\mathbf{G}$ that $\beta_{t,g}=P(G_i=g)$, since each of the $F_{Y(t)|G}(y|g)$ are valid CDFs. For the other $t,g$, use \eqref{eq:Fstarnew} to see that 
\begin{align*}
	\beta_{t,g}&=\lim_{y \uparrow \infty} \sum_{g'} [(A^{[t]})^+ A^{[t]}]_{g,g'} \cdot P(G_i=g')\cdot F_{Y(t)|G}(y|g') = \sum_{g'} [(A^{[t]})^+ A^{[t]}]_{g,g'} \cdot P(G_i=g')\\
	&= [(A^{[t]})^+ A^{[t]}P]_{g}
\end{align*}
where $P$ is a vector of $P(G_i=g)$ for all $g \in \mathcal{G}$. 

Unless $[(A^{[t]})^+ A^{[t]}P]_{g}=P_g$ for all $g \in \mathcal{G}$, the functions $\left[\mathbf{G}(y)\right]_{t,g}$ may thus not represent properly normalized CDFs. In fact, they may not even be monotonic in $y$. However, we can still use $\mathbf{G}$ as a building block to construct another set of functions that satisfy all of the properties of a CDF.

\subsubsection*{A broader class of candidates that also recover observables but represent CDFs}
Given some fixed $g^* \in \mathcal{G}$, let us define a vector valued function $\mathbf{D}: \mathcal{Y}\rightarrow \mathbbm{R}^{|\mathcal{T}|\cdot |\mathcal{G}|}$ with components:
\begin{equation} \label{eq:D}
	[\mathbf{D}(y)]_{t,g} := (P_g-\beta_{t,g})\cdot F_{Y(t)|G}(y|g^*)= [\left\{I-(A^{[t]})^+ A^{[t]}\right\}P]_{g}\cdot F_{Y(t)|G}(y|g^*)
\end{equation}
Now let us define for any $\lambda \in [0,1]$ the convex combination of $\mathbf{G}+\mathbf{D}$ and $\mathbf{G}^*$:
\begin{equation} \label{eq:Galpha}
	\mathbf{G}^{\lambda} := \lambda\left(\mathbf{G} + \mathbf{D} \right) + (1-\lambda)\mathbf{G}^* = \mathbf{G}^*+\lambda\left\{\mathbf{G}-\mathbf{G}^*+\mathbf{D}\right\}
\end{equation}
Our first observation will be that $\mathcal{A} \circ \mathbf{G}^{\lambda} = \mathbf{F}_{(YD)|Z}$, i.e. $\mathbf{G}^{\lambda}$ still recovers observables and thus $\mathbf{G}^{\lambda} \in \mathcal{S}$. To see this, note that:
\begin{align*}[&\mathcal{A} \circ \mathbf{G}^\lambda(y)]_{t,z} = [\mathcal{A} \circ \mathbf{G}(y)]_{t,z} + \lambda \cdot [\mathcal{A} \circ \left\{\mathbf{G}-\mathbf{G}^*+\mathbf{D}\right\}(y)]_{t,g}\\
	&=[F_{(YD)|Z}(y)]_{t,z} + \cancel{\lambda \cdot [\mathcal{A} \circ\mathbf{G}(y)]_{t,g}-\lambda \cdot [\mathcal{A} \circ\mathbf{G}^*(y)]_{t,g}} + \lambda \cdot [\mathcal{A} \circ\mathbf{D}(y)]_{t,g}\\
	&=[F_{(YD)|Z}(y)]_{t,z}+\lambda \cdot \sum_{g,g'} A^{[t]}_{z,g} \cdot [(I-(A^{[t]})^+ A^{[t]})]_{g,g'}\cdot P(G_i=g')\cdot  F_{Y(t)|G}(y|g^*)\\
	&=[F_{(YD)|Z}(y)]_{t,z}+\lambda \cdot \sum_{g'} [\cancel{A^{[t]}(I-(A^{[t]})^+ A^{[t]})}]_{z,g'}\cdot P(G_i=g')\cdot  F_{Y(t)|G}(y|g^*)\\
	&=[F_{(YD)|Z}(y)]_{t,z}
\end{align*}
since $\mathcal{A} \circ\mathbf{G}^*=\mathcal{A} \circ\mathbf{G}$ and $A^{[t]}(A^{[t]})^+ A^{[t]}=A^{[t]}$. 

Now, we verify that for a small enough $\lambda$, $\mathbf{G}^{\lambda}$ yields $F_{Y(t)|G}(y|g)$ that satisfy the properties of a CDF and hence $\mathbf{G}^{\lambda} \in \mathcal{R}$. First, note that $\left[\mathbf{G}^{\lambda}(y)\right]_{t,g}$ is right-continuous in $y$, since each of $[\mathbf{G}(y)]_{t,g}$, $[\mathbf{G}^*(y)]_{t,g}$, and $[\mathbf{D}(y)]_{t,g}$ are. We also have that $\lim_{y \downarrow -\infty} \left[\mathbf{G}^{\lambda}(y)\right]_{t,g}=0$, since $$\lim_{y \downarrow -\infty} \left[\mathbf{G}(y)\right]_{t,g}=\lim_{y \downarrow -\infty} \left[\mathbf{G}^{*}(y)\right]_{t,g}=\lim_{y \downarrow -\infty} \left[\mathbf{D}(y)\right]_{t,g}=0$$
Note as well that
\begin{align*}
	\lim_{y \uparrow \infty} \left[\mathbf{G}^{\lambda}(y)\right]_{t,g} &= \lim_{y \uparrow \infty} \left[\mathbf{G}^{*}(y)\right]_{t,g}+\lambda \cdot \lim_{y \uparrow \infty} \left[\left\{\mathbf{G}-\mathbf{G}^*+\mathbf{D}\right\}(y)\right]_{t,g} \\
	&=P_g+\lambda \cdot \left\{\lim_{y \uparrow \infty} \left[\mathbf{G}(y)\right]_{t,g}-\lim_{y \uparrow \infty} \left[\mathbf{G}^{*}(y)\right]_{t,g}+\lim_{y \uparrow \infty} \left[\mathbf{D}(y)\right]_{t,g}\right\}\\
	&=P_g+\lambda \cdot \left\{\beta_{t,g}-P_g+(P_g-\beta_{t,g})\cdot 1\right\} = P_g
\end{align*}
matching the correct normalization, i.e. $\lim_{y \uparrow \infty} \left[\mathbf{G}^{*}(y)\right]_{t,g} = P_g \cdot \lim_{y \uparrow \infty} F_{Y(t)|G=g}(y) = P_g$.

It only remains to be seen that for a small enough value of $\lambda$, $ \left[\mathbf{G}^{\lambda}(y)\right]_{t,g}$ is weakly increasing in $y$. This is always possible given that $\mathcal{P}_{latent}$ satisfies Condition REG:
\begin{proposition} \label{prop:mono}
	Given Condition REG, $\left[\mathbf{G}^{\lambda}(y)\right]_{t,g}$ is non-decreasing in $y$ for any $\lambda \in (0,\bar{\lambda}]$, where $\bar{\lambda} = \frac{\underbar{L}}{2|\mathcal{G}|\cdot \bar{L}} >0$.
\end{proposition}
\noindent Given Proposition \ref{prop:mono}, we have shown that for $\lambda \le \bar{\lambda}$, $\mathbf{G}^\lambda \in \mathcal{R}$ and hence $\mathbf{G}^\lambda \in (\mathcal{S} \cap \mathcal{R})$.

\subsubsection*{Outcome-nonrestrictive identification implies $c \in rs(A^{[t]})$}
Consider now any $\mathcal{P}_{latent} \in \bar{\mathscr{P},c}(\mathcal{G})$ and choose the $g^* \in \mathcal{G}$ in the definition of $\mathbf{D}$ so that REG holds for that $g^*$. We know that there exist $\lambda>0$ small enough that $\mathbf{G}^\lambda \in (\mathcal{S} \cap \mathcal{R}$). For any such $\lambda$, outcome-nonrestrictive identification of $\theta$ now requires that $\Theta \circ \mathbf{G}^\lambda = \Theta \circ \mathbf{G}^*$. This in turn requires, by Eq. \eqref{eq:Galpha}, that $\Theta \circ  \left\{\mathbf{G}-\mathbf{G}^*+\mathbf{D}\right\}=0$. Now:
\begin{align}
	\Theta &\circ  \left\{\mathbf{G}-\mathbf{G}^*+\mathbf{D}\right\} \nonumber \\
	&=\frac{1}{P(c(G_i)=1)} \sum_g c_g
	\cdot \left\{\int_\mathcal{Y} y\cdot  d\mathbf{G}(y)_{t,g}-\int_\mathcal{Y} y\cdot  d\mathbf{G}^*(y)_{t,g}+\int_\mathcal{Y} y\cdot  d\mathbf{D}(y)_{t,g}\right\} \nonumber\\
	&=\frac{1}{P(c(G_i)=1)} \sum_g c_g \sum_{g'} [I-(A^{[t]})^+ A^{[t]}]_{g,g'} \cdot P(G_i=g')\cdot \mathbbm{E}[Y_i(t)|G_i=g'] \nonumber\\
	&\hspace{.5in}+\frac{1}{P(c(G_i)=1)} \sum_g c_g \sum_{g'} [(I-(A^{[t]})^+ A^{[t]})P]_{g,g'}\cdot P(G_i=g') \cdot \mathbbm{E}[Y_i(t)|G_i=g^*] \nonumber\\
	&=\frac{1}{P(c(G_i)=1)} \sum_{g'} [c'(I-(A^{[t]})^+ A^{[t]})]_{g'} \cdot P(G_i=g')\cdot \mathbbm{E}[Y_i(t)|G_i=g'] \nonumber\\
	&\hspace{1in}+\frac{1}{P(c(G_i)=1)} \sum_{g'} [c'(I-(A^{[t]})^+ A^{[t]})]_{g'}\cdot P(G_i=g')\cdot \mathbbm{E}[Y_i(t)|G_i=g^*] \nonumber\\
	&=\frac{1}{P(c(G_i)=1)} \sum_{g'} [c'(I-(A^{[t]})^+ A^{[t]})]_{g'} \cdot P(G_i=g')\cdot \left\{ \mathbbm{E}[Y_i(t)|g']-\mathbbm{E}[Y_i(t)|g^*]\right\} \label{eq:finalcondtion}
\end{align}
Note that although the map $\Theta$ depends on the distribution $\mathcal{P}_G$, the constructions $\mathbf{G}$, $\mathbf{D}$ and $\mathbf{G}^\lambda$ all use the same distribution $\mathcal{P}_G$ from the actual distribution $\mathcal{P}_{latent}$. It is for this reason that $P(c(G_i)=1)$ factors out in Eq. \eqref{eq:finalcondtion}, and the RHS can only be non-zero if the sum over $g'$ appearing in it evaluates to zero.

Suppose that $c \notin rs(A^{[t]})$ so that $c'(I-(A^{[t]})^+ A^{[t]}) = \tilde{c}'$ for some non-zero vector $\tilde{c}$. Provided that $P(G_i=g')\cdot \left\{ \mathbbm{E}[Y_i(t)|G_i=g']-\mathbbm{E}[Y_i(t)|G_i=g^*]\right\}$, thought of as a vector across $g' \in \mathcal{G}$, is not perfectly orthogonal in $\mathbbm{R}^{|\mathcal{G}|}$ to $\tilde{c}$, we will have that $$\sum_{g'} \tilde{c}'_{g'} \cdot P(G_i=g')\cdot \left\{ \mathbbm{E}[Y_i(t)|G_i=g']-\mathbbm{E}[Y_i(t)|G_i=g^*]\right\} \ne 0$$
There is always a $\mathcal{P}_{latent} \in \bar{\mathscr{P}}_{latent,c}(\mathcal{G})$ such that this non-orthogonality holds, because the relative magnitudes of $P(G_i=g)$ and level-differences $\mathbbm{E}[Y_i(t)|G_i=g']-\mathbbm{E}[Y_i(t)|G_i=g^*]$ in $Y_i(t)$ can be varied without violating REG or changing the support of $G_i$. Thus if $c \notin rs(A^{[t]})$, we can obtain $\Theta \circ  \left\{\mathbf{G}-\mathbf{G}^*+\mathbf{D}\right\} \ne 0$ for some $\mathcal{P}_{latent} \in \mathscr{P}_{c,latent}(\mathcal{G})$, and $\theta$ is not outcome-nonrestrictive identified.

\subsubsection{Proof of Proposition \ref{prop:mono}}
The key to ensuring monotonicity will be to choose $\lambda$ small enough that any decreases with $y$ in the components of $\mathbf{G}^{\lambda}$ are dominated by increases in the corresponding components of $\mathbf{G}^*$, so that each $\left[\mathbf{G}^{\lambda}\right]_{t,g}$ is monotonically increasing. For $\left[\mathbf{G}^\lambda(y)\right]_{t,g}$ to be monotonically increasing in $y$ we need that for any $y'>y$: $\left[\mathbf{G}^\lambda(y')\right]_{t,g} - \left[\mathbf{G}^\lambda(y)\right]_{t,g}\ge 0$, i.e. that
\begin{equation} \label{eq:ineq}
	\left[\mathbf{G}^*(y')\right]_{t,g} - \left[\mathbf{G}^*\right]_{t,g}\ge \lambda\cdot  \left[(\mathbf{G}^*-\mathbf{G})(y')-(\mathbf{G}^*-\mathbf{G})(y)\right]_{t,g} - (\left[\mathbf{D}(y')\right]_{t,g} - \left[\mathbf{D}\right]_{t,g})\}
\end{equation}
Let us turn first to $\left[(\mathbf{G}^*-\mathbf{G})(y)\right]_{t,g}$. Fix a $g$ and $t$, and any $y' > y$. Then, by \eqref{eq:Fstarnew}:
\begin{equation} \label{eq:Fstarnew2}
	\left[\mathbf{G}(y')\right]_{t,g}-\left[\mathbf{G}^*(y)\right]_{t,g}= \begin{cases}
		P(G_i=g)\cdot\{F_{Y(t)|G}(y'|g)-F_{Y(t)|G}(y|g)\}\\
		\sum_{g'} [(A^{[t]})^+ A^{[t]}]_{g,g'}\cdot  P(G_i=g')\cdot \left\{F_{Y(t)|G}(y'|g')-F_{Y(t)|G}(y|g')\right\}
	\end{cases}
\end{equation}
where the first line indicates the case that $g$ is such that $\max_{z \in \mathcal{Z}} \mathbbm{1}(T_g(z)=t)=0$, and the second that $\max_{z \in \mathcal{Z}} \mathbbm{1}(T_g(z)=t)=1$. Thus $\left[(\mathbf{G}^*-\mathbf{G})(y')\right]_{t,g}-\left[(\mathbf{G}^*-\mathbf{G})(y)\right]_{t,g}$ is equal to $0$ if $\max_{z \in \mathcal{Z}} \mathbbm{1}(T_g(z)=t)=0$, and 
$$\sum_{g'} [I-(A^{[t]})^+ A^{[t]}]_{g,g'} \cdot P(G_i=g')\cdot \left\{F_{Y(t)|G}(y'|g')-F_{Y(t)|G}(y|g')\right\}$$
if $\max_{z \in \mathcal{Z}} \mathbbm{1}(T_g(z)=t)=1$.

Thus we have by REG that \small
\begin{align*}
	&\left|\left[(\mathbf{G}^*-\mathbf{G})(y')\right]_{t,g}-\left[(\mathbf{G}^*-\mathbf{G})(y)\right]_{t,g}\right|\\
	&= \left|\sum_{g'} [I-(A^{[t]})^+ A^{[t]}]_{g,g'} \cdot P(G_i=g')\cdot \left\{F_{Y(t)|G}(y'|g')- F_{Y(t)|G}(y|g')\right\}\right|\\
	&= \left\{F_{Y(t)|G}(y'|g^*)- F_{Y(t)|G}(y|g^*)\right\}\cdot \left|\sum_{g'} [I-(A^{[t]})^+ A^{[t]}]_{g,g'} \cdot P(G_i=g') \cdot \frac{F_{Y(t)|G}(y'|g')- F_{Y(t)|G}(y|g')}{F_{Y(t)|G}(y'|g^*)- F_{Y(t)|G}(y|g^*)}\right|\\	
	& \le \left\{F_{Y(t)|G}(y'|g^*)- F_{Y(t)|G}(y|g^*)\right\}\cdot  |\mathcal{G}|^{1/2} \cdot \sqrt{\sum_{g'} P(G_i=g')^2\cdot  \left(\frac{F_{Y(t)|G}(y'|g')- F_{Y(t)|G}(y|g')}{F_{Y(t)|G}(y'|g^*)- F_{Y(t)|G}(y|g^*)}\right)^2}\\
	& \le \left\{F_{Y(t)|G}(y'|g^*)- F_{Y(t)|G}(y|g^*)\right\}\cdot |\mathcal{G}| \cdot \max_{g'} P(G_i=g') \cdot \max_{g'} \left|\frac{F_{Y(t)|G}(y'|g')- F_{Y(t)|G}(y|g')}{F_{Y(t)|G}(y'|g^*)- F_{Y(t)|G}(y|g^*)}\right|\\
	& \le \left\{F_{Y(t)|G}(y'|g^*)- F_{Y(t)|G}(y|g^*)\right\}\cdot |\mathcal{G}| \cdot \max_{g'} \left|\frac{F_{Y(t)|G}(y'|g')- F_{Y(t)|G}(y|g')}{F_{Y(t)|G}(y'|g^*)- F_{Y(t)|G}(y|g^*)}\right|
\end{align*} \large
using that $[I-(A^{[t]})^+ A^{[t]}]$ is a projection (so that $|[I-(A^{[t]})^+ A^{[t]}]v| \le |v|$ for any vector $v \in \mathbbm{R}^{|\mathcal{G}|}$) and by the Cauchy-Schwarz inequality. Let $\delta_t^*(y',y):=F_{Y(t)|G}(y'|g^*)- F_{Y(t)|G}(y|g^*)$. Then, by REG:
$$ \left|\left[(\mathbf{G}^*-\mathbf{G})(y')\right]_{t,g}-\left[(\mathbf{G}^*-\mathbf{G})(y)\right]_{t,g}\right| \le \delta_t^*(y',y) \cdot |\mathcal{G}| \cdot \bar{L}$$	
Now consider $\left[\mathbf{D}(y)\right]_{t,g}$. Fix a $g$ and $t$, and any $y' > y$. Similarly, we have that 
\small
\begin{align*}
	\left|\left[\mathbf{D}(y')\right]_{t,g}-\left[\mathbf{D}(y)\right]_{t,g}\right|&= \left|\sum_{g'} [I-(A^{[t]})^+ A^{[t]}]_{g,g'} \cdot P(G_i=g')\cdot \{F_{Y(t)|G}(y'|g^*)-F_{Y(t)|G}(y|g^*)\}\right|\\
	&\le \{F_{Y(t)|G}(y'|g^*)-F_{Y(t)|G}(y|g^*)\} \cdot |P|\\
	&\le \{F_{Y(t)|G}(y'|g^*)-F_{Y(t)|G}(y|g^*)\} \cdot |\mathcal{G}|
\end{align*} \large
So, using Condition REG:
$$ \left|\left[\mathbf{D}(y')\right]_{t,g}-\left[\mathbf{D}(y)\right]_{t,g}\right| \le \delta_t^*(y',y) \cdot |\mathcal{G}|\cdot \bar{L}$$
We can thus put an upper bound on the RHS of \eqref{eq:ineq}
$$\lambda\cdot \{\left[(\mathbf{G}^*-\mathbf{G})(y')-(\mathbf{G}^*-\mathbf{G})(y)\right]_{t,g} - (\lim_{y' \downarrow y} \left[\mathbf{D}(y')\right]_{t,g} - \left[\mathbf{D}\right]_{t,g})\} \le 2\lambda \cdot \delta_t^*(y',y) \cdot |\mathcal{G}|\cdot \bar{L}$$
Meanwhile, by REG:
\begin{align*}
 &\left\{\left[\mathbf{G}^*(y')\right]_{t,g} - \left[\mathbf{G}^*(y)\right]_{t,g}\right\}\\  & \hspace{.4in}= \{F_{Y(t)|G}(y'|g^*)-F_{Y(t)|G}(y|g^*)\} \cdot \frac{F_{Y(t)|G}(y'|g')- F_{Y(t)|G}(y|g')}{F_{Y(t)|G}(y'|g^*)- F_{Y(t)|G}(y|g^*)}\ge \delta_t^*(y',y) \cdot \underbar{L}
\end{align*}
Thus inequality \eqref{eq:ineq} then holds provided that $\delta_t^*(y',y) \cdot \underbar{L} \ge  2\lambda \cdot \delta_t^*(y',y) \cdot |\mathcal{G}|\cdot \bar{L}$, which holds trivially if $\delta_t^*(y',y)=0$ and if and only if $\lambda \le \frac{\underbar{L}}{2|\mathcal{G}|\cdot \bar{L}}$ if $\delta_t^*(y',y)>0$.

A visualization of the intuition behind this result is depicted in Figure \ref{fig:mono}.
\begin{figure}[h!]
	\begin{center}
		\includegraphics[width=4.5in]{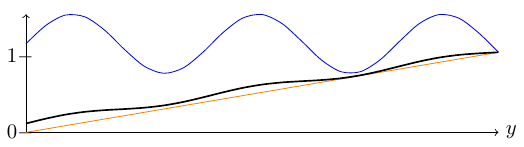}
		\caption{Depiction of Proposition \ref{prop:mono}. The blue sinusoidal function depicts an example of a $\left[(\mathbf{G}+\mathbf{D})^\lambda(y)\right]_{t,g}$ that is not weakly increasing. The orange curve depicts $\left[\mathbf{G}^*(y)\right]_{t,g}$ which is weakly increasing. The black curve depicts $\left[\mathbf{G}^\lambda(y)\right]_{t,g}$, which is a linear combination of the blue and orange functions with weights $\lambda=0.1$ and $1-\lambda = 0.9$, respectively. This value of $\lambda$ is small enough that the black curve is weakly increasing everywhere. \label{fig:mono}}
	\end{center}
\end{figure}

\subsection{Proof of Lemma \ref{lemma:fullrowrank}}
Suppose that $A^{[t]}$ does not have full row rank. This implies that for some $\mathcal{Z}_0 \subset \mathcal{Z}$, each of the remaining rows of $A^{[t]}$ for $z \notin \mathcal{Z}_0$ can be written as a linear combination of the rows of $A^{[t]}$ for $z \in \mathcal{Z}_0$. Take such a $z^* \notin \mathcal{Z}_0$, and accordingly let 
$$A^{[t]}_{z^*,g} = \sum_{z \in \mathcal{Z}_0} \gamma_z \cdot A^{[t]}_{z,g} \quad \textrm{ for all } g \in \mathcal{G}$$
Note then that Eq. \eqref{Aeq} implies that 
\begin{align*}
	F_{(YD)|Z=z^*}&(y,t) = \sum_{g \in \mathcal{G}} A^{[t]}_{z^*,g} \cdot P(G_i=g) \cdot F_{Y(t)|G=g}(y)\\
	&= \sum_{g \in \mathcal{G}} \left(\sum_{z \in \mathcal{Z}_0} \gamma_z \cdot  A^{[t]}_{z,g}\right) \cdot P(G_i=g) \cdot F_{Y(t)|G=g}(y)\\
	&= \sum_{z \in \mathcal{Z}_0} \gamma_z \cdot \sum_{g \in \mathcal{G}} A^{[t]}_{z,g} \cdot P(G_i=g) \cdot F_{Y(t)|G=g}(y) = \sum_{z \in \mathcal{Z}_0} \gamma_z \cdot F_{(YD)|Z=z}(y,t) 
\end{align*}
where the RHS on the last line does not depend on the distribution of observables for $i$ such that $Z_i=z^*$. Thus, $F_{(YD)|Z=z^*}(y,t)$ adds no information that is not contained in $F_{(YD)|Z=z}(y,t)$ for $z \in \mathcal{Z}_0$. If $\mu_g^t$ is outcome-nonrestrictive identified, it must be using the distribution $\mathcal{P}_{YTZ|Z \in \mathcal{Z}_0}$ rather than the full unconditional distribution $\mathcal{P}_{obs}=\mathcal{P}_{YTZ}$.

\subsection{Proof of Proposition \ref{prop:limitsearch}}
Suppose first that $|\mathcal{G}| > |\mathcal{Z}|$ and $A=A^{[t]}$ has full row rank of $|\mathcal{Z}|$. Then since $A$ has full row-rank of $|\mathcal{Z}|$, there exists a subset of $|\mathcal{Z}|$ columns that are linearly independent from one another. Write $A = [\tilde{A}, \tilde{A}^c]$ where $\tilde{A}$ is an invertible $|\mathcal{Z}| \times |\mathcal{Z}|$ matrix of these columns, and $\tilde{A}_c$ are the others. Write the system $A'\alpha = c$ in this notation as 
$$ \begin{bmatrix}
	\tilde{A}' \\ {\tilde{A}_c}'
\end{bmatrix} \alpha = \begin{pmatrix} \tilde{c}\\ \tilde{c}_c \end{pmatrix}$$
where $\tilde{c}$ denotes the $|\mathcal{Z}|$ components of $c$ corresponding to the columns of $A$ put into in $\tilde{A}$, and $\tilde{c}_c$ are the remaining entries $c_g$ of $c$. Then $\alpha = {\tilde{A}}^{'-1}\tilde{c}$, which can be seen by left-multiplying the above equation by the $|\mathcal{Z}| \times |\mathcal{G}|$ matrix $[{\tilde{A}}^{'-1},\mathbf{0}^{|\mathcal{Z}| \times |\mathcal{G}|-|\mathcal{Z}|}]$. Intuitively, the system $A'\alpha = c$ is over-determined, so we only only need the components $\tilde{c}$ of $c$ to uniquely determine the vector $\alpha$.

Now consider the case in which $|\mathcal{G}| < |\mathcal{Z}|$, so that the system $A'\alpha = c$ is now undetermined. Suppose for now that the rank of $A$ is $|\mathcal{G}|$ so that it has full column rank. One solution $\alpha$ can then be obtained by  writing $A=\begin{bmatrix}
	\tilde{A} \\ {\tilde{A}_c}
\end{bmatrix} $ where $\tilde{A}$ is an invertible $|\mathcal{G}| \times |\mathcal{G}|$ matrix representing $|\mathcal{G}|$ linearly independent \textit{rows} of $A$. Now consider $\alpha = \begin{pmatrix} \tilde{A}^{-1}c\\ \mathbf{0}^{(|\mathcal{Z}|-|\mathcal{G}|) \times 1} \end{pmatrix}$ where note that $\tilde{A}^{-1}c$ is $|\mathcal{G}|-$component vector. This represents a solution to $A'\alpha = c$ since
$$A'\begin{pmatrix} \tilde{A}^{'-1}c\\ \mathbf{0}^{(|\mathcal{Z}|-|\mathcal{G}|) \times 1} \end{pmatrix} = [\tilde{A}, \tilde{A}_c] \begin{pmatrix} \tilde{A}^{-1}c\\ \mathbf{0}^{(|\mathcal{Z}|-|\mathcal{G}|) \times 1} \end{pmatrix} = c $$

We can combine the constructions in the two special cases considered above to relax any assumptions about the cardinality of $\mathcal{Z}$ and $\mathcal{G}$ or the rank of $A$. Let the rank of $A$ be $k \le \min\{|\mathcal{Z}|, |\mathcal{G}|\}$. Write $A = A_k[I_k,M]$ where $A_k$ is a $k \times |\mathcal{G}|$ matrix composed of $k$ linearly independent columns of $A$, and $M$ is $(|\mathcal{G}|-k) \times k$ matrix that expresses the remaining $(|\mathcal{G}|-k)$ columns of $A$ as linear combinations of the columns of $A$ represented in $A_k$. Write $c = \begin{pmatrix} \tilde{c}_k\\ \tilde{c}_c \end{pmatrix}$ where $\tilde{c}_k$ collects the corresponding $k$ components of $c$. Note that if $c' = \alpha'A$ has a solution, then $c' = \tilde{c}_k'[I_k,M]$, since $c' = (\alpha_k' A_k)[I,M]$ where the $k$ components of $c'$ corresponding to the columns in $A_k$ are $\alpha_k' A_k$, so $\tilde{c}_k'=\alpha_k' A_k$. Now split the rows of $A_k$ as
$A_k=\begin{bmatrix}
	\tilde{A} \\ {\tilde{A}_c}
\end{bmatrix}$ where $\tilde{A}$ is a square invertible $k \times k$ matrix representing $k$ linearly independent rows of $A_k$ and $\tilde{A}_c$ is $(|\mathcal{Z}|-k) \times k$. Now $\alpha = \begin{pmatrix} \tilde{c}_k'\tilde{A}^{-1}\\ \mathbf{0}^{(|\mathcal{Z}|-k)\times 1} \end{pmatrix}$ represents a solution to $c'=\alpha'A$ because $[\tilde{c}_k' \tilde{A}^{-1},  \mathbf{0}^{1 \times (|\mathcal{Z}|-k)}] A = [\tilde{c}_k' \tilde{A}^{-1},  \mathbf{0}^{1 \times (|\mathcal{Z}|-k)}] \begin{bmatrix}
	\tilde{A} \\ {\tilde{A}_c}
\end{bmatrix}[I_k,M] = \tilde{c}_k'[I_k,M] = c'$.

In all of the three cases considered above, we can write any non-zero elements $\alpha_z$ of a $\alpha$ yielding a binary combination as components $x_z$ of $x=M^{-1}b$, where $M$ is an invertible $n \times n$ binary matrix (i.e. having entries of $0$ or $1$), and $b$ an n-component binary vector. Equivalently, $x$ represents the unique solution to $Mx=b$. Cramer's rule for such a solution establishes that the $x_z$ can be written as $x_z = \frac{det(M_z)}{det(M)}$, where $M_z$ is a matrix that replaces the column $z$ of the matrix $M$ with the vector $b$. Since both $M$ and $b$ are composed of binary entries, the matrix $M_z$ is always binary as well. The result now follows as stated in Proposition \ref{prop:limitsearch} since $0$ is always a possible value of $det(M_z)$.

\subsection{Proof of Proposition \ref{prop:misleading}}
Given $\mathbbm{E}[\nu_i|Z_i=0]$, the parameter $\gamma_3-\gamma_1-\gamma_2$ is given by
\begin{align*}
	\gamma_3-\gamma_1-\gamma_2 &= \mathbbm{E}[Y_i|Z_i=C]-\mathbbm{E}[Y_i|Z_i=A]-\mathbbm{E}[Y_i|Z_i=B]+\mathbbm{E}[Y_i|Z_i=0]\\
	&=\mathbbm{E}[Y_i(T_i(C))-Y_i(T_i(A))-Y_i(T_i(B))+Y_i(T_i(0))]\\
	&=\mathbbm{E}[Y_i(C)-Y_i(A)-Y_i(B)+Y_i(0)] + \mathbbm{E}[Y_i(T_i(C))-Y_i(C)]\\
	&\hspace{1cm}-\mathbbm{E}[Y_i(T_i(A))-Y_i(A)]-\mathbbm{E}[Y_i(T_i(B))-Y_i(B)]+\mathbbm{E}[Y_i(T_i(0))-Y_i(0)]
\end{align*}
Each of the last three terms in the final line can differ from zero in ways that do not offset one another, provided that imperfect compliance is allowed, i.e. $P(T_i(z) \ne z) > 0$ for some $z \in \mathcal{Z}$.

		\section{Relationship to \citet*{navjeevan2023identification}} \label{sec:nps}
This section discusses how the results of this paper relate to recent results characterizing identification in IV models by \citet*{navjeevan2023identification} (NPS).

NPS consider \textit{unconditional} expectations of functions taking the form $\mathbbm{E}[\ell(\tilde{Y}_i,G_i)]$, which in general are allowed to mix potential outcomes and potential treatments, as well as covariates. NPS do not define or explore in depth a notion of ``outcome-nonrestrictive'' identification, as their framework allows the researcher to impose restrictions on outcomes of the types discussed in Section \ref{sec:notapply}.

NPS do mention conditional average treatment effects as a motivation for specializing their general result to cases in which $\ell$ takes the separable form $Y_i(t)\cdot c(G_i)$, for some $t \in \mathcal{T}$ (see their Section 4.4). In these separable cases, NPS derive results that are related to but distinct from my Theorems \ref{thm:suff} and \ref{thm:necc} (which were obtained independently). 

In particular, Corollary 4.4 of NPS assumes discrete instruments and supposes that no additional restrictions are placed on the distribution of unobservables aside from the existence of finite first moments. This model is thus essentially the same as the model $M$ (see Footnote \ref{fn:idzoo}) I use to define outcome-nonrestrictive identification. From Corollary 4.4, NPS derive two important implications. Firstly, they find that the conditions on function $c(\cdot)$ for identification of $\mathbbm{E}[Y_i(t)\cdot c(G_i)]$ are equivalent to those for identification of $\mathbbm{E}[f(Y_i(t))\cdot c(G_i)]$ for any bounded function $f(\cdot)$. This implies that $\mathbbm{E}[f(Y_i(t))\cdot c(G_i)]$ is identified if and only if $P(c(G_i)=1)$ is (take $f(\cdot)=1$). Second, NPS find that a moment of the form $\mathbbm{E}[f(Y_i(t))\cdot c(G_i)]$ is identified if and only if the function $c(g)$ can be written as $\mathbbm{E}[\kappa(Z_i)\cdot \mathbbm{1}(T_i=t)|G_i=g]$ for some function $\kappa$. 
Though NPS do not characterize it in this way, one can see that this is equivalent to $c \in rs(A^{[t]})$ by applying the law of iterated expectations over $Z_i$.\footnote{The closest way in NPS of stating this condition to $c \in rs(A^{[t]})$ seems to be Eq. (28) from their discussion of the selection model of \citet{klinewalters}. In my notation their Eq. (28) reads as $\min_{\alpha \in \mathbbm{R}^{|\mathcal{Z}|}} \left(c(g) - \sum_z \alpha_z A^{[t]}_{zg}\right)^2 = 0$, which is equivalent to $c \in rs(A^{[t]})$.}

\begin{figure}[h!]
	\centering
	\begin{tikzpicture}
		
		\node (F1) at (1,-1.5) {$c \in rowspace(A^{[t]})$};
		\node (F2) at (6,-.5) [yshift=1.5ex] {$\mathbbm{E}[Y_i(t)\cdot c(G_i)]$ identified};
		\node (F3) at (6,-2.5) [yshift=-1.5ex] {$P(c(G_i)=1)$ identified};
		\node (F4) at (12.5,-1.5) {$\mathbbm{E}[Y_i(t)|c(G_i)=1]$ identified};
		\node (F5) at (9.5,-1.5) {$\implies$};
		
		\draw[implies-implies, double equal sign distance, purple] (F1) -- (F2) node[midway,above=.5mm, scale=.75]{NPS C4.4};
		\draw[implies-implies, double equal sign distance, purple] (F1) -- (F3) node[midway,below=.5mm, scale=.75]{NPS C4.4};
		\draw[implies-implies, double equal sign distance, purple] (F2) -- (F3) node[midway,right=.5mm, scale=.75]{NPS C4.4};
		\draw[decorate, decoration={brace, amplitude=10pt}] (8.7,-.5) -- (8.7,-2.5);
		
	\end{tikzpicture} \vspace{.2in}
	\caption{On left, $\iff$ symbols (in purple) depict implications Corollary 4.4 of NPS, for parameters of the form $\mathbbm{E}[Y_i(t)|c(G_i)=1]$. On right, $\implies$ symbol (in black) depicts an implication of $\mathbbm{E}[Y_i(t)|c(G_i)=1] = \frac{\mathbbm{E}[f(Y_i(t))\cdot c(G_i)]}{P(c(G_i)=1)}$. \label{fig:implications1}}
\end{figure}

These results of NPS are summarized in Figure \ref{fig:implications1}. Taken together, they imply Theorem \ref{thm:suff} but not Theorem \ref{thm:necc} of this paper. Since $\mathbbm{E}[Y_i(t)\cdot c(G_i)]$ and $P(c(G_i)=1)$ are both identified when $c$ is in the rowspace of $A^{[t]}$, the results of NPS readily establish that $\mathbbm{E}[Y_i(t)|c(G_i)=1]$ is identified provided that $P(c(G_i)=1)>0$, in the case of a binary valued function $c$. However, their results do not establish that the conditional expectation $\mathbbm{E}[Y_i(t)|c(G_i)=1]$ is \textit{only} identified when $c \in rs(A^{[t]})$ holds. Instead, they show that $\mathbbm{E}[Y_i(t)\cdot c(G_i)]$ and $P(c(G_i)=1)$ can only be identified separately if $c \in rs(A^{[t]})$ holds.

By contrast, Theorem \ref{thm:necc} establishes the necessary direction of $c \in rs(A^{[t]})$ for when $\mathbbm{E}[Y_i(t)|c(G_i)=1]$ is identified (given discrete instruments), as depicted in Figure \ref{fig:implications2} below. While Theorem \ref{thm:suff} establishes that $\mathbbm{E}[Y_i(t)\cdot c(G_i)]$, $P(c(G_i)=1)$ and $\mathbbm{E}[Y_i(t)|c(G_i)=1]$ are all identified if $c$ belongs to the rowspace of $A^{[t]}$, Theorem \ref{thm:necc} establishes that $\mathbbm{E}[Y_i(t)|c(G_i)=1]$ is \textit{only} identified if $c$ belongs to the rowspace of $A^{[t]}$. 
\begin{figure}[h!]
	\centering
	\begin{tikzpicture}
		\node (F1) at (1,-1.5) {$c \in rowspace(A^{[t]})$};
		\node (F2) at (6,-.5) [yshift=1.5ex] {$\mathbbm{E}[Y_i(t)\cdot c(G_i)]$ identified};
		\node (F3) at (6,-2.5) [yshift=-1.5ex] {$P(c(G_i)=1)$ identified};
		\node (F4) at (12.3,-1.5) {$\mathbbm{E}[Y_i(t)|c(G_i)=1]$ identified};
		
		\draw[-implies, double equal sign distance, orange] (2.75,-1.4) -- (9.5,-1.4) node[midway,above=.5mm, scale=.75]{Thm. 1};
		\draw[implies-, double equal sign distance, blue] (2.75,-1.6) -- (9.5,-1.6) node[midway,below=.5mm, scale=.75]{Thm. 2};
		\draw[-implies, double equal sign distance, orange] (F1) -- (F3) node[midway,left=.75mm,below=.5mm,scale=.75]{Thm. 1};
		\draw[-implies, double equal sign distance, orange] (F1) -- (F2) node[midway,left=.5mm,above=.5mm, scale=.75]{Thm. 1};
	\end{tikzpicture} \vspace{.2in}
	\caption{Implications of Theorems \ref{thm:necc} (in blue) and \ref{thm:suff} (in orange) of this paper. \label{fig:implications2}}
\end{figure}


Beyond Theorem \ref{thm:necc}, the present paper also differs from NPS by its exploration of the implications of $c \in rs(A^{[t]})$ for the identification of conditional average treatment effects, in the case that $c$ is binary-valued. This requires finding functions $c$ that belong to the \textit{intersection} of rowspaces of $A^{[t]}$ and $A^{[t']}$ for $t'\ne t$ together with the unit cube, as we saw in Section \ref{sec:geomTEs}. This analysis shows, in the positive direction (Theorem \ref{thm:suff}), how $c \in rs(A^{[t]})$ synthesizes many identification results for treatment effects from the literature (Appendix \ref{sec:examples}). In the other direction (Theorem \ref{thm:necc}), this allows one to exhaustively catalog identification results for a given support of $T_i$ and $Z_i$, as described in Section \ref{sec:bruteforce}.\\

\noindent \textit{An illustrative example:} To appreciate the difference between $\mathbbm{E}[Y_i(t)|c(G_i)=1]$ being identified and $\mathbbm{E}[Y_i(t)\cdot c(G_i)]$ being identified, consider a setting with a binary treatment and binary instrument in which $\mathcal{G}$ allows all four response types: always-takers, never-takers, compliers and defiers. The choice model $\mathcal{G}$ is represented by the matrix $A=A^{[1]}$:\\

\begin{table}[h!]
	\centering
	\begin{tabular}{c|ccccccc}
		&$n.t.$ & $comp.$ & $def.$ & $a.t.$\\
		\hline
		$\mathbf{z=0}$ & 0 & 0 & 1 & 1\\
		$\mathbf{z=1}$ & 0 & 1 & 0 & 1\\
		\hline
	\end{tabular}
\end{table} \vspace{-.25cm}
\noindent Since $c=(0,1,0,0)'$ does not belong to the rowspace of $A$, treatment effects or counterfactual means among compliers are not outcome-nonrestrictive identified. Similarly, the proportion of compliers is not identified.\footnote{The difference $\mathbbm{E}[D_i|Z_i=1]-\mathbbm{E}[D_i|Z_i=0]$ instead identifies $P(G_i=comp.)-P(G_i=def.)$. We can also identify the quantities $\{P(G_i=comp.)+P(G_i=a.t.)\}$, $\{P(G_i=def.)+P(G_i=a.t.)\}$, $\{P(G_i=comp.)+P(G_i=n.t.)\}$ and $\{P(G_i=def.)+P(G_i=n.t.)\}$, but not $P(G_i=comp.)+P(G_i=def.)$.} However, it is straightforward to see that if one maintains the assumption that compliers and defiers share the same average treatment effect, then the average treatment effect among compliers becomes identified and is equal to the conventional Wald ratio \citep{imbensangristrubin} $(\mathbbm{E}[Y_i|Z_i=1]-\mathbbm{E}[Y_i|Z_i=0])/(\mathbbm{E}[D_i|Z_i=1]-\mathbbm{E}[D_i|Z_i=0])$.
This example demonstrates that a parameter like $\mathbbm{E}[Y_i(1)-Y_i(0)|G_i=comp.]$ can in general be identified even when $P(G_i=comp.)$ and $\mathbbm{E}[\{Y_i(1)-Y_i(0)\} \cdot \mathbbm{1}(G_i=comp.)]$ are not, if restrictions are imposed on outcomes.\footnote{Note further that although we can also write $\mathbbm{E}[Y_i(1)-Y_i(0)|G_i=comp.]=\frac{\mathbbm{E}[Y_i(1) \cdot \mathbbm{1}(G_i=comp.)]}{\mathbbm{E}[\mathbbm{1}(G_i=comp.)]}-\frac{\mathbbm{E}[Y_i(0) \cdot \mathbbm{1}(G_i=comp.)]}{\mathbbm{E}[\mathbbm{1}(G_i=comp.)]}$
none of the quantities $\mathbbm{E}[Y_i(1) \cdot \mathbbm{1}(G_i=comp.)]$, $\mathbbm{E}[Y_i(0) \cdot \mathbbm{1}(G_i=comp.)]$, or $\mathbbm{E}[\mathbbm{1}(G_i=comp.)]$ are identified in isolation, even with the outcome restriction that $\mathbbm{E}[Y_i(1)-Y_i(0)|G_i=comp.]=\mathbbm{E}[Y_i(1)-Y_i(0)|G_i=def.]$.} Theorem \ref{thm:necc} shows that this however cannot occur when identification is outcome-nonrestrictive.

		\section{Illustrative examples from the brute force search} \label{app:examples}

\subsubsection{Binary treatment and binary instrument} \label{sec:cdmodel}
With a binary instrument and a binary treatment, the brute-force search reveals that there is exactly one additional choice model beyond the classic LATE model of \citet{Imbens2018} that admits of outcome-nonrestrictive identification of treatment effects. 

Instead of ruling out defiers, suppose that we rule out both always-takers and never-takers, i.e. all units are affected by the treatment. In this case the matrix $A$ becomes:
\begin{table}[h!]
	\centering
	\begin{tabular}{c|cc}
		&compliers & defiers \\
		\hline
		$\mathbf{z=0}$ & 0 & 1\\
		$\mathbf{z=1}$ & 1 & 0\\
		\hline
	\end{tabular}
\end{table}

\noindent The rowspaces of $A^{[1]}$ and $A^{[0]}$ are the same and both span $\mathbbm{R}^2$: $rs(A^{[1]})=rs(A^{[0]}) = span\left\{\begin{pmatrix}
	1 \\ 0
\end{pmatrix},\begin{pmatrix}
	0 \\ 1
\end{pmatrix}\right\}$. Thus, given the results of Section \ref{sec:discrete}, we know that binary collections correspond to any non-zero vertex of the unit cube in $\mathbbm{R}^2$, as depicted in Figure \ref{fig:cube2} below.
\begin{figure}[ht!]
	\begin{center}
		\hspace{1cm}\includegraphics[width=3.5in]{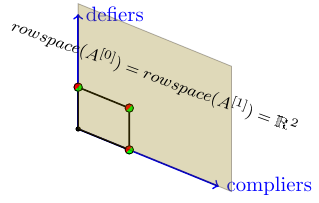}	
	\end{center}
	\caption{A geometric depiction of the model with compliers and defiers only. The vectors $c=(0,1)'$, $c=(1,0)'$ and $c=(1,1)'$ all belong to both $rs(A^{[1]})$ and $rs(A^{[0]})$ and hence the LATE for either response type or the ATE are identified. As in Figure \ref{fig:cube2}, the split-shading of a given vertex (red/green in color) of the unit square indicates that it lies in $rs(A^{[0]}) \cap rs(A^{[1]})$ and is not equal to the zero vector. \label{fig:cube2}}
\end{figure}

\noindent For example, we can identify the LATE for compliers by choosing $\alpha_1 = (0,1)$ $\alpha_0=(1,0)$ (using the notation $\alpha_0,\alpha_1$ introduced in Section \ref{sec:geomTEs}):
$$\mathbbm{E}[Y_i(1)-Y_i(0)|i \textrm{ is complier}] = \frac{\mathbbm{E}[Y_i \cdot D_i|Z_i=1]}{\mathbbm{E}[D_i|Z_i=1]}-\frac{\mathbbm{E}[Y_i \cdot (1-D_i)|Z_i=0]}{1-\mathbbm{E}[D_i|Z_i=0]}$$
using the typical notation $D_i = T_i$ for a binary treatment. An analagous construction identifies the average treatment effect among defiers.\footnote{\citet{leegoff} apply this choice model to study the effect of an NFL team deferring the kickoff on the game outcome, with the kickoff coin flip that decides which team is given the option to defer as the instrument. If it is common knowledge between the teams whether receiving the kickoff is beneficial in that particular game, then a simple model of optimizing play would predict that each game will either be a ``complier'' or a ``defier''.} Note that by Theorem \ref{thm:suff}, $\mathbbm{E}[D_i|Z_i=1]$ and $1-\mathbbm{E}[D_i|Z_i=0]$ are both measures for the same population parameter $P(c(G_i)=1) = P(i \textrm{ is complier})$. Thus $\mathbbm{E}[D_i|Z_i=1]=1-\mathbbm{E}[D_i|Z_i=0]$ can be used as an overidentification restriction for this choice model (there is no such restriction for the LATE model that simply rules out defiers).

\subsubsection{3 treatments, binary instrument} \label{sec:klinewalters}
Now suppose that the instrument is binary and $\mathcal{T} = \{0,1,2\}$. With no restrictions on selection behavior, there are $3^{|\mathcal{Z}|}=9$ conceivable response types. Table \ref{table:bruteforceresults} reports that in this case there are five selection models that afford a total of five distinct outcome-nonrestrictive identification results. These results are all listed in Appendix \ref{sec:catalog}.

As an example that may be empirically relevant, consider a selection model that has $T(z)$ increasing in $z$, but rules out always-1 takers and individuals that skip from $t=0$ to $t=2$ when $z$ is increased from 0 to 1. This restriction leaves four response types: always-0 takers, always-2 takers, individuals who move from treatment 0 to treatment 1, and individuals who move from treatment 1 to treatment 2. SM.3.2.1 reported in Appendix \ref{sec:catalog} reveals that two treatment effects $\Delta^{t,t'}_c$ are identified in this choice model. For example, the quantity
\begin{align*}
	\frac{\mathbbm{E}[Y_i\cdot D_{i}^{[1]}|Z_i=1]}{\mathbbm{E}[D_{i}^{[1]}|Z_i=1]}-\frac{\mathbbm{E}[Y_i\cdot D_{i}^{[0]}|Z_i=0]-\mathbbm{E}[Y_i\cdot D_{i}^{[0]}|Z_i=1]}{\mathbbm{E}[D_{i}^{[0]}|Z_i=0]-\mathbbm{E}[D_{i}^{[0]}|Z_i=1]}
\end{align*}
\noindent corresponds to the binary collection with $\alpha_0 = (0,1)'$ and $\alpha_1 = (1,-1)'$, and identifies $\mathbbm{E}[Y_i(1)-Y_i(0)|T_i(0)=0, T_i(1)=1]$. The treatment effect $\mathbbm{E}[Y_i(2)-Y_i(1)|T_i(0)=1, T_i(1)=2]$ is identified by a similar estimand.

This selection model could represent a setting in which $t=2$ represents passing a test outright, $t=1$ passing the test ``provisionally'', and $t=0$ failing. Suppose that a reform implemented for some schools lowers the score threshold $\tau_o$ for an outright pass to the old threshold $\tau_p$ for a provisional pass, while further lowering $\tau_p$, as depicted below:
\tikzstyle{line} = [draw]
\begin{center}
	\begin{tikzpicture}
		\node (F1) at (0,1) {$\mathbf{(z=0)}$};
		\node (F2) at (4.85,1) {$\tau_p(0)$};
		\node (F3) at (9,1) {$\tau_o(0)$};
		\node (F4) at (13,1) {};
		
		\node (G1) at (0,0.5) {$\mathbf{(z=1)}$};
		\node (G3) at (4.85,0.5) {$\tau_o(1)$};
		\node (G2) at (2.5,0.5) {$\tau_p(1)$};
		\node (G4) at (13,0.5) {};
		
		\draw[<-, red, very thick] (F1) -- (F2);
		\draw[-, orange, very thick] (F2) -- (F3);
		\draw[->, green, very thick] (F3) -- (F4);
		
		\draw[<-, red, very thick] (G1) -- (G2);
		\draw[-, orange, very thick] (G2) -- (G3);
		\draw[->, green, very thick] (G3) -- (G4);
	\end{tikzpicture}
\end{center}
\noindent Students will then belong to one of the four types described above, depending on their test score. The quantity $\mathbbm{E}[Y_i(1)-Y_i(0)|T_i(z)=z]$ represents the average effect of moving from a fail to a provisional pass among the students who are brought in to a provisional pass by the grading reform.\\

\noindent \textit{The selection model of \citet{klinewalters}:} Another observation from the $|\mathcal{T}|=3, |\mathcal{Z}|=2$ setting is that the selection model of \citet{klinewalters} (KW) does not appear in the catalog of Appendix \ref{sec:catalog}. KW study a setting in which the binary instrument is an offer to choose $t=2$, while treatments $t=0$ and $t=1$ are always available even if $z=0$. On the grounds of revealed preference, KW impose that $T_i(1) \ne T_i(0) \implies T_i(1)=2$ which results in a selection model with five response types. KW show that the parameter $\mathbbm{E}[Y_i(2)-Y_i(T_i(0))|T_i(1)=2,T_i(0)\ne 2]$ is then identified. The quantity $T_i(0)$ represents an individual's next preferred alternative to $t=2$, which may vary across those $i$ for whom $T_i(1)=2,T_i(0)\ne 2$. As a result, this parameter does not fit the form of the general family of treatment effect parameters $\Delta^{t,t'}_c$ introduced in Section \ref{sec:oaid}. Indeed, the brute force search confirms that unfortunately \textit{no} parameters of the form $\Delta^{t,t'}_c$ between two fixed treatmets $t$ and $t'$ are identified in the KW selection model.

\subsubsection{Binary treatment, 3 instrument values}
Suppose now that treatment is binary and $\mathcal{Z} = \{0,1,2\}$. Table \ref{table:bruteforceresults} reports that in this case there are 11 selection models that afford a total of 30 distinct outcome-nonrestrictive identification results, listed in Appendix \ref{sec:catalog}. One observation that emerges when we extend the analysis to instruments that take more than two values is that, there now exist binary collections that require coefficients $\alpha_z$ that do not belong to the set $\{-1,0,1\}$. Although this is entirely consistent with Proposition \ref{prop:limitsearch} for $|\mathcal{Z}|\ge 3$, a reasonable conjecture ex-ante might have been that $\alpha_z \in \{-1,0,1\}$ always holds, given the preponderance of this pattern in known identification results (see for example all of the results surveyed in Appendix \ref{sec:examples}).

For the sake of exposition, let us consider a ``judge-IV'' setting in which defendants $i$ receive a bail decision from a randomly-assigned judge $z$, with $t=1$ indicating that the defendent remains incarcerated and $t=0$ that they are released on bail. Suppose that the defendants are one of three types $g \in \mathcal{G}$ (typically unobserved to the researcher), comprising the columns of the table below. ``Prepared'' defendants dress formally and speak politely in their bail hearing, perhaps also presenting evidence that they are not a danger to the community. ``Unprepared'' defendants do not make such efforts. A third category of ``flight-risk'' defendants are thought to be particularly capable of and likely to fail to appear for trial if they are granted bail (while this is not true of the first two groups, e.g. due to strong personal ties to the jurisdiction or insufficient financial means to leave town).\\

\begin{table}[h!]
	\centering
	\begin{tabular}{c|ccccccc}
		&prepared & unprepared & flight-risk\\
		\hline
		$\mathbf{z=0}$ (standard) & 0 & 0 & 1\\
		$\mathbf{z=1}$ (character) & 0 & 1 & 0\\
		$\mathbf{z=2}$ (skeptics) & 1 & 0 & 1\\
		\hline
	\end{tabular}
\end{table}

\noindent The above table summarizes selection behavior when the judges also belong to one of three types, represented across rows. ``Standard'' judges are only concerned with failure to appear, and keep only the flight-risk defendants incarcerated. ``Character'' judges instead attempt to infer the risk of a defendant to public safety on the basis of the defendant's presentation and arguments to their character made in the bail hearing, but do not attempt to assess whether the defendant is likely to skip town. ``Skeptic'' judges are also sensitive to judgments about presentation, but in the opposite direction: they are suspicious of defendants precisely when they seem to be making a case that they are not dangerous. They deny bail for the prepared defendants, and also deny bail for flight-risk defendants.\footnote{This selection model is equivalent to SM.2.3.4 in Appendix \ref{sec:catalog}, after permuting instrument/treatment labels. Note that this model is merely illustrative: more types and nuance in their definitions could add some realism.}


Note that this model does not satisfy the strong LATE monotonicity assumption typically invoked in judge-IV settings, which has been challenged on empirical grounds \citep{frandsenetal2023,sigstad}. If there were no skeptic $(z=2)$ judges, then this model would instead consist of compliers, defiers, and never-takers, which we have already seen permits no outcome-nonrestrictive identification results for treatment effects. However, the presence of the skeptics aids here in identification, as we can then identify the average effect of incarceration among two groups $\mathbbm{E}[Y_i(1)-Y_i(0)|G_i \in \{unprepared, flight\textrm{-}risk\}]$ by $\frac{\mathbbm{E}[Y_iD_i|Z_i=0]+\mathbbm{E}[Y_iD_i|Z_i=1]}{\mathbbm{E}[D_i|Z_i=0]+\mathbbm{E}[D_i|Z_i=1]} - \frac{-\mathbbm{E}[Y_iD_i|Z_i=0]+\mathbbm{E}[Y_iD_i|Z_i=1]+2\cdot \mathbbm{E}[Y_iD_i|Z_i=2]}{-\mathbbm{E}[D_i|Z_i=0]+\mathbbm{E}[D_i|Z_i=1]+2\cdot \mathbbm{E}[D_i|Z_i=2]}$, corresponding to the binary collection with $\alpha_1 = (1,1,0)$ and $\alpha_0=(-1,1,2)'$.\footnote{A coefficient of two is inevitable for $t=0$, since we need $\alpha_1 = 1$ (using the $\alpha_z$ notation) for $c({flight\-risk})=0$, $\alpha_0 = -\alpha_1$ to get $c(prepared)=0$, but can only achieve $c(unprepared)=1$ if $\alpha_2+\alpha_1 = 1$.} Note that using this result requires judge types to be observable, or estimable from judges each seeing many cases.

\subsubsection{3 treatment values, 3 instrument values} \label{sec:3by3}
In the case of $\mathcal{Z} = \mathcal{T} = \{0,1,2\}$, the brute force approach returns 251 distinct binary collections spread across 251 unique selection models. These results nest for example two identification results presented in \citet*{kirkeboenleuvenmogstad} (KLM). KLM consider an unordered treatment which represents a student's field of study, where students are ``assigned'' to a given field, i.e $Z_i=j$ represents an incentive to choose field $j$. Proposition 2 of KLM presents three special cases in which a two stage least squares estimand with indicators for treatments $1$ and $2$ instrumented by indicators for $Z_i=1$ and $Z_i=2$ recovers causally interpretable coefficients. While their first result (restricting treatment effects to be homogeneous) is not outcome-nonrestrictive, the other two results are.

For example, the second result in KLM Proposition 2 shows that if preferences are further restricted so that $D^{[2]}_i(1)=D^{[2]}_i(0)$ and $D^{[1]}_i(2)=D^{[1]}_i(0)$ for all $i$ (an offer to one program does not affect whether or not the student chooses the other program), then $\mathbbm{E}[Y_i(1)-Y_i(0)|D^{[1]}_i(1)>D^{[1]}_i(0)]$ and $\mathbbm{E}[Y_i(2)-Y_i(0)|D^{[2]}_i(2)>D^{[2]}_i(0)]$ are each identified.\footnote{Throughout, KLM also maintain a version of unordered monotonicity (cf. \citealt{Heckman2018}) in which $D_i^{[1]}(1) \ge D_i^{[1]}(0)$ and $D_i^{[2]}(2) \ge D_i^{[2]}(0)$: an offer of admission never causes a student to select out of that field.} Let us consider how the first of these results appears in the comprehensive search (the second result proceeds similarly). Upon a relabeling of treatment/instrument values and removing one response type,\footnote{In particular, label the treatments $(0,1,2)$ as $(1,0,2)$, swap instrument values 1 and 2, and drop column 6. All results for the $3 \times 3$ case are enumerated in a working paper version of this paper: \url{https://arxiv.org/abs/2406.02835}.} selection model SM.3.3.63 in the catalog amounts to the following: $ A= \begin{bmatrix}
	0 & 0 & 0 & 0 & 1 & 2 & 0\\
	0 & 1 & 0 & 1 & 1 & 2 & 2\\
	0 & 0 & 2 & 2 & 1 & 2 & 1\\
\end{bmatrix}$. This selection model has seven response types, whereas the choice model considered by KLM contains only the first six columns of A. In the larger selection model with all seven groups, the treatment effect $\mathbbm{E}[Y_i(1)-Y_i(0)|T_i(1) \ne T_i(0)]$ is identified by the binary collection with $\alpha_0 = (-2,1,1)'$ and $\alpha_1=(1,-1,0)'$.\footnote{Accordingly, $\mathbbm{E}[Y_i(1)-Y_i(0)|T_i(1) \ne T_i(0)]= \frac{\mathbbm{E}[Y_i\cdot D_i^{[1]}|Z_i=2] + \mathbbm{E}[Y_i\cdot D_i^{[1]}|Z_i=1] - 2 \cdot \mathbbm{E}[Y_i\cdot D_i^{[1]}|Z_i=0]}{\mathbbm{E}[D_i^{[1]}|Z_i=2] + \mathbbm{E}[D_i^{[1]}|Z_i=1] - 2 \cdot \mathbbm{E}[ D_i^{[1]}|Z_i=0]}-\frac{\mathbbm{E}[Y_i\cdot D_i^{[0]}|Z_i=1] - \mathbbm{E}[Y_i\cdot D_i^{[0]}|Z_i=0]}{\mathbbm{E}[D_i^{[0]}|Z_i=1] - \mathbbm{E}[D_i^{[0]}|Z_i=0]}$. Identification of $\mathbbm{E}[Y_i(2)-Y_i(0)|T_i(2) \ne T_i(0)]$ is analogous.} Thus, we have seen that KLM's choice model can be relaxed to allow an additional response type, with the same estimand that identifies $\mathbbm{E}[Y_i(1)-Y_i(0)|D^{[1]}_i(1)>D^{[1]}_i(0)]$ in their more restrictive model identifying $\mathbbm{E}[Y_i(1)-Y_i(0)|T_i(1) \ne T_i(0)]$ more generally. Code available from the author allows one to check in general whether a given selection model can be relaxed in this way, using the catalog of identification results (available for $|\mathcal{T}|,|\mathcal{Z}| \le 3$).
		\subsection{4 treatment values, 4 instrument values: spillover effects within pairs} \label{sec:spillovers}
Although the $|\mathcal{T}|=|\mathcal{Z}|=4$ case is not included in the brute-force search of Table \ref{table:bruteforceresults} (due to the computational burden), it remains easy to check for outcome-nonrestrictive identification in any given choice model using the results of Section \ref{sec:discrete}. This section presents an alternative application in the $4 \times 4$ case to supplement the study of interaction effects from Section \ref{sec:empirical}.

Consider a setting in which each unit $i$ has one ``neighbor'' $n(i)$, and we allow for violations of SUTVA within neighbor pairs $(i,n(i))$. This can be accommodated by expanding the set of treatments $\mathcal{T}$ to accommodate values of such pairs, so that $Y_i = Y_i(T_{i},T_{n(i)})$ where $T_{n(i)}$ is the treatment of the neighbor of unit $i$, indexed by $n(i)$. I consider the case in which treatment $T$ itself is binary, so that $\tilde{T}_i:=(T_{i},T_{n(i)})$ may take one of four values $(0,0),(1,0),(0,1),(1,1)$. Following the notation in Section \ref{sec:empirical}, we denote these pair-level ``treatments'' as $0,A,B,C$, where $\tilde{T}_i=0$ indicates that neither unit is treated, $\tilde{T}_i=A$ that only unit $i$ is treated, $\tilde{T}_i=B$ that only their neighbor is treated, and $\tilde{T}_i=C$ that both $i$ and their neighbor is treated.

For each $i$, let $Z_i$ be a binary instrument that reflects whether $i$ is ``assigned'' to receive treatment. See \citet{kang2016peerencouragementdesignscausal} and \citet{Vazquez-Bare03072023} for related setups. Let $\tilde{T}_i(z,z')$ reflect the treatments for the pair as a function of treatment assignments $(z,z')$ for the pair. Let $\tilde{Z}_i:=(Z_i,Z_{n(i)})$ be the pairs realized treatment assignment, which can take any of four counterfactual values $z \in \tilde{\mathcal{Z}}=\{0,A,B,C\}$. I maintain throughout two assumptions about selection: i) first, that the $T_i$ component of $\tilde{T}_i(z,z')$ only depends on $z$, and that the $T_{n(i)}$ component of $\tilde{T}_i(z,z')$ only depends on $z'$; and ii) secondly, that each selection uptake is monotonic such that $T_i(1) \ge T_i(0)$, where we write $T_i(z,z')$ as $(T_i(z),Y_{n(i)}(z'))$. 

These restrictions leave nine response types, enumerated in the table below:
\begin{table}[h!]
	\centering \small
	\begin{tabular}{c|ccccccccccc}
		\textbf{assigned} $\downarrow$ & (nt,nt) & (cm,cm) & (cm,nt) & (nt,cm) & (at,at) & (at,cm) & (cm,at) & (nt,at) & (at,nt)\\
		\hline
		0=(0,0) & 0 & 0 & 0 & 0& C & A & B & B & A \\
		A=(1,0) & 0 & A & A & 0& C & A & C & B & A \\
		B=(0,1) & 0 & B & 0 & B& C & C & B & B & A \\
		C=(1,1) & 0 & C & A & B& C & C & C & B & A \\
		\hline
	\end{tabular} \vspace{.25cm}
\end{table}

\noindent Given a function $c(\cdot)$, the local average direct effect of one's own treatment on their outcome is:
\begin{align*}
	LADTE&:=\mathbbm{E}[Y_i(1,Z_{n(i)})-Y_i(0,Z_{n(i)})|c(G_i)=1]\\
	&=P(Z_{n(i)}=1)\cdot \mathbbm{E}[Y_i(C)-Y_i(B)|c(G_i)=1] + P(Z_{n(i)}=0)\cdot \mathbbm{E}[Y_i(A)-Y_i(0)|c(G_i)=1]
\end{align*}
using that $P(Z_{n(i)}=1|c(G_i)=1)=P(Z_{n(i)}=1)$ by independence.

Similarly, the local average spillover (indirect) effect is
\begin{align*}
	LASTE&:=\mathbbm{E}[Y_i(Z_i,1))-Y_i(Z_i,0)|c(G_i)=1]\\
	&=P(Z_{i}=1)\cdot \mathbbm{E}[Y_i(C)-Y_i(A)|c(G_i)=1] + P(Z_{i}=0)\cdot \mathbbm{E}[Y_i(B)-Y_i(0)|c(G_i)=1]
\end{align*}

\noindent Since the distributions of $Z_i$ and $Z_{n(i)}$ are identified, we can then point identify the LADTE and the LASTE provided that we can identify $\mu_{t}^c$ for all of $t \in \{0,A,B,C\}$. An application of Theorems \ref{thm:suff} and \ref{thm:necc} shows that this is possible without restricting outcomes in the above selection model if and only if $c(g) = \mathbbm{1}(g = cm,cm)$. 

This can be shown by direct enumeration of the $2^9$ possible functions $c: \mathcal{G} \rightarrow \{0,1\}$. The vector forms $\alpha_t$ of the resulting coefficient functions $\alpha^{[t]}(z)$ are:
\begin{align}
	\alpha_0 &= (+1, -1, -1, +1)', \quad \alpha_A = (-1, +1, +1, -1)' \label{eq:bcs2}\\
	\alpha_B &= (-1, +1, +1, -1)', \quad \alpha_C = (+1, -1, -1, +1)' \nonumber
\end{align}
Similar to the case of complementarities, the selection model therefore implies the overidentification restriction that
\normalsize
\begin{align*}
	p&:=P(T_i=0|Z_i=(1,1))-P(T_i=0|Z_i=(0,1))-P(T_i=0|Z_i=(1,0))+P(T_i=0|Z_i=(0,0))\\
	&=-P(T_i=A|Z_i=(1,1))+P(T_i=A|Z_i=(0,1))+P(T_i=A|Z_i=(1,0))-P(T_i=A|Z_i=(0,0))\\
	&=-P(T_i=B|Z_i=(1,1))+P(T_i=B|Z_i=(0,1))+P(T_i=B|Z_i=(1,0))-P(T_i=B|Z_i=(0,0))\\
	&=P(T_i=C|Z_i=(1,1))-P(T_i=C|Z_i=(0,1))-P(T_i=C|Z_i=(1,0))+P(T_i=C|Z_i=(0,0))
\end{align*}
\large
for some value $p \in [0,1]$, which identifies $P(g = \textrm{cm,cm})$.
		\section{Supplemental material for the application to interaction effects} \label{app:additionalinteraction}

\subsection{Motivating the restriction imposed by Proposition \ref{prop:interactionid}} \label{sec:sepchoice}
We can rationalize the restriction $\mathcal{G} \subseteq \mathcal{G}^{sep}$ made in Proposition \ref{prop:interactionid} by supposing that individuals choose \textit{separately} whether to receive treatment $A$ or $B$, rather than as a single joint decision. Let $S(z)$ denote the set of treatments among $\{A,B\}$ offered to an individual when their instrument realization is $z \in \{\textrm{neither}, \textrm{A}, \textrm{B}\, \textrm{both}\}$. That is, $S(\textrm{neither}) = \emptyset$, $S(\textrm{A})=\{A\}$, $S(\textrm{B})=\{B\}$, $S(\textrm{both})=\{A,B\}$.

\begin{definition*}
	We say that the population exhibits \textbf{separable choices} if their counterfactual selection satisfies for each $z \in \{\textrm{neither}, \textrm{A}, \textrm{B}\, \textrm{both}\}$:
	$$T_i(z) = \{t \in S(z): U_i(t) \ge 0\}$$
	where treatment $C$ is here understood as the set of treatments $\{A,B\}$, and treatment $0$ is understood as the null set $\emptyset$.
\end{definition*}
\noindent Separable (counterfactual) choices says that individuals choose treatment $A$ if and only if $U_i(A) \ge 0$ and $B$ if and only if $U_i(B) \ge 0$, subject to the options offered to them. This implies that $T_i(\textrm{both})=C \implies T_i(A)=A \textrm{ and } T_i(B)=B$, and similarly that $T_i(A)=A \textrm{ and } T_i(B)=B\implies T_i(\textrm{both})=C$. This eliminates exactly the remaining five groups displayed in gray in Table \ref{table:cross}.

\subsection{Identification with covariates} \label{app:covs}
Suppose that instead of \eqref{eq:independence} we have
\begin{equation} \label{eq:independencecovs}
	\{Z_i \indep (\tilde{Y}_i,G_i)\} | X_i
\end{equation}
where $X_i$ are observed covariates that are unaffected by treatment. This holds, for example, if the instruments are independent of these covariates jointly with the latent heterogeneity $(\tilde{Y}_i,G_i)$ across individual: $Z_i \indep (\tilde{Y}_i,G_i,X_i)$.

Consider a binary combination $(t,\alpha)$ such that $\sum_k \alpha_k D^{[t]}_i(z_k) = c^{[t,\alpha]}(G_i)$ for all $i$ where $c^{[t,\alpha]}(G_i) \in \{0,1\}$ for all $g \in \mathcal{G}$. I do not consider the case in which $\sum_k \alpha_k D^{[t]}_i(z_k) = c^{[t,\alpha]}(G_i,X_i)$ for some function $c^{[t,\alpha]}$ that depends both on $G_i$ and $X_i$, though such an extension would be possible. By the steps that establish Eq. \eqref{idresult} in the uncondional case, \eqref{idresult} generalizes to
\begin{equation} \label{idresultcovs}
	\mathbbm{E}\left[Y_i(t)\left|c^{[t,\alpha]}(G_i)=1,X_i=x\right.\right] = \frac{\sum_{k=1}^K\alpha_{k}\cdot \mathbbm{E}\left[Y_i\cdot D^{[t]}_i|Z_i=z_k,X_i=x\right]}{\sum_{k=1}^K\alpha_{k}\cdot \mathbbm{E}\left[D^{[t]}_i|Z_i=z_k,X_i=x\right]}
\end{equation}
for any value $x$. Notice that although $P(c^{[t,\alpha]}(G_i)=1|X_i=x)=\mathbbm{E}[c^{[t,\alpha]}(G_i)|X_i=x]$ might vary with $x$, it is identified by the denominator of the above for each: $P(c^{[t,\alpha]}(G_i)=1|X_i=x)=\sum_{k=1}^K\alpha_{k}\cdot \mathbbm{E}\left[D^{[t]}_i|Z_i=z_k,X_i=x\right]$. Consequently, the overall counterfactual mean that does not condition on $x$ is identified as
\begin{align*}
	\mathbbm{E}[Y_i(t)|c^{[t,\alpha]}(G_i)=1] &= \int dF_{X|c^{[t,\alpha]}(G)=1}(x) \cdot \mathbbm{E}[Y_i(t)|c^{[t,\alpha]}(G_i)=1,X_i=x]\\
	&= \int dF_{X|c^{[t,\alpha]}(G)=1}(x) \cdot\frac{\sum_{k=1}^K\alpha_{k}\cdot \mathbbm{E}\left[Y_i\cdot D^{[t]}_i|Z_i=z_k,X_i=x\right]}{P(c^{[t,\alpha]}(G_i)=1|X_i=x)}\\
	&= \int dF_{X}(x) \cdot\frac{\sum_{k=1}^K\alpha_{k}\cdot \mathbbm{E}\left[Y_i\cdot D^{[t]}_i|Z_i=z_k,X_i=x\right]}{P(c^{[t,\alpha]}(G_i)=1)}\\
	&= \frac{\sum_{k=1}^K\alpha_{k}\cdot \mathbbm{E}\left[\mathbbm{E}\left[Y_i\cdot D^{[t]}_i|Z_i=z_k,X_i\right]\right]}{P(c^{[t,\alpha]}(G_i)=1)}\\
	&= \frac{\sum_{k=1}^K\alpha_{k}\cdot \mathbbm{E}\left[\mathbbm{E}\left[Y_i\cdot D^{[t]}_i|Z_i=z_k,X_i\right]\right]}{\sum_{k=1}^K\alpha_{k}\cdot \mathbbm{E}\left[\mathbbm{E}[D^{[t]}_i|Z_i=z_k,X_i]\right]}
\end{align*}
applying Bayes' rule, echoing an argument for the LATE model by \citet{Frolich2007}. See also Appendix A of \citet{goff2024vector}. Given a binary collection, we can use these results to identify treatment effects that either do or do not condition on $X_i$.

Note that the conditional independence assumption \ref{eq:independencecovs} further allows us to identify the distribution of covariates $X_i$ among ``compliers'' for whom $c^{[t,\alpha]}(G_i)=1$ given a binary combination $(t,\alpha)$. Suppose that $X_i$ has $M$ components so that $X_i \in \mathbbm{R}^M$. Then for any Borel set $\mathcal{B}$ of $\mathbbm{R}^M$ we have that, by \eqref{eq:independencecovs}:
\begin{align*}
	\sum_k \alpha_k \cdot \mathbbm{E}[\mathbbm{1}(X_i \in \mathcal{B}) \cdot P(T_i=t|Z_i=z_k,X_i)] &= \sum_k \alpha_k \cdot \mathbbm{E}[\mathbbm{1}(X_i \in \mathcal{B})\cdot \mathbbm{E}[D_i^{[t]}(z_k)|X_i,Z_i=z_k]]\\
	&= \sum_k \alpha_k \cdot \mathbbm{E}[\mathbbm{1}(X_i \in \mathcal{B})\cdot \mathbbm{E}[D_i^{[t]}(z_k)|X_i]]\\
	&= \mathbbm{E}\left[\mathbbm{1}(X_i \in \mathcal{B})\cdot \mathbbm{E}\left[\left.\sum_k \alpha_k \cdot D_i^{[t]}(z_k)\right|X_i\right]\right]\\
	&= \mathbbm{E}[\mathbbm{1}(X_i \in \mathcal{B})\cdot \mathbbm{E}[c^{[t,\alpha]}(G_i)|X_i]]\\
	&= \mathbbm{E}[\mathbbm{E}[\mathbbm{1}(X_i \in \mathcal{B})\cdot c^{[t,\alpha]}(G_i)|X_i]]\\
	&= \mathbbm{E}[\mathbbm{1}(X_i \in \mathcal{B})\cdot c^{[t,\alpha]}(G_i)]\\
	&= P(c^{[t,\alpha]}(G_i)=1) \cdot P(X_i \in \mathcal{B}|c^{[t,\alpha]}(G_i)=1)
\end{align*}
Meanwhile
\begin{align*}
	\sum_k \alpha_k \cdot \mathbbm{E}[P(T_i=t|Z_i=z_k,X_i)] &= \sum_k \alpha_k \cdot \mathbbm{E}[\mathbbm{1}(X_i \in \mathcal{B})\cdot \mathbbm{E}[D_i^{[t]}(z_k)|X_i,Z_i=z_k]]\\
	&= \sum_k \alpha_k \cdot \mathbbm{E}[\mathbbm{E}[D_i^{[t]}(z_k)|X_i]] = \mathbbm{E}\left[\mathbbm{E}\left[\left.\sum_k \alpha_k \cdot D_i^{[t]}(z_k)\right|X_i\right]\right]\\
	&= \mathbbm{E}[\mathbbm{E}[c^{[t,\alpha]}(G_i)|X_i]] = \mathbbm{E}[c^{[t,\alpha]}(G_i)]= P(c^{[t,\alpha]}(G_i)=1)
\end{align*}
And thus 
\begin{equation} \label{eq:covsborel}
	P(X_i \in \mathcal{B}|c^{[t,\alpha]}(G_i)=1) = \frac{\sum_k \alpha_k \cdot \mathbbm{E}[\mathbbm{1}(X_i \in \mathcal{B}) \cdot P(T_i=t|Z_i=z_k,X_i)]}{\sum_k \alpha_k \cdot \mathbbm{E}[P(T_i=t|Z_i=z_k,X_i)]}
\end{equation}
This implies, for example, that we can identify the mean of $X_i$ among the $c^{[t,\alpha]}(G_i)=1$ sub-population as
$$\mathbbm{E}[X_i \in \mathcal{B}|c^{[t,\alpha]}(G_i)=1] = \frac{\sum_k \alpha_k \cdot \mathbbm{E}[X_i\cdot P(T_i=t|Z_i=z_k,X_i)]}{\sum_k \alpha_k \cdot \mathbbm{E}[P(T_i=t|Z_i=z_k,X_i)]}$$
which generalizes the seminal result of \citet{Abadie2003} for the case of the binary treatment, binary instrument LATE model.

If we have a binary collection $\{(t,\alpha^{[t]})\}_{t \in \psi}$, then Eq. \eqref{eq:covsborel} yields overidentification restrictions since it implies that
$$\frac{\sum_k \alpha^{[t]}_z \cdot \mathbbm{E}[\mathbbm{1}(X_i \in \mathcal{B}) \cdot P(T_i=t|Z_i=z_k,X_i)]}{\sum_k \alpha_k^{[t]} \cdot \mathbbm{E}[P(T_i=t|Z_i=z_k,X_i)]}=\frac{\sum_k \alpha^{[t']}_z \cdot \mathbbm{E}[\mathbbm{1}(X_i \in \mathcal{B}) \cdot P(T_i=t'|Z_i=z_k,X_i)]}{\sum_k \alpha_k^{[t']} \cdot \mathbbm{E}[P(T_i=t'|Z_i=z_k,X_i)]}$$
for any $t,t' \in \psi$. Note that this restriction is trivially satisfied for the binary collection that isolates compliers in the binary instrument, binary treatment LATE model.

\subsection{Details on empirical estimates including strata covariates} \label{sec:covsempirical}
Consider a binary combination $(t,\alpha)$ for a given treatment $t$, with associated function $c$. As shown in Appendix \ref{app:covs}, when Equation \ref{eq:independence} holds conditional on covariates $X_i$ we have:
\begin{align}
	\mathbbm{E}[Y_i(t)|c(G_i)=1] &= \frac{\sum_{k=1}^K\alpha_{k}\cdot \mathbbm{E}\left[\mathbbm{E}\left[Y_i\cdot D^{[t]}_i|Z_i=z_k,X_i\right]\right]}{P(c(G_i)=1)} \nonumber\\
	&\hspace{1in}=\frac{ \mathbbm{E}\left[\sum_{k=1}^K\alpha_{k}\cdot\mathbbm{E}\left[Y_i\cdot D^{[t]}_i|Z_i=z_k,X_i\right]\right]}{P(c(G_i)=1)} \label{eq:ytcov}
\end{align}
where 
\begin{equation} \label{eq:pccov}
	P(c(G_i)=1)=\sum_{k=1}^K\alpha_{k}\cdot \mathbbm{E}\left[\mathbbm{E}[D^{[t]}_i|Z_i=z_k,X_i]\right] =  \mathbbm{E}\left[\sum_{k=1}^K\alpha_{k}\cdot\mathbbm{E}[D^{[t]}_i|Z_i=z_k,X_i]\right]
\end{equation}
In the empirical application of \citet{depression}, randomization is performed within nine strata, which represents a discrete $X_i$ taking on nine values. To simplify estimation, I assume that the expectations $\mathbbm{E}[Y_i \cdot D_i^{[t]}|Z_i,X_i]$ and $\mathbbm{E}[D_i^{[t]}|Z_i,X_i]$ additively separable in $Z_i$ and $X_i$:
\begin{equation} \label{eq:ydz}
	\mathbbm{E}[Y_i \cdot D_i^{[t]}|Z_i,X_i] = \beta^{[t]}_{\textrm{both}} \cdot \mathbbm{1}(Z_i=\textrm{both}) + \beta^{[t]}_A \cdot \mathbbm{1}(Z_i=\textrm{just A}) + \beta^{[t]}_B \cdot \mathbbm{1}(Z_i=\textrm{just B}) + \sum_{s=1}^9 \lambda^{[t]}_s \cdot \mathbbm{1}(X_i=s) 
\end{equation}
and
\begin{equation} \label{eq:dz}
	\mathbbm{E}[D_i^{[t]}|Z_i,X_i] = \gamma^{[t]}_{\textrm{both}} \cdot \mathbbm{1}(Z_i=\textrm{both}) + \gamma^{[t]}_A \cdot \mathbbm{1}(Z_i=\textrm{just A}) + \gamma^{[t]}_B \cdot \mathbbm{1}(Z_i=\textrm{just B}) + \sum_{s=1}^9 \rho^{[t]}_s \cdot \mathbbm{1}(X_i=s) 
\end{equation}
i.e. linear regression equations with a full set of strata fixed effects (with none omitted) and instead omitting a dummy variable for $\mathbbm{1}(Z_i=\textrm{neither})$.

The four estimates of $p:=P(G_i=\textrm{complier})$ based on the choice model $\mathcal{G}^{sep}$ given in \eqref{eq:fourps} then become, using Eqs \eqref{eq:pccov} and \eqref{eq:dz}:
\begin{align}
	p &= \left\{\gamma_{\textrm{both}}^{[C]} + \sum_{s=1}^9 \rho^{[C]}_s \cdot P(X_i=s)\right\}  = \left\{\gamma_A^{[A]} - \gamma_{\textrm{both}}^{[A]}\right\}= \left\{\gamma_A^{[A]} - \gamma_{\textrm{both}}^{[A]}\right\}= \left\{\gamma_{\textrm{both}}^{[0]}- \gamma_A^{[0]} - \gamma_B^{[0]}\right\}  \label{eq:overid} 
\end{align}
Treatment effect estimates are then based on the following expressions using \eqref{eq:ydz}-\eqref{eq:dz}:
\begin{align}
	&\mathbbm{E}[Y_i(C)|i \textrm{ is  complier}] = \frac{\beta_{\textrm{both}}^{[C]} + \sum_{s=1}^9 \lambda^{[C]}_s \cdot P(X_i=s)}{\gamma_{\textrm{both}}^{[C]} + \sum_{s=1}^9 \rho^{[C]}_s \cdot P(X_i=s)}, \quad \mathbbm{E}[Y_i(A)|i \textrm{ is  complier}] = \frac{\beta_A^{[A]} - \beta_{\textrm{both}}^{[A]}}{\gamma_A^{[A]} - \gamma_{\textrm{both}}^{[A]}} \nonumber\\
	&\mathbbm{E}[Y_i(B)|i \textrm{ is  complier}] = \frac{\beta_B^{[B]} - \beta_{\textrm{both}}^{[B]}}{\gamma_B^{[B]} - \gamma_{\textrm{both}}^{[B]}}, \quad \quad \quad  \mathbbm{E}[Y_i(0)|i \textrm{ is  complier}] = \frac{\beta_{\textrm{both}}^{[0]}- \beta_A^{[0]} - \beta_B^{[0]}}{\gamma_{\textrm{both}}^{[0]}- \gamma_A^{[0]} - \gamma_B^{[0]}} \label{eq:muest}
\end{align}
and the local average interaction effect among compliers $LAIE$ is estimated accordingly. Some involved algebra shows that the expressions in \eqref{eq:muest} recover the results for complier average treatment effects in Theorem 1 of \citet{blackwell2017}, given one-sided noncompliance.

\subsection{GMM estimation} \label{sec:gmm}
Note that given the overidentification of $p:=P(G_i=\textrm{complier})$, any of the local counterfactual means $\eqref{eq:muest}$ could be estimated by swapping out an alternative estimate of $p$ in the denominator. In principle, we can increase efficiency by estimating treatment effects as well as $LAIE$ while imposing Eq. \eqref{eq:overid}, in a generalized method of moments (GMM) estimation approach. Column (4) of Table \ref{table:interaction} implements this. Given the logic of Corollary \ref{corr:itt}, GMM estimation of $LAIE$ combines the ITT regression \eqref{eq:ittinteractioncovs} with the first-stage regressions \eqref{eq:dz}, and imposing \eqref{eq:overid} as additional moments. For the treatment effect estimates $\mathbbm{E}[Y_i(t)-Y_i(0)|i \textrm{ is  complier}]$ for $t \in \{0,A,B\}$, GMM estimation combines regressions \eqref{eq:ydz} for treatments $t$ and $0$ with the first-stage regressions and \eqref{eq:overid}. All GMM estimates use the two-step GMM estimator, starting from an initial identity weight-matrix, and requesting a cluster robust final weight-matrix and standard errors.

\subsection{Deriving the expression  $\theta^{ITT}/p$ for local average interaction effect} \label{sec:ittexpression}
Consider first the case with no covariates. We have using Eqs. \eqref{idresult} and \eqref{eq:bcs}:
\begin{align*}
	&LAIE=\mathbbm{E}[Y_i(C)|c(G_i)=1]-\mathbbm{E}[Y_i(A)|c(G_i)=1]-\mathbbm{E}[Y_i(B)|c(G_i)=1]+\mathbbm{E}[Y_i(0)|c(G_i)=1]\\
	&\quad \quad =\frac{\mathbbm{E}[Y_i\cdot D_i^{[C]}|Z_i=\textrm{both}]}{\mathbbm{E}[D_i^{[C]}|Z_i=\textrm{both}]}-\frac{\mathbbm{E}[Y_i\cdot D_i^{[A]}|Z_i=\textrm{just A}]-\mathbbm{E}[Y_i\cdot D_i^{[A]}|Z_i=\textrm{both}]}{\mathbbm{E}[D_i^{[A]}|Z_i=\textrm{just A}]-\mathbbm{E}[D_i^{[A]}|Z_i=\textrm{both}]}\\
	&\quad \quad -\frac{\mathbbm{E}[Y_i\cdot D_i^{[B]}|Z_i=\textrm{just A}]-\mathbbm{E}[Y_i\cdot D_i^{[B]}|Z_i=\textrm{both}]}{\mathbbm{E}[D_i^{[B]}|Z_i=\textrm{just B}]-\mathbbm{E}[D_i^{[B]}|Z_i=\textrm{both}]}\\
	&+\frac{\mathbbm{E}[Y_i\cdot D_i^{[0]}|Z_i=\textrm{both}]-\mathbbm{E}[Y_i\cdot D_i^{[0]}|Z_i=\textrm{just A}]-\mathbbm{E}[Y_i\cdot D_i^{[0]}|Z_i=\textrm{just B}]+\mathbbm{E}[Y_i\cdot D_i^{[0]}|Z_i=\textrm{neither}]}{\mathbbm{E}[D_i^{[0]}|Z_i=\textrm{both}]-\mathbbm{E}[D_i^{[0]}|Z_i=\textrm{just A}]-\mathbbm{E}[D_i^{[0]}|Z_i=\textrm{just B}]-\mathbbm{E}[D_i^{[0]}|Z_i=\textrm{neither}]}\\
	&=\frac{1}{p}\cdot \left\{\mathbbm{E}[Y_i\cdot D_i^{[C]}|Z_i=\textrm{both}]-\mathbbm{E}[Y_i\cdot D_i^{[A]}|Z_i=\textrm{just A}]+\mathbbm{E}[Y_i\cdot D_i^{[A]}|Z_i=\textrm{both}]\right. \nonumber \\
	& \quad \quad \quad \left. -\mathbbm{E}[Y_i\cdot D_i^{[B]}|Z_i=\textrm{just B}]+\mathbbm{E}[Y_i\cdot D_i^{[B]}|Z_i=\textrm{both}] + \mathbbm{E}[Y_i\cdot D_i^{[0]}|Z_i=\textrm{both}] \right. \nonumber \\
	& \quad \quad \quad \left. -\mathbbm{E}[Y_i\cdot D_i^{[0]}|Z_i=\textrm{just A}]-\mathbbm{E}[Y_i\cdot D_i^{[0]}|Z_i=\textrm{just B}]+\mathbbm{E}[Y_i\cdot D_i^{[0]}|Z_i=\textrm{neither}] \right\} \nonumber \\
	&=\frac{1}{p}\cdot \left\{\mathbbm{E}[Y_i\cdot (D_i^{[0]}+D_i^{[A]}+D_i^{[B]}+D_i^{[C]})|Z_i=\textrm{both}]-\mathbbm{E}[Y_i\cdot (D_i^{[0]}+D_i^{[A]})|Z_i=\textrm{just A}]\right. \nonumber \\
	& \quad \quad \quad \left. -\mathbbm{E}[Y_i\cdot (D_i^{[0]}+D_i^{[B]})|Z_i=\textrm{just B}]+\mathbbm{E}[Y_i\cdot D_i^{[0]}|Z_i=\textrm{neither}] \right\} \nonumber \\
	&=\frac{1}{p}\cdot \left\{\mathbbm{E}[Y_i|Z_i=\textrm{both}]-\mathbbm{E}[Y_i|Z_i=\textrm{just A}]-\mathbbm{E}[Y_i|Z_i=\textrm{just B}]+\mathbbm{E}[Y_i|Z_i=\textrm{neither}] \right\} = \frac{\theta^{ITT}}{p}
\end{align*}
where $\theta^{ITT}:=\gamma_3-\gamma_1-\gamma_2$ from the ITT regression Eq. \eqref{eq:ittinteraction}. In the above I have used Eq. \eqref{eq:fourps} in the second step, then combined terms, and finally using that $(D_i^{[0]}+D_i^{[A]}+D_i^{[B]}+D_i^{[C]})=1$, that $(D_i^{[0]}+D_i^{[A]})$ conditional on $Z_i=\textrm{just A}$ (given one-sided noncompliance), that $(D_i^{[0]}+D_i^{[B]})$ conditional on $Z_i=\textrm{just B}$, and that $D_i^{[0]}$ conditional on $Z_i=\textrm{neither}$.

With covariates $X_i$, the standard intent-to-treat regression generalizes \eqref{eq:ittinteraction} by adding a linear function in the covariates that includes a constant:
\begin{equation} \label{eq:ittinteractioncovs}
	Y_i = \gamma_1 \cdot \mathbbm{1}(Z_i=A)+\gamma_2 \cdot \mathbbm{1}(Z_i=B) +\gamma_3 \cdot \mathbbm{1}(Z_i=C) + \pi'X_i + \nu_i
\end{equation}
In this case, $\theta^{ITT}:=\gamma_3-\gamma_1-\gamma_2$ is equal to
$$\mathbbm{E}[Y_i|Z_i=\textrm{both},X_i]-\mathbbm{E}[Y_i|Z_i=\textrm{just A},X_i]-\mathbbm{E}[Y_i|Z_i=\textrm{just B},X_i]+\mathbbm{E}[Y_i|Z_i=\textrm{neither},X_i]$$
with probability one (i.e. for all $X_i$). The same steps as above show that, using Eqs. \eqref{eq:ytcov} and \eqref{eq:pccov}:
\begin{align*}
	&LAIE=\frac{1}{p}\cdot \mathbbm{E}\left\{\mathbbm{E}[Y_i\cdot D_i^{[C]}|Z_i=\textrm{both},X_i]-\mathbbm{E}[Y_i\cdot D_i^{[A]}|Z_i=\textrm{just A},X_i]\right.\\
		&\quad \quad \left.+\mathbbm{E}[Y_i\cdot D_i^{[A]}|Z_i=\textrm{both},X_i]-\mathbbm{E}[Y_i\cdot D_i^{[B]}|Z_i=\textrm{just A},X_i]-\mathbbm{E}[Y_i\cdot D_i^{[B]}|Z_i=\textrm{both},X_i] \right.\\
		&\quad \quad \left.-\mathbbm{E}[Y_i\cdot D_i^{[0]}|Z_i=\textrm{both},X_i]+\mathbbm{E}[Y_i\cdot D_i^{[0]}|Z_i=\textrm{just A},X_i]+\mathbbm{E}[Y_i\cdot D_i^{[0]}|Z_i=\textrm{just B},X_i]\right.\\
		&\quad \quad \left.-\mathbbm{E}[Y_i\cdot D_i^{[0]}|Z_i=\textrm{neither},X_i]\right\} = \frac{1}{p}\cdot \mathbbm{E}[\gamma_3 - \gamma_1 - \gamma_2] = \frac{\theta^{ITT}}{p}
\end{align*}
provided that Eq. \eqref{eq:ittinteractioncovs} is correctly specified for the conditional mean $\mathbbm{E}[Y_i|Z_i,X_i]$.

\subsection{Setting up the linear program to test for offending types} \label{sec:lp}
This section considers the identification of bounds on the proportion of the population that belongs to a certain set of response types $\mathcal{G}^*$, within a larger selection model $\mathcal{G}$. This method is implemented in Section \ref{sec:empirical} to discuss whether first stage selection information is consistent with the choice model $\mathcal{G} \subseteq \mathcal{G}^{sep}$, under the maintained assumption that $\mathcal{G} \subseteq \mathcal{G}^{WARP}$. Thus for the remainder of this section we assume that $\mathcal{G} = \mathcal{G}^{WARP}$ defined in Section \ref{sec:empirical}. This section also ignores the randomization strata $X_i$, which is valid for testing ``first-stage'' restrictions if the response-type distribution is common across strata.

For any set of response types $\mathcal{G}^* \subseteq \mathcal{G}^{WARP}$, we can partially identify $P(G_i \in \mathcal{G}^*)$ as $P(G_i \in \mathcal{G}^*) \in [LB^*,UB^*]$ where
\begin{equation} \label{eq:lplb}
	LB^* = \min\limits_{x \in \mathbbm{R}^9} w'x \quad  \textrm{ subject to } \mathcal{A}x=\beta \textrm{ and } x \ge 0
\end{equation}
\begin{equation} \label{eq:lpub}
	UB^* = \max\limits_{x \in \mathbbm{R}^9} w'x \quad  \textrm{ subject to } \mathcal{A}x=\beta \textrm{ and } x \ge 0
\end{equation}
with $w$ a $9 \times 1$ vector (where $|\mathcal{G}^{WARP}|=9$) with components $w_g = \mathbbm{1}(g \in \mathcal{G}^*)$, and the constraint $x \ge 0$ is read as all components of the vector $x$ must be weakly positive. If $LB^*$ were found to be strictly positive with $\mathcal{G}^*$ chosen to be $\mathcal{G}^{WARP}-\mathcal{G}^{sep}$, this would constitute evidence that the restriction $\mathcal{G} \subseteq \mathcal{G}^{sep}$ is not satisfied, assuming that $\mathcal{G} \subseteq \mathcal{G}^{WARP}$.

The $16 \times 9$ matrix $\mathcal{A}$ can be obtained from the matrices $A^{[t]}$, and the $16$-component vector $\beta$ estimated from the data:
$$\mathcal{A} = \begin{bmatrix} A^{[0]} \\ A^{[A]} \\ A^{[B]} \\ A^{[both]} \end{bmatrix} = 
\begin{bmatrix} 
	1 & 1 & 1 & 1& {1} & {1} & {1} & {1} & {1}\\
	1 & 0& 0 & 1 & {1} & {0} & {1} & {0}& {0}\\
	1 & 0 & 1 & 0& {1} & {1} & {0}& {0} & {0}\\
	1 & 0 & 0 & 0& {0} & {0} & {0} & {0} & {0}\\
	\hline
	0 & 0 & 0 & 0& {0} & {0} & {0} & {0} & {0}\\
	0 &1& 1 & 0 & {0} & {1} & {0} & {1}& {1}\\
	0 & 0 & 0 & 0& {0} & {0} & {0}& {0} & {0}\\
	0 & 0 & 1 & 0& {0} & {0} & {0} & {1} & {0}\\
	\hline
	0 & 0 & 0 & 0& {0} & {0} & {0} & {0} & {0}\\
	0 &0& 0 & 0 & {0} & {0} & {0} & {0}& {0}\\
	0 & 1 & 0 & 1& {0} & {0} & {1}& {1} & {1}\\
	0 & 0 & 0 & 1& {0} & {0} & {0} & {0} & {1}\\
	\hline
	0 & 0 & 0 & 0& {0} & {0} & {0} & {0} & {0}\\
	0 & 0 & 0 & 0& {0} & {0} & {0} & {0} & {0}\\
	0 & 0 & 0 & 0& {0} & {0} & {0} & {0} & {0}\\
	0 & 1 & 0 & 0& {1} & {1} & {1} & {0} & {0}\\
\end{bmatrix} \quad \quad \quad \quad \beta = \begin{pmatrix}
	P(T_i=0|Z_i=\textrm{neither}) \\
	P(T_i=0|Z_i=\textrm{just A}) \\
	P(T_i=0|Z_i=\textrm{just B}) \\
	P(T_i=0|Z_i=\textrm{both}) \\ 
	\hline
	P(T_i=A|Z_i=\textrm{neither}) \\
	P(T_i=A|Z_i=\textrm{just A}) \\
	P(T_i=A|Z_i=\textrm{just B}) \\
	P(T_i=A|Z_i=\textrm{both}) \\
	\hline
	P(T_i=B|Z_i=\textrm{neither}) \\
	P(T_i=B|Z_i=\textrm{just A}) \\
	P(T_i=B|Z_i=\textrm{just B}) \\
	P(T_i=B|Z_i=\textrm{both}) \\
	\hline
	P(T_i=C|Z_i=\textrm{neither}) \\
	P(T_i=C|Z_i=\textrm{just A}) \\
	P(T_i=C|Z_i=\textrm{just B}) \\
	P(T_i=C|Z_i=\textrm{both})
\end{pmatrix}$$
Point estimates of the bounds $LB^*$ and $UB^*$ are readily obtained by solving the linear programs \eqref{eq:lplb} and \eqref{eq:lpub} with the sample estimator $\hat{\beta}$.

Given sampling error in $\hat{\beta}$ however, we would like to construct a valid confidence interval for the partially identified parameter $P(G_i \in \mathcal{G}^*) = w'x$ given its representation as a solution to the linear program $Ax=\beta, x \ge 0$. This problem is considered by \citet{fsst}, and I use the \texttt{fsst} command in the \texttt{lpinfer} package in $\texttt{R}$ to generate a confidence interval for the parameters $P(G_i \in \mathcal{G}^*)$ considered in the main text. The required inputs for $\texttt{fsst}$ are the matrix $\mathcal{A}$, the vector $\beta$ (specified as a function of the data, as the FSST procedure makes use of estimates of $\beta$ in bootstrap samples).

\subsubsection{Results for \citet{depression}}
In addition to the results for $\mathcal{G}^* = \mathcal{G}^{WARP}-\mathcal{G}^{sep}$ reported in the main text, I here present some further estimates. A 90\% confidence interval using the method of \citet*{fangsantos} (FSST) does rule out zero but is otherwise similar at $[0, 0.8281]$ (as opposed to $[0, 0.8297]$ for the 95\% interval). However, the p-value for the null hypothesis that  $P(G_i \in \mathcal{G}^{WARP}-\mathcal{G}^{sep})=0$ puts it just on the margin of being included in the 90\% confidence interval. The \texttt{lpinfer} package in \texttt{R} also allows for statistical inference on solutions to problems \eqref{eq:lplb} and \eqref{eq:lpub} using methods introduced by \citet{romanoshaikh} and \citet{chorussell}. The method of \citet{chorussell} yields $[0.02, 0.86]$ as a 95\% confidence interval. The method of \citet{romanoshaikh} yields $[0, 0.84]$ as a 95\% confidence interval. Confidence intervals for the share of the favor-B type (which the point estimates suggest may be the largest offending type) only are similar to the confidence intervals for all offending types in $\mathcal{G}^{WARP}-\mathcal{G}^{sep}$. Overall, the results offer little evidence against the assumption that $\mathcal{G}^{sep}$ represents the true choice model, and any violations that cannot be ruled out appear to be minor. This supports the conclusion that complementarities between the two treatments are identified among compliers without restricting outcomes, in line with Proposition \ref{prop:interactionid}.

In the above calculations, I do not condition on the nine strata used by \citet{depression} for randomization. This could be implemented by expanding $\mathcal{A}$ and $\beta$ to have $16 \times 9$ rows each, rather than $16$. However the above results are valid if the response-type distribution is common across strata, and under this assumption allow for a much more efficient use of the available sample.

\subsection{Financial incentives and support for academic achievement} \label{sec:star}
\citet*{angristlangoreopoulos} (ALO) report results from an the Project STAR intervention that cross-randomized academic support and financial incentives on academic achievement among first-year students at a large Canadian university. In this setting, I let treatment A represent the Student Support Program (SSP): a program which gave students access to peer advisers and supplemental instruction. I let treatment B represent the Student Fellowship Program (SFP), which made students eligible for merit scholarships based on good performance during the first year courses.

The STAR intervention randomized 250 students into an arm that was offered access to the SSP only ($Z_i=\textrm{just A}$), another 250 students to be offered access to the SFP only ($Z_i=\textrm{just B}$), and a third group of 150 students that was offered access to both programs ($Z_i=\textrm{both}$). A control group of 1,006 students were offered neither ($Z_i=\textrm{neither}$).

I use the replication data from ALO, which tracks program takeup as well as student performance among those students included in Project STAR. Treatment uptake for treatment A (SSP) is observable, and I define it as having attended a facilitated study groups or having met with an advisor. For treatment B, I follow ALO in defining treatment takeup as having responded to their invitation to sign up for the assigned treatment. ALO define compliance with respect to SSP (treatment A) similarly as having given their consent by simply signing up for their assigned treatment. With this definition however, no individual offered both treatments could opt for one treatment alone.\footnote{Given WARP, this would then limit the choice model to the groups n.t., complier, only both, A+, and B+ from Table \ref{table:cross}. This group yields a rather uninteresting selection model when intersected with $\mathcal{G}^{sep}$ (leaving just never takers and compliers) in order to afford outcome non-restrictive identification of complementarity between the treatments. However, defining compliance as ALO do also rejects the overidentification restriction \eqref{eq:overid}, with a chi-squared statistic (with 2 degrees of freedom) of 28.66.}  Since further information is available on whether individuals actually take part in SSP activities, I make use of this additional information.

I test the overidentification restriction of $\mathcal{G} \subseteq \mathcal{G}^{sep}$ in this setting as described in Eq. \eqref{eq:muest}, however note that in the present setting there are no randomization strata that need to be controlled for. The four point estimates for $p=P(i \textrm{ is complier})$ are 41\%, 21\%, 51\%, and 34\%, respectively. A test for equality of the four estimates returns a chi-squared statistic (with 2 degrees of freedom) of 26.16, a p-value of $0.0000$.\footnote{Inferential methods for the linear program described in \ref{sec:lp} at the 95\% level suggest that at least about 15\% of the population belongs to groups in $\mathcal{G}^{WARP}-\mathcal{G}^{sel}$, provided that $WARP$ holds.} Thus in contrast to the application of \citet{depression}, we find in the ALO context that we can clearly reject the choice model $\mathcal{G} \subseteq \mathcal{G}^{sep}$ that is required to identify complementarity between the two treatment effects in an outcome-agnostic manner.

		\section{Recovering existing identification results as binary combinations and collections} \label{sec:examples}

The notions of binary combinations and binary collections reveals a common structure among several existing IV point identification results. 

\subsection{Example 1: LATE monotonicity and marginal treatment effects}
Here I extend my analysis of the monotonicity assumption of \citet{Imbens2018} to cases with more than two instrument values. Treatment remains binary $\mathcal{T} = \{0,1\}$. Since $D^{[0]}_{i}(z) = 1-D^{[1]}_{i}(z)$, we can focus on the single treatment indicator $D_i(z):=D^{[1]}_i(z)$. \citet{Imbens2018} assume that:
\begin{assumption*}[IAM] For all $z, z' \in \mathcal{Z}$: $D_i(z) \ge D_i(z')$ for all $i$ or $D_i(z) \le D_i(z')$ all $i$.
\end{assumption*}
\noindent Suppose $z, z'$ are a pair such that the former case of assumption IAM obtains, and define a binary combination with $K=2$, $z_{1}=z'$, $\alpha_1 = 1$, $z_{2}=z$, $\alpha_2 = -1$. Then Eq. (\ref{idresult}) from the main paper yileds
\begin{equation}
	\mathbbm{E}[Y_i(1)|D_i(z')>D_i(z)] = \frac{\mathbbm{E}\left[Y_i\cdot D_{i}|Z_i=z'\right]-\mathbbm{E}\left[Y_i\cdot D_{i}|Z_i=z\right]}{ \mathbbm{E}\left[D_{i}|Z_i=z'\right]-\mathbbm{E}\left[D_{i}|Z_i=z\right]} \label{zzprime}
\end{equation}	
and similarly
\begin{equation*}
	\mathbbm{E}[Y_i(0)|D_i(z')>D_i(z)] = \frac{\mathbbm{E}\left[Y_i\cdot (1-D_{i})|Z_i=z'\right]-\mathbbm{E}\left[Y_i\cdot (1-D_{i})|Z_i=z\right]}{ \mathbbm{E}\left[(1-D_{i})|Z_i=z'\right]-\mathbbm{E}\left[(1-D_{i})|Z_i=z\right]}
\end{equation*}
Combining, we have that $\mathbbm{E}[Y_i(1)-Y_i(0)|D_i(z')>D_i(z)] = \frac{\mathbbm{E}\left[Y_i|Z_i=z'\right]-\mathbbm{E}\left[Y_i|Z_i=z\right]}{ \mathbbm{E}\left[D_{i}|Z_i=z'\right]-\mathbbm{E}\left[D_{i}|Z_i=z\right]}$, which is Theorem 1 of \citet{Imbens2018}.

Suppose that $\mathcal{Z}$ is continuous and for all $u \in [0,1]$ there exists a $z \in \mathcal{Z}$ such that $P(z):=P(D_i=1|Z_i=z) = u$. Let $U_i = \inf_{z\in \mathcal{Z}}\{P(z): D_i(z) = 1\}$. Given IAM, $U_i$ plays the role of $G_i$, indicating the ``first'' instrument value (as ordered by the propensity score function $P(z)$) at which $i$ would take treatment. 
For any given $u$, let $z$ be a point in $\mathcal{Z}$ such that $P(z)=u$ and take a sequence $z_j$ in $\mathcal{Z}$ such that $z_j \rightarrow z$ as $j \rightarrow \infty$. Taking the limit of Eq. (\ref{zzprime}) we have that:
$$\mathbbm{E}[Y_i(1)|U_i=u] = \lim_{j \rightarrow \infty} \frac{\mathbbm{E}\left[Y_i\cdot D_{i}|Z_i=z_j\right]-\mathbbm{E}\left[Y_i\cdot D_{i}|Z_i=z\right]}{ \mathbbm{E}\left[D_{i}|Z_i=z_j\right]-\mathbbm{E}\left[D_{i}|Z_i=z\right]} = \frac{d}{du}\mathbbm{E}\left[Y_i\cdot D_{i}|P(Z_i)=u\right]$$
and similarly for $Y_i(0)$, allowing us to identify marginal treatment effects (cf. \citet{Heckman2004}).

\subsection{Example 2: Vector monotonicity (\citealt{goff2024vector})}

\citet{goff2024vector} considers a binary treatment and finite $\mathcal{Z} \subseteq \mathcal{Z}_1 \times \mathcal{Z}_2 \times \dots \mathcal{Z}_J$, and the following monotonicity assumption.
\begin{assumption*}[2 (vector monotonicity)]
	There exists an ordering $\ge_j$ on $\mathcal{Z}_j$ for each $j \in \{1\dots J\}$ such that for all $z,z' \in \mathcal{Z}$, if $z \ge z'$ component-wise according to the $\{\ge_j\}$, then $D_i(z) \ge D_i(z')$ for all $i$.
\end{assumption*}

\noindent Theorem 1 of \citet{goff2024vector} shows that average counterfactual means are identified under vector monotonicity for groups defined by the condition $c(G_i,Z_i)=1$, where $c$ satisfies a condition called ``Property M'' and $Z_i$ has full rectangular support. His Proposition 6 shows that Property M is equivalent to $c(G_i,Z_i) = \sum_{k=1}^K\alpha_{k}\cdot D^{[t]}_i(z_k(Z_i))$ where $K$ is an even number no greater than $J/2$ and $\alpha_k = (-1)^k$ with $z_{k+1}(z) \ge z_{k}(z)$ component-wise according to the orders $\ge_j$, for all $k$. In what follows I for simplicity focus on the special case of target parameters in which $c$ depends on $G_i$ only, and not additionally on $Z_i$. See Appendix \ref{app:deponz} for a discussion of other parameters.

In the case of two binary instruments $J=2$, there are six selection types compatible with vector monotonicity, with names introduced by \citealt{Mogstad2020a}:

\begin{table}[h!]
	\centering
	\begin{tabular}{c|ccccccc}
		&$Z_1$ comp. & $Z_2$ comp. & eager-comp. & reluctant-comp. & n.t. & a.t. \\
		\hline
		$\mathbf{z=(0,0)'}$ & 0 & 0 & 0 & 0 & 0 & 1\\
		$\mathbf{z=(1,0)'}$ & 1 & 0 & 1 & 0 & 0 & 1\\
		$\mathbf{z=(0,1)'}$ & 0 & 1 & 1 & 0 & 0 & 1\\
		$\mathbf{z=(1,1)'}$ & 1 & 1 & 1 & 1 & 0 & 1\\
		\hline
	\end{tabular}
\end{table}
\noindent With this table defining matrix $A^{[1]}$, some algebra shows that for example $c=(1, 0, 0, 1, 0, 0)'$ occurs in the rowspace of both $A^{[1]}$ and $A^{[0]}$. One way to see this is to work out the row reduced echelon forms of $A^{[1]}$ and $A^{[0]}$, which preserve their row-spaces and are:
$$rref(A^{[1]}) = \begin{bmatrix}
		1 & 0 & 0 & 1 & 0 & 0\\
		0 & 1 & 0 & 1 & 0 & 0\\
		0 & 0 & 1 & -1 & 0 & 0\\
		0 & 0 & 0 & 0 & 1& 0
	\end{bmatrix} \quad \quad rref(A^{[0]}) = \begin{bmatrix}
		1 & 0 & 0 & 1 & 0 & 0\\
		0 & 1 & 0 & 1 & 0 & 0\\
		0 & 0 & 1 & -1 & 0 & 0\\
		0 & 0 & 0 & 0 & 0& 1
	\end{bmatrix}$$
From the first row of each of the reduced echelon forms, we can see immediately that e.g. the average treatment effect among $Z_1$ and reluctant compliers is outcome-nonrestrictive identified. Adding all four rows we see that the average treatment effect among all of the four compliers types is outcome-nonrestrictive identified, what \citet{goff2024vector} calls the ``all-compliers LATE''. \citet{goff2024vector} shows how these and similar point identification results generalize to any number of instruments under vector monotonicity. 

\subsection{Example 3: Unordered Monotonicity, \citet{Heckman2018}} \label{sec:HP}
\citet{Heckman2018} (HP) consider a finite $\mathcal{Z}$ and assume what they call \textit{unordered monotonicity} (UM) for a multi-valued treatment:
\begin{assumption*}[UM]
	For any $t \in \mathcal{T}$ and $z, z' \in \mathcal{Z}$, either $D^{[t]}_i(z) \ge D^{[t]}_i(z')$ for all $i$ or $D^{[t]}_i(z') \ge D^{[t]}_i(z)$ for all $i$.
\end{assumption*}
\noindent Given UM, let us fix a $t \in \mathcal{T}$ and label the points in $\mathcal{Z}$ as $z_m$ for $m=1,2 \dots |\mathcal{Z}|$, where the points are labeled in increasing order of the propensity score for treatment $t$: $P_t(z_{m+1}) \ge P_t(Z_m)$, where we define $P_t(z) = \mathbbm{E}[D^{[t]}_i|Z_i=z]$. The order is not important in the case of ties. Let $\Sigma_{ti}=|\{z \in \mathcal{Z}:T_i(z)=t\}|$ be the number of $z\in \mathcal{Z}$ for which $i$ takes treatment $t$. Note that $\Sigma_{ti}$ is exactly equal to $|\mathcal{Z}|-m+1$ for the smallest $m$ such that $D^{[t]}_i(z_m)=1$. Thus we have a binary combination for any treatment $t$ and value $s \in \{0,1, \dots |\mathcal{Z}|\}$: in particular $D_i^{[t]}(z_{m}) - D_i^{[t]}(z_{m-1}) \in \{0,1\}$ for all $i$, and is equal to $i$ for those units having $\Sigma_{ti}=s$.

It then follows immediately from Eq. (\ref{idresult}), as in the IAM case with a binary treatment (cf. Eq. \ref{zzprime}), that $E[Y_i(t)|\Sigma_{ti}=s]$ is identified for any $s = 1 \dots |\mathcal{Z}|$ as:
\begin{align}
	\mathbbm{E}[Y_i(t)|\Sigma_{ti}=s] = \begin{cases}
		\frac{\mathbbm{E}\left[Y_i\cdot D^{[t]}_i|Z_i=z_m\right]-\mathbbm{E}\left[Y_i\cdot D^{[t]}_i|Z_i=z_{m-1}\right]}{ \mathbbm{E}\left[D^{[t]}_i|Z_i=z_m\right]-\mathbbm{E}\left[D^{[t]}_i|Z_i=z_{m-1}\right]} & \textrm{ if } s < |\mathcal{Z}|\\
		\frac{\mathbbm{E}\left[Y_i\cdot D^{[t]}_i|Z_i=z_1\right]}{\mathbbm{E}\left[D^{[t]}_i|Z_i=z_1\right]}& \textrm{ if } s = |\mathcal{Z}|\end{cases} \label{UMresult}
\end{align}
where $m = |\mathcal{Z}| - s + 1$. 

This provides a simple proof of HP's Theorem T-6, which shows that $E[Y_i(t)|\Sigma_{ti}=s]$ is identified. HP show that
\begin{equation} \label{UMresultHP}
	E[Y_i(t)|\Sigma_{ti}=s] = \frac{c'{A^{[t]}}^+ Q_Z}{c'{A^{[t]}}^+ P_Z}
\end{equation}
where $B^+$ is the Moore-Penrose pseudo-inverse of a matrix $B$, $Q_Z$ is a vector of $\mathbbm{E}[Y_iD^{[t]}_i|Z_i=z]$ across $z$ and $P_Z$ is a vector of $\mathbbm{E}[D^{[t]}_i|Z_i=z]$ across $z$. Here $c$ corresponds to our parameter of interest (indexed by the pair $(t,s)$), with an entry of one if $\Sigma_{tg}=s$ for that selection type and zero otherwise.

To see the equivalence between this result and (\ref{UMresult}), we can take advantage of the structure of $A^{[t]}$ under UM to replace it with a smaller matrix whose inverse is very simple. Note that any two selection types $g$ sharing a value of $\Sigma_{ti}$ will have identical entries in $c$, and will have identical corresponding columns in the matrix ${A^{[t]}}$. This implies that they will have identical rows in ${A^{[t]}}^+$. We can remove the redundant columns of ${A^{[t]}}$ by indexing columns by values of $\Sigma_{ti}$ rather than by full response vectors $g$, and similarly indexing elements of $c$ by values of $\Sigma_{ti}$. This yields the same vector $c'{A^{[t]}}^+$ as before, up to a scalar factor that counts the number of values of $g$ such that $\Sigma_{ti}=s$. However, this factor cancels out in the numerator and denominator of (\ref{UMresultHP}). With this modification, $A^{[t]}$ is now a $|\mathcal{Z}|$ by $|\mathcal{Z}|+1$ matrix and $c$ is now a standard basis vector equal to one in its $s^{th}$ element (and zero elsewhere).

Let us now order the rows of this modified $A^{[t]}$ according to $z_1$, $z_2$, etc, and it's columns in decreasing order of $\Sigma_{ti}$. With this ordering, $A^{[t]}$ is simply a lower triangular matrix of ones, appended to the right by a single column of zeros. It can then be verified that rows $s=2\dots(|\mathcal{Z}|-1)$ of ${A^{[t]}}^+$ are of the form $(0,\dots,-1,1,\dots 0)'$ with $s-2$ zeroes on the left (while the first row is composed of a single $1$ in the first column, and the last row is all zeros).\footnote{E.g. the modified forms with $|\mathcal{Z}|=4$ are ${A^{[t]}} = \begin{pmatrix}
	1 & 0 & 0 & 0 & 0\\
	1 & 1 & 0 & 0 & 0\\
	1 & 1 & 1 & 0 & 0\\
	1 & 1 & 1 & 1 & 0\\
\end{pmatrix}$, ${A^{[t]}}^+ = \begin{pmatrix}
	1 & 0 & 0 & 0\\
	-1 & 1 & 0 & 0\\
	0 & -1 & 1 & 0\\
	0 & 0 & -1 & 1\\
	0 & 0 & 0 & 0
\end{pmatrix}$.} Note that given the definition of $c$, $c'{A^{[t]}}^+$picks out the $s^{th}$ row of ${A^{[t]}}^+$ in (\ref{UMresultHP}), and we have that (\ref{UMresult}) and (\ref{UMresultHP}) are equivalent.

Remark 7.1 of HP observes that given the above, treatment effects can be identified if (in my notation) for some $s,s'$ and $t,t' \in \mathcal{T}$, $\mathbbm{1}(\Sigma_{ti}=s)=\mathbbm{1}(\Sigma_{t'i}=s')$ almost surely, since then we can identifiy $\mathbbm{E}[Y_i(t')-Y_i(t)|\Sigma_{t'i}=s']=\mathbbm{E}[Y_i(t')|\Sigma_{t'i}=s']-\mathbbm{E}[Y_i(t)|\Sigma_{ti}=s]$. The idea of binary collections can be though of as a generalization of this type of result beyond the case of unordered monotonicity.

\subsection{Example 4: \citet{Lee2018a}}
\citet{Lee2018a} (LS) consider a class of models in which unit $i$'s selection type depends upon a $J$-dimensional vector $V_i \in [0,1]^J$ and a vector valued function $Q: \mathcal{Z} \rightarrow \mathcal{Q}$ where $\mathcal{Q} \subseteq \mathbbm{R}^J$. Selection is assumed to follow:
\begin{equation} \label{eq:leesalanieselectionmodel}
	D^{[t]}_i(z) = \sum_{l \subseteq \{1\dots J\}} c^t_{l} \cdot \prod_{j \in l} \mathbbm{1}(V_{ji} \le Q_j(z))
\end{equation}
for some set of coefficients $c^t_l$ defined over the subsets of $\{1\dots J\}$, for each $t \in \mathcal{T}$, and where $V_{ji}$ is the $j^{th}$ component of $V_i$, and $Q_j$ the $j^{th}$ component of $Q$. This model nests the marginal treatment effects (MTE) framework when we have a binary treatment and $J=1$, in which case we may let $Q(z) = \mathbbm{E}[D_i|Z_i=z]$ be the propensity score function.

The second part of LS's Theorem 3.1 shows that under support/regularity conditions:
\begin{equation} \label{LSresult} E[Y_i(t)|V_i=q]= \frac{\frac{\partial^J}{\partial_{q_1}\dots \partial_{q_J}} \mathbbm{E}[Y_i D^{[t]}_i|Q(Z_i)=q]}{\frac{\partial^J}{\partial_{q_1}\dots \partial_{q_J}} \mathbbm{E}[D^{[t]}_i|Q(Z_i)=q]}\end{equation}
Now let's see how this result can also be obtained through Theorem \ref{thm:suff}.
For any vector $q \in \mathbbm{R}^J$, let $S_i(q):=\{j \in \{1 \dots J\}: V_{ji} \le q_j\}$ be the set of indices for which $V_{ji} \le q_j$. Then $D^{[t]}_i(z) = \sum_{l \subseteq S_i(Q(z))} c_l^t$. Note that $D^{[t]}_i(z)$ only depends on $z$ through $Q(z)$. Thus, we could for each $q$ consider an arbitrary value $z \in \mathcal{Z}$ such that $Q(z)=q$, call it $Q^{-1}(q)$, and think of $D^{[t]}_i$ as a function $D^{[t]}_i(Q^{-1}(q))$ of $q$.

Let us consider a binary combination constructed to capture all units such that $V_i$ belongs to a rectangle $(q,q+h_1] \times (q,q+h_2]\dots \times (q,q+h_J]$ in $\mathbbm{R}^J$ for some ``corner'' location $q \in \mathbbm{R}^J$ and widths $h_1 \dots h_J$. For any $s \subseteq \{1 \dots J\}$, let $h_s:=\sum_{j \in s}h_j \mathbf{e}_j$, where $\mathbf{e}_j$ is the $j^{th}$ standard basis vector. This takes the form of a binary combination $(\alpha,t)$ having $K = 2^{J}$ and coefficients $\alpha_k = (-1)^{|s_k|}/\lambda$ for a certain scalar $\lambda$. The corresponding instrument values are $z_k = Q^{-1}(q + h_{s})$ given an arbitrary ordering $s_1 \dots s_K$ on the $K$ distinct subsets of $\{1 \dots J\}$. Below we will verify that the corresponding linear combination of the $D^{[t]}_i(Q^{-1}(q))$ is equal to $c(G_i)$ with probability one, where we let $c(G_i)$ be an indicator for $V_{i} \in (q,q+h_1] \dots \times (q,q+h_J]$. Via Eq. (\ref{idresult}) in the main text, we can thus identify $\mathbbm{E}[Y_i(t)|V_{i} \in (q,q+h_1] \dots \times (q,q+h_J]]$ as:
\begin{align}
	\mathbbm{E}[Y_i(t)|c(G_i)=1] = \frac{\sum_{s \subseteq \{1 \dots J\}}(-1)^{|s|}\cdot \mathbbm{E}\left[Y_i\cdot D^{[t]}_i|Z_i=Q^{-1}(q + h_{s})\right]}{\sum_{s \subseteq \{1 \dots J\}}(-1)^{|s|}\cdot \mathbbm{E}\left[D^{[t]}_i|Z_i=Q^{-1}(q + h_{s})\right]} \label{eq:rectangleid}
\end{align}
The scalar $\lambda$ depends on the selection mechanism (\ref{eq:leesalanieselectionmodel}) and is $\lambda:=\sum_{s \subseteq \{1 \dots J\}} (-1)^{|s|} \sum_{l \subseteq s} c_l^t$.\footnote{We can simplify this expression of $\lambda$ as follows. Note that given \ref{eq:leesalanieselectionmodel}) the coefficients $c_l^t$ must be such that $\sum_{l \subseteq s} c^t_{l} \in \{0,1\}$ for any $s \subseteq \{1 \dots J\}$. Let $S_{t}$ be the collection of $s$ such that it is equal to one. This is the collection of subsets of the thresholds that when crossed correspond to taking treatment $t$. Then $\lambda=\sum_{s \in S_{t}} (-1)^{|s|}$. We can derive an alternative expression for $\lambda$ by making use of the identity that for any $\sum_{f \subseteq S} (-1)^{|f|} = 0$ for any $S \ne \emptyset$. Then: 
	\begin{align*}
		\lambda &= \sum_{l \subseteq \{1 \dots J\}} c_l^t \sum_{s \supseteq l} (-1)^{|s|} = \ \sum_{l \subseteq \{1 \dots J\}} c_l^t \left(\cancel{\sum_{s \subseteq \{1 \dots J\}} (-1)^{|s|}}-\cancel{\sum_{s \subseteq l} (-1)^{|s|}} + (-1)^{|l|}\right) =\sum_{l \subseteq \{1 \dots J\}} (-1)^{|l|} \cdot c_l^t
	\end{align*}}

We now verify that with this notation $c(G_i)=\sum_{k=1}^{2^J} \alpha_k \cdot D^{[t]}_i(z_k)$. That is:
\begin{align} \label{eq:firstci}
	c(G_i)=\frac{1}{\lambda}\sum_{s \subseteq \{1\dots J\}} (-1)^{|s|} \cdot D^{[t]}_i\left(Q^{-1}(q+\sum_{j \in s}h_j\mathbf{e}_j)\right) &= \frac{1}{\lambda}\sum_{s \subseteq \{1\dots J\}} (-1)^{|s|} \sum_{l \subseteq S_i(q+h_S)} c_l^t \nonumber \\
\end{align}
Note that for any $S' \subseteq S$, $S_i(q+h_S') \subseteq S_i(q+h_S)$. Thus $S_i(q+h_{\{1\dots J\}})$ is the ``largest'' $S_i(q+h_{s})$ and $S_i(q)$ is the smallest.  Define: $$A_i = S_i(q+h_{\{1\dots J\}})-S_i(q) = \{j \in \{1 \dots J\}: q < V_{ji} \le q+h_{\{1\dots J\}}\}$$
$A_i$ is simply the set of dimensions in which $V_{ji}$ falls within the rectangle starting at $q$ with widths $h_{\{1 \dots J\}}$. Now comes the crucial step: we'll now show that (\ref{eq:firstci}) is zero for any individual $i$ for which $A_i$ does not contain all of $\{1 \dots J\}$. Indeed, if there were any $j \notin A_i$, each set $S$ in the first summation of (\ref{eq:firstci}) that did not contain $j$ would be canceled out by the set $S \cup j$, because $(-1)^{|S \cup j|} = -(-1)^{|S|}$, while $S_i(q+h_S)=S_i(q+h_{S\cup j})$. Pairing all sets in this way, we see that evaluates to zero unless $A_i = \{1 \dots J\}$. Now, $A_i = \{1 \dots J\}$ implies that $S_i(q)=\emptyset$, and we can now write $c(G_i)$ as:
$$ c(G_i) = \mathbbm{1}(A_i=\{1 \dots J\}) \cdot \frac{1}{\lambda} \cdot \left(\sum_{s \subseteq \{1 \dots J\}} (-1)^{|s|} \sum_{l \subseteq s} c_l^t\right) = \mathbbm{1}(A_i=\{1 \dots J\})$$
observing that we've defined $c$ to be equal to the quantity in parentheses, which depends on the selection model but not on $i$.

LS's Theorem 3.1 considers the $J^{th}$ order derivative
\begin{align*}
	&\frac{\partial^J}{\partial_{q_1}\dots \partial_{q_J}} \mathbbm{E}[Y_i D^{[t]}_i|Q(Z_i)=q]\\
	&\hspace{1in} = \frac{\partial^J}{\partial_{q_2}\dots \partial_{q_J}} \lim_{h_1 \downarrow 0}\frac{1}{h_1}\left(\mathbbm{E}[Y_i D^{[t]}_i|Q(Z_i)=q+h_1\mathbf{e}_1]-\mathbbm{E}[Y_i D^{[t]}_i|Q(Z_i)=q]\right)\\
	&\hspace{1in} = \lim_{h_1 \dots h_J \downarrow 0}\frac{1}{\prod_{j=1}^J h_j}\cdot \sum_{s \subseteq \{1\dots J\}} (-1)^{|s|} \cdot \mathbbm{E}\left[\left. Y_i D^{[t]}_i \right| Q(Z_i)=q+\sum_{j \in s}h_j\mathbf{e}_j\right]
\end{align*}
and takes the ratio:
\begin{align*}
	\frac{\frac{\partial^J}{\partial_{q_1}\dots \partial_{q_J}} \mathbbm{E}[Y_i D^{[t]}_i|Q(Z_i)=q]}{\frac{\partial^J}{\partial_{q_1}\dots \partial_{q_J}} \mathbbm{E}[D^{[t]}_i|Q(Z_i)=q]} &= \lim_{h_1 \dots h_J \downarrow 0}\frac{\sum_{s \subseteq \{1\dots J\}} (-1)^{|s|} \cdot \mathbbm{E}\left[\left. Y_i D^{[t]}_i \right| Q(Z_i)=q+\sum_{j \in s}h_j\mathbf{e}_j\right]}{\sum_{s \subseteq \{1\dots J\}} (-1)^{|s|} \cdot \mathbbm{E}\left[\left. D^{[t]}_i \right| Q(Z_i)=q+\sum_{j \in s}h_j\mathbf{e}_j\right]}
\end{align*}
LS's result (\ref{LSresult}) thus considers the limit of Eq. (\ref{eq:rectangleid}) as the width of the rectangle goes to zero in all dimensions.

\subsection{Example 5: Unordered (generalized) partial monotonicity}
We can define a generalization of vector and partial monotonicity to settings with multi-valued treatments, also nesting unordered monotonicity:
\begin{assumption*}[UPM]
	For any $t \in \mathcal{T}$, there exists a partial order $\succeq_t$ on $\mathcal{Z}$, such that if $z' \succeq_t z$, $D^{[t]}_i(z') \ge D^{[t]}_i(z)$ for all $i$.
\end{assumption*}
\noindent Note that even in the case of a binary treatment, UPM represents a generalization of partial monotonicity (PM), defined by \citet{MTW} for settings with multiple instruments. UPM allows for an arbitrary partial order on $\mathcal{Z}$, while PM considers a partial order that is based on holding all instruments but one at fixed values.

Assumption UPM implies that for any such $z, z'$: $D^{[t]}_i(z') - D^{[t]}_i(z) \in \{0,1\}$ and thus
$$\mathbbm{E}[Y_i(t)|D^{[t]}_i(z') > D^{[t]}_i(z)] = \frac{\mathbbm{E}[Y_iD^{[t]}_i|Z_i=z']-\mathbbm{E}[Y_iD^{[t]}_i|Z_i=z]}{\mathbbm{E}[D^{[t]}_i|Z_i=z']-\mathbbm{E}[D^{[t]}_i|Z_i=z]} $$
UPM holds, for example, when instruments correspond to choice sets and agents choose rationally from them, as in \citet{aroragoffhjort}. In such a setting instrument values $z$ are subsets of the treatments $\mathcal{T}$ that are available to the agent, and $D^{[t]}_i(z) \ge D^{[t]}_i(z')$ whenever ($z/t \subseteq z'/t$ and $t \in z$ if $t \in z'$). In words, $D^{[t]}_i(z)$ is weakly increasing with respect to the inclusion of $t$ in $z$ (since $i$ can only choose $t$ if it is available), and weakly decreasing with respect to the inclusion of any $t'\ne t$ in $z$ (since $i$ may prefer $t'$ to $t$).

\subsection{Example 6: Pairwise notions of monotonicity}

\citet{sun2024pairwise} proposes a notion of IV-validity that is specific to two values $z,z'$ of the instrument (which may be a vector). This includes the standard LATE model assumptions (independence, exclusion, and monotonicity). However, if independence and exclusion are maintained,  the notion of pairwise valid instruments reduces to what we might call \textit{pairwise-monotonicity}, i.e. that $D^{[t]}_i(z') \ge D^{[t]}_i(z')$ almost surely, or vice versa. 

\citet*{cclate} consider a notion of ``limited monotonicity'' for settings with multiple binary instruments and a binary treatment, which in the notation above corresponds to a setting in which $z'=(1,\dots 1)$ and $z=(0,\dots 0)$. \citet{hoff2023identifying} extends this notion to ordered treatments that need not be binary.

In the context of ``judge designs'' where the instrument is a scalar continuous measure of ``leniency'' with respect to a binary treatment, \citet{sigstad} and \citet{sigstad2024marginal} introduce a notion of ``extreme-pair'' monotonicity $D_i(\bar{j}) \ge D_i(\underline{j})$ almost surely, where $\bar{j}$ is the strictest judge, and $\underline{j}$ the most lenient.

In the case of a binary treatment $D_i$, the above papers point out that under a limited version of ``monotonicity'' between a pair of values $z,z'$, a particular local average treatment effect can be identified from a simple Wald estimand:
$$\mathbbm{E}[Y_i(1)-Y_i(0)|D_i(z') > D_i(z)] = \frac{\mathbbm{E}[Y_i|Z_i=z']-\mathbbm{E}[Y_i|Z_i=z]}{\mathbbm{E}[D_i|Z_i=z']-\mathbbm{E}[D_i|Z_i=z]}$$
This corresponds to a binary collection in which $\alpha_{z'} = 1$ and $\alpha_{z} = -1$ for $t=1$ , while $\alpha_{z'} = -1$ and $\alpha_{z} = 1$ for $t=0$.
		\section{Letting local causal parameters depend on $Z_i$} \label{app:deponz}
Let $z_k: \mathcal{Z} \rightarrow \mathcal{Z}$ be a function that maps an instrument value $Z_i$ to some possibly different value in $\mathcal{Z}$. Non-constant functions $z_k(\cdot)$ will allow us to nest parameters such as the average treatment effect on the treated, as well as some parameters from \citet{goff2024vector}. In that paper $z_k(z)$ could for instance change one component of $z$, and the $\alpha_k$ and $z_k(\cdot)$ can be chosen so that $c(G_i,Z_i):=\sum_{k}\alpha_{k}\cdot D^{[t]}_i(z_k(Z_i))$ only takes values of zero or one, i.e. $\alpha_z = \left\{\alpha_k, z_k(z)\right\}_{k=1}^K$ yields a binary combination for any $z \in \mathcal{Z}$. Then, by the law of iterated expectations: $ \mathbbm{E}\left[Y_i(t)\left|c(G_i,Z_i)=1\right. \right] = \sum_{z \in \mathcal{Z}} P(Z_i=z) \cdot  \mathbbm{E}[Y_i(t)|c(G_i,Z_i)=1]$ where each term in the summand is identified by (\ref{idresult}) and the distribution of $Z_i$.

Let us maintain the assumption that the support of the instruments $\mathcal{Z}$ is discrete and finite. Consider any counterfactual mean of the form $\theta=\mathbbm{E}[Y_i(t)|c(G_i,Z_i)=1]$ where now $c: \mathcal{G} \times \mathcal{Z} \rightarrow \{0,1\}$. By the law of iterated expectations over $Z_i$ and independence Eq. \eqref{eq:independence}, we can write $\theta$ as:
\begin{align} 
	\theta &= \sum_{z \in \mathcal{Z}} P(Z_i = z|c(G_i,z)=1) \cdot \mathbbm{E}[Y_i(t)|c(G_i,z)=1,Z_i=z] \nonumber\\
	&= \sum_{z \in \mathcal{Z}} P(Z_i = z) \cdot \mathbbm{E}[Y_i(t)|c_z(G_i)=1] \label{eq:deponz}
\end{align}
where we let $c_z(g)$ denote $c(g,z)$. Eq \eqref{eq:deponz} shows that $\theta$ can be written as a convex combination of $|\mathcal{Z}|$ counterfactual means of the form $\mu^t_c$ considered by Theorems \ref{thm:suff} and \ref{thm:necc}, with complier groups $c_z(\cdot)$ that depend on $z$. It is clear then by Theorem \ref{thm:suff}, a sufficient condition for $\theta$ to be outcome-nonrestrictive identified is that $c_z$ lies in the rowspace of matrix $A^{[t]}$ for each $z \in \mathcal{Z}$.

Theorem \ref{thm:necc} similarly extends to the more general class of functions $c(G_i,Z_i)$, provided that the family $\mathscr{P}_Z$ of distributions over the instruments allows for degenerate distributions at each value of $Z_i$. Then $c_z$ must lie in the rowspace of $A^{[t]}$ for all $z \in \mathcal{Z}$. If it were not, then for some $z \in \mathcal{Z}$, $c_z \notin rs(A^{[t]})$ and hence $\mu^t_{c_z}$ is not outcome-nonrestrictive identified, by Theorem \ref{thm:necc}. For a degenerate distribution $P_\mathcal{Z}$ that sets $P(Z_i=z)=1$, $\theta = \mu^t_{c_z}$ and hence $\theta$ is not outcome-nonrestrictive identified if $c_z \notin rs(A^{[t]})$. With this extension Theorem \ref{thm:necc} of this paper nests Theorem 2 of \citet{goff2024vector} as a special case, and expands its reach even in the case that vector monotonicity is maintained, if the outcome variable is continuous.
		\section{Extended analysis of identification under NSOG} \label{app:nsog}

It is known that unconditional means $\mathbbm{E}[Y_i(t)]$ of a given potential outcome $Y_i(t)$ can be point-identified, given an order condition on the instruments, under an assumption of ``no-selection on gains'' (NSOG) (see e.g. \citet{Kolesar2013, aroragoffhjort} for versions of this result).\footnote{\citet{Kolesar2013} calls this ``constant average treatment effects'', and does not use the term NSOG.} Note that identification of $\mathbbm{E}[Y_i(t)]$ and $\mathbbm{E}[Y_i(t')]$ immediately implies identification of unconditional average treatment effects $\mathbbm{E}[Y_i(t')-Y_i(t)]$ as well.

NSOG says that treatment effects are mean independent of actual treatment, given any realization of the instruments:
\begin{assumption*}[NOSG (no selection on gains)]
	For any $t,t',t_1,t_2 \in \mathcal{T}$ and $z \in \mathcal{Z}$:
	$$\mathbbm{E}[Y_i(t')-Y_i(t)|T_i=t_1,Z_i=z] = \mathbbm{E}[Y_i(t')-Y_i(t)|T_i=t_2,Z_i=z]$$
\end{assumption*}
\noindent NSOG implies that if we consider any fixed treatment value $0 \in \mathcal{T}$, then $\mathbbm{E}[Y_i(t')-Y_i(0)|T_i=t,Z_i=z] = \mathbbm{E}[Y_i(t')-Y_i(0)|Z_i=z]$ for any $t,z$, which coupled with independence \eqref{eq:independence} in turn implies that $\mathbbm{E}[Y_i(t')-Y_i(0)|T_i=t,Z_i=z] = \mathbbm{E}[Y_i(t')-Y_i(0)] := \Delta_{t'}$, where note that $\Delta_{t'}$ does not depend on $z$ or $t$. This normalization against an arbitrary treatment $0 \in \mathcal{T}$ allows us to carry around one less index in our expressions.

\subsection{Identification under NSOG}
This subsection first shows that $\mathbbm{E}[Y_i(t)]$ can be point identified for each $t \in \mathcal{T}$ under NSOG, given rich enough support of the instruments. The proof essentially follows that of \citet{aroragoffhjort}, which adapts an argument from \citet{Kolesar2013} to cases in which the treatments $\mathcal{T}$ are not necessarily ordered.

NSOG implies that:
$$\mathbbm{E}[Y_i-Y_i(0)|T_i=t,Z_i=z] = \mathbbm{E}[Y_i(t)-Y_i(0)|T_i=t,Z_i=z] = \Delta_{t}$$
Averaging over the conditional distribution of $T_i$ given $Z_i=z$, we have by the law of iterated expectations that
\begin{equation} \label{eq:nsog}
	\mathbbm{E}[Y_i-Y_i(0)|Z_i=z] = \sum_{t \in \mathcal{T}} P(T_i=t|Z_i=z) \cdot \Delta_t
\end{equation}
To now see that $\mathbbm{E}[Y_i(t)]$ can be identified under NSOG given rich enough instrument support, let us assume that $|\mathcal{Z}| \ge |\mathcal{T}|$ and suppose that there exists a set of $|\mathcal{T}|$ instrument values $\tilde{\mathcal{Z}} \subseteq \mathcal{Z}$ such that the $|\mathcal{T}| \times |\mathcal{T}|$ matrix $\Sigma$ with entries $\Sigma_{zt} = P(Z_i=z, T_i=t)$ over all $z \in \tilde{\mathcal{Z}}$ is invertible, with $P(Z_i = z) > 0$ for each $z \in \tilde{\mathcal{Z}}$.

Eq. \eqref{eq:nsog} can be re-written by multiplying both sides by $P(Z_i=z)$ as
$$\mathbbm{E}[\{Y_i-Y_i(0)
\} \cdot \mathbbm{1}(Z_i=z)] = \sum_{t \in \mathcal{T}} \Sigma_{zt} \cdot \Delta_t$$
for each $z \in \tilde{\mathcal{Z}}$. Equivalently, using independence:
\begin{align*}
	\mathbbm{E}[Y_i \cdot \mathbbm{1}(Z_i=z)] &= P(Z_i=z)\cdot \mathbbm{E}[Y_i(0)] +  \sum_{t \in \mathcal{T}} \Sigma_{zt} \cdot \Delta_t\\
	&= P(Z_i=z)\cdot \mathbbm{E}[Y_i(0)] +  \sum_{t \in \mathcal{T}, t \ne 0} \Sigma_{zt} \cdot \Delta_t\\
	&= \left\{P(Z_i=z,T_i=0)+\sum_{t \in \mathcal{T}, t \ne 0} \Sigma_{zt}\right\}\cdot \mathbbm{E}[Y_i(0)] +  \sum_{t \in \mathcal{T}, t \ne 0} \Sigma_{zt} \cdot \Delta_t\\
	&= P(Z_i=z,T_i=0)\cdot \mathbbm{E}[Y_i(0)] + \sum_{t \in \mathcal{T}, t \ne 0} \Sigma_{zt} \cdot \mathbbm{E}[Y_i(0)] + \sum_{t \in \mathcal{T}, t \ne 0} \Sigma_{zt} \cdot \Delta_t\\
	&= P(Z_i=z,T_i=0)\cdot \mathbbm{E}[Y_i(0)] + \sum_{t \in \mathcal{T}, t \ne 0} \Sigma_{zt} \cdot \mathbbm{E}[Y_i(t)]\\
	& = \sum_{t \in \mathcal{T}} \Sigma_{zt} \cdot \mathbbm{E}[Y_i(t)]
\end{align*}
using that $\Delta_0 = 0$ in the second equality. This yields a system of $|\mathcal{T}|$ equations in the $|\mathcal{T}|$ unknowns $\mathbbm{E}[Y_i(t)]$ with identified coefficients $\Sigma_{zt}$. Given that $\Sigma^{-1}$ is invertible, we have then that 
\begin{equation} \label{eq:nsogresult}
	\mathbbm{E}[Y_i(t)] = \sum_{z \in \tilde{\mathcal{Z}}} \Sigma^{-1}_{tz} \cdot \mathbbm{E}[Y_i \cdot \mathbbm{1}(Z_i=z)]
\end{equation}
Note that if $|\mathcal{Z}| \ge |\mathcal{T}|$ there may be overidentification restrictions implied by NSOG, that the RHS of \eqref{eq:nsogresult} is the same for different possible choices of $\tilde{\mathcal{Z}} \subset \mathcal{Z}$ (note that $\Sigma$ also depends on the choice of $\tilde{\mathcal{Z}}$). Furthermore, the RHS of \eqref{eq:nsogresult} is the estimand of a two-stage least squares regression of $Y_i$ on indicators for the mutually-exclusive treatments in $\mathcal{T}$ (and no constant), instrumented by indicators for the mutually-exclusive instrument values in $\tilde{\mathcal{Z}}$.

\subsection{How Theorem \ref{thm:necc} does not cover NSOG}
Since the result of the last section makes no assumption about which response types can show up in the population, it is compatible with any selection model $\mathcal{G} \subseteq \{0,1\}^{\mathcal{T}^\mathcal{Z}}$, including for example the full powerset $\{0,1\}^{\mathcal{T}^\mathcal{Z}}$ of possible response types $\mathcal{T}^\mathcal{Z}$.

Whatever $\mathcal{G}$ is, unconditional means like $\mathbbm{E}[Y_i(t)]$ correspond to the choice $c = (1, \dots, 1)'$ in $\mathbbm{R}^{|\mathcal{G}|}$. As long as $\mathcal{G}$ allows never-takers with respect to treatment $t$, this choice of $c$ will not lie in the rowspace of $A^{[t]}$. The unrestricted selection model $\mathcal{G} = \{0,1\}^{\mathcal{T}^\mathcal{Z}}$, for example, features such never-takers for any $t \in \mathcal{T}$. Thus the result of the last section demonstrates that it is possible to achieve point identification of $\mu_c^t$ without $c \in rs(A^{[t]})$, if we impose NSOG and that $\Sigma^{-1}$ exists.

Note that the imposing of NSOG makes this identification \textit{not} outcome-nonrestrictive. However, it is illustrative to see where the proof of Theorem \ref{thm:necc} breaks down in the case of the NSOG identification result. Let $NSOG$ denote the set of distributions $\mathcal{P}$ for which $P_{latent}(\mathcal{P})$ satisfies NSOG. In this notation, the last section establishes that $\{\theta(\mathcal{P}): \mathcal{P} \in (M \cap NSOG) \textrm{ and } \phi(\mathcal{P}) = \mathcal{P}_{obs}\}$ is a singleton for all $\mathcal{P}_{obs}$ that satisfy the rich support condition that $\Sigma^{-1}$ exists, which requires that there be no $\mathcal{P},\mathcal{P}' \in M \cap NSOG$ such that $\phi(\mathcal{P})=\phi(\mathcal{P}')$ but $\theta(\mathcal{P}) \ne \theta(\mathcal{P}')$ and such that $\Sigma^{-1}$ exists under $\mathcal{P}$ or $\mathcal{P}'$.

To see that there is no contradiction with Theorem \ref{thm:necc}, I below show that given a $\mathcal{P} \in (M \cap NSOG)$, the alternative distribution $\mathcal{P}'$ defined from it in the proof of Theorem \ref{thm:necc} does not lie within $NSOG$ when $(1, \dots 1)' \notin rs(A^{[t]})$. In partcular, the remainder of this section shows that if $(1, \dots 1)' \notin rs(A^{[t]})$ for any given $t \in \mathcal{T}$, the construction $\mathcal{P}'$ utilized in the proof of Theorem \ref{thm:necc} cannot lie in $NSOG$. If on the other hand $(1, \dots 1)' \in rs(A^{[t]})$, then Eq. \eqref{eq:finalcondtion} in the proof of Theorem \ref{thm:necc} shows that $\theta(\mathcal{P})=\theta(\mathcal{P}')$, consistent with $\theta$ being identified.

Recall that the way in which the proof of Theorem \ref{thm:necc} builds a candidate $\mathcal{P}'$ from the actual distribution $\mathcal{P}$ is to construct from the set of true potential outcome CDFs $\mathbf{G}^*$\: $[\mathbf{G}^*(y)]_{t,g}:=P(G_i=g)\cdot F_{Y(t)|G=g}(y)$ a new set of such CDFs $\mathbf{G}^\lambda$. For continuity with the notation used in this discussion so far, let $\mathcal{P}'$ correspond to the collection of CDFs $\mathbf{G}^\lambda$, and let us make explicit whether outcome expectations are with respect to the distribution $\mathcal{P}$ or $\mathcal{P}'$,\footnote{The response type probabilities $P(G_i=g)$ are the same for both $\mathcal{P}$ and $\mathcal{P}'$ so I leave this implicit for ease of exposition.}. Then we have by integrating Eq. \eqref{eq:Fstarnew} that:
\begin{align*}
	P(G_i=g)\cdot &\mathbbm{E}_\mathcal{P'}[Y_i(t)|G_i=g]= P(G_i=g)\cdot \mathbbm{E}_\mathcal{P}[Y_i(t)|G_i=g]\\
	&\hspace{-0.5in}+\lambda \cdot \sum_{g'} [I-(A^{[t]})^+ A^{[t]}]_{g,g'}\cdot P(G_i=g')\cdot \left\{\mathbbm{E}_\mathcal{P}[Y_i(t)|G_i=g']-\mathbbm{E}_\mathcal{P}[Y_i(t)|G_i=g^*]\right\}
\end{align*}
Then using independence \eqref{eq:independence}:
\begin{align*}
	\mathbbm{E}_\mathcal{P'}[&Y_i(t)-Y_i(0)|G_i=g, Z_i=z] = \mathbbm{E}_\mathcal{P'}[Y_i(t)|G_i=g]-\mathbbm{E}_\mathcal{P'}[Y_i(0)|G_i=g]\\
	&=\mathbbm{E}_\mathcal{P}[Y_i(t)-Y_i(0)|G_i=g]\\
	&+\lambda \cdot \sum_{g'} [I-(A^{[t]})^+ A^{[t]}]_{g,g'}\cdot \frac{P(G_i=g')}{P(G_i=g)}\cdot \left\{\mathbbm{E}_\mathcal{P}[Y_i(t)|G_i=g']-\mathbbm{E}_\mathcal{P}[Y_i(t)|G_i=g^*]\right\}\\
	&-\lambda \cdot \sum_{g'} [I-(A^{[0]})^+ A^{[0]}]_{g,g'}\cdot \frac{P(G_i=g')}{P(G_i=g)}\cdot \left\{\mathbbm{E}_\mathcal{P}[Y_i(0)|G_i=g']-\mathbbm{E}_\mathcal{P}[Y_i(0)|G_i=g^*]\right\}
\end{align*}
Therefore, for any $t_1 \in \mathcal{T}$: \normalsize
\begin{align}
	&\mathbbm{E}_\mathcal{P'}[Y_i(t)-Y_i(0)|T_i=t_1, Z_i=z] = \mathbbm{E}_\mathcal{P'}[Y_i(t)-Y_i(0)|A^{[t_1]}_{z,G_i} = 1, Z_i=z]  \nonumber \\
	&= \sum_g P(G_i=g|A^{[t_1]}_{z,G_i} = 1)\cdot \mathbbm{E}_\mathcal{P'}[Y_i(t)-Y_i(0)|G_i=g, Z_i=z] \nonumber \\
	&= \frac{1}{P(A^{[t_1]}_{z,G_i} = 1)} \sum_g P(G_i=g)\cdot A^{[t_1]}_{z,g} \cdot \mathbbm{E}_\mathcal{P'}[Y_i(t)-Y_i(0)|G_i=g, Z_i=z] \nonumber \\
	&= \frac{1}{P(A^{[t_1]}_{z,G_i} = 1)} \cdot \sum_g P(G_i=g)\cdot A^{[t_1]}_{z,g}  \cdot \left[\mathbbm{E}_\mathcal{P}[Y_i(t)-Y_i(0)|G_i=g] \color{white}{\frac{1}{1}} \right. \nonumber\\
	& \left.+\lambda \cdot \sum_{g'} [I-(A^{[t]})^+ A^{[t]}]_{g,g'}\cdot \frac{P(G_i=g')}{P(G_i=g)}\cdot \left\{\mathbbm{E}_\mathcal{P}[Y_i(t)|G_i=g']-\mathbbm{E}_\mathcal{P}[Y_i(t)|G_i=g^*]\right\} \right. \nonumber \\
	&\left.-\lambda \cdot \sum_{g'} [I-(A^{[0]})^+ A^{[0]}]_{g,g'}\cdot \frac{P(G_i=g')}{P(G_i=g)}\cdot \left\{\mathbbm{E}_\mathcal{P}[Y_i(0)|G_i=g']-\mathbbm{E}_\mathcal{P}[Y_i(0)|G_i=g^*]\right\}\right] \nonumber \\
	&= \Delta_t+\frac{\lambda}{P(A^{[t_1]}_{z,G_i} = 1)} \cdot \left[\sum_{g'} [A^{[t_1]}(I-(A^{[t]})^+ A^{[t]})]_{z,g'}\cdot P(G_i=g')\cdot \left\{\mathbbm{E}_\mathcal{P}[Y_i(t)|G_i=g']-\mathbbm{E}_\mathcal{P}[Y_i(t)|G_i=g^*]\right\} \right. \nonumber \\
	&\hspace{.5cm}\left.- \sum_{g'} [A^{[t_1]}(I-(A^{[0]})^+ A^{[0]})]_{z,g'}\cdot P(G_i=g')\cdot \left\{\mathbbm{E}_\mathcal{P}[Y_i(0)|G_i=g']-\mathbbm{E}_\mathcal{P}[Y_i(0)|G_i=g^*]\right\}\right] \label{eq:nsogcondition}
\end{align} \large
where $\Delta_t:=\mathbbm{E}_{\mathcal{P}}[Y_i(t)-Y_i(0)]$. Note that we can simplify the denominator as $P(A^{[t_1]}_{z,G_i} = 1)= \sum_g P(G_i=g) \cdot A^{[t_1]}_{z,g}=[A^{[t_1]}P]_{z}$, where $P$ is a vector of response type probabilities $P_{g}=P(G_i=g)$. Since $\Sigma_{zt} = P(Z_i=z,T_i=t) = P(Z_i=z)\cdot P(T_i=t|Z_i=z) = P(Z_i=z)\cdot \sum_{g} P(G_i=g)\cdot A^{[t]}_{z,g} = P(Z_i=z)\cdot [A^{[t]}P]_z$. We can thus rewrite $P(A^{[t_1]}_{z,G_i} = 1)$ as $\Sigma_{zt}/P(Z_i=z)$.

For us to have $\mathcal{P'} \in NSOG$, it must be the case that the RHS of \eqref{eq:nsogcondition} does not depend on $z$ or $t_1$, and equals $\Delta_t'(\lambda):=\mathbbm{E}_{\mathcal{P}'}[Y_i(t)-Y_i(0)]$ for any $\mathcal{P} \in (M \cap NSOG)$. In the notation $\Delta_t'(\lambda)$ we make explicit that the value of $\mathbbm{E}_{\mathcal{P}'}[Y_i(t)-Y_i(0)]$ could depend on $\lambda$. In the case of $t_1=t$, expression  \eqref{eq:nsogcondition} for $\Delta_t'(\lambda)$ simplifies to \normalsize
$$\Delta_t - \frac{\lambda}{\Sigma_{zt}} \cdot P(Z_i=z) \cdot \sum_{g'} [A^{[t]}-A^{[t]}(A^{[0]})^+ A^{[0]}]_{z,g'}\cdot P(G_i=g')\cdot \left\{\mathbbm{E}_\mathcal{P}[Y_i(0)|G_i=g']-\mathbbm{E}_\mathcal{P}[Y_i(0)|G_i=g^*]\right\}$$ \large
using that $A^{[t]}(A^{[t]})^+ A^{[t]}=A^{[t]}$.
Similarly, taking $t_1=0$, we have that $\Delta_t'(\lambda)$ is equal to \normalsize
$$\Delta_t + \frac{\lambda}{\Sigma_{zt}} \cdot P(Z_i=z) \cdot \sum_{g'} [A^{[0]}-A^{[0]}(A^{[t]})^+ A^{[t]}]_{z,g'}\cdot P(G_i=g')\cdot \left\{\mathbbm{E}_\mathcal{P}[Y_i(t)|G_i=g']-\mathbbm{E}_\mathcal{P}[Y_i(t)|G_i=g^*]\right\}$$ \large
Note that for any $\mathcal{P}$, there exists a small enough $\lambda>0$ that $\mathcal{P}' \in M$. For the above equations to simultaneously hold for any such $\lambda > 0$, we must have for any $z$ such that $P(Z_i=z)>0$:
\begin{align} \label{eq:needforanylambda}
&\sum_{g'} [A^{[0]}(I-(A^{[t]})^+ A^{[t]})]_{z,g'}\cdot P(G_i=g')\cdot \left\{\mathbbm{E}_\mathcal{P}[Y_i(t)|G_i=g']-\mathbbm{E}_\mathcal{P}[Y_i(t)|G_i=g^*]\right\} \nonumber\\
&+\sum_{g'} [A^{[t]}(I-(A^{[0]})^+ A^{[0]})]_{z,g'}\cdot P(G_i=g')\cdot \left\{\mathbbm{E}_\mathcal{P}[Y_i(0)|G_i=g']-\mathbbm{E}_\mathcal{P}[Y_i(0)|G_i=g^*]\right\}=0
\end{align}
for all $\mathcal{P} \in M \cap REG \cap NSOG$. Consider a distribution $\mathcal{P}$ for which $\mathcal{P}_Z$ has full support $\mathcal{Z}$, and for which conditional average treatment effects take the separable form $\mathbbm{E}[Y_i(t)|G_i=g] = \lambda_g + \Delta_t$, where $\Delta_0:=0$. Defining $\tilde{\lambda}_g:=\lambda_g - \lambda_{g^*}$, Eq. \eqref{eq:needforanylambda} reads in this case: \normalsize
\begin{align*} \label{eq:needforanylambda2}
	\sum_{g'} [A^{[0]}(I-&(A^{[t]})^+ A^{[t]})]_{z,g'}\cdot P(G_i=g')\cdot \tilde{\lambda}_{g'}+\sum_{g'} [A^{[t]}(I-(A^{[0]})^+ A^{[0]})]_{z,g'}\cdot P(G_i=g')\cdot \tilde{\lambda}_{g'}=0
\end{align*} \large
Given that $\tilde{\lambda}_{g'}$ can be freely chosen such that $P(G_i=g')\cdot \tilde{\lambda}_{g'}=\mathbbm{1}(g'=g)$ for any $g \in \mathcal{G}$ and $\mathcal{P}_G$, this can only be true when $A^{[0]}(I-(A^{[t]})^+ A^{[t]})=A^{[t]}(I-(A^{[0]})^+ A^{[0]})$ entry by entry as matrices. We'll now see that this can only be true for all $t \in \mathcal{T}$ if $c=(1, \dots 1)' \in rs(A^{[t]})$ for all $t \in \mathcal{T}$.

Note that the matrix $(A^{[t]})^+ A^{[t]}$ is an orthogonal projector onto into $rs(A^{[t]})$, and $(A^{[0})^+ A^{[0]}$ is an orthogonal projector onto into $rs(A^{[0]})$, and the required condition is
$${A^{[t]}_z}'(I-(A^{[0})^+ A^{[0]}) = -{A^{[0]}_z}'(I-(A^{[t})^+ A^{[t]})$$
for all $z \in \mathcal{Z}$, where the row-vector ${A^{[t]}_z}'$ denotes row $z$ of the matrix $A^{[t]}$, and similarly for $A^{[0]}$. Note that the row-vector ${A^{[t]}_z}'(I-(A^{[0})^+ A^{[0]})$ belongs to the orthogonal complement of $rs(A^{[0]})$ in $\mathbbm{R}^{|\mathcal{G}|}$. It is thus orthogonal to any row of $A^{[0]}$, including ${A^{[0]}_z}'$. But $-{A^{[0]}_z}'(I-B)$ cannot be orthogonal to $c_z$ unless ${A^{[0]}_z}'(A^{[t})^+ A^{[t]} = {A^{[0]}_z}'$ so that $-{A^{[0]}_z}'(I-(A^{[t})^+ A^{[t]})$ is the zero vector. In that case, note that ${A^{[t]}_z}'(I-(A^{[0})^+ A^{[0]})$ is the zero vector as well, so we have that ${A^{[t]}_z}' \in rs(A^{[0]})$ and ${A^{[0]}_z}' \in rs(A^{[t]})$. Compiling over all $z \in \mathcal{Z}$, we have that $A^{[0]}$ and $A^{[t]}$ have the same row-space. Repeating this argument over all $t \in \mathcal{T}$, we have that $rs(A^{[t]})$ is the same for all $t \in \mathcal{T}$.

Now let us see that this in turn implies that $(1, \dots 1)' \in rs(A^{[t]})$. Note that $\sum_{t' \in \mathcal{T}} {A_z^{[t']}}'=(1, \dots 1)'$ for any $z$, because all response types take one and only one treatment when $Z_i=z$. But since ${A_z^{[t']}}' \in rs(A^{[t']})$, it must also be in the rowspace of $A^{[t]}$. Since ${A_z^{[t']}}' \in rs(A^{[t]})$ for each $t'$, the linear combination $\sum_{t' \in \mathcal{T}} {A_z^{[t']}}'=(1, \dots 1)'$ is also in $rs(A^{[t]})$. Thus we have shown that $\mathcal{P}' \in NSOG$ implies that $(1, \dots 1)' \in rs(A^{[t]})$ for all $t$.

		\section{Partial identification when $c \notin rowspace(A^{[t]})$} \label{sec:partialid}

\subsection{Relationship to \citet*{bai2024identifyingpowermonotonicityaverage}} \label{sec:bhmsv}

\citet{bai2024identifyingpowermonotonicityaverage} (BHMSV) study the identifying power for ATEs and unconditional counterfactual means of a restriction on selection that they call \textit{generalized monotonicity} (GM). In my notation, GM says that for a given $\mathcal{P}_{latent}$ and each $t \in \mathcal{T}$, there exists an instrument value $z^* = z^*(t,\mathcal{P}_{latent})$ such that
\begin{equation} \label{eq:gm}
	P(D_i(z^*) \ne t \textrm{ and } D_i(z)=t \textrm{ for some } z \in \mathcal{Z})=0
\end{equation}
according to $\mathcal{P}_{latent}$. That is, no individual takes treatment $t$ when $z \ne z^*$ unless they also do when $z = z^*$. BHMSV show that GM or any strengthening of it (that does not restrict outcomes) does not reduce the size of identified sets for unconditional parameters of the form $\mathbbm{E}[Y_i(t)]$, when the outcome variable has finite support $\mathcal{Y}$ and the instruments are also finite.

While GM nests many notions of monotonicity from the literature that have been used for positive point identification results, it generalizes them in a different way than the criterion $c \in rs(A^{[t]})$ of the present paper does. While $c \in rs(A^{[t]})$ ensures point identification of $\mathbbm{E}[Y_i(t)|c_{G_i}=1]$, GM represents a double-edged sword when the parameter of interest is an unconditional mean or ATE with $c = (1,1,\dots 1)'$. Using Theorem \ref{thm:necc} of this paper, we can see that GM is in fact sufficient to establish either that i) $\mathbbm{E}[Y_i(t)]$ is point identified in an outcome-nonrestrictive way; or ii) that it is \textit{not} point identified in an outcome-nonrestrictive way. Which of these cases i) or ii) holds can be determined by the observable distribution $\mathcal{P}_{obs}$, and does not depend on $\mathcal{G}$ beyond it satisfying GM.

Let $\tilde{\mathcal{G}}(\mathcal{P}_{latent})$ be the support of $G_i$ under $\mathcal{P}_{latent}$, and note that \eqref{eq:gm} is equivalent to:
\begin{equation} \label{eq:gm2}
	\textrm{For all } g \in \tilde{\mathcal{G}}(\mathcal{P}_{latent}): \quad A^{[t]}_{z^*,g} =0 \implies A^{[t]}_{z,g} =0 \textrm{ for all } z \in \mathcal{Z}
\end{equation}
Consider a given distribution of observables $\mathcal{P}_{obs}$. Either $P(T_i=t|Z_i=z^*) = 1$ or $P(T_i=t|Z_i=z^*) < 1$ according to $\mathcal{P}_{obs}$. If the first case holds, then $\mathbbm{E}[Y_i(t)]=\mathbbm{E}[Y_i|Z_i=z^*]$ and $\mathbbm{E}[Y_i(t)]$ is thus point-identified without requiring any restrictions on selection. Thus, assuming GM or any strengthening of it cannot reduce the identified set for $\mathbbm{E}[Y_i(t)]$ further, unless it results in model rejection.

If on the other hand $P(T_i=t|Z_i=z^*) < 1$ and GM holds, then there must exist a $g \in \tilde{\mathcal{G}}(\mathcal{P}_{latent})$ such that $A^{[t]}_{z^*,g} =0$. Therefore by \eqref{eq:gm2}, for this $g$ it must be that $A^{[t]}_{z,g} =0$ for all $z \in \mathcal{Z}$, i.e. there are never-takers with respect to treatment $t$. This in turn implies that $(1,1,\dots 1)' \notin rs(A^{[t]})$, precisely the case in which we know that $\mathbbm{E}[Y_i(t)]$ is \textit{not} outcome-nonrestrictive point identified, by Theorem \ref{thm:necc}.

By showing that the bounds on $\mathbbm{E}[Y_i(t)]$ in partially identified settings are not improved by imposing restrictions stronger than GM, BHMSV underscore the importance of: i) focusing on other parameters of interest beyond the ATE (i.e. $c \ne (1,1,\dots 1)'$) when one is willing to impose restrictions on selection; and ii) finding restrictions on selection that are outside of the scope of GM. Indeed, many of the selection models reported in Appendix \ref{sec:catalog} below do not satisfy GM, yet are sufficient for point identification of more localized treatment effect parameters than the ATE (and in some cases the ATE as well).

The remainder of this section provides more detail to build intuition about the connection between BHMSV's result and the proof of Theorem \ref{thm:necc} in this paper. For a given $\mathcal{P}_{obs}$, let us write the identified set for $\mathbbm{E}[Y_i(t)]$ under model $M$ as
$$\Theta(\mathcal{P}_{obs},M) = \{\theta(\mathcal{P}): \phi(\mathcal{P})=\mathcal{P}_{obs} \textrm{ and } \mathcal{P} \in M\}$$
Given BHMSV's assumption that $\mathcal{Y}$ is finite, let us for each $y \in \mathcal{Y}$ define $x^y$ to be a $|\mathcal{G}|$-component vector with components $x^y_g=P(Y_i(t)=y|G_i=g)$ and $\beta$ to be a $|\mathcal{Z}|$-component vector with components $\beta^y_z = P(Y_i=y,T_i=t|Z_i=z)$. The restriction $\phi(\mathcal{P})=\mathcal{P}_{obs}$ corresponds to the set of solutions to finite system of linear equations $A^{[t]}x^y=\beta^y$, for each $y\in \mathcal{Y}$. Given $|\mathcal{Y}| < \infty$, we can collect these into a single finite linear system $\mathcal{A}^{[t]}\tilde{x}=\tilde{\beta}$, where $\mathcal{A}^{[t]}$ is a block diagonal matrix of $A^{[t]}$ copied $|\mathcal{Y}|$ times, $\tilde{x}$ is a $|\mathcal{Y}| \times |\mathcal{G}|$ component vector, and $\tilde{\beta}$ is a $|\mathcal{Y}| \times |\mathcal{Z}|$ component vector. The set $\mathcal{X}:=\{\tilde{x}(\mathcal{P}): \phi(\mathcal{P})=\mathcal{P}_{obs}\}$ is thus a vector space, where we let $\tilde{x}(\mathcal{P})$ represent $\tilde{x}$ as a function of the distribution of model fundamentals $\mathcal{P}$.

Whether GM or any strengthening of it reduces the identified set for $\mathbbm{E}[Y_i(t)]$ thus depends upon whether the action of $\theta(\cdot)$ on the $\mathcal{P} \in M$ such that $\tilde{x}(\mathcal{P}) \in \mathcal{X}$ reduces $\Theta(\mathcal{P}_{obs},M)$ relative to a case with no restrictions on selection. Eq. \eqref{eq:finalcondtion} from the proof of Theorem \ref{thm:necc} suggests that $\Theta(\mathcal{P}_{obs},M)$ satisfies
\begin{equation} \label{eq:idsetexpression}
	\Theta(\mathcal{P}_{obs},M) \subseteq \left\{\mathbbm{1}'(A^{[t]})^+\beta + \sum_{g'} [\mathbbm{1}'(I-(A^{[t]})^+ A^{[t]})]_{g'} \cdot w_{g'}: w \in \mathbbm{R}^{|\mathcal{G}|} \right\}
\end{equation}
where $\mathbbm{1}:=(1,1,\dots 1)'$ and $\beta$ a $|\mathcal{Z}|$-component vector with components $\beta_z = \mathbbm{E}[Y_i \cdot \mathbbm{1}(T_i=t)|Z_i=z]$. GM implies that the set of the RGS is not a singleton if $P(T_i=t|Z_i=z^*) < 1$. The subset relation appearing in \eqref{eq:idsetexpression} reflects that, as in Theorem \ref{thm:necc}, some $\tilde{x}$ for which $\mathcal{A}^{[t]}\tilde{x}=\tilde{\beta}$ may not be attainable from $\mathcal{P}$ that are valid distributions and reflect any further assumptions of the model $M$, for example that $Y_i$ has bounded support.

BHMSV show that if $M$ does not restrict outcomes, $\Theta(\mathcal{P}_{obs},M)$ is in fact equal to the identified set under no selection restrictions, which is (given the finite support $\mathcal{Y}$): $$\left\{\beta_{z^*} - \min\{\mathcal{Y}\} \cdot P(T_i \ne t|Z_i=z^*),\beta_{z^*} + \max\{\mathcal{Y}\} \cdot P(T_i \ne t|Z_i=z^*) \right\}$$
An interesting question for further study is in what manner the result of BHMSV extends to the more general class target parameters indexed by vectors $c$ that may differ from $(1,1,\dots,1)'$.  A reasonable conjecture would be that if, given $\mathcal{P}_{obs}$, a class of restrictions on selection cannot change the fact that $c \notin rs(A^{[t]})$ , there is limited scope for such restrictions to reduce the size of the identified set for $\mu_c^t$.

\subsection{Partial identification in general}
Accordingly, consider an arbitrary $c \in \{0,1\}^{|\mathcal{G}}$ where we may have that $c \notin rs(A^{[t]})$. By similar logic as above, we can deduce that the identified set $\Theta(\mathcal{P}_{obs},M)$ for $\mu_c^t$ satisfies:
\begin{align*}
	\Theta(\mathcal{P}_{obs},M) \subseteq \frac{1}{P(c(G_i)=1)}\cdot \left\{c'(A^{[t]})^+\beta + \sum_{g'} [c'(I-(A^{[t]})^+ A^{[t]})]_{g'} \cdot w_{g'}: w \in \mathbbm{R}^{|\mathcal{G}|} \right\}
\end{align*}
The RHS may again be an outer set for $\Theta(\mathcal{P}_{obs},M)$, for example when $Y_i$ has bounded support. An added complication now, as compared to unconditional means, is that the probability $P(c(G_i)=1)$ is no longer known to be equal to one, and our only identifying information for it is that $\sum_{g \in \mathcal{G}} A^{[t]}_{gz} = d_z$ for all $z \in \mathcal{Z}$, where $d_g:=P(T_i=t|Z_i=z)$. 


		\section{Algorithms to enumerate outcome-nonrestrictive identification results} \label{app:algoone}
\begin{tcolorbox}[title=Algorithm 1:, colback=white, breakable=true]
	Begin with a given instrument support $\mathcal{Z}$ and set of treatments $\mathcal{T}$, and $t' \ne t$ in $\mathcal{T}$:
	\begin{enumerate}
		\item Loop over all possible choice models $\mathcal{G}$ given $\mathcal{Z}$ and $\mathcal{T}$. There are $2^{|\mathcal{G}^m|=2^{{|\mathcal{T}|}^{|\mathcal{Z}|}}}$ of these, where we let $\mathcal{G}^m$ denote the set of all $|\mathcal{T}|^{|\mathcal{Z}|}$ conceivable response types (mappings from $\mathcal{Z}$ to $\mathcal{T}$)
		\item Given the results of Section \ref{sec:geomTEs}, find a basis for the left null-space $ns(A^{[t',t]})$ of $A^{[t',t]}:=\begin{bmatrix}
			A^{[t']}\\
			A^{[t]}
		\end{bmatrix}$ via a QR decomposition of $A^{[t',t]}$. Represent this basis by a $k \times 2|\mathcal{Z}|$ matrix $N^{[t',t]}$, where $k$ is the dimension of $ns(A^{[t',t]})$. For any vector $\alpha \in ns(A^{[t',t]'})$, let $\alpha_1(\alpha) = [I_{|\mathcal{Z}|},\mathbf{0}_{|\mathcal{Z}| \times |\mathcal{Z}|}]\alpha$ be its first $|\mathcal{Z}|$ components, and let $\mathcal{C}^{[t,t']} = \{A^{[t']'}\alpha_1(\alpha): \alpha \in ns(A^{[t',t]'})\}$ be the subspace of $\mathbbm{R}^{|\mathcal{G}|}$ corresponding to these $\alpha$. $\mathcal{C}^{[t,t']}$ is a k-dimensional vector space with a basis represented by the $k \times |\mathcal{G}|$ matrix $B^{[t',t]}:=A^{[t']'}N^{[t',t]}[I_{|\mathcal{Z}|},\mathbf{0}_{|\mathcal{Z}| \times |\mathcal{Z}|}]$. 
		\item If $k \ge 1$, we now determine the intersection of $\mathcal{C}^{[t,t']}$ with the unit cube. This is done by looping over the $2^{|\mathcal{G}|}-1$ non-zero vectors $c$ in $\{0,1\}^{|\mathcal{G}|}$, and checking whether $c \in \mathcal{C}^{[t,t']}$ (when $B$ has full row rank, this can be done e.g. by checking that ${B^{[t',t]}}^+B^{[t',t]}c=c$, where $B^+$ is the Moore-Penrose pseudo-inverse of $B$).
	\end{enumerate}
\end{tcolorbox}
\noindent Note that since the computational problem as a whole is symmetric with respect to permutations of the (arbitrary) treatment labels, we can focus on binary collections containing the two treatment values $t'=1$ and $t=0$, and then generate new binary collections by then applying all re-labelings to the treatment values.

\begin{table}[h!]
	\begin{center}
		\begin{tabular}{cc||cccc}
			$|\mathcal{T}|$ & $|\mathcal{Z}|$ & $|\mathcal{C}_{|\mathcal{Z}|}|$ & \# $\alpha$'s (i.e. $(|\mathcal{C}_{|\mathcal{Z}|}|)^{2|\mathcal{Z}|}$) & $|\mathcal{G}^m|=|\mathcal{T}|^{|\mathcal{Z}|}$ & \# selection models (i.e. $2^{|\mathcal{G}^m|}$) \\ \hline\hline
			2 & 2 & 3           & 81       & 4           & 16                  \\
			3 & 2 & 3           & 81       & 8           & 256                 \\
			2 & 3 & 7           & 117,649   & 9           & 512                 \\
			3 & 3 & 7           & 117,649   & 27          & 1.34$\cdot 10^8$            \\
			4 & 3 & 7           & 117,649   & 81          & 2.42$\cdot 10^{24}$     \\
			4 & 4 & 16          & 4.29 $\cdot 10^9$ & 256         & 1.16$\cdot 10^{77}$           
		\end{tabular}
	\end{center}
	\caption{Comparison of the computational complexity of Algorithms 1 and 2 \label{table:complexity}}
\end{table}
Table \ref{table:complexity} compares the complexity of Algorithms 1 and 2. It does not account for the full computational cost of running each algorithm (e.g. computations within each choice of $\alpha$ in the case of Algorithm 2, or within a selection model in the case of Algorithm 1), but it is nevertheless clear that Algorithm 1 quickly becomes infeasible, while there remains hope for Algorithm 2 with $|\mathcal{Z}|=|\mathcal{T}|= 4$.\\

\begin{tcolorbox}[title=Algorithm 2:, colback=white, breakable=true]
	Begin with a given instrument support $\mathcal{Z}$ and set of treatments $\mathcal{T}$.\\
	
	\noindent \textbf{Part One: generate binary collections by }$\alpha$
	\begin{enumerate}
		\item Loop over all vectors $2\cdot |\mathcal{Z}|$-component vectors $\alpha$ having components in the set $\mathcal{C}_{|\mathcal{Z}|}$ (there are $(|\mathcal{C}_{|\mathcal{Z}|}|)^{2|\mathcal{Z}|}$ of these)
		\item With $t'=1$ and $t=0$ fixed (as with Algorithm 1), construct the matrix $A^{[t',t]}:=\begin{bmatrix}
			A^{[t']}\\
			A^{[t]}
		\end{bmatrix}$ where now each of $A^{[t']}$ and $A^{[t]}$ representing the full set of conceivable response types $\mathcal{G}^m$ (having $|\mathcal{G}^m|=|\mathcal{T}|^{|\mathcal{Z}|}$ columns). Compute for each $\alpha$ the row vector $\alpha'A^{[t',t]}$.
		\item Consider the columns $g$ of $\alpha'A^{[t',t]}$ that take the value of $0$, and call this set $\mathcal{G}^0(\alpha)$. Note that $\mathcal{G}^0(\alpha)$ is the set of $g$ for which $[\alpha_1(\alpha)'A^{[t']}]_g = [\alpha_0(\alpha)'A^{[t]}]_g$ (using the notation introduced in Algorithm 1).
		\item Now find the set $\mathcal{G}(\alpha) \subseteq \mathcal{G}^0(\alpha)$ such that $[\alpha_1(\alpha)'A^{[t']}]_g \in \{0,1\}$ for all $g \in \mathcal{G}(\alpha)$. Only response types $g$ in the set $\mathcal{G}^0(\alpha)$ can exist in a binary collection having $\alpha^{[t]} = \alpha_0(\alpha)$ and $\alpha^{[t']} = \alpha_1(\alpha)$. Further, the set $\mathcal{G}(\alpha)$ is \textit{maximal} (given $\alpha$) in the sense that we get a binary collection from $\alpha$ for $t,t'$ for any selection model $\mathcal{G} \subseteq \mathcal{G}(\alpha)$.
		\item Some of the binary collections (indexed by $\alpha$) constructed in this way will be redundant in the following sense. Define $c(\alpha)=\alpha_1(\alpha)'A^{[t']}$, and let vectors $\alpha$ and $\beta$ be two $2|\mathcal{Z}|$-component vectors such that $c(\alpha) = c(\beta)$ but $\mathcal{G}(\alpha) \subset \mathcal{G}(\beta)$. Then $\beta$ delivers the same largest complier group as $\alpha$ but while allowing for a strictly larger selection model. In this case remove $\alpha$, since the identification result for $\beta$ nests that of $\alpha$. If $c(\alpha)=c(\beta)$ as above and $\alpha$ and $\beta$ deliver \textit{the same} maximal selection model, i.e. $\mathcal{G}(\alpha) = \mathcal{G}(\beta)$, then drop whichever vector has more non-zero elements than the other, i.e. drop $\alpha$ if $||\alpha||_0 > ||\beta||_0$ where $||\cdot||_0$ indicates the $\ell_0$ norm. If $||\alpha||_0 = ||\beta||_0$, then keep whichever vector has a smaller $\ell_2$ norm is kept (this choice is arbitrary).
	\end{enumerate}
	
	\noindent \textbf{Part Two: organize by selection model and pare redundancies}
	\begin{enumerate}
		\item Extend the binary collections obtained in Part One of the algorithm for $(t',t) = (1,0)$ to all other choices of $t'>t$. Binary collections can now be indexed by the tuple $(t',t,\alpha)$. Any $(1,0,\alpha)$ obtained in Part One above yields a binary collection for $(t',t,\alpha)$ with the same vector $\alpha$, with the response types suitably re-defined based on relabeling the treatment values.
		\item Now collect all binary combinations that share a maximal selection model $\mathcal{G}$, which based on the last step may allow treatment effects that contemplate differing treatment contrasts (e.g. treatment value 2 vs. 0 or treatment 1 vs. 0) to be associated with the same selection model. 
		\item We now have a list of selection models $\mathcal{G}$ that admit of at least one binary collection, and for each such $\mathcal{G}$ a list of these binary collections. Recall that each selection model can be expressed by the matrix $A$. To distill out selection models with a unique structure, eliminate any redundancies where one selection model can be transformed into another by re-labeling treatment values, or by permuting the labels of the instrument values and re-ordering the columns of $A$.
	\end{enumerate}
\end{tcolorbox}
		\section{Catalog of outcome-nonrestrictive identification results} \label{sec:catalog}

The following results are for various small values of $|\mathcal{Z}|$ $|\mathcal{T}|$. Results for $|\mathcal{Z}|=|\mathcal{T}|=3$ are available upon request from the author (these add roughly 40 pages of output).

For a given $\mathcal{T}$ and $\mathcal{Z}$, binary collections are organized by selection models, given unique identifiers of the format \texttt{SM.}$|\mathcal{T}|$.$|\mathcal{Z}|$\texttt{.s}, where \texttt{s} is an index of the various selection models in that setting. Within each selection model, binary collections are enumerated by ascending numbers $i), ii)$ etc. Each binary collection is presented via the coefficient vectors $\alpha_{t'}$ and $\alpha_{t}$ (following the notation of Sec. \ref{sec:geomTEs} but keeping $t$ and $t'$ explicit).

Binary collections that share a common maximal selection model $\mathcal{G}(\alpha)$ organized under that selection model, and are not re-listed for $\mathcal{G} \subseteq \mathcal{G}(\alpha)$. Further, some binary collections for $\mathcal{G}$ might be listed under a $\mathcal{G}(\alpha)$ that nests $\mathcal{G}$ only after suitable re-labeling of the treatments and instruments. It is for this reason that the set of binary collections listed under a given selection model may not be closed under addition, even when adding the $c$ for two such collections results in another vector composed of all zeroes and ones.

For example, consider the $A$ matrices for SM.2.3.8 and SM2.3.1 below:
$$ \texttt{SM.2.3.8}: \begin{bmatrix}
	1 & 0 \\
	1 & 0 \\
	0 & 1 \\
\end{bmatrix} \quad \quad \quad \quad \texttt{SM.2.3.1.swapped}: \begin{bmatrix}
0 & 0 & 1\\
0 & 1 & 1\\
1 & 0 & 0\\
\end{bmatrix} $$
where by \texttt{SM.2.3.1.swapped} I indicate that I have swapped the first and third rows of the A matrix listed in the catalog that follows for SM.2.3.1. This swapping corresponds to a re-labelling of the instrument values.

SM.2.3.8 consists of two selection types, and the catalog shows that it admits of binary collection i) with $\alpha_1 = (0.5,0.5,0)'$ and $\alpha_0 = (0,0,1)'$ yielding $c=(1,0)'$ as well as binary collection ii) with $\alpha_1 = (0,0,1)'$ and $\alpha_0 = (0.5,0.5,0)'$ yielding $c=(1,0)'$. This implies that SM.2.3.8 also admits of a binary collection yielding $c=(1,1)'$, with $\alpha_0=\alpha_1 = (0.5,0.5,1)'$. 

The reason that this third binary collection is not listed under SM.2.3.8 is that SM.2.3.8 is not maximal for it: unlike collections i) and ii) which just include one of the two types in SM.2.3.8, identification of the average treatment effect for both of the types in SM.2.3.8 holds in the less restrictive selection model SM.2.3.1.swapped, which contains the selection types of SM.2.3.8 in its first and third columns. The sole binary collection listed under SM.2.3.1 corresponds to $c=(1,1)'$ in SM.2.3.8.

\subsection{2 treatments, 2 instrument values}
\subsubsection*{SM.2.2.1}
$$ A= \begin{bmatrix}
0 & 0 & 1\\
0 & 1 & 1\\
\end{bmatrix} $$
\begin{enumerate}[i)]
\item $(t',t)=(1,0)$; $\alpha_{t'}=(-1, 1)'; \alpha_{t}=(1, -1)'$; $c=(0, 1, 0)'$
\end{enumerate}
\subsubsection*{SM.2.2.2}
$$ A= \begin{bmatrix}
1 & 0\\
0 & 1\\
\end{bmatrix} $$
\begin{enumerate}[i)]
\item $(t',t)=(1,0)$; $\alpha_{t'}=(0, 1)'; \alpha_{t}=(1, 0)'$; $c=(0, 1)'$
\item $(t',t)=(1,0)$; $\alpha_{t'}=(1, 0)'; \alpha_{t}=(0, 1)'$; $c=(1, 0)'$
\item $(t',t)=(1,0)$; $\alpha_{t'}=(1, 1)'; \alpha_{t}=(1, 1)'$; $c=(1, 1)'$
\end{enumerate}

\subsection{3 treatments, 2 instrument values}
\subsubsection*{SM.3.2.1}
$$ A= \begin{bmatrix}
0 & 0 & 1 & 2\\
0 & 1 & 2 & 2\\
\end{bmatrix} $$
\begin{enumerate}[i)]
\item $(t',t)=(1,0)$; $\alpha_{t'}=(0, 1)'; \alpha_{t}=(1, -1)'$; $c=(0, 1, 0, 0)'$
\end{enumerate}
\subsubsection*{SM.3.2.2}
$$ A= \begin{bmatrix}
0 & 0 & 1 & 2\\
0 & 1 & 1 & 2\\
\end{bmatrix} $$
\begin{enumerate}[i)]
\item $(t',t)=(1,0)$; $\alpha_{t'}=(-1, 1)'; \alpha_{t}=(1, -1)'$; $c=(0, 1, 0, 0)'$
\end{enumerate}
\subsubsection*{SM.3.2.3}
$$ A= \begin{bmatrix}
1 & 2 & 0 & 1 & 2\\
0 & 0 & 1 & 2 & 2\\
\end{bmatrix} $$
\begin{enumerate}[i)]
\item $(t',t)=(1,0)$; $\alpha_{t'}=(0, 1)'; \alpha_{t}=(1, 0)'$; $c=(0, 0, 1, 0, 0)'$
\end{enumerate}
\subsubsection*{SM.3.2.4}
$$ A= \begin{bmatrix}
2 & 0 & 1 & 2\\
0 & 1 & 1 & 2\\
\end{bmatrix} $$
\begin{enumerate}[i)]
\item $(t',t)=(1,0)$; $\alpha_{t'}=(-1, 1)'; \alpha_{t}=(1, 0)'$; $c=(0, 1, 0, 0)'$
\end{enumerate}
\subsubsection*{SM.3.2.5}
$$ A= \begin{bmatrix}
1 & 0 & 2\\
0 & 1 & 2\\
\end{bmatrix} $$
\begin{enumerate}[i)]
\item $(t',t)=(1,0)$; $\alpha_{t'}=(1, 1)'; \alpha_{t}=(1, 1)'$; $c=(1, 1, 0)'$
\end{enumerate}

\subsection{2 treatments, 3 instrument values}
\subsubsection*{SM.2.3.1}
$$ A= \begin{bmatrix}
0 & 0 & 1\\
1 & 0 & 0\\
0 & 1 & 1\\
\end{bmatrix} $$
\begin{enumerate}[i)]
\item $(t',t)=(1,0)$; $\alpha_{t'}=(1, 1, 0)'; \alpha_{t}=(-1, 1, 2)'$; $c=(1, 0, 1)'$
\end{enumerate}
\subsubsection*{SM.2.3.2}
$$ A= \begin{bmatrix}
0 & 1 & 0 & 1\\
0 & 0 & 1 & 1\\
0 & 0 & 0 & 1\\
\end{bmatrix} $$
\begin{enumerate}[i)]
\item $(t',t)=(1,0)$; $\alpha_{t'}=(1, 1, -2)'; \alpha_{t}=(-1, -1, 2)'$; $c=(0, 1, 1, 0)'$
\end{enumerate}
\subsubsection*{SM.2.3.3}
$$ A= \begin{bmatrix}
1 & 1 & 0\\
1 & 0 & 1\\
0 & 1 & 1\\
\end{bmatrix} $$
\begin{enumerate}[i)]
\item $(t',t)=(1,0)$; $\alpha_{t'}=(0.5, 0.5, 0.5)'; \alpha_{t}=(1, 1, 1)'$; $c=(1, 1, 1)'$
\item $(t',t)=(1,0)$; $\alpha_{t'}=(1, 0, 0)'; \alpha_{t}=(0, 1, 1)'$; $c=(1, 1, 0)'$
\item $(t',t)=(1,0)$; $\alpha_{t'}=(0, 1, 0)'; \alpha_{t}=(1, 0, 1)'$; $c=(1, 0, 1)'$
\item $(t',t)=(1,0)$; $\alpha_{t'}=(0.5, 0.5, -0.5)'; \alpha_{t}=(0, 0, 1)'$; $c=(1, 0, 0)'$
\item $(t',t)=(1,0)$; $\alpha_{t'}=(0, 0, 1)'; \alpha_{t}=(1, 1, 0)'$; $c=(0, 1, 1)'$
\item $(t',t)=(1,0)$; $\alpha_{t'}=(0.5, -0.5, 0.5)'; \alpha_{t}=(0, 1, 0)'$; $c=(0, 1, 0)'$
\item $(t',t)=(1,0)$; $\alpha_{t'}=(-0.5, 0.5, 0.5)'; \alpha_{t}=(1, 0, 0)'$; $c=(0, 0, 1)'$
\end{enumerate}
\subsubsection*{SM.2.3.4}
$$ A= \begin{bmatrix}
0 & 1 & 0\\
1 & 0 & 1\\
0 & 1 & 1\\
\end{bmatrix} $$
\begin{enumerate}[i)]
\item $(t',t)=(1,0)$; $\alpha_{t'}=(2, 1, -1)'; \alpha_{t}=(0, 1, 1)'$; $c=(1, 1, 0)'$
\end{enumerate}
\subsubsection*{SM.2.3.5}
$$ A= \begin{bmatrix}
0 & 1 & 0 & 1\\
1 & 1 & 0 & 0\\
0 & 0 & 1 & 1\\
\end{bmatrix} $$
\begin{enumerate}[i)]
\item $(t',t)=(1,0)$; $\alpha_{t'}=(0, 1, 1)'; \alpha_{t}=(0, 1, 1)'$; $c=(1, 1, 1, 1)'$
\item $(t',t)=(1,0)$; $\alpha_{t'}=(1, 0, 0)'; \alpha_{t}=(-1, 1, 1)'$; $c=(0, 1, 0, 1)'$
\item $(t',t)=(1,0)$; $\alpha_{t'}=(0, 1, 0)'; \alpha_{t}=(0, 0, 1)'$; $c=(1, 1, 0, 0)'$
\item $(t',t)=(1,0)$; $\alpha_{t'}=(0, 0, 1)'; \alpha_{t}=(0, 1, 0)'$; $c=(0, 0, 1, 1)'$
\item $(t',t)=(1,0)$; $\alpha_{t'}=(-1, 1, 1)'; \alpha_{t}=(1, 0, 0)'$; $c=(1, 0, 1, 0)'$
\end{enumerate}
\subsubsection*{SM.2.3.6}
$$ A= \begin{bmatrix}
0 & 1 & 1 & 1\\
0 & 1 & 0 & 1\\
0 & 0 & 1 & 1\\
\end{bmatrix} $$
\begin{enumerate}[i)]
\item $(t',t)=(1,0)$; $\alpha_{t'}=(2, -1, -1)'; \alpha_{t}=(-2, 1, 1)'$; $c=(0, 1, 1, 0)'$
\end{enumerate}
\subsubsection*{SM.2.3.7}
$$ A= \begin{bmatrix}
1 & 0 & 1\\
0 & 0 & 1\\
0 & 1 & 1\\
\end{bmatrix} $$
\begin{enumerate}[i)]
\item $(t',t)=(1,0)$; $\alpha_{t'}=(1, -1, 0)'; \alpha_{t}=(0, 0, 1)'$; $c=(1, 0, 0)'$
\item $(t',t)=(1,0)$; $\alpha_{t'}=(0, -1, 1)'; \alpha_{t}=(1, 0, 0)'$; $c=(0, 1, 0)'$
\end{enumerate}
\subsubsection*{SM.2.3.8}
$$ A= \begin{bmatrix}
1 & 0\\
1 & 0\\
0 & 1\\
\end{bmatrix} $$
\begin{enumerate}[i)]
\item $(t',t)=(1,0)$; $\alpha_{t'}=(0.5, 0.5, 0)'; \alpha_{t}=(0, 0, 1)'$; $c=(1, 0)'$
\item $(t',t)=(1,0)$; $\alpha_{t'}=(0, 0, 1)'; \alpha_{t}=(0.5, 0.5, 0)'$; $c=(0, 1)'$
\end{enumerate}
\subsubsection*{SM.2.3.9}
$$ A= \begin{bmatrix}
1 & 0 & 0\\
0 & 1 & 0\\
0 & 0 & 1\\
\end{bmatrix} $$
\begin{enumerate}[i)]
\item $(t',t)=(1,0)$; $\alpha_{t'}=(1, 1, 0)'; \alpha_{t}=(0, 0, 1)'$; $c=(1, 1, 0)'$
\item $(t',t)=(1,0)$; $\alpha_{t'}=(1, 1, 1)'; \alpha_{t}=(0.5, 0.5, 0.5)'$; $c=(1, 1, 1)'$
\item $(t',t)=(1,0)$; $\alpha_{t'}=(1, 0, 0)'; \alpha_{t}=(-0.5, 0.5, 0.5)'$; $c=(1, 0, 0)'$
\item $(t',t)=(1,0)$; $\alpha_{t'}=(0, 1, 0)'; \alpha_{t}=(0.5, -0.5, 0.5)'$; $c=(0, 1, 0)'$
\item $(t',t)=(1,0)$; $\alpha_{t'}=(1, 0, 1)'; \alpha_{t}=(0, 1, 0)'$; $c=(1, 0, 1)'$
\item $(t',t)=(1,0)$; $\alpha_{t'}=(0, 1, 1)'; \alpha_{t}=(1, 0, 0)'$; $c=(0, 1, 1)'$
\item $(t',t)=(1,0)$; $\alpha_{t'}=(0, 0, 1)'; \alpha_{t}=(0.5, 0.5, -0.5)'$; $c=(0, 0, 1)'$
\end{enumerate}
\subsubsection*{SM.2.3.10}
$$ A= 
 $$
\begin{enumerate}[i)]
\item $(t',t)=(1,0)$; $\alpha_{t'}=(1, 0, -1)'; \alpha_{t}=(-1, 0, 1)'$; $c=(0, 1, 0, 1, 0, 0)'$
\end{enumerate}
\subsubsection*{SM.2.3.11}
$$ A= %
 $$
\begin{enumerate}[i)]
\item $(t',t)=(1,0)$; $\alpha_{t'}=(0, 1, 0)'; \alpha_{t}=(-1, 0, 1)'$; $c=(0, 1, 0)'$
\item $(t',t)=(1,0)$; $\alpha_{t'}=(0, 0, 1)'; \alpha_{t}=(-1, 1, 0)'$; $c=(0, 0, 1)'$
\end{enumerate}

		\subsection{3 treatments, 3 instrument values}
\subsubsection*{SM.3.3.1}
$$ A= %
 $$
\begin{enumerate}[i)]
\item $(t',t)=(1,0)$; $\alpha_{t'}=(1, 1, -2)'; \alpha_{t}=(-1, 1, 2)'$; $c=(1, 0, 0, 1, 0)'$
\end{enumerate}
\subsubsection*{SM.3.3.2}
$$ A= %
 $$
\begin{enumerate}[i)]
\item $(t',t)=(1,0)$; $\alpha_{t'}=(1, 1, -1)'; \alpha_{t}=(-1, 1, 2)'$; $c=(1, 0, 0, 0, 1, 0)'$
\end{enumerate}
\subsubsection*{SM.3.3.3}
$$ A= %
 $$
\begin{enumerate}[i)]
\item $(t',t)=(1,0)$; $\alpha_{t'}=(2, 1, -1)'; \alpha_{t}=(-1, 1, 2)'$; $c=(1, 1, 0, 0, 0)'$
\end{enumerate}
\subsubsection*{SM.3.3.4}
$$ A= %
 $$
\begin{enumerate}[i)]
\item $(t',t)=(1,0)$; $\alpha_{t'}=(1, 1, 0)'; \alpha_{t}=(-1, 1, 2)'$; $c=(1, 0, 1, 0, 0, 1, 0)'$
\end{enumerate}
\subsubsection*{SM.3.3.5}
$$ A= %
 $$
\begin{enumerate}[i)]
\item $(t',t)=(1,0)$; $\alpha_{t'}=(-2, 1, 1)'; \alpha_{t}=(-1, 1, 2)'$; $c=(1, 1, 0, 0, 0)'$
\end{enumerate}
\subsubsection*{SM.3.3.6}
$$ A= %
 $$
\begin{enumerate}[i)]
\item $(t',t)=(1,0)$; $\alpha_{t'}=(-1, 1, 1)'; \alpha_{t}=(-1, 1, 2)'$; $c=(1, 1, 0, 0, 0, 0)'$
\end{enumerate}
\subsubsection*{SM.3.3.7}
$$ A= %
 $$
\begin{enumerate}[i)]
\item $(t',t)=(1,0)$; $\alpha_{t'}=(0, 1, 1)'; \alpha_{t}=(-1, 1, 2)'$; $c=(1, 1, 1, 0, 0, 0)'$
\end{enumerate}
\subsubsection*{SM.3.3.8}
$$ A= %
 $$
\begin{enumerate}[i)]
\item $(t',t)=(1,0)$; $\alpha_{t'}=(1, 1, 1)'; \alpha_{t}=(-1, 1, 2)'$; $c=(1, 1, 0, 1, 0)'$
\end{enumerate}
\subsubsection*{SM.3.3.9}
$$ A= %
 $$
\begin{enumerate}[i)]
\item $(t',t)=(1,0)$; $\alpha_{t'}=(-1, 1, 2)'; \alpha_{t}=(-1, 1, 2)'$; $c=(1, 1, 0, 0, 0)'$
\end{enumerate}
\subsubsection*{SM.3.3.10}
$$ A= %
 $$
\begin{enumerate}[i)]
\item $(t',t)=(1,0)$; $\alpha_{t'}=(1, 1, -2)'; \alpha_{t}=(-1, -1, 2)'$; $c=(0, 1, 1, 0, 0)'$
\end{enumerate}
\subsubsection*{SM.3.3.11}
$$ A= %
 $$
\begin{enumerate}[i)]
\item $(t',t)=(1,0)$; $\alpha_{t'}=(1, 1, -1)'; \alpha_{t}=(-1, -1, 2)'$; $c=(0, 1, 1, 0, 0, 0)'$
\end{enumerate}
\subsubsection*{SM.3.3.12}
$$ A= %
 $$
\begin{enumerate}[i)]
\item $(t',t)=(1,0)$; $\alpha_{t'}=(1, 1, 0)'; \alpha_{t}=(-1, -1, 2)'$; $c=(0, 1, 1, 0, 0)'$
\end{enumerate}
\subsubsection*{SM.3.3.13}
$$ A= %
 $$
\begin{enumerate}[i)]
\item $(t',t)=(1,0)$; $\alpha_{t'}=(1, 1, -2)'; \alpha_{t}=(1, 1, 1)'$; $c=(1, 1, 0, 1, 1, 0)'$
\end{enumerate}
\subsubsection*{SM.3.3.14}
$$ A= %
 $$
\begin{enumerate}[i)]
\item $(t',t)=(1,0)$; $\alpha_{t'}=(1, 1, -1)'; \alpha_{t}=(1, 1, 1)'$; $c=(1, 1, 0, 0, 1, 1, 0)'$
\end{enumerate}
\subsubsection*{SM.3.3.15}
$$ A= %
 $$
\begin{enumerate}[i)]
\item $(t',t)=(1,0)$; $\alpha_{t'}=(2, 1, -1)'; \alpha_{t}=(1, 1, 1)'$; $c=(1, 1, 0, 1, 0)'$
\end{enumerate}
\subsubsection*{SM.3.3.16}
$$ A= %
 $$
\begin{enumerate}[i)]
\item $(t',t)=(1,0)$; $\alpha_{t'}=(1, 1, 0)'; \alpha_{t}=(1, 1, 1)'$; $c=(1, 1, 1, 1, 0, 1, 1, 0)'$
\end{enumerate}
\subsubsection*{SM.3.3.17}
$$ A= %
 $$
\begin{enumerate}[i)]
\item $(t',t)=(1,0)$; $\alpha_{t'}=(0.5, 0.5, 0.5)'; \alpha_{t}=(1, 1, 1)'$; $c=(1, 1, 1, 0)'$
\end{enumerate}
\subsubsection*{SM.3.3.18}
$$ A= %
 $$
\begin{enumerate}[i)]
\item $(t',t)=(1,0)$; $\alpha_{t'}=(1, 1, 1)'; \alpha_{t}=(1, 1, 1)'$; $c=(1, 1, 1, 1, 1, 1, 0)'$
\end{enumerate}
\subsubsection*{SM.3.3.19}
$$ A= %
 $$
\begin{enumerate}[i)]
\item $(t',t)=(1,0)$; $\alpha_{t'}=(1, 1, -2)'; \alpha_{t}=(0, 1, 1)'$; $c=(1, 1, 1, 0, 1, 0, 0)'$
\end{enumerate}
\subsubsection*{SM.3.3.20}
$$ A= %
 $$
\begin{enumerate}[i)]
\item $(t',t)=(1,0)$; $\alpha_{t'}=(1, -1, -1)'; \alpha_{t}=(0, 1, 1)'$; $c=(1, 0, 1, 0, 0, 0)'$
\end{enumerate}
\subsubsection*{SM.3.3.21}
$$ A= %
 $$
\begin{enumerate}[i)]
\item $(t',t)=(1,0)$; $\alpha_{t'}=(2, -1, -1)'; \alpha_{t}=(0, 1, 1)'$; $c=(1, 1, 0, 0, 0)'$
\end{enumerate}
\subsubsection*{SM.3.3.22}
$$ A= %
 $$
\begin{enumerate}[i)]
\item $(t',t)=(1,0)$; $\alpha_{t'}=(1, 0, -1)'; \alpha_{t}=(0, 1, 1)'$; $c=(1, 1, 0, 0, 1, 0, 0, 0, 0)'$
\end{enumerate}
\subsubsection*{SM.3.3.23}
$$ A= %
 $$
\begin{enumerate}[i)]
\item $(t',t)=(1,0)$; $\alpha_{t'}=(1, 1, -1)'; \alpha_{t}=(0, 1, 1)'$; $c=(1, 1, 1, 0, 0, 0, 1, 0, 0)'$
\end{enumerate}
\subsubsection*{SM.3.3.24}
$$ A= %
 $$
\begin{enumerate}[i)]
\item $(t',t)=(1,0)$; $\alpha_{t'}=(2, 1, -1)'; \alpha_{t}=(0, 1, 1)'$; $c=(1, 1, 1, 0, 0, 0, 0)'$
\end{enumerate}
\subsubsection*{SM.3.3.25}
$$ A= %
 $$
\begin{enumerate}[i)]
\item $(t',t)=(1,0)$; $\alpha_{t'}=(1, -0.5, -0.5)'; \alpha_{t}=(0, 1, 1)'$; $c=(1, 0, 1, 0, 0)'$
\end{enumerate}
\subsubsection*{SM.3.3.26}
$$ A= %
 $$
\begin{enumerate}[i)]
\item $(t',t)=(1,0)$; $\alpha_{t'}=(1, -1, 0)'; \alpha_{t}=(0, 1, 1)'$; $c=(1, 1, 0, 0, 0, 1, 0, 0, 0)'$
\end{enumerate}
\subsubsection*{SM.3.3.27}
$$ A= %
 $$
\begin{enumerate}[i)]
\item $(t',t)=(1,0)$; $\alpha_{t'}=(1, 0, 0)'; \alpha_{t}=(0, 1, 1)'$; $c=(1, 1, 1, 0, 0, 0, 0, 1, 0, 0, 0, 0)'$
\end{enumerate}
\subsubsection*{SM.3.3.28}
$$ A= %
 $$
\begin{enumerate}[i)]
\item $(t',t)=(1,0)$; $\alpha_{t'}=(1, 1, 0)'; \alpha_{t}=(0, 1, 1)'$; $c=(1, 1, 1, 1, 0, 0, 1, 0, 0)'$
\end{enumerate}
\subsubsection*{SM.3.3.29}
$$ A= %
 $$
\begin{enumerate}[i)]
\item $(t',t)=(1,0)$; $\alpha_{t'}=(1, -1, 1)'; \alpha_{t}=(0, 1, 1)'$; $c=(1, 1, 1, 0, 0, 1, 0, 0, 0)'$
\end{enumerate}
\subsubsection*{SM.3.3.30}
$$ A= %
 $$
\begin{enumerate}[i)]
\item $(t',t)=(1,0)$; $\alpha_{t'}=(1, 0, 1)'; \alpha_{t}=(0, 1, 1)'$; $c=(1, 1, 1, 1, 1, 0, 0, 0, 0)'$
\end{enumerate}
\subsubsection*{SM.3.3.31}
$$ A= %
 $$
\begin{enumerate}[i)]
\item $(t',t)=(1,0)$; $\alpha_{t'}=(-2, 1, 1)'; \alpha_{t}=(0, 1, 1)'$; $c=(1, 1, 1, 1, 0, 0, 0)'$
\end{enumerate}
\subsubsection*{SM.3.3.32}
$$ A= %
 $$
\begin{enumerate}[i)]
\item $(t',t)=(1,0)$; $\alpha_{t'}=(-1, 1, 1)'; \alpha_{t}=(0, 1, 1)'$; $c=(1, 1, 1, 1, 0, 0, 0, 0)'$
\end{enumerate}
\subsubsection*{SM.3.3.33}
$$ A= %
 $$
\begin{enumerate}[i)]
\item $(t',t)=(1,0)$; $\alpha_{t'}=(0, 1, 1)'; \alpha_{t}=(0, 1, 1)'$; $c=(1, 1, 1, 1, 1, 1, 0, 0, 0)'$
\end{enumerate}
\subsubsection*{SM.3.3.34}
$$ A= %
 $$
\begin{enumerate}[i)]
\item $(t',t)=(1,0)$; $\alpha_{t'}=(1, 1, 1)'; \alpha_{t}=(0, 1, 1)'$; $c=(1, 1, 1, 1, 1, 1, 0, 0)'$
\end{enumerate}
\subsubsection*{SM.3.3.35}
$$ A= %
 $$
\begin{enumerate}[i)]
\item $(t',t)=(1,0)$; $\alpha_{t'}=(-1, 2, 1)'; \alpha_{t}=(0, 1, 1)'$; $c=(1, 1, 1, 0, 0, 0)'$
\end{enumerate}
\subsubsection*{SM.3.3.36}
$$ A= %
 $$
\begin{enumerate}[i)]
\item $(t',t)=(1,0)$; $\alpha_{t'}=(1, 1, -2)'; \alpha_{t}=(-1, 1, 1)'$; $c=(1, 0, 1, 0, 0, 1, 0)'$
\end{enumerate}
\subsubsection*{SM.3.3.37}
$$ A= %
 $$
\begin{enumerate}[i)]
\item $(t',t)=(1,0)$; $\alpha_{t'}=(1, -1, -1)'; \alpha_{t}=(-1, 1, 1)'$; $c=(0, 1, 0, 0, 1, 0, 0)'$
\end{enumerate}
\subsubsection*{SM.3.3.38}
$$ A= %
 $$
\begin{enumerate}[i)]
\item $(t',t)=(1,0)$; $\alpha_{t'}=(2, -1, -1)'; \alpha_{t}=(-1, 1, 1)'$; $c=(1, 0, 1, 0, 0, 0)'$
\end{enumerate}
\subsubsection*{SM.3.3.39}
$$ A= %
 $$
\begin{enumerate}[i)]
\item $(t',t)=(1,0)$; $\alpha_{t'}=(1, 0, -1)'; \alpha_{t}=(-1, 1, 1)'$; $c=(0, 1, 0, 1, 0, 0, 0, 1, 0, 0)'$
\end{enumerate}
\subsubsection*{SM.3.3.40}
$$ A= %
 $$
\begin{enumerate}[i)]
\item $(t',t)=(1,0)$; $\alpha_{t'}=(0.5, 0.5, -1)'; \alpha_{t}=(-1, 1, 1)'$; $c=(1, 0, 0, 0, 0)'$
\end{enumerate}
\subsubsection*{SM.3.3.41}
$$ A= %
 $$
\begin{enumerate}[i)]
\item $(t',t)=(1,0)$; $\alpha_{t'}=(-1, 1, -1)'; \alpha_{t}=(-1, 1, 1)'$; $c=(1, 0, 0, 0, 0, 0)'$
\end{enumerate}
\subsubsection*{SM.3.3.42}
$$ A= %
 $$
\begin{enumerate}[i)]
\item $(t',t)=(1,0)$; $\alpha_{t'}=(0, 1, -1)'; \alpha_{t}=(-1, 1, 1)'$; $c=(1, 1, 0, 0, 0, 0, 0, 0)'$
\end{enumerate}
\subsubsection*{SM.3.3.43}
$$ A= %
 $$
\begin{enumerate}[i)]
\item $(t',t)=(1,0)$; $\alpha_{t'}=(1, 1, -1)'; \alpha_{t}=(-1, 1, 1)'$; $c=(1, 0, 1, 0, 0, 0, 1, 0)'$
\end{enumerate}
\subsubsection*{SM.3.3.44}
$$ A= %
 $$
\begin{enumerate}[i)]
\item $(t',t)=(1,0)$; $\alpha_{t'}=(2, 1, -1)'; \alpha_{t}=(-1, 1, 1)'$; $c=(1, 0, 1, 0, 0, 0)'$
\end{enumerate}
\subsubsection*{SM.3.3.45}
$$ A= %
 $$
\begin{enumerate}[i)]
\item $(t',t)=(1,0)$; $\alpha_{t'}=(1, -0.5, -0.5)'; \alpha_{t}=(-1, 1, 1)'$; $c=(0, 1, 0, 0, 1, 0)'$
\end{enumerate}
\subsubsection*{SM.3.3.46}
$$ A= %
 $$
\begin{enumerate}[i)]
\item $(t',t)=(1,0)$; $\alpha_{t'}=(0.5, 0.5, -0.5)'; \alpha_{t}=(-1, 1, 1)'$; $c=(1, 0, 0, 0, 0, 0)'$
\end{enumerate}
\subsubsection*{SM.3.3.47}
$$ A= %
 $$
\begin{enumerate}[i)]
\item $(t',t)=(1,0)$; $\alpha_{t'}=(1, -1, 0)'; \alpha_{t}=(-1, 1, 1)'$; $c=(0, 1, 0, 1, 0, 0, 0, 1, 0, 0)'$
\end{enumerate}
\subsubsection*{SM.3.3.48}
$$ A= %
 $$
\begin{enumerate}[i)]
\item $(t',t)=(1,0)$; $\alpha_{t'}=(1, 0, 0)'; \alpha_{t}=(-1, 1, 1)'$; $c=(0, 1, 0, 1, 0, 1, 0, 0, 0, 1, 0, 0)'$
\end{enumerate}
\subsubsection*{SM.3.3.49}
$$ A= %
 $$
\begin{enumerate}[i)]
\item $(t',t)=(1,0)$; $\alpha_{t'}=(0.5, 0.5, 0)'; \alpha_{t}=(-1, 1, 1)'$; $c=(1, 0, 0, 0, 0, 0)'$
\end{enumerate}
\subsubsection*{SM.3.3.50}
$$ A= %
 $$
\begin{enumerate}[i)]
\item $(t',t)=(1,0)$; $\alpha_{t'}=(-1, 1, 0)'; \alpha_{t}=(-1, 1, 1)'$; $c=(1, 0, 0, 0, 0, 0, 0, 0)'$
\end{enumerate}
\subsubsection*{SM.3.3.51}
$$ A= %
 $$
\begin{enumerate}[i)]
\item $(t',t)=(1,0)$; $\alpha_{t'}=(0, 1, 0)'; \alpha_{t}=(-1, 1, 1)'$; $c=(1, 1, 0, 0, 0, 0, 0, 0, 0)'$
\end{enumerate}
\subsubsection*{SM.3.3.52}
$$ A= %
 $$
\begin{enumerate}[i)]
\item $(t',t)=(1,0)$; $\alpha_{t'}=(1, 1, 0)'; \alpha_{t}=(-1, 1, 1)'$; $c=(1, 0, 1, 0, 1, 0, 0, 1, 0)'$
\end{enumerate}
\subsubsection*{SM.3.3.53}
$$ A= %
 $$
\begin{enumerate}[i)]
\item $(t',t)=(1,0)$; $\alpha_{t'}=(0, 0, 1)'; \alpha_{t}=(-1, 1, 1)'$; $c=(0, 0, 1, 1, 0, 0, 0, 0, 0)'$
\end{enumerate}
\subsubsection*{SM.3.3.54}
$$ A= %
 $$
\begin{enumerate}[i)]
\item $(t',t)=(1,0)$; $\alpha_{t'}=(1, 0, 1)'; \alpha_{t}=(-1, 1, 1)'$; $c=(0, 1, 0, 1, 1, 0, 1, 0, 0)'$
\end{enumerate}
\subsubsection*{SM.3.3.55}
$$ A= %
 $$
\begin{enumerate}[i)]
\item $(t',t)=(1,0)$; $\alpha_{t'}=(-2, 1, 1)'; \alpha_{t}=(-1, 1, 1)'$; $c=(1, 0, 1, 0, 0, 0)'$
\end{enumerate}
\subsubsection*{SM.3.3.56}
$$ A= %
 $$
\begin{enumerate}[i)]
\item $(t',t)=(1,0)$; $\alpha_{t'}=(-1, 1, 1)'; \alpha_{t}=(-1, 1, 1)'$; $c=(1, 0, 1, 0, 0, 0, 0)'$
\end{enumerate}
\subsubsection*{SM.3.3.57}
$$ A= %
 $$
\begin{enumerate}[i)]
\item $(t',t)=(1,0)$; $\alpha_{t'}=(0, 1, 1)'; \alpha_{t}=(-1, 1, 1)'$; $c=(1, 1, 0, 1, 1, 0, 0, 0)'$
\end{enumerate}
\subsubsection*{SM.3.3.58}
$$ A= %
 $$
\begin{enumerate}[i)]
\item $(t',t)=(1,0)$; $\alpha_{t'}=(1, 1, 1)'; \alpha_{t}=(-1, 1, 1)'$; $c=(1, 0, 1, 1, 0, 1, 0)'$
\end{enumerate}
\subsubsection*{SM.3.3.59}
$$ A= %
 $$
\begin{enumerate}[i)]
\item $(t',t)=(1,0)$; $\alpha_{t'}=(-1, 2, 1)'; \alpha_{t}=(-1, 1, 1)'$; $c=(1, 0, 1, 0, 0, 0)'$
\end{enumerate}
\subsubsection*{SM.3.3.60}
$$ A= %
 $$
\begin{enumerate}[i)]
\item $(t',t)=(1,0)$; $\alpha_{t'}=(1, 1, -2)'; \alpha_{t}=(-2, 1, 1)'$; $c=(0, 1, 1, 0, 1, 0)'$
\end{enumerate}
\subsubsection*{SM.3.3.61}
$$ A= %
 $$
\begin{enumerate}[i)]
\item $(t',t)=(1,0)$; $\alpha_{t'}=(1, -1, -1)'; \alpha_{t}=(-2, 1, 1)'$; $c=(0, 1, 0, 1, 0, 0)'$
\end{enumerate}
\subsubsection*{SM.3.3.62}
$$ A= %
 $$
\begin{enumerate}[i)]
\item $(t',t)=(1,0)$; $\alpha_{t'}=(2, -1, -1)'; \alpha_{t}=(-2, 1, 1)'$; $c=(0, 1, 1, 0, 0)'$
\end{enumerate}
\subsubsection*{SM.3.3.63}
$$ A= %
 $$
\begin{enumerate}[i)]
\item $(t',t)=(1,0)$; $\alpha_{t'}=(1, 0, -1)'; \alpha_{t}=(-2, 1, 1)'$; $c=(0, 1, 1, 0, 0, 1, 0, 0)'$
\end{enumerate}
\subsubsection*{SM.3.3.64}
$$ A= %
 $$
\begin{enumerate}[i)]
\item $(t',t)=(1,0)$; $\alpha_{t'}=(1, 1, -1)'; \alpha_{t}=(-2, 1, 1)'$; $c=(0, 1, 1, 0, 0, 1, 0)'$
\end{enumerate}
\subsubsection*{SM.3.3.65}
$$ A= %
 $$
\begin{enumerate}[i)]
\item $(t',t)=(1,0)$; $\alpha_{t'}=(2, 1, -1)'; \alpha_{t}=(-2, 1, 1)'$; $c=(0, 1, 1, 0, 0)'$
\end{enumerate}
\subsubsection*{SM.3.3.66}
$$ A= %
 $$
\begin{enumerate}[i)]
\item $(t',t)=(1,0)$; $\alpha_{t'}=(1, -0.5, -0.5)'; \alpha_{t}=(-2, 1, 1)'$; $c=(0, 1, 0, 1, 0)'$
\end{enumerate}
\subsubsection*{SM.3.3.67}
$$ A= %
 $$
\begin{enumerate}[i)]
\item $(t',t)=(1,0)$; $\alpha_{t'}=(1, 0, 0)'; \alpha_{t}=(-2, 1, 1)'$; $c=(0, 1, 1, 1, 0, 0, 1, 0, 0)'$
\end{enumerate}
\subsubsection*{SM.3.3.68}
$$ A= %
 $$
\begin{enumerate}[i)]
\item $(t',t)=(1,0)$; $\alpha_{t'}=(1, 1, 0)'; \alpha_{t}=(-2, 1, 1)'$; $c=(0, 1, 1, 1, 0, 1, 0)'$
\end{enumerate}
\subsubsection*{SM.3.3.69}
$$ A= %
 $$
\begin{enumerate}[i)]
\item $(t',t)=(1,0)$; $\alpha_{t'}=(-1, 1, 1)'; \alpha_{t}=(-2, 1, 1)'$; $c=(0, 1, 1, 0, 0, 0)'$
\end{enumerate}
\subsubsection*{SM.3.3.70}
$$ A= %
 $$
\begin{enumerate}[i)]
\item $(t',t)=(1,0)$; $\alpha_{t'}=(0, 1, 1)'; \alpha_{t}=(-2, 1, 1)'$; $c=(0, 1, 1, 1, 1, 0, 0)'$
\end{enumerate}
\subsubsection*{SM.3.3.71}
$$ A= %
 $$
\begin{enumerate}[i)]
\item $(t',t)=(1,0)$; $\alpha_{t'}=(1, 1, 1)'; \alpha_{t}=(-2, 1, 1)'$; $c=(0, 1, 1, 1, 1, 0)'$
\end{enumerate}
\subsubsection*{SM.3.3.72}
$$ A= %
 $$
\begin{enumerate}[i)]
\item $(t',t)=(1,0)$; $\alpha_{t'}=(-1, 2, 1)'; \alpha_{t}=(-2, 1, 1)'$; $c=(0, 1, 1, 0, 0)'$
\end{enumerate}
\subsubsection*{SM.3.3.73}
$$ A= %
 $$
\begin{enumerate}[i)]
\item $(t',t)=(1,0)$; $\alpha_{t'}=(0, 1, -1)'; \alpha_{t}=(1, 0, 1)'$; $c=(1, 1, 0, 0, 0, 0, 1, 0, 0)'$
\end{enumerate}
\subsubsection*{SM.3.3.74}
$$ A= %
 $$
\begin{enumerate}[i)]
\item $(t',t)=(1,0)$; $\alpha_{t'}=(1, 1, -1)'; \alpha_{t}=(1, 0, 1)'$; $c=(1, 1, 1, 0, 0, 0, 0, 1, 0)'$
\end{enumerate}
\subsubsection*{SM.3.3.75}
$$ A= %
 $$
\begin{enumerate}[i)]
\item $(t',t)=(1,0)$; $\alpha_{t'}=(-1, 1, 0)'; \alpha_{t}=(1, 0, 1)'$; $c=(1, 0, 1, 0, 0, 0, 1, 0, 0)'$
\end{enumerate}
\subsubsection*{SM.3.3.76}
$$ A= %
 $$
\begin{enumerate}[i)]
\item $(t',t)=(1,0)$; $\alpha_{t'}=(0, 1, 0)'; \alpha_{t}=(1, 0, 1)'$; $c=(1, 1, 0, 0, 1, 0, 0, 0, 0, 1, 0, 0)'$
\end{enumerate}
\subsubsection*{SM.3.3.77}
$$ A= %
 $$
\begin{enumerate}[i)]
\item $(t',t)=(1,0)$; $\alpha_{t'}=(1, 1, 0)'; \alpha_{t}=(1, 0, 1)'$; $c=(1, 1, 1, 0, 1, 0, 0, 1, 0)'$
\end{enumerate}
\subsubsection*{SM.3.3.78}
$$ A= %
 $$
\begin{enumerate}[i)]
\item $(t',t)=(1,0)$; $\alpha_{t'}=(1, 0, 1)'; \alpha_{t}=(1, 0, 1)'$; $c=(1, 1, 1, 1, 1, 1, 0, 0, 0)'$
\end{enumerate}
\subsubsection*{SM.3.3.79}
$$ A= %
 $$
\begin{enumerate}[i)]
\item $(t',t)=(1,0)$; $\alpha_{t'}=(-1, 1, 1)'; \alpha_{t}=(1, 0, 1)'$; $c=(1, 1, 0, 1, 0, 0, 1, 0, 0)'$
\end{enumerate}
\subsubsection*{SM.3.3.80}
$$ A= %
 $$
\begin{enumerate}[i)]
\item $(t',t)=(1,0)$; $\alpha_{t'}=(0, 1, 1)'; \alpha_{t}=(1, 0, 1)'$; $c=(1, 1, 1, 1, 0, 0, 1, 0, 0)'$
\end{enumerate}
\subsubsection*{SM.3.3.81}
$$ A= %
 $$
\begin{enumerate}[i)]
\item $(t',t)=(1,0)$; $\alpha_{t'}=(1, 1, -2)'; \alpha_{t}=(0, 0, 1)'$; $c=(1, 1, 1, 1, 0, 0, 0, 0, 0)'$
\end{enumerate}
\subsubsection*{SM.3.3.82}
$$ A= %
 $$
\begin{enumerate}[i)]
\item $(t',t)=(1,0)$; $\alpha_{t'}=(1, -1, -1)'; \alpha_{t}=(0, 0, 1)'$; $c=(1, 1, 0, 0, 0, 0, 0, 0, 0)'$
\end{enumerate}
\subsubsection*{SM.3.3.83}
$$ A= %
 $$
\begin{enumerate}[i)]
\item $(t',t)=(1,0)$; $\alpha_{t'}=(1, 0, -1)'; \alpha_{t}=(0, 0, 1)'$; $c=(1, 1, 1, 0, 0, 0, 0, 0, 0, 0, 0, 0)'$
\end{enumerate}
\subsubsection*{SM.3.3.84}
$$ A= %
 $$
\begin{enumerate}[i)]
\item $(t',t)=(1,0)$; $\alpha_{t'}=(0.5, 0.5, -1)'; \alpha_{t}=(0, 0, 1)'$; $c=(1, 0, 0, 0, 0, 0)'$
\end{enumerate}
\subsubsection*{SM.3.3.85}
$$ A= %
 $$
\begin{enumerate}[i)]
\item $(t',t)=(1,0)$; $\alpha_{t'}=(-1, 1, -1)'; \alpha_{t}=(0, 0, 1)'$; $c=(1, 1, 0, 0, 0, 0, 0, 0, 0)'$
\end{enumerate}
\subsubsection*{SM.3.3.86}
$$ A= %
 $$
\begin{enumerate}[i)]
\item $(t',t)=(1,0)$; $\alpha_{t'}=(0, 1, -1)'; \alpha_{t}=(0, 0, 1)'$; $c=(1, 1, 1, 0, 0, 0, 0, 0, 0, 0, 0, 0)'$
\end{enumerate}
\subsubsection*{SM.3.3.87}
$$ A= %
 $$
\begin{enumerate}[i)]
\item $(t',t)=(1,0)$; $\alpha_{t'}=(1, 1, -1)'; \alpha_{t}=(0, 0, 1)'$; $c=(1, 1, 1, 1, 0, 0, 0, 0, 0, 0, 0, 0)'$
\end{enumerate}
\subsubsection*{SM.3.3.88}
$$ A= %
 $$
\begin{enumerate}[i)]
\item $(t',t)=(1,0)$; $\alpha_{t'}=(0.5, 0.5, -0.5)'; \alpha_{t}=(0, 0, 1)'$; $c=(1, 0, 0, 0, 0, 0, 0, 0, 0)'$
\end{enumerate}
\subsubsection*{SM.3.3.89}
$$ A= %
 $$
\begin{enumerate}[i)]
\item $(t',t)=(1,0)$; $\alpha_{t'}=(1, -1, 0)'; \alpha_{t}=(0, 0, 1)'$; $c=(1, 1, 0, 0, 0, 0, 0, 0, 0, 0, 0, 0)'$
\end{enumerate}
\subsubsection*{SM.3.3.90}
$$ A= %
 $$
\begin{enumerate}[i)]
\item $(t',t)=(1,0)$; $\alpha_{t'}=(1, 0, 0)'; \alpha_{t}=(0, 0, 1)'$; $c=(1, 1, 1, 0, 0, 0, 0, 0, 0, 0, 0, 0, 0, 0, 0)'$
\end{enumerate}
\subsubsection*{SM.3.3.91}
$$ A= %
 $$
\begin{enumerate}[i)]
\item $(t',t)=(1,0)$; $\alpha_{t'}=(0.5, 0.5, 0)'; \alpha_{t}=(0, 0, 1)'$; $c=(1, 0, 0, 0, 0, 0, 0, 0, 0)'$
\end{enumerate}
\subsubsection*{SM.3.3.92}
$$ A= %
 $$
\begin{enumerate}[i)]
\item $(t',t)=(1,0)$; $\alpha_{t'}=(-1, 1, 0)'; \alpha_{t}=(0, 0, 1)'$; $c=(1, 1, 0, 0, 0, 0, 0, 0, 0, 0, 0, 0)'$
\end{enumerate}
\subsubsection*{SM.3.3.93}
$$ A= %
 $$
\begin{enumerate}[i)]
\item $(t',t)=(1,0)$; $\alpha_{t'}=(0, 1, 0)'; \alpha_{t}=(0, 0, 1)'$; $c=(1, 1, 1, 0, 0, 0, 0, 0, 0, 0, 0, 0, 0, 0, 0)'$
\end{enumerate}
\subsubsection*{SM.3.3.94}
$$ A= %
 $$
\begin{enumerate}[i)]
\item $(t',t)=(1,0)$; $\alpha_{t'}=(1, 1, 0)'; \alpha_{t}=(0, 0, 1)'$; $c=(1, 1, 1, 1, 0, 0, 0, 0, 0, 0, 0, 0)'$
\end{enumerate}
\subsubsection*{SM.3.3.95}
$$ A= %
 $$
\begin{enumerate}[i)]
\item $(t',t)=(1,0)$; $\alpha_{t'}=(1, -1, 1)'; \alpha_{t}=(0, 0, 1)'$; $c=(1, 1, 0, 0, 0, 0, 0, 0, 0)'$
\end{enumerate}
\subsubsection*{SM.3.3.96}
$$ A= %
 $$
\begin{enumerate}[i)]
\item $(t',t)=(1,0)$; $\alpha_{t'}=(-1, 1, 1)'; \alpha_{t}=(0, 0, 1)'$; $c=(1, 1, 0, 0, 0, 0, 0, 0, 0)'$
\end{enumerate}
\subsubsection*{SM.3.3.97}
$$ A= %
 $$
\begin{enumerate}[i)]
\item $(t',t)=(1,0)$; $\alpha_{t'}=(1, 1, -2)'; \alpha_{t}=(-1, 0, 1)'$; $c=(0, 1, 1, 0, 1, 0, 0, 0)'$
\end{enumerate}
\subsubsection*{SM.3.3.98}
$$ A= %
 $$
\begin{enumerate}[i)]
\item $(t',t)=(1,0)$; $\alpha_{t'}=(1, -1, -1)'; \alpha_{t}=(-1, 0, 1)'$; $c=(0, 1, 0, 1, 0, 0, 0, 0, 0)'$
\end{enumerate}
\subsubsection*{SM.3.3.99}
$$ A= %
 $$
\begin{enumerate}[i)]
\item $(t',t)=(1,0)$; $\alpha_{t'}=(1, 0, -1)'; \alpha_{t}=(-1, 0, 1)'$; $c=(0, 1, 0, 1, 0, 1, 0, 0, 0, 0, 0, 0)'$
\end{enumerate}
\subsubsection*{SM.3.3.100}
$$ A= %
 $$
\begin{enumerate}[i)]
\item $(t',t)=(1,0)$; $\alpha_{t'}=(0.5, 0.5, -1)'; \alpha_{t}=(-1, 0, 1)'$; $c=(0, 1, 0, 0, 0, 0)'$
\end{enumerate}
\subsubsection*{SM.3.3.101}
$$ A= %
 $$
\begin{enumerate}[i)]
\item $(t',t)=(1,0)$; $\alpha_{t'}=(-1, 1, -1)'; \alpha_{t}=(-1, 0, 1)'$; $c=(0, 1, 0, 0, 0, 0, 0)'$
\end{enumerate}
\subsubsection*{SM.3.3.102}
$$ A= %
 $$
\begin{enumerate}[i)]
\item $(t',t)=(1,0)$; $\alpha_{t'}=(0, 1, -1)'; \alpha_{t}=(-1, 0, 1)'$; $c=(0, 1, 1, 0, 0, 0, 0, 0, 0, 0)'$
\end{enumerate}
\subsubsection*{SM.3.3.103}
$$ A= %
 $$
\begin{enumerate}[i)]
\item $(t',t)=(1,0)$; $\alpha_{t'}=(1, 1, -1)'; \alpha_{t}=(-1, 0, 1)'$; $c=(0, 1, 1, 0, 1, 0, 0, 0, 0, 0)'$
\end{enumerate}
\subsubsection*{SM.3.3.104}
$$ A= %
 $$
\begin{enumerate}[i)]
\item $(t',t)=(1,0)$; $\alpha_{t'}=(0.5, 0.5, -0.5)'; \alpha_{t}=(-1, 0, 1)'$; $c=(0, 1, 0, 0, 0, 0, 0, 0)'$
\end{enumerate}
\subsubsection*{SM.3.3.105}
$$ A= %
 $$
\begin{enumerate}[i)]
\item $(t',t)=(1,0)$; $\alpha_{t'}=(1, -1, 0)'; \alpha_{t}=(-1, 0, 1)'$; $c=(0, 1, 0, 1, 0, 0, 0, 0, 0, 0)'$
\end{enumerate}
\subsubsection*{SM.3.3.106}
$$ A= %
 $$
\begin{enumerate}[i)]
\item $(t',t)=(1,0)$; $\alpha_{t'}=(1, 0, 0)'; \alpha_{t}=(-1, 0, 1)'$; $c=(0, 1, 0, 1, 0, 1, 0, 0, 0, 0, 0, 0)'$
\end{enumerate}
\subsubsection*{SM.3.3.107}
$$ A= %
 $$
\begin{enumerate}[i)]
\item $(t',t)=(1,0)$; $\alpha_{t'}=(0.5, 0.5, 0)'; \alpha_{t}=(-1, 0, 1)'$; $c=(0, 1, 0, 0, 0, 0, 0)'$
\end{enumerate}
\subsubsection*{SM.3.3.108}
$$ A= %
 $$
\begin{enumerate}[i)]
\item $(t',t)=(1,0)$; $\alpha_{t'}=(-1, 1, 0)'; \alpha_{t}=(-1, 0, 1)'$; $c=(0, 1, 0, 0, 0, 0, 0, 0, 0)'$
\end{enumerate}
\subsubsection*{SM.3.3.109}
$$ A= %
 $$
\begin{enumerate}[i)]
\item $(t',t)=(1,0)$; $\alpha_{t'}=(0, 1, 0)'; \alpha_{t}=(-1, 0, 1)'$; $c=(0, 1, 1, 0, 0, 0, 0, 0, 0, 0, 0, 0)'$
\end{enumerate}
\subsubsection*{SM.3.3.110}
$$ A= %
 $$
\begin{enumerate}[i)]
\item $(t',t)=(1,0)$; $\alpha_{t'}=(1, 1, 0)'; \alpha_{t}=(-1, 0, 1)'$; $c=(0, 1, 1, 0, 1, 0, 0, 0, 0)'$
\end{enumerate}
\subsubsection*{SM.3.3.111}
$$ A= %
 $$
\begin{enumerate}[i)]
\item $(t',t)=(1,0)$; $\alpha_{t'}=(1, -1, 1)'; \alpha_{t}=(-1, 0, 1)'$; $c=(0, 1, 0, 1, 0, 0, 0, 0)'$
\end{enumerate}
\subsubsection*{SM.3.3.112}
$$ A= %
 $$
\begin{enumerate}[i)]
\item $(t',t)=(1,0)$; $\alpha_{t'}=(-1, 1, 1)'; \alpha_{t}=(-1, 0, 1)'$; $c=(0, 1, 0, 0, 0, 0, 0, 0)'$
\end{enumerate}
\subsubsection*{SM.3.3.113}
$$ A= %
 $$
\begin{enumerate}[i)]
\item $(t',t)=(1,0)$; $\alpha_{t'}=(1, 1, -1)'; \alpha_{t}=(-0.5, -0.5, 1)'$; $c=(0, 1, 1, 0, 0, 0)'$
\end{enumerate}
\subsubsection*{SM.3.3.114}
$$ A= %
 $$
\begin{enumerate}[i)]
\item $(t',t)=(1,0)$; $\alpha_{t'}=(1, 1, 0)'; \alpha_{t}=(-0.5, -0.5, 1)'$; $c=(0, 1, 1, 0, 0)'$
\end{enumerate}
\subsubsection*{SM.3.3.115}
$$ A= %
 $$
\begin{enumerate}[i)]
\item $(t',t)=(1,0)$; $\alpha_{t'}=(0, 1, -1)'; \alpha_{t}=(1, -1, 1)'$; $c=(0, 0, 1, 1, 0, 0, 0, 1, 0, 0)'$
\end{enumerate}
\subsubsection*{SM.3.3.116}
$$ A= %
 $$
\begin{enumerate}[i)]
\item $(t',t)=(1,0)$; $\alpha_{t'}=(1, 0, 0)'; \alpha_{t}=(1, -1, 1)'$; $c=(0, 1, 1, 0, 0, 0, 0, 0, 0)'$
\end{enumerate}
\subsubsection*{SM.3.3.117}
$$ A= %
 $$
\begin{enumerate}[i)]
\item $(t',t)=(1,0)$; $\alpha_{t'}=(-1, 1, 0)'; \alpha_{t}=(1, -1, 1)'$; $c=(0, 1, 0, 1, 0, 0, 0, 1, 0, 0)'$
\end{enumerate}
\subsubsection*{SM.3.3.118}
$$ A= %
 $$
\begin{enumerate}[i)]
\item $(t',t)=(1,0)$; $\alpha_{t'}=(0, 1, 0)'; \alpha_{t}=(1, -1, 1)'$; $c=(0, 0, 1, 1, 0, 1, 0, 0, 0, 1, 0, 0)'$
\end{enumerate}
\subsubsection*{SM.3.3.119}
$$ A= %
 $$
\begin{enumerate}[i)]
\item $(t',t)=(1,0)$; $\alpha_{t'}=(1, 1, 0)'; \alpha_{t}=(1, -1, 1)'$; $c=(0, 1, 1, 0, 1, 0, 0, 1, 0)'$
\end{enumerate}
\subsubsection*{SM.3.3.120}
$$ A= %
 $$
\begin{enumerate}[i)]
\item $(t',t)=(1,0)$; $\alpha_{t'}=(0, 0, 1)'; \alpha_{t}=(1, -1, 1)'$; $c=(0, 0, 1, 1, 0, 0, 0, 0, 0)'$
\end{enumerate}
\subsubsection*{SM.3.3.121}
$$ A= %
 $$
\begin{enumerate}[i)]
\item $(t',t)=(1,0)$; $\alpha_{t'}=(0, 1, 1)'; \alpha_{t}=(1, -1, 1)'$; $c=(0, 0, 1, 1, 1, 0, 1, 0, 0)'$
\end{enumerate}
\subsubsection*{SM.3.3.122}
$$ A= %
 $$
\begin{enumerate}[i)]
\item $(t',t)=(1,0)$; $\alpha_{t'}=(1, 0, -1)'; \alpha_{t}=(0, -1, 1)'$; $c=(0, 0, 1, 1, 0, 0, 0, 0, 0, 0)'$
\end{enumerate}
\subsubsection*{SM.3.3.123}
$$ A= %
 $$
\begin{enumerate}[i)]
\item $(t',t)=(1,0)$; $\alpha_{t'}=(-1, 1, -1)'; \alpha_{t}=(0, -1, 1)'$; $c=(0, 0, 1, 1, 0, 0, 0, 0, 0)'$
\end{enumerate}
\subsubsection*{SM.3.3.124}
$$ A= %
 $$
\begin{enumerate}[i)]
\item $(t',t)=(1,0)$; $\alpha_{t'}=(0, 1, -1)'; \alpha_{t}=(0, -1, 1)'$; $c=(0, 0, 0, 1, 1, 1, 0, 0, 0, 0, 0, 0)'$
\end{enumerate}
\subsubsection*{SM.3.3.125}
$$ A= %
 $$
\begin{enumerate}[i)]
\item $(t',t)=(1,0)$; $\alpha_{t'}=(1, 1, -1)'; \alpha_{t}=(0, -1, 1)'$; $c=(0, 0, 1, 1, 1, 0, 0, 0, 0, 0)'$
\end{enumerate}
\subsubsection*{SM.3.3.126}
$$ A= %
 $$
\begin{enumerate}[i)]
\item $(t',t)=(1,0)$; $\alpha_{t'}=(1, -1, 0)'; \alpha_{t}=(0, -1, 1)'$; $c=(0, 0, 1, 0, 0, 0, 0, 0, 0)'$
\end{enumerate}
\subsubsection*{SM.3.3.127}
$$ A= %
 $$
\begin{enumerate}[i)]
\item $(t',t)=(1,0)$; $\alpha_{t'}=(1, 0, 0)'; \alpha_{t}=(0, -1, 1)'$; $c=(0, 0, 1, 1, 0, 0, 0, 0, 0, 0, 0, 0)'$
\end{enumerate}
\subsubsection*{SM.3.3.128}
$$ A= %
 $$
\begin{enumerate}[i)]
\item $(t',t)=(1,0)$; $\alpha_{t'}=(-1, 1, 0)'; \alpha_{t}=(0, -1, 1)'$; $c=(0, 0, 1, 1, 0, 0, 0, 0, 0, 0)'$
\end{enumerate}
\subsubsection*{SM.3.3.129}
$$ A= %
 $$
\begin{enumerate}[i)]
\item $(t',t)=(1,0)$; $\alpha_{t'}=(0, 1, 0)'; \alpha_{t}=(0, -1, 1)'$; $c=(0, 0, 0, 1, 1, 1, 0, 0, 0, 0, 0, 0)'$
\end{enumerate}
\subsubsection*{SM.3.3.130}
$$ A= %
 $$
\begin{enumerate}[i)]
\item $(t',t)=(1,0)$; $\alpha_{t'}=(1, 1, 0)'; \alpha_{t}=(0, -1, 1)'$; $c=(0, 0, 1, 1, 1, 0, 0, 0, 0)'$
\end{enumerate}
\subsubsection*{SM.3.3.131}
$$ A= %
 $$
\begin{enumerate}[i)]
\item $(t',t)=(1,0)$; $\alpha_{t'}=(1, 1, -2)'; \alpha_{t}=(-1, -1, 1)'$; $c=(0, 1, 0, 1, 0, 0)'$
\end{enumerate}
\subsubsection*{SM.3.3.132}
$$ A= %
 $$
\begin{enumerate}[i)]
\item $(t',t)=(1,0)$; $\alpha_{t'}=(1, -1, -1)'; \alpha_{t}=(-1, -1, 1)'$; $c=(0, 0, 1, 0, 0, 0)'$
\end{enumerate}
\subsubsection*{SM.3.3.133}
$$ A= %
 $$
\begin{enumerate}[i)]
\item $(t',t)=(1,0)$; $\alpha_{t'}=(1, 0, -1)'; \alpha_{t}=(-1, -1, 1)'$; $c=(0, 0, 1, 0, 1, 0, 0, 0, 0)'$
\end{enumerate}
\subsubsection*{SM.3.3.134}
$$ A= %
 $$
\begin{enumerate}[i)]
\item $(t',t)=(1,0)$; $\alpha_{t'}=(0.5, 0.5, -1)'; \alpha_{t}=(-1, -1, 1)'$; $c=(0, 1, 0, 0, 0)'$
\end{enumerate}
\subsubsection*{SM.3.3.135}
$$ A= %
 $$
\begin{enumerate}[i)]
\item $(t',t)=(1,0)$; $\alpha_{t'}=(0, 1, -1)'; \alpha_{t}=(-1, -1, 1)'$; $c=(0, 0, 1, 1, 0, 0, 0, 0, 0)'$
\end{enumerate}
\subsubsection*{SM.3.3.136}
$$ A= %
 $$
\begin{enumerate}[i)]
\item $(t',t)=(1,0)$; $\alpha_{t'}=(1, 1, -1)'; \alpha_{t}=(-1, -1, 1)'$; $c=(0, 1, 0, 1, 0, 0, 0)'$
\end{enumerate}
\subsubsection*{SM.3.3.137}
$$ A= %
 $$
\begin{enumerate}[i)]
\item $(t',t)=(1,0)$; $\alpha_{t'}=(0.5, 0.5, -0.5)'; \alpha_{t}=(-1, -1, 1)'$; $c=(0, 1, 0, 0, 0, 0)'$
\end{enumerate}
\subsubsection*{SM.3.3.138}
$$ A= %
 $$
\begin{enumerate}[i)]
\item $(t',t)=(1,0)$; $\alpha_{t'}=(1, -1, 0)'; \alpha_{t}=(-1, -1, 1)'$; $c=(0, 0, 1, 0, 0, 0, 0)'$
\end{enumerate}
\subsubsection*{SM.3.3.139}
$$ A= %
 $$
\begin{enumerate}[i)]
\item $(t',t)=(1,0)$; $\alpha_{t'}=(1, 0, 0)'; \alpha_{t}=(-1, -1, 1)'$; $c=(0, 0, 1, 0, 1, 0, 0, 0, 0)'$
\end{enumerate}
\subsubsection*{SM.3.3.140}
$$ A= %
 $$
\begin{enumerate}[i)]
\item $(t',t)=(1,0)$; $\alpha_{t'}=(0.5, 0.5, 0)'; \alpha_{t}=(-1, -1, 1)'$; $c=(0, 1, 0, 0, 0)'$
\end{enumerate}
\subsubsection*{SM.3.3.141}
$$ A= %
 $$
\begin{enumerate}[i)]
\item $(t',t)=(1,0)$; $\alpha_{t'}=(0, 1, 0)'; \alpha_{t}=(-1, -1, 1)'$; $c=(0, 0, 1, 1, 0, 0, 0, 0, 0)'$
\end{enumerate}
\subsubsection*{SM.3.3.142}
$$ A= %
 $$
\begin{enumerate}[i)]
\item $(t',t)=(1,0)$; $\alpha_{t'}=(1, 1, 0)'; \alpha_{t}=(-1, -1, 1)'$; $c=(0, 1, 0, 1, 0, 0)'$
\end{enumerate}
\subsubsection*{SM.3.3.143}
$$ A= %
 $$
\begin{enumerate}[i)]
\item $(t',t)=(1,0)$; $\alpha_{t'}=(1, -1, 1)'; \alpha_{t}=(-1, -1, 1)'$; $c=(0, 0, 1, 0, 0, 0)'$
\end{enumerate}
\subsubsection*{SM.3.3.144}
$$ A= %
 $$
\begin{enumerate}[i)]
\item $(t',t)=(1,0)$; $\alpha_{t'}=(0, 1, 0)'; \alpha_{t}=(1, -2, 1)'$; $c=(0, 1, 1, 1, 0, 0, 1, 0, 0)'$
\end{enumerate}
\subsubsection*{SM.3.3.145}
$$ A= %
 $$
\begin{enumerate}[i)]
\item $(t',t)=(1,0)$; $\alpha_{t'}=(1, 1, 1)'; \alpha_{t}=(0.5, 0.5, 0.5)'$; $c=(1, 1, 1, 0)'$
\end{enumerate}
\subsubsection*{SM.3.3.146}
$$ A= %
 $$
\begin{enumerate}[i)]
\item $(t',t)=(1,0)$; $\alpha_{t'}=(1, -1, -1)'; \alpha_{t}=(0, 0.5, 0.5)'$; $c=(1, 0, 0, 0, 0)'$
\end{enumerate}
\subsubsection*{SM.3.3.147}
$$ A= %
 $$
\begin{enumerate}[i)]
\item $(t',t)=(1,0)$; $\alpha_{t'}=(1, 0, -1)'; \alpha_{t}=(0, 0.5, 0.5)'$; $c=(1, 0, 0, 0, 0, 0, 0)'$
\end{enumerate}
\subsubsection*{SM.3.3.148}
$$ A= %
 $$
\begin{enumerate}[i)]
\item $(t',t)=(1,0)$; $\alpha_{t'}=(1, 1, -1)'; \alpha_{t}=(0, 0.5, 0.5)'$; $c=(1, 0, 0, 0, 0, 0)'$
\end{enumerate}
\subsubsection*{SM.3.3.149}
$$ A= %
 $$
\begin{enumerate}[i)]
\item $(t',t)=(1,0)$; $\alpha_{t'}=(1, 0, 0)'; \alpha_{t}=(0, 0.5, 0.5)'$; $c=(1, 0, 0, 0, 0, 0, 0, 0, 0)'$
\end{enumerate}
\subsubsection*{SM.3.3.150}
$$ A= %
 $$
\begin{enumerate}[i)]
\item $(t',t)=(1,0)$; $\alpha_{t'}=(1, -1, -1)'; \alpha_{t}=(-0.5, 0.5, 0.5)'$; $c=(1, 0, 0, 0, 0, 0)'$
\end{enumerate}
\subsubsection*{SM.3.3.151}
$$ A= %
 $$
\begin{enumerate}[i)]
\item $(t',t)=(1,0)$; $\alpha_{t'}=(1, 0, -1)'; \alpha_{t}=(-0.5, 0.5, 0.5)'$; $c=(1, 0, 0, 0, 0, 0, 0, 0)'$
\end{enumerate}
\subsubsection*{SM.3.3.152}
$$ A= %
 $$
\begin{enumerate}[i)]
\item $(t',t)=(1,0)$; $\alpha_{t'}=(1, 1, -1)'; \alpha_{t}=(-0.5, 0.5, 0.5)'$; $c=(1, 0, 0, 0, 0, 0)'$
\end{enumerate}
\subsubsection*{SM.3.3.153}
$$ A= %
 $$
\begin{enumerate}[i)]
\item $(t',t)=(1,0)$; $\alpha_{t'}=(1, 0, 0)'; \alpha_{t}=(-0.5, 0.5, 0.5)'$; $c=(1, 0, 0, 0, 0, 0, 0, 0, 0)'$
\end{enumerate}
\subsubsection*{SM.3.3.154}
$$ A= %
 $$
\begin{enumerate}[i)]
\item $(t',t)=(1,0)$; $\alpha_{t'}=(1, -1, -1)'; \alpha_{t}=(-1, 0.5, 0.5)'$; $c=(0, 1, 0, 0, 0)'$
\end{enumerate}
\subsubsection*{SM.3.3.155}
$$ A= %
 $$
\begin{enumerate}[i)]
\item $(t',t)=(1,0)$; $\alpha_{t'}=(1, 0, -1)'; \alpha_{t}=(-1, 0.5, 0.5)'$; $c=(0, 1, 0, 0, 0, 0)'$
\end{enumerate}
\subsubsection*{SM.3.3.156}
$$ A= %
 $$
\begin{enumerate}[i)]
\item $(t',t)=(1,0)$; $\alpha_{t'}=(1, 1, -1)'; \alpha_{t}=(-1, 0.5, 0.5)'$; $c=(0, 1, 0, 0, 0)'$
\end{enumerate}
\subsubsection*{SM.3.3.157}
$$ A= %
 $$
\begin{enumerate}[i)]
\item $(t',t)=(1,0)$; $\alpha_{t'}=(1, 0, 0)'; \alpha_{t}=(-1, 0.5, 0.5)'$; $c=(0, 1, 0, 0, 0, 0)'$
\end{enumerate}
\subsubsection*{SM.3.3.158}
$$ A= %
 $$
\begin{enumerate}[i)]
\item $(t',t)=(1,0)$; $\alpha_{t'}=(0, 1, 0)'; \alpha_{t}=(0.5, 0, 0.5)'$; $c=(1, 0, 0, 0, 0, 0, 0, 0, 0)'$
\end{enumerate}
\subsubsection*{SM.3.3.159}
$$ A= %
 $$
\begin{enumerate}[i)]
\item $(t',t)=(1,0)$; $\alpha_{t'}=(0, 1, 0)'; \alpha_{t}=(0.5, -0.5, 0.5)'$; $c=(0, 0, 1, 0, 0, 0, 0, 0, 0)'$
\end{enumerate}
\subsubsection*{SM.3.3.160}
$$ A= %
 $$
\begin{enumerate}[i)]
\item $(t',t)=(1,0)$; $\alpha_{t'}=(1, 1, 0)'; \alpha_{t}=(1, 1, 0)'$; $c=(1, 1, 0, 1, 1, 0, 1, 1, 0)'$
\end{enumerate}
\subsubsection*{SM.3.3.161}
$$ A= %
 $$
\begin{enumerate}[i)]
\item $(t',t)=(1,0)$; $\alpha_{t'}=(0, -1, 1)'; \alpha_{t}=(1, 1, 0)'$; $c=(0, 0, 1, 1, 0, 0, 1, 0, 0)'$
\end{enumerate}
\subsubsection*{SM.3.3.162}
$$ A= %
 $$
\begin{enumerate}[i)]
\item $(t',t)=(1,0)$; $\alpha_{t'}=(1, -1, 1)'; \alpha_{t}=(1, 1, 0)'$; $c=(1, 0, 0, 1, 0, 1, 1, 0, 0)'$
\end{enumerate}
\subsubsection*{SM.3.3.163}
$$ A= %
 $$
\begin{enumerate}[i)]
\item $(t',t)=(1,0)$; $\alpha_{t'}=(-1, 0, 1)'; \alpha_{t}=(1, 1, 0)'$; $c=(0, 0, 1, 1, 0, 1, 0, 0, 0)'$
\end{enumerate}
\subsubsection*{SM.3.3.164}
$$ A= %
 $$
\begin{enumerate}[i)]
\item $(t',t)=(1,0)$; $\alpha_{t'}=(0, 0, 1)'; \alpha_{t}=(1, 1, 0)'$; $c=(0, 0, 0, 0, 1, 1, 1, 1, 0, 0, 0, 0)'$
\end{enumerate}
\subsubsection*{SM.3.3.165}
$$ A= %
 $$
\begin{enumerate}[i)]
\item $(t',t)=(1,0)$; $\alpha_{t'}=(1, 0, 1)'; \alpha_{t}=(1, 1, 0)'$; $c=(1, 0, 0, 1, 1, 1, 1, 0, 0)'$
\end{enumerate}
\subsubsection*{SM.3.3.166}
$$ A= %
 $$
\begin{enumerate}[i)]
\item $(t',t)=(1,0)$; $\alpha_{t'}=(-1, 1, 1)'; \alpha_{t}=(1, 1, 0)'$; $c=(1, 0, 0, 1, 1, 0, 1, 0, 0)'$
\end{enumerate}
\subsubsection*{SM.3.3.167}
$$ A= %
 $$
\begin{enumerate}[i)]
\item $(t',t)=(1,0)$; $\alpha_{t'}=(0, 1, 1)'; \alpha_{t}=(1, 1, 0)'$; $c=(1, 0, 0, 1, 1, 1, 1, 0, 0)'$
\end{enumerate}
\subsubsection*{SM.3.3.168}
$$ A= %
 $$
\begin{enumerate}[i)]
\item $(t',t)=(1,0)$; $\alpha_{t'}=(1, -1, -1)'; \alpha_{t}=(0, 1, 0)'$; $c=(1, 0, 0, 0, 0, 1, 0, 0, 0)'$
\end{enumerate}
\subsubsection*{SM.3.3.169}
$$ A= %
 $$
\begin{enumerate}[i)]
\item $(t',t)=(1,0)$; $\alpha_{t'}=(1, 0, -1)'; \alpha_{t}=(0, 1, 0)'$; $c=(1, 0, 0, 0, 0, 0, 0, 1, 0, 0, 0, 0)'$
\end{enumerate}
\subsubsection*{SM.3.3.170}
$$ A= %
 $$
\begin{enumerate}[i)]
\item $(t',t)=(1,0)$; $\alpha_{t'}=(1, 1, -1)'; \alpha_{t}=(0, 1, 0)'$; $c=(1, 0, 0, 0, 0, 0, 1, 0, 0)'$
\end{enumerate}
\subsubsection*{SM.3.3.171}
$$ A= %
 $$
\begin{enumerate}[i)]
\item $(t',t)=(1,0)$; $\alpha_{t'}=(1, -1, 0)'; \alpha_{t}=(0, 1, 0)'$; $c=(1, 0, 0, 0, 1, 0, 0, 0, 1, 0, 0, 0)'$
\end{enumerate}
\subsubsection*{SM.3.3.172}
$$ A= %
 $$
\begin{enumerate}[i)]
\item $(t',t)=(1,0)$; $\alpha_{t'}=(1, 0, 0)'; \alpha_{t}=(0, 1, 0)'$; $c=(1, 0, 0, 0, 0, 1, 0, 0, 0, 0, 1, 0, 0, 0, 0)'$
\end{enumerate}
\subsubsection*{SM.3.3.173}
$$ A= %
 $$
\begin{enumerate}[i)]
\item $(t',t)=(1,0)$; $\alpha_{t'}=(0.5, -0.5, 0.5)'; \alpha_{t}=(0, 1, 0)'$; $c=(0, 0, 0, 1, 0, 0, 0, 0, 0)'$
\end{enumerate}
\subsubsection*{SM.3.3.174}
$$ A= %
 $$
\begin{enumerate}[i)]
\item $(t',t)=(1,0)$; $\alpha_{t'}=(0.5, 0, 0.5)'; \alpha_{t}=(0, 1, 0)'$; $c=(0, 0, 0, 0, 1, 0, 0, 0, 0)'$
\end{enumerate}
\subsubsection*{SM.3.3.175}
$$ A= %
 $$
\begin{enumerate}[i)]
\item $(t',t)=(1,0)$; $\alpha_{t'}=(1, -2, 1)'; \alpha_{t}=(0, 1, 0)'$; $c=(1, 0, 0, 1, 1, 0, 1, 0, 0)'$
\end{enumerate}
\subsubsection*{SM.3.3.176}
$$ A= %
 $$
\begin{enumerate}[i)]
\item $(t',t)=(1,0)$; $\alpha_{t'}=(-1, -1, 1)'; \alpha_{t}=(0, 1, 0)'$; $c=(0, 0, 1, 1, 0, 0, 0, 0, 0)'$
\end{enumerate}
\subsubsection*{SM.3.3.177}
$$ A= %
 $$
\begin{enumerate}[i)]
\item $(t',t)=(1,0)$; $\alpha_{t'}=(0, -1, 1)'; \alpha_{t}=(0, 1, 0)'$; $c=(0, 0, 0, 1, 1, 1, 0, 0, 0, 0, 0, 0)'$
\end{enumerate}
\subsubsection*{SM.3.3.178}
$$ A= %
 $$
\begin{enumerate}[i)]
\item $(t',t)=(1,0)$; $\alpha_{t'}=(1, -1, 1)'; \alpha_{t}=(0, 1, 0)'$; $c=(1, 0, 0, 0, 1, 1, 0, 0, 1, 0, 0, 0)'$
\end{enumerate}
\subsubsection*{SM.3.3.179}
$$ A= %
 $$
\begin{enumerate}[i)]
\item $(t',t)=(1,0)$; $\alpha_{t'}=(-1, 0, 1)'; \alpha_{t}=(0, 1, 0)'$; $c=(0, 0, 0, 0, 1, 1, 0, 0, 0, 0, 0, 0)'$
\end{enumerate}
\subsubsection*{SM.3.3.180}
$$ A= %
 $$
\begin{enumerate}[i)]
\item $(t',t)=(1,0)$; $\alpha_{t'}=(0, 0, 1)'; \alpha_{t}=(0, 1, 0)'$; $c=(0, 0, 0, 0, 0, 0, 1, 1, 1, 0, 0, 0, 0, 0, 0)'$
\end{enumerate}
\subsubsection*{SM.3.3.181}
$$ A= %
 $$
\begin{enumerate}[i)]
\item $(t',t)=(1,0)$; $\alpha_{t'}=(1, 0, 1)'; \alpha_{t}=(0, 1, 0)'$; $c=(1, 0, 0, 0, 0, 1, 1, 1, 0, 0, 0, 0)'$
\end{enumerate}
\subsubsection*{SM.3.3.182}
$$ A= %
 $$
\begin{enumerate}[i)]
\item $(t',t)=(1,0)$; $\alpha_{t'}=(-1, 1, 1)'; \alpha_{t}=(0, 1, 0)'$; $c=(0, 0, 0, 1, 1, 0, 0, 0, 0)'$
\end{enumerate}
\subsubsection*{SM.3.3.183}
$$ A= %
 $$
\begin{enumerate}[i)]
\item $(t',t)=(1,0)$; $\alpha_{t'}=(1, -1, -1)'; \alpha_{t}=(-1, 1, 0)'$; $c=(0, 1, 0, 0, 0, 0, 1, 0, 0)'$
\end{enumerate}
\subsubsection*{SM.3.3.184}
$$ A= %
 $$
\begin{enumerate}[i)]
\item $(t',t)=(1,0)$; $\alpha_{t'}=(1, 0, -1)'; \alpha_{t}=(-1, 1, 0)'$; $c=(0, 1, 0, 0, 0, 0, 0, 1, 0, 0)'$
\end{enumerate}
\subsubsection*{SM.3.3.185}
$$ A= %
 $$
\begin{enumerate}[i)]
\item $(t',t)=(1,0)$; $\alpha_{t'}=(1, -1, 0)'; \alpha_{t}=(-1, 1, 0)'$; $c=(0, 1, 0, 0, 0, 1, 0, 0, 0, 1, 0, 0)'$
\end{enumerate}
\subsubsection*{SM.3.3.186}
$$ A= %
 $$
\begin{enumerate}[i)]
\item $(t',t)=(1,0)$; $\alpha_{t'}=(1, 0, 0)'; \alpha_{t}=(-1, 1, 0)'$; $c=(0, 1, 0, 0, 0, 1, 0, 0, 0, 1, 0, 0)'$
\end{enumerate}
\subsubsection*{SM.3.3.187}
$$ A= %
 $$
\begin{enumerate}[i)]
\item $(t',t)=(1,0)$; $\alpha_{t'}=(0, -1, 1)'; \alpha_{t}=(-1, 1, 0)'$; $c=(0, 0, 0, 1, 1, 0, 0, 0, 0, 0)'$
\end{enumerate}
\subsubsection*{SM.3.3.188}
$$ A= %
 $$
\begin{enumerate}[i)]
\item $(t',t)=(1,0)$; $\alpha_{t'}=(1, -1, 1)'; \alpha_{t}=(-1, 1, 0)'$; $c=(0, 1, 0, 0, 1, 0, 0, 1, 0, 0)'$
\end{enumerate}
\subsubsection*{SM.3.3.189}
$$ A= %
 $$
\begin{enumerate}[i)]
\item $(t',t)=(1,0)$; $\alpha_{t'}=(-1, 0, 1)'; \alpha_{t}=(-1, 1, 0)'$; $c=(0, 0, 0, 1, 0, 0, 0, 0, 0)'$
\end{enumerate}
\subsubsection*{SM.3.3.190}
$$ A= %
 $$
\begin{enumerate}[i)]
\item $(t',t)=(1,0)$; $\alpha_{t'}=(0, 0, 1)'; \alpha_{t}=(-1, 1, 0)'$; $c=(0, 0, 0, 0, 0, 1, 1, 0, 0, 0, 0, 0)'$
\end{enumerate}
\subsubsection*{SM.3.3.191}
$$ A= %
 $$
\begin{enumerate}[i)]
\item $(t',t)=(1,0)$; $\alpha_{t'}=(1, 0, 1)'; \alpha_{t}=(-1, 1, 0)'$; $c=(0, 1, 0, 0, 1, 0, 1, 0, 0)'$
\end{enumerate}
\subsubsection*{SM.3.3.192}
$$ A= %
 $$
\begin{enumerate}[i)]
\item $(t',t)=(1,0)$; $\alpha_{t'}=(0, 0, 1)'; \alpha_{t}=(0.5, 0.5, 0)'$; $c=(0, 0, 0, 0, 1, 0, 0, 0, 0)'$
\end{enumerate}
\subsubsection*{SM.3.3.193}
$$ A= %
 $$
\begin{enumerate}[i)]
\item $(t',t)=(1,0)$; $\alpha_{t'}=(-1, 1, -1)'; \alpha_{t}=(1, 0, 0)'$; $c=(0, 1, 0, 0, 0, 0, 1, 0, 0)'$
\end{enumerate}
\subsubsection*{SM.3.3.194}
$$ A= %
 $$
\begin{enumerate}[i)]
\item $(t',t)=(1,0)$; $\alpha_{t'}=(0, 1, -1)'; \alpha_{t}=(1, 0, 0)'$; $c=(0, 0, 1, 0, 0, 0, 0, 0, 0, 1, 0, 0)'$
\end{enumerate}
\subsubsection*{SM.3.3.195}
$$ A= %
 $$
\begin{enumerate}[i)]
\item $(t',t)=(1,0)$; $\alpha_{t'}=(1, 1, -1)'; \alpha_{t}=(1, 0, 0)'$; $c=(0, 1, 0, 0, 0, 0, 0, 1, 0)'$
\end{enumerate}
\subsubsection*{SM.3.3.196}
$$ A= %
 $$
\begin{enumerate}[i)]
\item $(t',t)=(1,0)$; $\alpha_{t'}=(-1, 1, 0)'; \alpha_{t}=(1, 0, 0)'$; $c=(0, 1, 0, 0, 0, 1, 0, 0, 0, 1, 0, 0)'$
\end{enumerate}
\subsubsection*{SM.3.3.197}
$$ A= %
 $$
\begin{enumerate}[i)]
\item $(t',t)=(1,0)$; $\alpha_{t'}=(0, 1, 0)'; \alpha_{t}=(1, 0, 0)'$; $c=(0, 0, 1, 0, 0, 0, 0, 1, 0, 0, 0, 0, 1, 0, 0)'$
\end{enumerate}
\subsubsection*{SM.3.3.198}
$$ A= %
 $$
\begin{enumerate}[i)]
\item $(t',t)=(1,0)$; $\alpha_{t'}=(-0.5, 0.5, 0.5)'; \alpha_{t}=(1, 0, 0)'$; $c=(0, 0, 0, 0, 1, 0, 0, 0, 0)'$
\end{enumerate}
\subsubsection*{SM.3.3.199}
$$ A= %
 $$
\begin{enumerate}[i)]
\item $(t',t)=(1,0)$; $\alpha_{t'}=(0, 0.5, 0.5)'; \alpha_{t}=(1, 0, 0)'$; $c=(0, 0, 0, 0, 1, 0, 0, 0, 0)'$
\end{enumerate}
\subsubsection*{SM.3.3.200}
$$ A= %
 $$
\begin{enumerate}[i)]
\item $(t',t)=(1,0)$; $\alpha_{t'}=(-1, -1, 1)'; \alpha_{t}=(1, 0, 0)'$; $c=(0, 0, 1, 0, 0, 1, 0, 0, 0)'$
\end{enumerate}
\subsubsection*{SM.3.3.201}
$$ A= %
 $$
\begin{enumerate}[i)]
\item $(t',t)=(1,0)$; $\alpha_{t'}=(0, -1, 1)'; \alpha_{t}=(1, 0, 0)'$; $c=(0, 0, 0, 0, 1, 0, 0, 1, 0, 0, 0, 0)'$
\end{enumerate}
\subsubsection*{SM.3.3.202}
$$ A= %
 $$
\begin{enumerate}[i)]
\item $(t',t)=(1,0)$; $\alpha_{t'}=(1, -1, 1)'; \alpha_{t}=(1, 0, 0)'$; $c=(0, 0, 0, 1, 0, 1, 0, 0, 0)'$
\end{enumerate}
\subsubsection*{SM.3.3.203}
$$ A= %
 $$
\begin{enumerate}[i)]
\item $(t',t)=(1,0)$; $\alpha_{t'}=(-1, 0, 1)'; \alpha_{t}=(1, 0, 0)'$; $c=(0, 0, 0, 1, 0, 1, 0, 1, 0, 0, 0, 0)'$
\end{enumerate}
\subsubsection*{SM.3.3.204}
$$ A= %
 $$
\begin{enumerate}[i)]
\item $(t',t)=(1,0)$; $\alpha_{t'}=(0, 0, 1)'; \alpha_{t}=(1, 0, 0)'$; $c=(0, 0, 0, 0, 0, 0, 1, 1, 1, 0, 0, 0, 0, 0, 0)'$
\end{enumerate}
\subsubsection*{SM.3.3.205}
$$ A= %
 $$
\begin{enumerate}[i)]
\item $(t',t)=(1,0)$; $\alpha_{t'}=(-2, 1, 1)'; \alpha_{t}=(1, 0, 0)'$; $c=(0, 1, 0, 1, 0, 1, 0, 1, 0)'$
\end{enumerate}
\subsubsection*{SM.3.3.206}
$$ A= %
 $$
\begin{enumerate}[i)]
\item $(t',t)=(1,0)$; $\alpha_{t'}=(-1, 1, 1)'; \alpha_{t}=(1, 0, 0)'$; $c=(0, 1, 0, 0, 1, 0, 1, 0, 0, 1, 0, 0)'$
\end{enumerate}
\subsubsection*{SM.3.3.207}
$$ A= %
 $$
\begin{enumerate}[i)]
\item $(t',t)=(1,0)$; $\alpha_{t'}=(0, 1, 1)'; \alpha_{t}=(1, 0, 0)'$; $c=(0, 0, 1, 0, 0, 1, 1, 0, 0, 1, 0, 0)'$
\end{enumerate}
\subsubsection*{SM.3.3.208}
$$ A= %
 $$
\begin{enumerate}[i)]
\item $(t',t)=(1,0)$; $\alpha_{t'}=(-1, 1, -1)'; \alpha_{t}=(1, -1, 0)'$; $c=(0, 1, 0, 0, 0, 0, 1, 0, 0)'$
\end{enumerate}
\subsubsection*{SM.3.3.209}
$$ A= %
 $$
\begin{enumerate}[i)]
\item $(t',t)=(1,0)$; $\alpha_{t'}=(0, 1, -1)'; \alpha_{t}=(1, -1, 0)'$; $c=(0, 1, 0, 0, 0, 0, 0, 1, 0, 0)'$
\end{enumerate}
\subsubsection*{SM.3.3.210}
$$ A= %
 $$
\begin{enumerate}[i)]
\item $(t',t)=(1,0)$; $\alpha_{t'}=(-1, 1, 0)'; \alpha_{t}=(1, -1, 0)'$; $c=(0, 1, 0, 0, 0, 1, 0, 0, 0, 1, 0, 0)'$
\end{enumerate}
\subsubsection*{SM.3.3.211}
$$ A= %
 $$
\begin{enumerate}[i)]
\item $(t',t)=(1,0)$; $\alpha_{t'}=(0, 1, 0)'; \alpha_{t}=(1, -1, 0)'$; $c=(0, 1, 0, 0, 0, 1, 0, 0, 0, 1, 0, 0)'$
\end{enumerate}
\subsubsection*{SM.3.3.212}
$$ A= %
 $$
\begin{enumerate}[i)]
\item $(t',t)=(1,0)$; $\alpha_{t'}=(0, -1, 1)'; \alpha_{t}=(1, -1, 0)'$; $c=(0, 0, 0, 0, 0, 1, 0, 0, 0)'$
\end{enumerate}
\subsubsection*{SM.3.3.213}
$$ A= %
 $$
\begin{enumerate}[i)]
\item $(t',t)=(1,0)$; $\alpha_{t'}=(-1, 0, 1)'; \alpha_{t}=(1, -1, 0)'$; $c=(0, 0, 0, 1, 0, 1, 0, 0, 0, 0)'$
\end{enumerate}
\subsubsection*{SM.3.3.214}
$$ A= %
 $$
\begin{enumerate}[i)]
\item $(t',t)=(1,0)$; $\alpha_{t'}=(0, 0, 1)'; \alpha_{t}=(1, -1, 0)'$; $c=(0, 0, 0, 0, 0, 1, 1, 0, 0, 0, 0, 0)'$
\end{enumerate}
\subsubsection*{SM.3.3.215}
$$ A= %
 $$
\begin{enumerate}[i)]
\item $(t',t)=(1,0)$; $\alpha_{t'}=(-1, 1, 1)'; \alpha_{t}=(1, -1, 0)'$; $c=(0, 1, 0, 0, 1, 0, 0, 1, 0, 0)'$
\end{enumerate}
\subsubsection*{SM.3.3.216}
$$ A= %
 $$
\begin{enumerate}[i)]
\item $(t',t)=(1,0)$; $\alpha_{t'}=(0, 1, 1)'; \alpha_{t}=(1, -1, 0)'$; $c=(0, 1, 0, 0, 1, 0, 1, 0, 0)'$
\end{enumerate}
\subsubsection*{SM.3.3.217}
$$ A= %
 $$
\begin{enumerate}[i)]
\item $(t',t)=(1,0)$; $\alpha_{t'}=(0, 0, 1)'; \alpha_{t}=(0.5, 0.5, -0.5)'$; $c=(0, 0, 0, 0, 1, 0, 0, 0, 0)'$
\end{enumerate}
\subsubsection*{SM.3.3.218}
$$ A= %
 $$
\begin{enumerate}[i)]
\item $(t',t)=(1,0)$; $\alpha_{t'}=(1, 0, 0)'; \alpha_{t}=(1, 1, -1)'$; $c=(0, 0, 0, 1, 0, 0, 1, 0, 0)'$
\end{enumerate}
\subsubsection*{SM.3.3.219}
$$ A= %
 $$
\begin{enumerate}[i)]
\item $(t',t)=(1,0)$; $\alpha_{t'}=(0, 1, 0)'; \alpha_{t}=(1, 1, -1)'$; $c=(0, 0, 0, 1, 0, 0, 1, 0, 0)'$
\end{enumerate}
\subsubsection*{SM.3.3.220}
$$ A= %
 $$
\begin{enumerate}[i)]
\item $(t',t)=(1,0)$; $\alpha_{t'}=(0, -1, 1)'; \alpha_{t}=(1, 1, -1)'$; $c=(0, 0, 0, 1, 1, 0, 0, 1, 0, 0)'$
\end{enumerate}
\subsubsection*{SM.3.3.221}
$$ A= %
 $$
\begin{enumerate}[i)]
\item $(t',t)=(1,0)$; $\alpha_{t'}=(-1, 0, 1)'; \alpha_{t}=(1, 1, -1)'$; $c=(0, 0, 0, 1, 1, 0, 1, 0, 0, 0)'$
\end{enumerate}
\subsubsection*{SM.3.3.222}
$$ A= %
 $$
\begin{enumerate}[i)]
\item $(t',t)=(1,0)$; $\alpha_{t'}=(0, 0, 1)'; \alpha_{t}=(1, 1, -1)'$; $c=(0, 0, 0, 0, 1, 1, 1, 1, 0, 0, 0, 0)'$
\end{enumerate}
\subsubsection*{SM.3.3.223}
$$ A= %
 $$
\begin{enumerate}[i)]
\item $(t',t)=(1,0)$; $\alpha_{t'}=(1, 0, 1)'; \alpha_{t}=(1, 1, -1)'$; $c=(0, 0, 0, 1, 1, 1, 1, 0, 0)'$
\end{enumerate}
\subsubsection*{SM.3.3.224}
$$ A= %
 $$
\begin{enumerate}[i)]
\item $(t',t)=(1,0)$; $\alpha_{t'}=(0, 1, 1)'; \alpha_{t}=(1, 1, -1)'$; $c=(0, 0, 0, 1, 1, 1, 1, 0, 0)'$
\end{enumerate}
\subsubsection*{SM.3.3.225}
$$ A= %
 $$
\begin{enumerate}[i)]
\item $(t',t)=(1,0)$; $\alpha_{t'}=(1, 0, -1)'; \alpha_{t}=(0, 1, -1)'$; $c=(0, 0, 0, 0, 1, 0, 0, 0, 0)'$
\end{enumerate}
\subsubsection*{SM.3.3.226}
$$ A= %
 $$
\begin{enumerate}[i)]
\item $(t',t)=(1,0)$; $\alpha_{t'}=(1, -1, 0)'; \alpha_{t}=(0, 1, -1)'$; $c=(0, 0, 1, 0, 0, 0, 1, 0, 0, 0)'$
\end{enumerate}
\subsubsection*{SM.3.3.227}
$$ A= %
 $$
\begin{enumerate}[i)]
\item $(t',t)=(1,0)$; $\alpha_{t'}=(1, 0, 0)'; \alpha_{t}=(0, 1, -1)'$; $c=(0, 0, 1, 0, 0, 0, 0, 1, 0, 0, 0, 0)'$
\end{enumerate}
\subsubsection*{SM.3.3.228}
$$ A= %
 $$
\begin{enumerate}[i)]
\item $(t',t)=(1,0)$; $\alpha_{t'}=(-1, -1, 1)'; \alpha_{t}=(0, 1, -1)'$; $c=(0, 0, 1, 1, 0, 0, 0, 0, 0)'$
\end{enumerate}
\subsubsection*{SM.3.3.229}
$$ A= %
 $$
\begin{enumerate}[i)]
\item $(t',t)=(1,0)$; $\alpha_{t'}=(0, -1, 1)'; \alpha_{t}=(0, 1, -1)'$; $c=(0, 0, 0, 1, 1, 1, 0, 0, 0, 0, 0, 0)'$
\end{enumerate}
\subsubsection*{SM.3.3.230}
$$ A= %
 $$
\begin{enumerate}[i)]
\item $(t',t)=(1,0)$; $\alpha_{t'}=(1, -1, 1)'; \alpha_{t}=(0, 1, -1)'$; $c=(0, 0, 1, 1, 0, 0, 1, 0, 0, 0)'$
\end{enumerate}
\subsubsection*{SM.3.3.231}
$$ A= %
 $$
\begin{enumerate}[i)]
\item $(t',t)=(1,0)$; $\alpha_{t'}=(-1, 0, 1)'; \alpha_{t}=(0, 1, -1)'$; $c=(0, 0, 1, 1, 0, 0, 0, 0, 0, 0)'$
\end{enumerate}
\subsubsection*{SM.3.3.232}
$$ A= %
 $$
\begin{enumerate}[i)]
\item $(t',t)=(1,0)$; $\alpha_{t'}=(0, 0, 1)'; \alpha_{t}=(0, 1, -1)'$; $c=(0, 0, 0, 1, 1, 1, 0, 0, 0, 0, 0, 0)'$
\end{enumerate}
\subsubsection*{SM.3.3.233}
$$ A= %
 $$
\begin{enumerate}[i)]
\item $(t',t)=(1,0)$; $\alpha_{t'}=(1, 0, 1)'; \alpha_{t}=(0, 1, -1)'$; $c=(0, 0, 1, 1, 1, 0, 0, 0, 0)'$
\end{enumerate}
\subsubsection*{SM.3.3.234}
$$ A= %
 $$
\begin{enumerate}[i)]
\item $(t',t)=(1,0)$; $\alpha_{t'}=(1, -1, 0)'; \alpha_{t}=(-1, 1, -1)'$; $c=(0, 0, 1, 0, 0, 0, 1, 0, 0)'$
\end{enumerate}
\subsubsection*{SM.3.3.235}
$$ A= %
 $$
\begin{enumerate}[i)]
\item $(t',t)=(1,0)$; $\alpha_{t'}=(1, 0, 0)'; \alpha_{t}=(-1, 1, -1)'$; $c=(0, 0, 1, 0, 0, 0, 1, 0, 0)'$
\end{enumerate}
\subsubsection*{SM.3.3.236}
$$ A= %
 $$
\begin{enumerate}[i)]
\item $(t',t)=(1,0)$; $\alpha_{t'}=(0, -1, 1)'; \alpha_{t}=(-1, 1, -1)'$; $c=(0, 0, 1, 1, 0, 0, 0, 0, 0)'$
\end{enumerate}
\subsubsection*{SM.3.3.237}
$$ A= %
 $$
\begin{enumerate}[i)]
\item $(t',t)=(1,0)$; $\alpha_{t'}=(0, 0, 1)'; \alpha_{t}=(-1, 1, -1)'$; $c=(0, 0, 1, 1, 0, 0, 0, 0, 0)'$
\end{enumerate}
\subsubsection*{SM.3.3.238}
$$ A= %
 $$
\begin{enumerate}[i)]
\item $(t',t)=(1,0)$; $\alpha_{t'}=(0, 1, -1)'; \alpha_{t}=(1, 0, -1)'$; $c=(0, 0, 0, 0, 0, 0, 1, 0, 0)'$
\end{enumerate}
\subsubsection*{SM.3.3.239}
$$ A= %
 $$
\begin{enumerate}[i)]
\item $(t',t)=(1,0)$; $\alpha_{t'}=(-1, 1, 0)'; \alpha_{t}=(1, 0, -1)'$; $c=(0, 0, 0, 1, 0, 0, 0, 1, 0, 0)'$
\end{enumerate}
\subsubsection*{SM.3.3.240}
$$ A= %
 $$
\begin{enumerate}[i)]
\item $(t',t)=(1,0)$; $\alpha_{t'}=(0, 1, 0)'; \alpha_{t}=(1, 0, -1)'$; $c=(0, 0, 0, 0, 1, 0, 0, 0, 0, 1, 0, 0)'$
\end{enumerate}
\subsubsection*{SM.3.3.241}
$$ A= %
 $$
\begin{enumerate}[i)]
\item $(t',t)=(1,0)$; $\alpha_{t'}=(-1, -1, 1)'; \alpha_{t}=(1, 0, -1)'$; $c=(0, 0, 1, 0, 0, 1, 0, 0, 0)'$
\end{enumerate}
\subsubsection*{SM.3.3.242}
$$ A= %
 $$
\begin{enumerate}[i)]
\item $(t',t)=(1,0)$; $\alpha_{t'}=(0, -1, 1)'; \alpha_{t}=(1, 0, -1)'$; $c=(0, 0, 1, 0, 0, 1, 0, 0, 0, 0)'$
\end{enumerate}
\subsubsection*{SM.3.3.243}
$$ A= %
 $$
\begin{enumerate}[i)]
\item $(t',t)=(1,0)$; $\alpha_{t'}=(-1, 0, 1)'; \alpha_{t}=(1, 0, -1)'$; $c=(0, 0, 0, 1, 0, 1, 0, 1, 0, 0, 0, 0)'$
\end{enumerate}
\subsubsection*{SM.3.3.244}
$$ A= %
 $$
\begin{enumerate}[i)]
\item $(t',t)=(1,0)$; $\alpha_{t'}=(0, 0, 1)'; \alpha_{t}=(1, 0, -1)'$; $c=(0, 0, 0, 1, 1, 1, 0, 0, 0, 0, 0, 0)'$
\end{enumerate}
\subsubsection*{SM.3.3.245}
$$ A= %
 $$
\begin{enumerate}[i)]
\item $(t',t)=(1,0)$; $\alpha_{t'}=(-1, 1, 1)'; \alpha_{t}=(1, 0, -1)'$; $c=(0, 0, 1, 0, 1, 0, 0, 1, 0, 0)'$
\end{enumerate}
\subsubsection*{SM.3.3.246}
$$ A= %
 $$
\begin{enumerate}[i)]
\item $(t',t)=(1,0)$; $\alpha_{t'}=(0, 1, 1)'; \alpha_{t}=(1, 0, -1)'$; $c=(0, 0, 1, 1, 0, 0, 1, 0, 0)'$
\end{enumerate}
\subsubsection*{SM.3.3.247}
$$ A= %
 $$
\begin{enumerate}[i)]
\item $(t',t)=(1,0)$; $\alpha_{t'}=(-1, 1, 0)'; \alpha_{t}=(1, -1, -1)'$; $c=(0, 0, 1, 0, 0, 0, 1, 0, 0)'$
\end{enumerate}
\subsubsection*{SM.3.3.248}
$$ A= %
 $$
\begin{enumerate}[i)]
\item $(t',t)=(1,0)$; $\alpha_{t'}=(0, 1, 0)'; \alpha_{t}=(1, -1, -1)'$; $c=(0, 0, 1, 0, 0, 0, 1, 0, 0)'$
\end{enumerate}
\subsubsection*{SM.3.3.249}
$$ A= %
 $$
\begin{enumerate}[i)]
\item $(t',t)=(1,0)$; $\alpha_{t'}=(-1, 0, 1)'; \alpha_{t}=(1, -1, -1)'$; $c=(0, 0, 1, 0, 1, 0, 0, 0, 0)'$
\end{enumerate}
\subsubsection*{SM.3.3.250}
$$ A= %
 $$
\begin{enumerate}[i)]
\item $(t',t)=(1,0)$; $\alpha_{t'}=(0, 0, 1)'; \alpha_{t}=(1, -1, -1)'$; $c=(0, 0, 1, 1, 0, 0, 0, 0, 0)'$
\end{enumerate}
\subsubsection*{SM.3.3.251}
$$ A= %
 $$
\begin{enumerate}[i)]
\item $(t',t)=(1,0)$; $\alpha_{t'}=(0, 0, 1)'; \alpha_{t}=(1, 1, -2)'$; $c=(0, 1, 1, 1, 1, 0, 0, 0, 0)'$
\end{enumerate}

	\end{appendices}
	
\end{document}